\documentclass[11pt]{article}

\usepackage[T1]{fontenc}
\usepackage[utf8]{inputenc}
\usepackage{lmodern}
\usepackage{microtype}

\usepackage{amsmath,amssymb,amsfonts,amsthm,amscd,mathtools,mathrsfs}

\usepackage[margin=0.87in]{geometry}
\usepackage{titling}
\usepackage{enumitem}
\usepackage{array}
\usepackage{titlesec}

\setlist{topsep=0.35em,itemsep=0.12em,parsep=0pt,partopsep=0pt}
\titlespacing*{\section}{0pt}{2.0ex plus .5ex minus .2ex}{0.8ex}
\titlespacing*{\subsection}{0pt}{1.45ex plus .35ex minus .15ex}{0.45ex}
\titleformat{\subsection}{\normalfont\large\itshape}{\thesubsection}{0.7em}{}
\usepackage{xcolor}
\usepackage[most]{tcolorbox}
\tcbset{
  abstractbox/.style={
    enhanced,
    breakable,
    colback=gray!10,
    colframe=gray!10,
    boxrule=0pt,
    arc=0pt,
    left=3mm,
    right=3mm,
    top=3mm,
    bottom=3mm,
    before skip=0pt,
    after skip=0.8em,
    before upper={\setlength{\parindent}{0pt}}
  }
}

\usepackage[authoryear]{natbib}
\bibpunct{(}{)}{,}{a}{}{,}

\usepackage{hyperref}
\usepackage{cleveref}
\hypersetup{
  colorlinks=true,
  linkcolor=magenta,
  citecolor=teal,
  urlcolor=magenta,
  pdfauthor={Abhiram Sripat},
  pdftitle={Quantum Mechanics on Non-Hausdorff One-Manifolds with Finite Graph Resolutions: Analytic Completion, Branching, and Invariant-Sector Scattering},
  pdfsubject={Closed quantum domains and scattering on finite graph-resolved non-Hausdorff one-manifolds},
  pdfkeywords={non-Hausdorff manifolds, quantum mechanics, vector bundles, section extension, transported jets, Sobolev closure, quantum graphs, branching, holonomy, compact groups, invariant sectors, scattering}
}

\pretitle{\begin{flushleft}\LARGE\bfseries}
\posttitle{\par\end{flushleft}\vspace{0.5em}}
\preauthor{\begin{flushleft}\large}
\postauthor{\par\end{flushleft}}
\predate{\begin{flushleft}}
\postdate{\par\end{flushleft}}

\theoremstyle{plain}
\newtheorem{theorem}{Theorem}[section]
\newtheorem{proposition}[theorem]{Proposition}
\newtheorem{corollary}[theorem]{Corollary}
\newtheorem{lemma}[theorem]{Lemma}
\newtheorem{principle}[theorem]{Principle}

\theoremstyle{definition}
\newtheorem{definition}[theorem]{Definition}
\newtheorem{hypothesis}[theorem]{Hypothesis}

\newtheorem{remark}[theorem]{Remark}

\newcommand{\R}{\mathbb{R}}
\newcommand{\C}{\mathbb{C}}

\newcommand{\Ltwo}{\mathbb{L}_{2}}

\title{Quantum Mechanics on Non-Hausdorff One-Manifolds\\
with Finite Graph Resolutions:\\
Analytic Completion, Branching,\\
and Invariant-Sector Scattering}
\author{Abhiram Sripat\\
\textit{Florence Quantum Labs}\\
\href{mailto:abhiram@florencequantumlabs.com}{\textit{abhiram@florencequantumlabs.com}}}
\date{}

\begin{document}
\pagestyle{plain}
\maketitle
\thispagestyle{empty}

\begin{tcolorbox}[abstractbox]
\textbf{Abstract.}
We develop quantum mechanics on a finite graph-resolved class of second-countable, \(T_1\), non-Hausdorff one-manifolds by separating local differential expressions from the global extension and completion data that determine closed operator domains. For binary adjunctions, transported-jet equalisers give the formal compatibility conditions, while Whitney realisability supplies the residual finite-order obstruction; exponentially tangent branches show that formal compatibility alone can fail. Under explicit trace hypotheses, Sobolev closures are determined by the realisable incidence traces in the noncritical range \(j<m-c/p\). For finite smooth graph resolutions this yields exact integer-order closures, formal endpoint-jet equalisers, signed return holonomy, and finite prolongation profiles.

At first order, scalar completion gives the continuous \(H^1\)-domain on the Hausdorff propagation graph. A marked resolving-presentation measure pushes forward to edge weights and the flat kinetic form gives the weighted Kirchhoff Hamiltonian. Multiple-origin lines and circles, branching junctions, finite trees, and split-and-rejoin networks provide explicit deficiency indices, bright--dark decompositions, transmission criteria, and compact dark states. For a rank-\(r\) Hermitian bundle on the line with \(k\) origins, relative transitions \(W_2,\ldots,W_k\in U(r)\) select the transmitting space \(V_{\mathbf W}=\bigcap_{a=2}^k\ker(W_a-I)\). Compact-group representations identify this with the invariant sector of the generated subgroup; connected compact semisimple groups require at most two relative gluing matrices for full invariant completion, and low-degree \(SU(3)\) tensors give explicit optimal ranks and a quartic rank-two invariant projector. The resulting hierarchy distinguishes weighted-graph data, bundle-selected trace domains, cycle holonomy, and higher transported jets.
\end{tcolorbox}


\vspace{1em}

\tableofcontents
\clearpage
\setlength{\emergencystretch}{3em}

\part{Analytic foundations and completed domains}
\label{part:analytic-foundations}

\section{Introduction}
\label{sec:introduction}

By a non-Hausdorff one-manifold we mean a second-countable, \(T_1\), locally Euclidean smooth space without the Hausdorff axiom. Such spaces occur in adjunction and leaf-space constructions \citep{HaefligerReeb1957,OConnell2023Adjunction,OConnell2024Bundles,LysynskyiMaksymenko2026Line,LysynskyiMaksymenko2026Union,VlasenkoMaksymenko2026} and have also appeared in mathematical-physics models \citep{HellerPysiakSasin11}.

We study finite smooth graph resolutions \(\pi:M\to\Gamma\) with finite metric Hausdorff reflection \(\Gamma\). Metric-measure models use a marked Hausdorff presentation \(q:\widetilde M\to M\), an edge metric, and a positive presentation measure \(\widetilde\mu\), inducing \(\mu_\Gamma=(\pi\circ q)_*\widetilde\mu\). The multiplicities entering the pushed-forward measure belong to this marked presentation; bundle transport, connections, and quadratic forms are additional quantum data. The class includes multiple-origin lines and circles, finite split sets, branching trees, and split-and-rejoin systems.

On the line with two origins, Hausdorff-valued scalars factor through \(\mathbb R\), whereas a relative-sign rank-one bundle forces smooth component representatives to be flat at the origin. Its quantum domain is therefore determined only after this extension condition is completed in the operator-relevant topology.

For binary adjunctions, component restriction becomes a Hausdorff prolongation problem. Transported jets give necessary equaliser conditions, but exponentially tangent branches show that full Whitney realisability is an independent obstruction; finite clean conically tame systems remove it. Under explicit trace and density hypotheses, \(W^{m,p}\)-closure retains precisely the realisable incidence traces in the noncritical range \(j<m-c/p\). Finite graph resolutions sharpen this to exact integer-order endpoint closures, signed formal holonomy, and prolongation profiles.

At first order the scalar domain is \(H^1_{\mathrm{cont}}(\Gamma)\); the pushed-forward marked measure supplies edge weights, and the flat derivative form yields the weighted Kirchhoff Hamiltonian. The oriented derivative has its own boundary form and deficiency theory. Branching, trees, and split-and-rejoin geometries then give explicit channel decompositions, transmission criteria, resonant transparency, and compact dark states.

For rank-\(r\) unitary gluing on the line with \(k\) origins, the relative transitions select \(V_{\mathbf W}=\bigcap_{a=2}^k\ker(W_a-I)\): the flat Friedrichs operator is free on \(V_{\mathbf W}\), Dirichlet on its orthogonal complement, and has transmission projector \(P_{V_{\mathbf W}}\). If \(W_a=\rho(g_a)\), this becomes the invariant sector of the generated compact subgroup. Connected compact semisimple groups require at most two relative gluing matrices for invariant completion, with explicit low-degree \(SU(3)\) ranks and a quartic rank-two projector.

The analysis uses standard Whitney extension, Sobolev trace, quadratic-form, self-adjoint extension, quantum-graph, and compact-group theory \citep{Whitney1934Extension,Fefferman2006Whitney,AdamsFournier2003,Kato1995,KostrykinSchrader99,BerkolaikoKuchment2013,Hall2015,FultonHarris1991}. The structural contribution is the derivation, within the stated non-Hausdorff class, of a single chain from incidence and bundle transport to extension images, closed trace domains, resolution-induced weights, and invariant-sector scattering projectors.

Throughout, \emph{resolution-induced} means determined by the marked analytic resolution together with its graph metric, presentation measure, stated bundle transport or flat connection, and unperturbed quadratic form. Topology alone does not determine these choices; other connections, potentials, measures, defect interactions, or self-adjoint extensions may retain additional information.

Parts~\ref{part:analytic-foundations}--\ref{part:detectability} develop the analytic foundations, graph-resolved dynamics, unitary gluing, and detectability hierarchy; the appendices contain the cutoff and Sobolev perturbation estimates.

\section{The line with two origins: the motivating quantum problem}
\label{sec:L2}

The line with two origins provides the basic contrast: its scalar sector is elementary, whereas the relative-sign rank-one sector has a nontrivial extension defect. Section~\ref{sec:L2-return} returns to the corresponding closed dynamics after the Sobolev analysis.

\subsection{Geometry and scalar calibration}
\label{subsec:L2-geometry}

Let
\[
  \Ltwo
  :=
  (\mathbb{R}_1\sqcup\mathbb{R}_2)/\!\sim,
  \qquad
  (x,1)\sim(x,2)
  \quad
  \text{for }x\neq0.
\]
Write
\(o_1=[0]_1, \qquad o_2=[0]_2\).
The Hausdorff reflection is
\(q:\Ltwo\longrightarrow\mathbb{R}, \qquad q([x])=x, \qquad q(o_1)=q(o_2)=0\).

\begin{proposition}[Hausdorff factorisation]
\label{prop:L2-hausdorff-factorisation}
Every continuous map from \(\Ltwo\) to a Hausdorff space identifies the two
origins and factors uniquely through \(q\).
\end{proposition}

\begin{proof}
If the two origins had distinct images, disjoint neighbourhoods of those
images would pull back to disjoint neighbourhoods of \(o_1\) and \(o_2\),
contradicting their inseparability. Since \(q\) is the quotient collapsing
only the exceptional fibre, the induced factor map is unique and
continuous.
\end{proof}

For the identity-glued smooth structure,
\(C^\infty(\Ltwo,\mathbb C) \cong C^\infty(\mathbb R,\mathbb C)\).
For the scalar calibration, equip the Hausdorff reflection with Lebesgue
measure. Equivalently, after an irrelevant overall normalization, one may use
any marked resolving presentation whose pushed-forward measure is a positive
constant multiple of \(dx\). If \(f=\widehat f\circ q\), define its scalar
\(L^2\)-norm by
\(\|f\|_{L^2(\Ltwo)}^2 := \int_{\mathbb R}|\widehat f(x)|^2\,dx\).
\(q^*:L^2(\mathbb R)\to L^2(\Ltwo)\) is unitary for this calibration.
We take the free scalar kinetic form to be
\[
  \mathfrak q_0[u]
  :=
  \int_{\mathbb R}|u'(x)|^2\,dx,
  \qquad
  D(\mathfrak q_0)=H^1(\mathbb R).
\]
Its associated self-adjoint operator is the free Laplacian
\(-d^2/dx^2\) on \(H^2(\mathbb R)\).

\begin{theorem}[Scalar reduction]
\label{thm:L2-scalar-blindness}
The doubled origin creates neither an additional scalar \(L^2\)-degree of
freedom nor a scalar point interaction for the resolution-induced free
Hamiltonian.
\end{theorem}

\begin{proof}
Smooth compatibility in the identity-glued atlas identifies the two local
representatives with one function in \(C^\infty(\mathbb R)\). The two origins
have zero mass, so pullback by \(q\) is unitary from \(L^2(\mathbb R)\) onto
the scalar \(L^2\)-space. Under this unitary the flat kinetic form is
\(\int_{\mathbb R}|u'|^2\,dx\) on \(H^1(\mathbb R)\); its associated
self-adjoint operator is the free Laplacian on \(H^2(\mathbb R)\).
\end{proof}

\subsection{The nontrivial rank-one sector}
\label{subsec:L2-relative-sign}

For the standard orientation and flat metric,
\(\operatorname{Spin}(1)\cong\mathbb Z_2\).
We use the usual cocycle definition of a spin structure
\citep{LawsonMichelsohn89}: a locally trivial principal
\(\operatorname{Spin}(1)\)-bundle together with an equivariant fibrewise
double covering of the oriented orthonormal frame bundle. All bundles and
bundle maps are taken over the stated non-Hausdorff atlas. Since
\(SO(1)=\{1\}\), this is equivalently a lift of the trivial oriented-frame
cocycle by locally constant \(\{\pm1\}\)-valued transitions. This direct
cocycle formulation avoids any appeal to classification theorems requiring a
Hausdorff or paracompact base.

The overlap of the two standard charts has two connected components, so a
spin transition is determined by two signs
\(\varepsilon_-, \qquad \varepsilon_+\).
After a change of local frame, only the relative sign \(\eta=\varepsilon_+\varepsilon_-\) remains.

\begin{theorem}[Spin structures on \(\Ltwo\)]
\label{thm:L2-spin-classification}
There are exactly two spin-structure isomorphism classes: the trivial class
and the relative-sign class represented by
\[
  g(x)
  =
  \begin{cases}
    +1,&x<0,\\
    -1,&x>0.
  \end{cases}
\]
\end{theorem}

\begin{proof}
The transition function is locally constant because
\(\operatorname{Spin}(1)\) is discrete, hence is determined by
\((\varepsilon_-,\varepsilon_+)\in\{\pm1\}^2\). A change of spin frame on
the two connected charts multiplies both signs by the same relative gauge
factor. It can normalize \(\varepsilon_-=1\), leaving
\(\eta=\varepsilon_+\varepsilon_-\in\{\pm1\}\). The value of \(\eta\) is
unchanged by every further gauge transformation, and the two displayed
normal forms are therefore inequivalent.
\end{proof}

For the relative-sign class, smooth compatibility through the two origins
forces every one-sided derivative to vanish at the exceptional coordinate.
Equivalently, restriction to either Hausdorff chart identifies the smooth
global sections with
\[
  \mathcal F^\infty_0(\mathbb R)
  :=
  \left\{
    f\in C^\infty(\mathbb R):
    f^{(j)}(0)=0
    \text{ for every }j\ge0
  \right\}.
\]
The finite-\(C^k\) restriction image is identified in
Theorem~\ref{thm:sign-restriction-images}. Although the relative-sign bundle
is smooth on the non-Hausdorff atlas, its representative on either Hausdorff
chart must prolong through the missing origin with the prescribed regularity.
The resulting component-extension problem is formulated next for binary
adjunctions and finite branch systems; its Sobolev closure will determine the
closed quantum domain used in Part~\ref{part:quantum-dynamics}.

\section{Adjunction bundles and failure of unrestricted section extension}
\label{sec:extension-failure}

Throughout this section,
\(r\in\mathbb N_0\cup\{\infty\}\).
Let \(M_1\) and \(M_2\) be smooth Hausdorff manifolds, let \(U_i\subseteq M_i\) be open, and let \(f:U_1\longrightarrow U_2\) be a diffeomorphism. Write \(M=M_1\cup_f M_2\) for the resulting binary adjunction and set
\(S_i:=M_i\setminus U_i\).

Let \(E_i\longrightarrow M_i\) be smooth real or complex vector bundles of equal rank, and let \(\Phi:E_1|_{U_1}\longrightarrow E_2|_{U_2}\) be a vector-bundle isomorphism covering \(f\). Gluing by \(\Phi\)
produces the adjunction bundle
\(E=E_1\cup_\Phi E_2\longrightarrow M\)
\citep{OConnell2024Bundles}.

A \(C^r\)-section of \(E\) is represented uniquely by a compatible
pair \((s_1,s_2) \in \Gamma^r(M_1,E_1) \times \Gamma^r(M_2,E_2)\) satisfying
\(s_2\circ f=\Phi\circ s_1 \qquad \text{on }U_1\).
Equivalently,
\[
  \Gamma^r(M,E)
  \cong
  \left\{
    (s_1,s_2):
    s_i\in\Gamma^r(M_i,E_i),\
    s_2\circ f=\Phi\circ s_1
    \text{ on }U_1
  \right\}.
\]
This compatible-pair description concerns sections already defined on
both components. It does not imply that an arbitrary section on one
component admits a compatible section on the other.

\subsection{Restriction as a prolongation problem}
\label{subsec:exact-prolongation}

Let \(\operatorname{res}_i^r: \Gamma^r(M,E) \longrightarrow \Gamma^r(M_i,E_i)\) denote component restriction. Define \(\rho_2^r: \Gamma^r(M_2,E_2) \longrightarrow \Gamma^r(U_2,E_2|_{U_2})\) by ordinary restriction, and define \(T_\Phi^r: \Gamma^r(M_1,E_1) \longrightarrow \Gamma^r(U_2,E_2|_{U_2})\) by
\(T_\Phi^r(s_1) := \Phi\circ s_1\circ f^{-1}\).

\begin{proposition}[Exact prolongation criterion]
\label{prop:exact-prolongation}
For every
\[
  r\in\mathbb N_0\cup\{\infty\},
\]
\[
  \operatorname{Im}\operatorname{res}_1^r
  =
  \left(T_\Phi^r\right)^{-1}
  \left(
    \operatorname{Im}\rho_2^r
  \right).
\]
Equivalently, a section
\[
  s_1\in\Gamma^r(M_1,E_1)
\]
extends to a \(C^r\)-section of \(E\) if and only if its transported
section
\[
  T_\Phi^r(s_1)
  \in
  \Gamma^r(U_2,E_2|_{U_2})
\]
extends across \(S_2\) to a \(C^r\)-section of \(E_2\) on \(M_2\).
\end{proposition}

\begin{proof}
If a global section is represented by \((s_1,s_2)\), compatibility
gives
\(s_2|_{U_2}=T_\Phi^r(s_1)\),
and therefore
\(T_\Phi^r(s_1)\in\operatorname{Im}\rho_2^r\).

Conversely, suppose that \(s_2\in\Gamma^r(M_2,E_2)\) satisfies
\(s_2|_{U_2}=T_\Phi^r(s_1)\).
Then
\(s_2\circ f=\Phi\circ s_1 \qquad \text{on }U_1\),
so \((s_1,s_2)\) defines a global \(C^r\)-section of \(E\).
\end{proof}

\begin{corollary}[Surjectivity criterion]
\label{cor:restriction-surjectivity}
The map
\[
  \operatorname{res}_1^r
\]
is surjective if and only if
\[
  T_\Phi^r
  \left(
    \Gamma^r(M_1,E_1)
  \right)
  \subseteq
  \operatorname{Im}\rho_2^r.
\]
In particular, surjectivity of \(\rho_2^r\) implies surjectivity of
\(\operatorname{res}_1^r\).
\end{corollary}

Component extension is therefore an ordinary prolongation problem on the opposite Hausdorff component; smoothness on the overlap alone does not imply prolongability across \(S_2\).

\subsection{The relative-sign line bundle}
\label{subsec:relative-sign-bundle}

Let
\[
  M_1=M_2=\mathbb R,
  \qquad
  U_1=U_2=\mathbb R\setminus\{0\},
  \qquad
  f=\operatorname{id}_{\mathbb R\setminus\{0\}}.
\]
The resulting adjunction is the line with two origins,
\(\Ltwo = \mathbb R_1 \cup_{\mathbb R\setminus\{0\}} \mathbb R_2\).

Take the trivial real line bundles \(E_i=\mathbb R_i\times\mathbb R \longrightarrow \mathbb R_i\) and define \(g:\mathbb R\setminus\{0\} \longrightarrow \mathbb R^\times\) by
\[
  g(x)
  =
  \begin{cases}
    1,&x<0,\\
    -1,&x>0.
  \end{cases}
\]
Because the overlap is disconnected, \(g\) is smooth. The map \(\Phi(x,\lambda) = \bigl(x,g(x)\lambda\bigr)\) is therefore a smooth vector-bundle isomorphism over the overlap and
defines a smooth real line bundle
\(E_g := E_1\cup_\Phi E_2 \longrightarrow \Ltwo\).

\begin{proposition}[Compatibility with the stated adjunction framework]
\label{prop:sign-example-hypotheses}
The base adjunction and bundle gluing above satisfy the hypotheses
imposed on the binary adjunction and colimit bundle in
\citep{OConnell2024Bundles}:

\begin{enumerate}[label=\textup{(\roman*)}]
  \item \(M_1\) and \(M_2\) are smooth Hausdorff manifolds;

  \item \(U_1\) and \(U_2\) are proper open submanifolds;

  \item \(f\) is a smooth open embedding and extends to the identity
        diffeomorphism
        \[
          \overline f:
          \overline{U_1}^{\,M_1}
          \longrightarrow
          \overline{U_2}^{\,M_2};
        \]

  \item the closures
        \[
          \overline{U_i}^{\,M_i}
          =
          \mathbb R
        \]
        are smooth closed submanifolds;

  \item \(\Phi\) is a smooth vector-bundle isomorphism on the open
        overlap, and the binary cocycle conditions are automatic.
\end{enumerate}
\end{proposition}

\begin{proof}
All assertions concerning the base spaces and the base gluing map are
immediate. Smoothness of \(\Phi\) follows because \(g\) is locally
constant on each connected component of
\(\mathbb R\setminus\{0\}\). In a binary adjunction there are no
nontrivial triple-overlap conditions.
\end{proof}

The fibre transition does not extend continuously to the closure of
the overlap: its one-sided limits at the origin are \(1\) and \(-1\).
This is not excluded by the stated colimit-bundle hypotheses.

In the chosen trivialisations, a \(C^r\)-section of \(E_g\) is a pair \((u,v) \in C^r(\mathbb R)\times C^r(\mathbb R)\) such that
\(v(x)=g(x)u(x) \qquad (x\neq0)\),
or equivalently,
\(v=u \quad\text{on }(-\infty,0), \qquad v=-u \quad\text{on }(0,\infty)\).

The transported representative of \(u\) on the punctured second
component is
\[
  T_g^r u(x)
  =
  \begin{cases}
    u(x),&x<0,\\
    -u(x),&x>0.
  \end{cases}
\]
By Proposition~\ref{prop:exact-prolongation}, \(u\) lies in the
restriction image precisely when this piecewise-defined function
extends through the origin with the required regularity.

\subsection{Exact restriction images}
\label{subsec:sign-restriction-images}

For
\(k\in\mathbb N_0\),
let \(J_0^k: C^k(\mathbb R) \longrightarrow J_0^k(\mathbb R,\mathbb R)\) be the \(k\)-jet map at the origin.

\begin{theorem}[Restriction images of the relative-sign bundle]
\label{thm:sign-restriction-images}
For every
\[
  k\in\mathbb N_0
\]
and \(i=1,2\),
\[
  \boxed{
  \operatorname{Im}\operatorname{res}_i^k
  =
  \left\{
    u\in C^k(\mathbb R):
    u^{(j)}(0)=0,\
    0\leq j\leq k
  \right\}
  =
  \ker J_0^k.
  }
\]
\[
  C^k(\mathbb R)
  \big/
  \operatorname{Im}\operatorname{res}_i^k
  \cong
  J_0^k(\mathbb R,\mathbb R).
\]

In the smooth category,
\[
  \boxed{
  \operatorname{Im}\operatorname{res}_i^\infty
  =
  \mathcal F_0,
  }
\]
where
\[
  \mathcal F_0
  :=
  \left\{
    u\in C^\infty(\mathbb R):
    u^{(j)}(0)=0
    \text{ for every }j\geq0
  \right\}
\]
is the ideal of smooth functions flat at the origin.
\end{theorem}

\begin{proof}
It suffices to treat the first component. If \(u\) is the restriction of a global section, its companion satisfies \(v=u\) on \(( -\infty,0)\) and \(v=-u\) on \((0,\infty)\). Continuity of \(v^{(j)}\) at \(0\) therefore gives \(u^{(j)}(0)=-u^{(j)}(0)\) for \(0\le j\le k\). Conversely, if these derivatives vanish, the piecewise function \(v=u\) on the left, \(v(0)=0\), \(v=-u\) on the right is \(C^k\), so \((u,v)\) is compatible. The second component follows from \(g^{-1}=g\), and the quotient statement from surjectivity of \(J_0^k\). Applying the finite-order result for every \(k\) gives the smooth flat ideal, while the same piecewise construction shows that every flat \(u\) extends smoothly.
\end{proof}

\begin{corollary}[Failure of unrestricted component extension]
\label{cor:sign-nonsurjectivity}
Neither restriction map
\[
  \operatorname{res}_i^\infty:
  \Gamma^\infty(\Ltwo,E_g)
  \longrightarrow
  C^\infty(\mathbb R)
\]
is surjective. In particular, the constant section \(u=1\) on either
component does not extend to a global smooth section of \(E_g\).
\end{corollary}

\begin{theorem}[Failure of unrestricted component extension under the stated hypotheses]
\label{thm:oconnell-extension-counterexample}
The relative-sign bundle \(E_g\to\Ltwo\) satisfies the binary adjunction and
colimit-bundle hypotheses used in \cite[Theorem~2.5]{OConnell2024Bundles},
but its component restriction maps are not surjective. Hence the stated hypotheses do not imply unrestricted component extension.
\end{theorem}

\begin{proof}
Proposition~\ref{prop:sign-example-hypotheses} verifies the stated base and
bundle-gluing hypotheses. Corollary~\ref{cor:sign-nonsurjectivity} shows that
the constant component section \(u=1\) has no global smooth extension.

The obstruction occurs at the closure of the gluing region. In the published
proof of \cite[Theorem~2.5]{OConnell2024Bundles}, a component section is
restricted to the closure of the overlap and the bundle isomorphism on the
open overlap is then used to interpret the transported data as a section over
the corresponding closed subbundle. The colimit-bundle hypotheses specify
that bundle isomorphism only on the open overlap. In the present example its
fibre multiplier has one-sided limits \(1\) and \(-1\), so no continuous
extension of the bundle transition to the closure exists. The transported
constant section therefore has incompatible one-sided limits and cannot be
prolonged.
\end{proof}

\begin{remark}[Scope of the correction]
\label{rem:oconnell-correction-scope}
The counterexample does not affect the construction of the colimit bundle, nor
the statement that a family of component sections already satisfying the
overlap compatibility equations determines a global section. It concerns the
additional surjectivity assertion that an arbitrary prescribed component
section always admits compatible companions. No claim is made here about later results that use the theorem; the point at issue is only the unrestricted component-surjectivity step. A sufficient replacement hypothesis for the present paper is stated in Theorem~\ref{thm:boundary-prolongable-gluing}.
\end{remark}

The two one-sided \(k\)-jets forced on the second component are
\(J_0^k u \qquad\text{and}\qquad -J_0^k u\).
They are the one-sided jets of a common \(C^k\)-germ exactly when
\(J_0^k u=0\).
The failure is therefore a boundary-prolongation defect rather than a
failure of bundle gluing.

\subsection{A sufficient boundary-prolongation criterion}
\label{subsec:boundary-prolongable-gluing}

Surjectivity is restored if the bundle transition itself prolongs smoothly to the closures.

Set
\(K_i:=\overline{U_i}^{\,M_i}\).

\begin{theorem}[Boundary-prolongable gluing]
\label{thm:boundary-prolongable-gluing}
Assume that:

\begin{enumerate}[label=\textup{(\roman*)}]
  \item \(K_i\) is a closed embedded smooth submanifold with boundary
        of \(M_i\);

  \item \(f\) extends to a diffeomorphism
        \[
          \overline f:
          K_1\longrightarrow K_2;
        \]

  \item \(\Phi\) extends to a smooth vector-bundle isomorphism
        \[
          \overline\Phi:
          E_1|_{K_1}
          \longrightarrow
          E_2|_{K_2}
        \]
        covering \(\overline f\).
\end{enumerate}

Then
\[
  \operatorname{res}_i^\infty:
  \Gamma^\infty(M,E)
  \longrightarrow
  \Gamma^\infty(M_i,E_i)
\]
is surjective for \(i=1,2\).
\end{theorem}

\begin{proof}
Given \(s_1\in\Gamma^\infty(M_1,E_1)\), transport its restriction to \(K_1\) by the prolonged data:
\[
 t(y)=\overline\Phi\bigl(s_1(\overline f^{-1}(y))\bigr),\qquad y\in K_2.
\]
In adapted coordinates for the embedded submanifold-with-boundary \(K_2\), the coefficient functions of \(t\) extend smoothly across the boundary, so the standard extension lemma for sections on a closed subset yields \(s_2\in\Gamma^\infty(M_2,E_2)\) with \(s_2|_{K_2}=t\) \citep{Lee2012SmoothManifolds}. On \(U_2\), this is exactly \(\Phi\circ s_1\circ f^{-1}\); hence \((s_1,s_2)\) is compatible. The other restriction map is symmetric.
\end{proof}

\begin{remark}
\label{rem:boundary-prolongation-sufficient}
The hypotheses of
Theorem~\ref{thm:boundary-prolongable-gluing}
are sufficient, not necessary. Proposition~\ref{prop:exact-prolongation}
remains the exact criterion: a prescribed component section extends
precisely when its transported overlap section admits a prolongation
on the opposite component.
\end{remark}

The finite-regularity analogue holds under the corresponding \(C^k\)-extension hypotheses.

In the relative-sign example, the base gluing map extends smoothly to
the closures, but the fibre transition has incompatible one-sided
limits. \(\overline\Phi\) does not exist even continuously, and
Theorem~\ref{thm:boundary-prolongable-gluing} does not apply.
\section{Transported jets and finite-order extension defects}
\label{sec:transported-jets}

We retain the binary adjunction and bundle data \(M=M_1\cup_f M_2, \qquad E=E_1\cup_\Phi E_2\) from Section~\ref{sec:extension-failure}. Set
\(B_i := \overline{U_i}^{\,M_i}\setminus U_i, \qquad i=1,2\).
Only the incidence of \(U_2\) with \(B_2\) enters the construction
below.

For finite \(k\), the spaces of \(C^k\)-sections are equipped with
their compact-open \(C^k\)-topologies. Spaces of continuous sections
over incidence strata carry the compact-open topology. Unless
explicitly stated otherwise, all quotients in this section are
algebraic quotients; their images need not be closed.

All constructions are relative to a fixed finite branch system and
fixed smooth branchwise prolongations. Their dependence on these data
is controlled by truncation, refinement, and naturality below.

\subsection{Finite branch systems}
\label{subsec:finite-branch-systems}

\begin{definition}[Finite branch system]
\label{def:finite-branch-system}
A finite branch system for \(U_2\) along \(B_2\) consists of:

\begin{enumerate}[label=\textup{(\roman*)}]
  \item a finite index set \(A\);

  \item an open neighbourhood \(N\subseteq M_2\) of \(B_2\);

  \item pairwise disjoint open subsets
        \[
          V_a\subseteq U_2\cap N,
          \qquad
          a\in A,
        \]
        satisfying
        \[
          U_2\cap N
          =
          \bigsqcup_{a\in A}V_a.
        \]
\end{enumerate}

The incidence carrier of branch \(a\) is
\[
  S_a
  :=
  \overline{V_a}^{\,M_2}\cap B_2,
\]
and the set of branches incident at \(y\in B_2\) is
\[
  A_y
  :=
  \{a\in A:y\in S_a\}.
\]
Branches with empty incidence carrier are discarded.
\end{definition}

We assume that \(B_2\) has a finite stratification \(B_2 = \bigsqcup_{\lambda\in\Lambda}S_\lambda\) by locally closed embedded smooth submanifolds satisfying the frontier
condition, such that \(A_y\) is constant on every \(S_\lambda\). Its
common value is denoted by
\(A_\lambda\).
For every \(a\in A\) and \(\lambda\in\Lambda\),
\(a\in A_\lambda \quad\Longleftrightarrow\quad S_\lambda\subseteq S_a\).

\begin{definition}[Branchwise prolongation]
\label{def:branchwise-prolongation}
A smooth prolongation of the gluing data along \(V_a\) consists of:

\begin{enumerate}[label=\textup{(\roman*)}]
  \item an open neighbourhood
        \[
          \Omega_a\subseteq M_2
        \]
        of \(S_a\);

  \item a smooth map
        \[
          g_a:\Omega_a\longrightarrow M_1;
        \]

  \item a smooth vector-bundle isomorphism
        \[
          \Psi_a:
          g_a^*E_1
          \longrightarrow
          E_2|_{\Omega_a};
        \]
\end{enumerate}
such that, on \(V_a\cap\Omega_a\),
\[
  g_a=f^{-1},
\]
and
\[
  \Psi_a(y,\xi)
  =
  \Phi\bigl(g_a(y),\xi\bigr)
\]
under the natural identification
\[
  (g_a^*E_1)_y=E_{1,g_a(y)}.
\]
\end{definition}

The limiting source map is
\(\gamma_a := g_a|_{S_a}: S_a\longrightarrow M_1\).
Only the order-\(k\) germ of \((g_a,\Psi_a)\) along \(S_a\) enters the
order-\(k\) construction.

\begin{remark}
Smooth prolongations are imposed so that the transported jet maps vary
smoothly along every incidence stratum at all finite orders. More
generally, \(C^{k+q}\)-prolongation data give \(C^q\)-dependence of the
order-\(k\) transported maps.
\end{remark}

\subsection{Transported jets}
\label{subsec:transported-jets}

Fix
\(k\in\mathbb N_0\).
Functoriality of finite jet bundles
\citep{Saunders1989JetBundles,KolarMichorSlovak1993}
gives, for every \(a\in A\) and \(y\in S_a\), a linear map \(\Theta_{a,y}^k: J_{\gamma_a(y)}^kE_1 \longrightarrow J_y^kE_2\) defined by
\[
  \Theta_{a,y}^k
  \left(
    j_{\gamma_a(y)}^k s
  \right)
  :=
  j_y^k
  \left(
    \Psi_a(g_a^*s)
  \right).
\]

\begin{proposition}[Well-defined transported jets]
\label{prop:transported-jet-well-defined}
The following statements hold.

\begin{enumerate}[label=\textup{(\roman*)}]
  \item The value
        \[
          \Theta_{a,y}^k
          \left(
            j_{\gamma_a(y)}^k s
          \right)
        \]
        depends only on the source \(k\)-jet
        \(j_{\gamma_a(y)}^k s\).

  \item The family \(y\mapsto\Theta_{a,y}^k\) is continuous on
        \(S_a\).

  \item For every incidence stratum
        \[
          S_\lambda\subseteq S_a,
        \]
        the restriction is a smooth vector-bundle morphism
        \[
          \Theta_{a,\lambda}^k:
          \left(\gamma_a^*J^kE_1\right)|_{S_\lambda}
          \longrightarrow
          J^kE_2|_{S_\lambda}.
        \]

  \item The transported map depends only on the order-\(k\) germ of
        \((g_a,\Psi_a)\) along \(S_a\).
\end{enumerate}
\end{proposition}

\begin{proof}
In local coordinates and frames the prolonged transport has the form \(y\mapsto G_a(y)u(g_a(y))\). Its derivatives through order \(k\) are universal polynomial expressions in the derivatives of \(u\) through order \(k\) at \(g_a(y)\) and in the derivatives of \(g_a,G_a\) through order \(k\). Hence the transported jet depends only on the source \(k\)-jet, varies continuously on \(S_a\), smoothly along each incidence stratum, and is unchanged when the prolongation data have the same order-\(k\) germ along \(S_a\).
\end{proof}

For
\(s_1\in\Gamma^k(M_1,E_1)\),
write
\[
  T_a^k(s_1)(y)
  :=
  \Theta_{a,y}^k
  \left(
    j_{\gamma_a(y)}^k s_1
  \right),
  \qquad
  y\in S_a.
\]
\(T_a^k(s_1) \in \Gamma^0 \left( S_a, J^kE_2|_{S_a} \right)\),
and its restriction to each incidence stratum is smooth whenever
\(s_1\) is smooth.

\subsection{The incidence quotient and compatibility operator}
\label{subsec:incidence-defect}

For a stratum \(S_\lambda\), define the incidence-jet bundle
\(\mathcal J_\lambda^k := \bigoplus_{a\in A_\lambda} J^kE_2|_{S_\lambda}\).
Its diagonal subbundle is
\[
  \operatorname{Diag}_\lambda^k
  :=
  \left\{
    (\xi_a)_{a\in A_\lambda}:
    \xi_a=\xi_b
    \text{ for all }a,b\in A_\lambda
  \right\}.
\]
Set
\(\mathcal Q_\lambda^k := \mathcal J_\lambda^k/ \operatorname{Diag}_\lambda^k\),
and let \(q_\lambda^k: \mathcal J_\lambda^k \longrightarrow \mathcal Q_\lambda^k\) be the quotient morphism.

If
\(\operatorname{rank}E_i=r, \qquad \dim M_2=n\),
then
\(\operatorname{rank}\mathcal Q_\lambda^k = (|A_\lambda|-1) r \binom{n+k}{k}\).

Define
\[
  \mathcal Q_A^k
  :=
  \bigoplus_{\lambda\in\Lambda}
  \Gamma^0
  \left(
    S_\lambda,
    \mathcal Q_\lambda^k
  \right).
\]
For
\(s_1\in\Gamma^k(M_1,E_1)\),
set
\[
  \mathfrak d_{A,\lambda}^k(s_1)
  :=
  q_\lambda^k
  \left(
    \left(
      T_a^k(s_1)|_{S_\lambda}
    \right)_{a\in A_\lambda}
  \right),
\]
and define
\(\mathfrak d_A^k := \bigoplus_{\lambda\in\Lambda} \mathfrak d_{A,\lambda}^k\).
Then \(\mathfrak d_A^k: \Gamma^k(M_1,E_1) \longrightarrow \mathcal Q_A^k\) is continuous and linear.

The operator is associated with the chosen branchwise prolongation
germs. Propositions~\ref{prop:branch-refinement} and
\ref{thm:defect-naturality} below describe the invariance available
under controlled changes of presentation.

\begin{definition}[Finite transported-jet equaliser]
\label{def:formal-equaliser}
The order-\(k\) transported-jet equaliser is
\[
  \mathfrak E_A^k
  :=
  \ker\mathfrak d_A^k.
\]
The corresponding finite compatibility quotient is
\[
  \mathfrak F_A^k
  :=
  \Gamma^k(M_1,E_1)/\mathfrak E_A^k.
\]
The first isomorphism theorem gives
\[
  \mathfrak F_A^k
  \cong
  \operatorname{Im}\mathfrak d_A^k.
\]
\end{definition}

Here ``formal'' records only equality of limiting transported jets, not existence of a companion section on \(M_2\).

\begin{theorem}[Transported-jet necessity]
\label{thm:universal-jet-obstruction}
For every
\[
  k\in\mathbb N_0,
\]
\[
  \operatorname{Im}\operatorname{res}_1^k
  \subseteq
  \mathfrak E_A^k.
\]
\end{theorem}

\begin{proof}
Let a global \(C^k\)-section be represented by
\((s_1,s_2) \in \Gamma^k(M_1,E_1) \times \Gamma^k(M_2,E_2)\).
On \(V_a\cap\Omega_a\),
\(s_2 = \Phi\circ s_1\circ f^{-1} = \Psi_a(g_a^*s_1)\).
Both sides extend as \(C^k\)-sections to \(\Omega_a\). Hence, for
every
\(y\in S_\lambda \quad\text{and}\quad a\in A_\lambda\),
\(T_a^k(s_1)(y) = j_y^k s_2\).
The transported incidence family is diagonal at every point of every
stratum, so
\(\mathfrak d_A^k(s_1)=0\).
\end{proof}

\subsection{Truncation, refinement, and naturality}
\label{subsec:truncation-naturality}

For
\(0\leq\ell\leq k\),
let \(\rho_{k,\ell}^{(i)}: J^kE_i \longrightarrow J^\ell E_i\) denote jet truncation.

\begin{proposition}[Truncation compatibility]
\label{prop:defect-truncation}
For every branch \(a\),
\[
  \rho_{k,\ell}^{(2)}
  \circ
  \Theta_a^k
  =
  \Theta_a^\ell
  \circ
  \gamma_a^*\rho_{k,\ell}^{(1)}.
\]
Truncation induces a continuous linear map
\[
  \rho_{k,\ell,A}:
  \mathcal Q_A^k
  \longrightarrow
  \mathcal Q_A^\ell
\]
such that
\[
  \mathfrak d_A^\ell(s_1)
  =
  \rho_{k,\ell,A}
  \left(
    \mathfrak d_A^k(s_1)
  \right)
\]
for every
\[
  s_1\in\Gamma^k(M_1,E_1).
\]
In particular,
\[
  \mathfrak E_A^k
  \subseteq
  \mathfrak E_A^\ell
  \cap
  \Gamma^k(M_1,E_1).
\]
\end{proposition}

\begin{proof}
The first identity is the naturality of jet truncation under
prolongation. Since truncation maps diagonal incidence families to
diagonal incidence families, it descends to the quotient bundles and
then to their section spaces.
\end{proof}

A refinement of \(A\) consists of a finite branch system \(A'\), a
surjection
\(\varpi:A'\longrightarrow A\),
and branch data satisfying, after possibly shrinking the common
neighbourhood of \(B_2\),

\(V_a = \bigsqcup_{\varpi(a')=a}V_{a'}\),
with each prolongation \((g_{a'},\Psi_{a'})\) equal to the restriction
of \((g_a,\Psi_a)\) along the corresponding subbranch.

\begin{proposition}[Refinement invariance]
\label{prop:branch-refinement}
If \(A'\) is a refinement of \(A\), then
\[
  \mathfrak E_{A'}^k
  =
  \mathfrak E_A^k.
\]
\end{proposition}

\begin{proof}
Subbranches of one original branch have identical transported jets.
Refinement therefore adds only repeated copies of existing equality
conditions. Conversely, every original branch incidence remains
represented by its subbranches. The two systems impose the same
pointwise transported-jet equalities.
\end{proof}

Refinement may change the quotient target by repeating branch coordinates; only the equalisers are asserted equal.

An isomorphism of complete branch-labelled adjunction germs consists of:

\begin{enumerate}[label=\textup{(\roman*)}]
  \item diffeomorphisms
        \(F_i:M_i\longrightarrow M_i'\);

  \item vector-bundle isomorphisms \(\widehat F_i:E_i\longrightarrow E_i'\)         covering \(F_i\);

  \item a bijection of branch and stratum labels;
\end{enumerate}
intertwining the base adjunction maps, bundle gluing maps, incidence
carriers, source maps, and prolongation germs.

The induced map on sections is
\(F_{i,\sharp}^k(s) := \widehat F_i \circ s\circ F_i^{-1}\).
Jet functoriality and the branch-label bijection similarly induce
\[
  \mathcal G_{A,\sharp}^k:
  \mathcal Q_A^k
  \longrightarrow
  \mathcal Q_{A'}^{\prime\,k}.
\]

\begin{theorem}[Naturality]
\label{thm:defect-naturality}
The diagram
\[
\begin{CD}
  \Gamma^k(M_1,E_1)
  @>{\mathfrak d_A^k}>>
  \mathcal Q_A^k
  \\
  @V{F_{1,\sharp}^k}VV
  @VV{\mathcal G_{A,\sharp}^k}V
  \\
  \Gamma^k(M_1',E_1')
  @>{\mathfrak d_{A'}^{\prime\,k}}>>
  \mathcal Q_{A'}^{\prime\,k}
\end{CD}
\]
commutes.
\end{theorem}

\begin{proof}
For every corresponding pair of branches, jet functoriality
intertwines the source and target jet transports. On the incidence
direct sums, the induced map permutes branch coordinates and applies
the target jet isomorphism to every coordinate. It maps each diagonal
subbundle onto the corresponding diagonal subbundle, so it descends to
the quotient bundles and their section spaces.
\end{proof}

\begin{corollary}
The isomorphism \(F_{1,\sharp}^k\) maps
\[
  \mathfrak E_A^k
  \quad\text{onto}\quad
  \mathfrak E_{A'}^{\prime\,k},
\]
and \(\mathcal G_{A,\sharp}^k\) maps
\[
  \operatorname{Im}\mathfrak d_A^k
  \quad\text{onto}\quad
  \operatorname{Im}\mathfrak d_{A'}^{\prime\,k}.
\]
Hence the finite compatibility quotient is invariant under isomorphism
of the branch-labelled adjunction germ.
\end{corollary}

Changes of coordinates and bundle frames are special cases.

\subsection{Fibrewise equaliser and compatibility bundles}
\label{subsec:equaliser-bundles}

On a positive-dimensional stratum, several branches may have the same
limiting source point. To avoid treating repeated copies of one source
jet as independent variables, assume that the equivalence relation \(a\sim_y b \quad\Longleftrightarrow\quad \gamma_a(y)=\gamma_b(y)\) is constant on \(S_\lambda\). After further stratification, this may be
imposed whenever the source-coincidence pattern is locally constant.

Let \(\mathcal P_\lambda = \{B_1,\ldots,B_{m_\lambda}\}\) be the resulting partition of \(A_\lambda\). For every block \(B\),
the maps \(\gamma_a\), \(a\in B\), agree on \(S_\lambda\); denote their
common value by
\(\gamma_B:S_\lambda\longrightarrow M_1\).
For distinct blocks \(B\neq B'\),
\(\gamma_B(y)\neq\gamma_{B'}(y) \qquad (y\in S_\lambda)\).

Define the source-jet bundle
\(\mathcal V_\lambda^k := \bigoplus_{B\in\mathcal P_\lambda} \gamma_B^*J^kE_1\).
For \(a\in A_\lambda\), let \(B(a)\) be its source block. Repetition
of the appropriate source jet followed by branch transport defines
\[
  \widetilde\Delta_\lambda^k:
  \mathcal V_\lambda^k
  \longrightarrow
  \mathcal J_\lambda^k
\]
by
\[
  \widetilde\Delta_{\lambda,y}^k
  \left(
    (\eta_B)_{B\in\mathcal P_\lambda}
  \right)
  =
  \left(
    \Theta_{a,y}^k
    \left(
      \eta_{B(a)}
    \right)
  \right)_{a\in A_\lambda}.
\]
Set
\[
  \Delta_\lambda^k
  :=
  q_\lambda^k
  \circ
  \widetilde\Delta_\lambda^k:
  \mathcal V_\lambda^k
  \longrightarrow
  \mathcal Q_\lambda^k.
\]

\begin{definition}
The fibrewise equaliser and compatibility-image spaces are
\[
  \mathcal E_{\lambda,y}^k
  :=
  \ker\Delta_{\lambda,y}^k,
  \qquad
  \mathcal D_{\lambda,y}^k
  :=
  \operatorname{Im}\Delta_{\lambda,y}^k.
\]
\end{definition}

\begin{theorem}[Constant-rank equaliser theorem]
\label{thm:constant-rank-equaliser}
If \(\Delta_\lambda^k\) has locally constant rank, then
\[
  \mathcal E_\lambda^k
  :=
  \ker\Delta_\lambda^k
\]
and
\[
  \mathcal D_\lambda^k
  :=
  \operatorname{Im}\Delta_\lambda^k
\]
are smooth vector subbundles, and the induced morphism gives a
canonical bundle isomorphism
\[
  \mathcal V_\lambda^k/
  \mathcal E_\lambda^k
  \cong
  \mathcal D_\lambda^k.
\]
\end{theorem}

\begin{proof}
Apply the constant-rank theorem to the smooth vector-bundle morphism
\(\Delta_\lambda^k\)
\citep{Lee2012SmoothManifolds}.
\end{proof}

This statement is fibrewise. A smooth section of
\(\mathcal E_\lambda^k\) need not be holonomic and need not arise as
the restriction of the jet of a global section of \(E_1\). Likewise, \(\operatorname{Im}\mathfrak d_A^k\) may be strictly smaller than the full section space
\(\Gamma^0(S_\lambda,\mathcal D_\lambda^k)\).

\subsection{Isolated incidence}
\label{subsec:isolated-incidence}

Let \(x\in B_2\) be an isolated incidence point. Define \(P_x := \{\gamma_a(x):a\in A_x\} \subseteq M_1\) and
\(\mathcal V_x^k := \bigoplus_{p\in P_x}J_p^kE_1\).
Let \(\mathcal Q_x^k\) be the fibre at \(x\) of the incidence quotient belonging to the
zero-dimensional stratum containing \(x\). Define \(D_x^k: \mathcal V_x^k \longrightarrow \mathcal Q_x^k\) by
\[
  D_x^k
  \left(
    (\eta_p)_{p\in P_x}
  \right)
  :=
  \left[
    \left(
      \Theta_{a,x}^k
      \left(
        \eta_{\gamma_a(x)}
      \right)
    \right)_{a\in A_x}
  \right].
\]

Let \(\mathfrak d_{A,x}^k: \Gamma^k(M_1,E_1) \longrightarrow \mathcal Q_x^k\) denote evaluation at \(x\) of the corresponding component of
\(\mathfrak d_A^k\).

\begin{proposition}[Isolated-point reduction]
\label{prop:isolated-point-reduction}
At an isolated incidence point,
\[
  \operatorname{Im}\mathfrak d_{A,x}^k
  =
  \operatorname{Im}D_x^k.
\]
The local finite compatibility problem is represented by the
finite-dimensional linear map \(D_x^k\).
\end{proposition}

\begin{proof}
Evaluation of \(k\)-jets at the finite set \(P_x\) is surjective.
Indeed, choose pairwise disjoint coordinate neighbourhoods around the
points of \(P_x\), realise the prescribed finite jets in local bundle
frames by polynomial sections, and multiply by compactly supported
cutoffs equal to one near the relevant points. Summing the resulting
sections realises the prescribed jet family.

The map \(\mathfrak d_{A,x}^k\) is the composition of this surjective
evaluation map with \(D_x^k\), which proves the equality of images.
\end{proof}

Suppose now that every branch has the same source point
\(\gamma_a(x)=p \qquad (a\in A_x)\),
and that every transported map \(\Theta_{a,x}^k: J_p^kE_1 \longrightarrow J_x^kE_2\) is invertible. Fix a reference branch \(a_0\) and define
\[
  H_a^k
  :=
  \left(
    \Theta_{a_0,x}^k
  \right)^{-1}
  \Theta_{a,x}^k
  \in
  \operatorname{GL}(J_p^kE_1).
\]

\begin{proposition}[Common-source fixed-space formula]
\label{prop:common-source-fixed-space}
The local equaliser is
\[
  \ker D_x^k
  =
  \bigcap_{a\in A_x}
  \operatorname{Fix}(H_a^k).
\]
Moreover,
\[
  \operatorname{Im}D_x^k
  \cong
  J_p^kE_1
  \bigg/
  \bigcap_{a\in A_x}
  \operatorname{Fix}(H_a^k).
\]
\end{proposition}

\begin{proof}
A source jet \(\eta\in J_p^kE_1\) belongs to the equaliser precisely
when
\(\Theta_{a,x}^k\eta = \Theta_{a_0,x}^k\eta \qquad (a\in A_x)\).
Applying \(\left( \Theta_{a_0,x}^k \right)^{-1}\) gives
\(H_a^k\eta=\eta \qquad (a\in A_x)\).
The quotient formula follows from the first isomorphism theorem.
\end{proof}

When isolated incidence points have pairwise disjoint source sets, the
corresponding local maps depend on independent finite jet variables,
and their images form a direct sum. Shared source points may couple the
local compatibility conditions.

\subsection{First obstruction order}
\label{subsec:first-obstruction-order}

For
\(s_1\in\Gamma^\infty(M_1,E_1)\),
define its first transported-jet obstruction order by
\[
  \operatorname{ord}_A(s_1)
  :=
  \min
  \left\{
    k\geq0:
    \mathfrak d_A^k(s_1)\neq0
  \right\},
\]
with \(\operatorname{ord}_A(s_1)=\infty\) when
\(\mathfrak d_A^k(s_1)=0 \qquad \text{for every }k\).

For \(k\geq1\), define the new order-\(k\) compatibility quotient by
\[
  \mathfrak N_A^k
  :=
  \frac{
    \left\{
      s\in\Gamma^k(M_1,E_1):
      \mathfrak d_A^{k-1}(s)=0
    \right\}
  }{
    \mathfrak E_A^k
  }.
\]

\begin{proposition}[New compatibility conditions]
\label{prop:new-order-defect}
There is a canonical linear isomorphism
\[
  \mathfrak N_A^k
  \cong
  \operatorname{Im}\mathfrak d_A^k
  \cap
  \ker\rho_{k,k-1,A}.
\]
\end{proposition}

\begin{proof}
Restrict \(\mathfrak d_A^k\) to
\[
  \left\{
    s\in\Gamma^k(M_1,E_1):
    \mathfrak d_A^{k-1}(s)=0
  \right\}.
\]
By truncation compatibility, its image lies in
\(\ker\rho_{k,k-1,A}\).
Conversely, if
\(\eta\in \operatorname{Im}\mathfrak d_A^k \cap \ker\rho_{k,k-1,A}\),
write
\(\eta=\mathfrak d_A^k(s)\).
Then
\(\mathfrak d_A^{k-1}(s) = \rho_{k,k-1,A}(\eta) = 0\).
The kernel of the restricted map is exactly
\(\mathfrak E_A^k\).
The first isomorphism theorem proves the result.
\end{proof}

\subsection{Smooth necessity and the residual obstruction}
\label{subsec:formal-actual-defects}

Because the fixed smooth branchwise prolongations define a compatible
transported-jet system at every finite order, set
\[
  \mathfrak E_A^\infty
  :=
  \left\{
    s\in\Gamma^\infty(M_1,E_1):
    \mathfrak d_A^k(s)=0
    \text{ for every }k\geq0
  \right\}.
\]
Equivalently,
\[
  \mathfrak E_A^\infty
  =
  \Gamma^\infty(M_1,E_1)
  \cap
  \bigcap_{k\geq0}\mathfrak E_A^k.
\]

\begin{corollary}[Smooth transported-jet necessity]
\label{cor:smooth-universal-necessity}
\[
  \operatorname{Im}\operatorname{res}_1^\infty
  \subseteq
  \mathfrak E_A^\infty.
\]
\end{corollary}

\begin{proof}
Apply
Theorem~\ref{thm:universal-jet-obstruction}
at every finite order.
\end{proof}

No converse is asserted here.

Set
\(\mathfrak R_A^k := \operatorname{Im}\operatorname{res}_1^k\).
The algebraic extension cokernel is
\(\mathfrak C_A^k := \Gamma^k(M_1,E_1)/\mathfrak R_A^k\),
while the finite transported-compatibility quotient is
\[
  \mathfrak F_A^k
  :=
  \Gamma^k(M_1,E_1)/\mathfrak E_A^k
  \cong
  \operatorname{Im}\mathfrak d_A^k.
\]
Their difference is the jet-invisible residual obstruction
\(\mathfrak W_A^k := \mathfrak E_A^k/\mathfrak R_A^k\).

\begin{proposition}[Residual-obstruction sequence]
\label{prop:residual-obstruction-sequence}
There is a canonical short exact sequence of vector spaces
\[
  0
  \longrightarrow
  \mathfrak W_A^k
  \longrightarrow
  \mathfrak C_A^k
  \longrightarrow
  \mathfrak F_A^k
  \longrightarrow
  0.
\]
Equivalently,
\[
  0
  \longrightarrow
  \frac{
    \ker\mathfrak d_A^k
  }{
    \operatorname{Im}\operatorname{res}_1^k
  }
  \longrightarrow
  \frac{
    \Gamma^k(M_1,E_1)
  }{
    \operatorname{Im}\operatorname{res}_1^k
  }
  \longrightarrow
  \operatorname{Im}\mathfrak d_A^k
  \longrightarrow
  0.
\]
\end{proposition}

\begin{proof}
The universal transported-jet obstruction gives
\(\mathfrak R_A^k \subseteq \mathfrak E_A^k \subseteq \Gamma^k(M_1,E_1)\).
The result is the canonical short exact sequence associated with these
nested subspaces.
\end{proof}

\(\mathfrak W_A^k\) is the jet-invisible extension defect, identified below with failure of Whitney compatibility.
\section{Whitney realisation and conically tame extension}
\label{sec:whitney-realisation}

The equaliser \(\mathfrak E_A^k\) records limiting-jet compatibility; this section identifies the remaining Whitney obstruction and a sufficient condition for its vanishing.

Throughout,
\(K_2 := \overline{U_2}^{\,M_2}\).
Since \(U_2\) is open, \(K_2\) is a closed subset of \(M_2\).

\subsection{Bundle-valued Whitney fields}
\label{subsec:bundle-valued-whitney}

Let \(F\longrightarrow X\) be a smooth real or complex vector bundle over a smooth paracompact
manifold, and let \(K\subseteq X\) be closed.

A section \(\mathcal F:K\longrightarrow J^k F|_K\) is called a \emph{Whitney \(C^k\)-field} if, in every coordinate chart
and local bundle frame, its coefficient functions satisfy the
classical Whitney remainder conditions of order \(k\). A formal jet
field \(\mathcal F^\infty\) is called a \emph{smooth Whitney field} if each finite truncation is a
Whitney \(C^k\)-field.

Coordinate and frame changes act by smooth triangular jet automorphisms, so the definitions are intrinsic.

\begin{proposition}[Bundle-valued Whitney realisation]
\label{prop:bundle-valued-whitney-extension}
Let \(K\subseteq X\) be closed.

\begin{enumerate}[label=\textup{(\roman*)}]
  \item Every Whitney \(C^k\)-field
        \[
          \mathcal F:K\longrightarrow J^k F|_K
        \]
        is realised by a section
        \[
          s\in\Gamma^k(X,F)
        \]
        satisfying
        \[
          j^k s|_K=\mathcal F.
        \]

  \item Every smooth Whitney field on \(K\) is realised by a section
        \[
          s\in\Gamma^\infty(X,F)
        \]
        with the prescribed full jet on \(K\).
\end{enumerate}
\end{proposition}

\begin{proof}
Choose a locally finite family of bundle-trivialising coordinate neighbourhoods \(O_i'\Subset O_i\) covering \(K\), and put \(K_i=K\cap\overline{O_i'}\). In each frame, \(\mathcal F|_{K_i}\) is a finite-dimensional vector-valued Whitney field on the closed Euclidean set \(\varphi_i(K_i)\); applying the scalar Whitney theorem componentwise gives \(s_i\in\Gamma^k(O_i,F)\) with \(j^ks_i|_{K_i}=\mathcal F|_{K_i}\) \citep{Whitney1934Extension,Fefferman2006Whitney,FeffermanLuli2014}.

Add \(O_0=X\setminus K\), set \(s_0=0\), and choose a partition of unity \(\{\eta_0,\eta_i\}\) subordinate to \(\{O_0,O_i'\}\). The locally finite sum \(s=\sum_i\eta_i s_i\) is \(C^k\). At any \(x\in K\), all local extensions meeting \(x\) have the same \(k\)-jet \(\mathcal F(x)\), while \(\eta_0\) vanishes near \(x\); the Leibniz rule and \(\sum_i\eta_i=1\) therefore give \(j_x^ks=\mathcal F(x)\). The \(C^\infty\) statement follows from the smooth Whitney theorem and the same patching argument.
\end{proof}

\subsection{Assembled transported fields}
\label{subsec:assembled-fields}

For
\(s_1\in\Gamma^k(M_1,E_1)\),
let \(\sigma[s_1] := \Phi\circ s_1\circ f^{-1} \in \Gamma^k(U_2,E_2|_{U_2})\) be the forced section on the second gluing region. For every branch
\(a\), define its prolonged representative by
\[
  \overline{\sigma}_a[s_1]
  :=
  \Psi_a(g_a^*s_1)
  \in
  \Gamma^k(\Omega_a,E_2|_{\Omega_a}).
\]
On
\(V_a\cap\Omega_a\),
one has
\(\overline{\sigma}_a[s_1] = \sigma[s_1]\).

Suppose that
\(s_1\in\mathfrak E_A^k\).
At every \(y\in B_2\), the transported jets
\(T_a^k(s_1)(y) = j_y^k \overline{\sigma}_a[s_1], \qquad a\in A_y\),
are equal. Denote their common value by
\(\xi_y^k(s_1)\in J_y^k E_2\).

\begin{definition}[Assembled transported field]
\label{def:assembled-transported-field}
For
\[
  s_1\in\mathfrak E_A^k,
\]
define
\[
  \mathcal A_A^k(s_1):
  K_2\longrightarrow J^kE_2|_{K_2}
\]
by
\[
  \mathcal A_A^k(s_1)(y)
  =
  \begin{cases}
    j_y^k \sigma[s_1],
    &y\in U_2,\\[1mm]
    \xi_y^k(s_1),
    &y\in B_2.
  \end{cases}
\]
\end{definition}

The field is relative to the fixed prolongation germs and transforms naturally under the refinements and isomorphisms of Section~\ref{sec:transported-jets}.

\begin{proposition}[Continuity of the assembled field]
\label{prop:assembled-field-continuity}
For every
\[
  s_1\in\mathfrak E_A^k,
\]
the assembled field
\[
  \mathcal A_A^k(s_1)
\]
is a continuous section of \(J^kE_2|_{K_2}\).
\end{proposition}

\begin{proof}
On \(U_2\), \(\mathcal A_A^k(s_1)=j^k\sigma[s_1]\). Let \(y_n\to y\in B_2\). Finiteness of the branch set allows passage to a subsequence with a fixed incident branch \(a\). If \(y_n\in U_2\), then eventually \(y_n\in V_a\) and
\[
 \mathcal A_A^k(s_1)(y_n)=j_{y_n}^k\sigma_a[s_1]\to j_y^k\sigma_a[s_1]=\xi_y^k(s_1).
\]
If \(y_n\in B_2\), choose the fixed branch from \(A_{y_n}\); closedness of \(S_a\) in \(B_2\) gives \(y\in S_a\), and continuity of \(T_a^k(s_1)\) gives the same limit. Every convergent sequence has the required limit.
\end{proof}

\begin{proposition}[Local holonomicity]
\label{prop:assembled-local-holonomicity}
Let
\[
  s_1\in\mathfrak E_A^k.
\]

\begin{enumerate}[label=\textup{(\roman*)}]
  \item On \(U_2\),
        \[
          \mathcal A_A^k(s_1)
          =
          j^k \sigma[s_1].
        \]

  \item Near a point \(y\in B_2\) with
        \[
          |A_y|=1,
        \]
        the assembled field is the jet of the unique prolonged branch
        section.

  \item For every branch \(a\),
        \[
          \mathcal A_A^k(s_1)
          \big|
          _{\overline{V_a}^{\,M_2}\cap K_2\cap\Omega_a}
          =
          j^k \overline{\sigma}_a[s_1]
          \big|
          _{\overline{V_a}^{\,M_2}\cap K_2\cap\Omega_a}.
        \]
\end{enumerate}
\end{proposition}

\begin{proof}
Only the boundary values require verification. If
\(y\in S_a\),
then
\(\xi_y^k(s_1) = T_a^k(s_1)(y) = j_y^k \overline{\sigma}_a[s_1]\).
The remaining statements follow from the definitions.
\end{proof}

Only cross-branch approach to a multiple-incidence locus can therefore violate Whitney compatibility.

Define
\[
  \mathfrak X_A^k
  :=
  \left\{
    s_1\in\mathfrak E_A^k:
    \mathcal A_A^k(s_1)
    \text{ is a Whitney }C^k\text{-field}
  \right\}.
\]

\subsection{Exact Whitney criterion}
\label{subsec:exact-whitney-criterion}

\begin{theorem}[Exact finite-order Whitney realisation]
\label{thm:exact-whitney-realisation}
For every
\[
  k\in\mathbb N_0,
\]
\[
  \boxed{
  \operatorname{Im}\operatorname{res}_1^k
  =
  \mathfrak X_A^k.
  }
\]
Equivalently, a section
\[
  s_1\in\Gamma^k(M_1,E_1)
\]
extends if and only if:

\begin{enumerate}[label=\textup{(\roman*)}]
  \item
        \[
          s_1\in\ker\mathfrak d_A^k;
        \]

  \item the assembled field
        \[
          \mathcal A_A^k(s_1)
        \]
        is a Whitney \(C^k\)-field on \(K_2\).
\end{enumerate}
\end{theorem}

\begin{proof}
Suppose first that \(s_1\) is the restriction of a global section
represented by
\((s_1,s_2) \in \Gamma^k(M_1,E_1) \times \Gamma^k(M_2,E_2)\).
The universal transported-jet obstruction gives
\(s_1\in\mathfrak E_A^k\).
Moreover,
\(j^k s_2|_{K_2} = \mathcal A_A^k(s_1)\).
The assembled field is therefore Whitney \(C^k\).

Conversely, let
\(s_1\in\mathfrak X_A^k\).
By
Proposition~\ref{prop:bundle-valued-whitney-extension},
there exists \(s_2\in\Gamma^k(M_2,E_2)\) such that
\(j^k s_2|_{K_2} = \mathcal A_A^k(s_1)\).
On the open set \(U_2\),
\(j^k s_2 = j^k \sigma[s_1]\).
Equality of the zero-order components gives
\(s_2|_{U_2} = \sigma[s_1] = \Phi\circ s_1\circ f^{-1}\).
Hence \((s_1,s_2)\) is a compatible \(C^k\)-pair.
\end{proof}

\begin{corollary}[Residual Whitney obstruction]
\label{cor:residual-whitney-obstruction}
There is a canonical identification
\[
  \mathfrak W_A^k
  =
  \frac{
    \mathfrak E_A^k
  }{
    \operatorname{Im}\operatorname{res}_1^k
  }
  =
  \frac{
    \mathfrak E_A^k
  }{
    \mathfrak X_A^k
  }.
\]
\(\mathfrak W_A^k\) is exactly the failure of Whitney
realisability after all limiting order-\(k\) transported-jet equations
have been satisfied.
\end{corollary}

The criterion reduces to cross-branch remainders. In a coordinate chart and local frame, write
\(F_\alpha: K_2\longrightarrow\mathbb K^r, \qquad |\alpha|\leq k\),
for the coefficient functions of
\(\mathcal A_A^k(s_1)\), and define
\[
  R_\alpha(x,y)
  :=
  F_\alpha(x)
  -
  \sum_{|\beta|\leq k-|\alpha|}
  \frac{
    F_{\alpha+\beta}(y)
  }{
    \beta!
  }
  (x-y)^\beta.
\]

\begin{proposition}[Cross-branch reduction]
\label{prop:cross-branch-reduction}
Let
\[
  s_1\in\mathfrak E_A^k.
\]
Then \(\mathcal A_A^k(s_1)\) is a Whitney \(C^k\)-field if and only
if, locally near every multiple-incidence stratum,
\[
  R_\alpha(x,y)
  =
  o
  \left(
    |x-y|^{k-|\alpha|}
  \right)
\]
locally uniformly for
\[
  |\alpha|\leq k,
\]
that is, uniformly when \(x\) and \(y\) remain in a fixed compact
subset of a coordinate neighbourhood and approach the stratum through
the closures of two distinct incident branches.
\end{proposition}

\begin{proof}
By
Proposition~\ref{prop:assembled-local-holonomicity},
the Whitney conditions already hold on the closure of each individual
branch and near every single-incidence point.

On a sufficiently small neighbourhood of a multiple-incidence
stratum, the remaining pairs are those belonging to distinct incident
branch closures. Since the branch system is finite, uniformity over
all such pairs is equivalent to the finitely many displayed
cross-branch conditions.
\end{proof}

\subsection{Exact smooth realisation}
\label{subsec:exact-smooth-whitney}

Recall
\[
  \mathfrak E_A^\infty
  =
  \left\{
    s_1\in\Gamma^\infty(M_1,E_1):
    \mathfrak d_A^k(s_1)=0
    \text{ for every }k
  \right\}.
\]
For
\(s_1\in\mathfrak E_A^\infty\),
truncation compatibility gives a formal field \(\mathcal A_A^\infty(s_1)\) whose order-\(k\) truncation is
\(\mathcal A_A^k(s_1)\).

Set
\[
  \mathfrak X_A^\infty
  :=
  \left\{
    s_1\in\mathfrak E_A^\infty:
    \mathcal A_A^\infty(s_1)
    \text{ is a smooth Whitney field}
  \right\}.
\]

\begin{theorem}[Exact smooth Whitney realisation]
\label{thm:exact-smooth-whitney}
\[
  \boxed{
  \operatorname{Im}\operatorname{res}_1^\infty
  =
  \mathfrak X_A^\infty.
  }
\]
\end{theorem}

\begin{proof}
A global smooth companion supplies the assembled formal jet field.

Conversely, let
\(s_1\in\mathfrak X_A^\infty\).
By the smooth part of
Proposition~\ref{prop:bundle-valued-whitney-extension},
there exists \(s_2\in\Gamma^\infty(M_2,E_2)\) whose full jet on \(K_2\) is
\(\mathcal A_A^\infty(s_1)\).
Its zero-order component on \(U_2\) is \(\sigma[s_1]\), so
\(s_2|_{U_2} = \Phi\circ s_1\circ f^{-1}\).
Hence \((s_1,s_2)\) is compatible.
\end{proof}

\subsection{Exponentially tangent branches}
\label{subsec:exponentially-tangent-branches}

Formal compatibility at every finite order need not imply Whitney
realisability.

Define
\[
  \delta(x)
  =
  \begin{cases}
    e^{-1/x^2},&x>0,\\
    0,&x\leq0,
  \end{cases}
  \qquad
  \rho(x)
  =
  \begin{cases}
    e^{-1/(2x^2)},&x>0,\\
    0,&x\leq0.
  \end{cases}
\]
Then
\(\delta=\rho^2\).
Set
\(w(x) := \frac{1}{10}\delta(x)^2\).

In
\(M_2=\mathbb R^2\),
define the open ribbons
\[
  V_0
  :=
  \left\{
    (x,y):
    x>0,\
    |y|<w(x)
  \right\},
\]
\[
  V_1
  :=
  \left\{
    (x,y):
    x>0,\
    |y-\delta(x)|<w(x)
  \right\}.
\]
For \(x>0\),
\(\delta(x)-2w(x) = \delta(x) \left( 1-\frac{\delta(x)}5 \right) >0\),
so \(V_0\) and \(V_1\) are disjoint. Their closures meet one another
only at
\(o=(0,0)\).

Let
\(M_1 = \mathbb R_{(0)}^2 \sqcup \mathbb R_{(1)}^2\).
In the \(i\)-th component, set
\[
  W_i
  :=
  \left\{
    (x,y):
    x>0,\
    |y|<w(x)
  \right\}.
\]
Let
\(U_1=W_0\sqcup W_1, \qquad U_2=V_0\sqcup V_1\),
and define the gluing diffeomorphism branchwise by
\(f_0(x,y)=(x,y), \qquad f_1(x,y)=(x,y+\delta(x))\).
Take trivial real line bundles and identity fibre transport.

The inverse branches prolong smoothly to all of \(M_2\) by \(g_0(x,y)=(x,y)\) with values in \(\mathbb R_{(0)}^2\), and \(g_1(x,y)=(x,y-\delta(x))\) with values in \(\mathbb R_{(1)}^2\).

Define
\(s_1|_{\mathbb R_{(0)}^2}=0, \qquad s_1(x,y)|_{\mathbb R_{(1)}^2} = \rho(x)\).
The forced section on \(U_2\) is
\(\sigma[s_1]|_{V_0}=0, \qquad \sigma[s_1]|_{V_1}=\rho(x)\).

\begin{proposition}[Formal compatibility at every order]
\label{prop:exponential-formal-compatibility}
For every
\[
  k\in\mathbb N_0,
\]
\[
  \mathfrak d_A^k(s_1)=0.
\]
Hence
\[
  s_1\in\mathfrak E_A^\infty.
\]
\end{proposition}

\begin{proof}
The only multiple-incidence point is \(o\). The two prolonged branch
sections at \(o\) are
\(0 \qquad\text{and}\qquad \rho(x)\).
Since \(\rho\) is flat at \(x=0\), both have zero jet of every order
at \(o\).

At all remaining boundary points, only one branch is incident, so no
cross-branch transported equality is imposed.
\end{proof}

For \(x>0\), set
\(p_x=(x,0)\in V_0, \qquad q_x=(x,\delta(x))\in V_1\).
Then
\(|p_x-q_x|=\delta(x)\),
while
\(\sigma[s_1](p_x)=0, \qquad \sigma[s_1](q_x)=\rho(x)\).

\begin{theorem}[Formal compatibility without \(C^1\)-realisation]
\label{thm:exponential-no-realisation}
The section \(s_1\) does not belong to
\[
  \operatorname{Im}\operatorname{res}_1^1.
\]
Equivalently,
\[
  \mathcal A_A^1(s_1)
\]
is not a Whitney \(C^1\)-field.
\end{theorem}

\begin{proof}
For the pair \((p_x,q_x)\), the order-zero Whitney remainder satisfies
\[
  \frac{
    |\sigma[s_1](p_x)-\sigma[s_1](q_x)|
  }{
    |p_x-q_x|
  }
  =
  \frac{\rho(x)}{\delta(x)}
  =
  e^{1/(2x^2)}
  \longrightarrow\infty.
\]
The Whitney \(C^1\)-condition fails.

Equivalently, suppose that a compatible \(s_2\in C^1(\mathbb R^2)\) existed. On a closed ball centred at \(o\), the derivative \(Ds_2\)
would be bounded. The mean-value estimate would therefore give
\(\rho(x) = |s_2(q_x)-s_2(p_x)| \leq C\delta(x)\),
contradicting
\(\rho(x)/\delta(x)\longrightarrow\infty\).
\end{proof}

\begin{corollary}
\label{cor:exponential-residual-nonzero}
\[
  \mathfrak W_A^1\neq0.
\]
\end{corollary}

The obstruction is Whitney-theoretic: the value gap \(\rho(x)=\sqrt{\delta(x)}\) decays too slowly relative to the branch separation \(\delta(x)\), despite agreement of all limiting jets.

\subsection{Conically tame multiple incidence}
\label{subsec:conical-tameness}

Set
\[
  B_2^{\mathrm{mult}}
  :=
  \left\{
    y\in B_2:
    |A_y|\geq2
  \right\}.
\]
Only this locus requires cross-branch control.

\begin{definition}[Clean multiple-incidence stratum]
\label{def:clean-multiple-incidence}
A connected submanifold
\[
  S\subseteq B_2^{\mathrm{mult}}
\]
is a clean multiple-incidence stratum if:

\begin{enumerate}[label=\textup{(\roman*)}]
  \item \(S\) is closed and embedded in \(M_2\);

  \item \(A_y=A_S\) is constant on \(S\);

  \item there exists an open neighbourhood
        \[
          N_S\subseteq N
        \]
        of \(S\) such that
        \[
          B_2^{\mathrm{mult}}\cap N_S=S,
        \]
        and
        \[
          U_2\cap N_S
          =
          \bigsqcup_{a\in A_S}
          \left(
            V_a\cap N_S
          \right);
        \]
        equivalently, no branch not incident to \(S\) enters \(N_S\);

  \item every branch prolongation indexed by \(A_S\) is defined on a
        neighbourhood of \(S\).
\end{enumerate}
\end{definition}

Single-incidence faces carried by branches in \(A_S\) may remain in
\(N_S\). A finite clean
multiple-incidence system is a finite disjoint union \(B_2^{\mathrm{mult}} = \bigsqcup_{\mu=1}^N S^{(\mu)}\) of clean strata.

Fix a clean stratum \(S\), a Riemannian metric on \(M_2\), and a
tubular map
\(\tau: \mathcal N_\varepsilon \subseteq NS \longrightarrow \mathcal U\)
\citep{Lee2012SmoothManifolds}. For
\(x=\tau(v)\notin S\),
write
\(\pi_S(x)=\pi_{NS}(v), \qquad r_S(x)=|v|, \qquad \omega_S(x)=\frac{v}{|v|}\),
where \(\pi_{NS}:NS\longrightarrow S\) is the normal-bundle projection.

For an open set \(\Omega\subseteq\mathbb S(NS)\) and a positive function
\(\eta:S\longrightarrow(0,\infty)\),
set
\[
  \mathcal C(\Omega,\eta)
  :=
  \left\{
    \tau(r\omega):
    \omega\in\Omega,\
    0<r<\eta(\pi_{NS}(\omega))
  \right\}.
\]
Write \(\Omega^-\Subset_S\Omega^+\) when \(\overline{\Omega^-} \subseteq \Omega^+\) in \(\mathbb S(NS)\).

\begin{definition}[Conical tameness]
\label{def:conical-tameness}
The clean stratum \(S\) is conically tame if there exist angular
regions
\[
  \Omega_a^-\Subset_S\Omega_a^+,
  \qquad
  a\in A_S,
\]
and a positive smooth function
\[
  \eta:S\longrightarrow(0,\infty)
\]
such that
\[
  \overline{\Omega_a^+}
  \cap
  \overline{\Omega_b^+}
  =
  \varnothing
  \qquad
  (a\neq b),
\]
and, after shrinking the tubular neighbourhood,
\[
  V_a
  \cap
  \tau
  \left\{
    0<|v|<\eta(\pi_{NS}(v))
  \right\}
  \subseteq
  \mathcal C(\Omega_a^-,\eta).
\]
A finite clean system is conically tame when each of its strata is.
\end{definition}

\begin{proposition}[Separated normal directions imply conical tameness]
\label{prop:separated-normal-directions-conical}
Let \(S\) be a compact clean multiple-incidence stratum. Suppose that, for
each \(a\in A_S\), there is a smooth unit normal field
\[
  \nu_a:S\longrightarrow\mathbb S(NS)
\]
and a constant \(\theta>0\) such that
\[
  d_{\mathbb S(NS)}\bigl(\nu_a(p),\nu_b(p)\bigr)\geq\theta
  \qquad
  (p\in S,\ a\neq b),
\]
and, after shrinking a tubular neighbourhood, every point
\(x=\tau(v)\in V_a\) sufficiently near \(S\) satisfies
\[
  d_{\mathbb S(NS)}
  \left(
    \frac{v}{|v|},
    \nu_a(\pi_{NS}(v))
  \right)
  <\frac{\theta}{4}.
\]
Then \(S\) is conically tame.
\end{proposition}

\begin{proof}
Let \(\Omega_a^-\) and \(\Omega_a^+\) be the fibrewise angular balls about
\(\nu_a\) of radii \(\theta/3\) and \(2\theta/5\), respectively. The
hypothesis on the branch directions puts \(V_a\) inside
\(\mathcal C(\Omega_a^-,\eta)\) after shrinking the radial function
\(\eta\). Moreover,
\(\overline{\Omega_a^-}\subset\Omega_a^+\),
while the triangle inequality and
\(2(2\theta/5)<\theta\) give
\(\overline{\Omega_a^+}\cap\overline{\Omega_b^+}=\varnothing \qquad(a\neq b)\).
Compactness of \(S\) allows the local radial bounds to be replaced by one
positive smooth function \(\eta\). These are precisely the conditions of
Definition~\ref{def:conical-tameness}.
\end{proof}

\begin{corollary}[Isolated transverse incidence]
\label{cor:isolated-transverse-conical}
Let \(S=\{y\}\) be an isolated clean multiple-incidence stratum. Suppose
that, in a smooth coordinate chart centred at \(y\), each incident branch is
contained near \(y\) in a \(C^1\) sector with a limiting tangent ray, and the
limiting tangent rays of distinct branches are distinct. Then the incidence is
conically tame.
\end{corollary}

\begin{proof}
The finitely many limiting unit tangent directions have a positive minimum
angular separation. Each branch direction lies in an arbitrarily small
angular neighbourhood of its limiting ray sufficiently near \(y\). Apply
Proposition~\ref{prop:separated-normal-directions-conical}.
\end{proof}

\begin{proposition}[The exponentially tangent system is not conically tame]
\label{prop:exponential-not-conically-tame}
The multiple-incidence point in the exponentially tangent ribbon system of
Section~\ref{subsec:exponentially-tangent-branches} is clean but not
conically tame.
\end{proposition}

\begin{proof}
The multiple-incidence locus is the single point \(o=(0,0)\), all incident
branches are isolated from any further multiple-incidence stratum, and the
chosen branch prolongations are defined near \(o\); hence the stratum is
clean. In the standard tubular coordinates at \(o\), the centre curves of
the two ribbons contain
\(p_x=(x,0), \qquad q_x=(x,\delta(x)), \qquad x>0\).
Since \(\delta(x)/x\to0\), both normalized direction vectors converge to
\((1,0)\in\mathbb S^1\). If conical tameness held, the tails of the two
branches would lie in angular regions \(\Omega_0^-\) and \(\Omega_1^-\)
whose larger closures \(\overline{\Omega_0^+}\) and
\(\overline{\Omega_1^+}\) are disjoint. The common limiting direction
\((1,0)\) would belong to both closures, a contradiction.
\end{proof}

\begin{theorem}[Invariance of conical tameness]
\label{thm:conical-tameness-invariance}
Let \(S\subset M_2\) and \(S'\subset M_2'\) be clean multiple-incidence
strata. Suppose that a diffeomorphism of neighbourhood germs
\[
  F:(M_2,S)\longrightarrow(M_2',S')
\]
maps the incident branch germs bijectively,
\(F(V_a)=V'_{\sigma(a)}\). Then \(S\) is conically tame if and only if
\(S'\) is conically tame. In particular, conical tameness is independent of
the auxiliary Riemannian metric and tubular neighbourhood used to state
Definition~\ref{def:conical-tameness}.
\end{theorem}

\begin{proof}
The differential of \(F\) induces a normal-bundle isomorphism \(NF:NS\to NS'\), hence after radial normalization a diffeomorphism \(\widehat{NF}:\mathbb S(NS)\to\mathbb S(NS')\). For tubular maps \(\tau,\tau'\), the normalized normal direction of \(\tau'^{-1}F(\tau(p,v))\) converges, as \(|v|\to0\), locally uniformly on compact subsets of \(S\times\mathbb S(NS)\), to \(\widehat{NF}(p,v/|v|)\).

If \(S\) is conically tame, the finitely many closed angular regions \(\overline{\Omega_a^+}\) are disjoint. Their images under \(\widehat{NF}\) therefore admit pairwise disjoint open neighbourhoods with smaller nested angular regions containing \(\widehat{NF}(\Omega_a^-)\). The preceding convergence places each branch germ \(F(V_a)\) in the corresponding target cone after shrinking the radial neighbourhood; a positive smooth radius below the local bounds is obtained by a locally finite partition of unity. Hence \(S'\) is conically tame. Applying the same argument to \(F^{-1}\) proves the converse. Taking \(F\) to be the identity germ gives independence of the auxiliary metric and tubular map.
\end{proof}

For every \(a\in A_S\), choose an angular cutoff \(\chi_a: \mathcal U\setminus S \longrightarrow[0,1]\) such that, near \(S\),
\(\chi_a=1 \quad \text{on }\mathcal C(\Omega_a^-,\eta)\),
\(\chi_a=0 \quad \text{on }\mathcal C(\Omega_b^-,\eta), \qquad b\neq a\),
and, in normal-bundle coordinates \((p,v)\), the estimates
are locally uniform along \(S\): for every compact set \(K\Subset S\) and every pair of multi-indices \(\alpha,\beta\), there exists \(C_{K,\alpha,\beta}>0\) such that
\[
  \left|
    \partial_p^\alpha
    \partial_v^\beta
    \chi_a(p,v)
  \right|
  \leq
  C_{K,\alpha,\beta}
  |v|^{-|\beta|}
\]
for \(p\in K\) and sufficiently small \(v\neq0\).
The construction and estimates are given in
Appendix~\ref{app:conical-cutoffs}.

\subsection{Cutoff-flatness}
\label{subsec:cutoff-flatness}

For
\(r\in\Gamma^k(\mathcal U,F)\),
the notation \(j_S^k r=0\) means vanishing of the full ambient order-\(k\) jet along \(S\).
Equivalently, in every adapted coordinate chart and local bundle
frame,
\(\partial_p^\alpha\partial_v^\beta r(p,0)=0 \qquad (|\alpha|+|\beta|\leq k)\).
For such an \(r\), define
\[
  [\chi_a r]_0(x)
  :=
  \begin{cases}
    \chi_a(x)r(x),
    &x\notin S,\\
    0,
    &x\in S.
  \end{cases}
\]

\begin{theorem}[Finite-order cutoff-flatness]
\label{thm:finite-order-cutoff-flatness}
If
\[
  r\in\Gamma^k(\mathcal U,F),
  \qquad
  j_S^k r=0,
\]
then
\[
  [\chi_a r]_0\in\Gamma^k(\mathcal U,F),
  \qquad
  j_S^k[\chi_a r]_0=0.
\]
Moreover, after shrinking the neighbourhood of \(S\),
\[
  [\chi_a r]_0=r
  \quad
  \text{on }V_a,
\]
and
\[
  [\chi_a r]_0=0
  \quad
  \text{on }V_b,
  \qquad
  b\neq a.
\]
\end{theorem}

\begin{proof}
This is
Proposition~\ref{prop:quantitative-cutoff-flatness}
of Appendix~\ref{app:conical-cutoffs}.
\end{proof}

\begin{corollary}[Smooth cutoff-flatness]
\label{cor:smooth-cutoff-flatness}
If
\[
  r\in\Gamma^\infty(\mathcal U,F)
\]
is flat along \(S\), then
\[
  [\chi_a r]_0
\]
is smooth and flat along \(S\).
\end{corollary}

\begin{proof}
Apply
Theorem~\ref{thm:finite-order-cutoff-flatness}
at every finite order.
\end{proof}

\subsection{Conically tame realisation}
\label{subsec:conically-tame-realisation}

Assume that \(B_2^{\mathrm{mult}} = \bigsqcup_{\mu=1}^N S^{(\mu)}\) is a finite clean conically tame system.

Fix
\(s_1\in\mathfrak E_A^k\).
Let \(S=S^{(\mu)}\) be one multiple-incidence stratum, and choose a reference branch
\(a_0\in A_S\).
After shrinking to a common neighbourhood on which every relevant
prolonged branch section is defined, set
\[
  r_a[s_1]
  :=
  \overline{\sigma}_a[s_1]
  -
  \overline{\sigma}_{a_0}[s_1],
  \qquad
  a\in A_S.
\]
Since \(s_1\in\mathfrak E_A^k\),
\(j_S^k r_a[s_1]=0\).

Define
\[
  s_{2,S}^{\mathrm{loc}}[s_1]
  :=
  \overline{\sigma}_{a_0}[s_1]
  +
  \sum_{a\in A_S\setminus\{a_0\}}
  [\chi_a r_a[s_1]]_0.
\]

\begin{proposition}[Local conical companion]
\label{prop:local-conical-companion}
After shrinking the neighbourhood of \(S\),
\[
  s_{2,S}^{\mathrm{loc}}[s_1]
  \in
  \Gamma^k(\mathcal U_S,E_2)
\]
and
\[
  s_{2,S}^{\mathrm{loc}}[s_1]
  =
  \sigma[s_1]
  \qquad
  \text{on }U_2\cap\mathcal U_S.
\]
\end{proposition}

\begin{proof}
Finite-order cutoff-flatness gives \(C^k\)-regularity of every
correction.

On the reference branch \(V_{a_0}\), all correction terms vanish, so
\(s_{2,S}^{\mathrm{loc}}[s_1] = \overline{\sigma}_{a_0}[s_1] = \sigma[s_1]\).

On a branch
\(V_b\cap\mathcal U_S, \qquad b\in A_S\setminus\{a_0\}\),
the cutoff \(\chi_b\) equals one and all other correction cutoffs
vanish. Hence
\[
  s_{2,S}^{\mathrm{loc}}[s_1]
  =
  \overline{\sigma}_{a_0}[s_1]
  +
  r_b[s_1]
  =
  \overline{\sigma}_b[s_1]
  =
  \sigma[s_1].
\]
\end{proof}

The branch-isolation condition in
Definition~\ref{def:clean-multiple-incidence}
ensures that these are all branch pieces meeting
\(U_2\cap\mathcal U_S\). At a single-incidence point away from the
multiple-incidence strata, the unique prolonged branch section is
already a local companion. Hence every point of \(B_2\) has a
neighbourhood carrying a \(C^k\)-section agreeing with
\(\sigma[s_1]\) on its intersection with \(U_2\).

Choose a locally finite cover \(\{\mathcal U_i\}_{i\in I}\) of a neighbourhood \(\mathcal U\) of \(B_2\), local companions
\(t_i[s_1]\in\Gamma^k(\mathcal U_i,E_2)\),
and a smooth partition of unity \(\{\eta_i\}_{i\in I}\) subordinate to that cover. Set
\(s_2^{\mathrm{loc}}[s_1] := \sum_{i\in I}\eta_i t_i[s_1]\).
Every local companion agrees with \(\sigma[s_1]\) on the overlap, so
\(s_2^{\mathrm{loc}}[s_1] = \sigma[s_1] \qquad \text{on }U_2\cap\mathcal U\).

Choose \(\zeta\in C^\infty(M_2,[0,1])\) such that \(\zeta=1\) on a neighbourhood of \(B_2\) and
\(\operatorname{supp}\zeta\subseteq\mathcal U\).
Since \(1-\zeta\) vanishes near the boundary \(B_2\), the section \((1-\zeta)\sigma[s_1]\) extends by zero from \(U_2\) to a \(C^k\)-section of \(E_2\) on
\(M_2\). Since \(\operatorname{supp}\zeta\subseteq\mathcal U\), the
section \(\zeta s_2^{\mathrm{loc}}[s_1]\) also extends by zero from \(\mathcal U\) to \(M_2\). Their sum is a
global companion.

\begin{theorem}[Conically tame finite-order realisation]
\label{thm:conically-tame-finite-realisation}
For every
\[
  k\in\mathbb N_0,
\]
\[
  \boxed{
  \operatorname{Im}\operatorname{res}_1^k
  =
  \ker\mathfrak d_A^k.
  }
\]
\end{theorem}

\begin{proof}
The inclusion from left to right is
Theorem~\ref{thm:universal-jet-obstruction}.
For the converse, the construction above gives, for every
\(s_1\in\ker\mathfrak d_A^k\),
a section \(s_2\in\Gamma^k(M_2,E_2)\) satisfying
\(s_2|_{U_2} = \Phi\circ s_1\circ f^{-1}\).
\((s_1,s_2)\) is a compatible global section.
\end{proof}

\begin{corollary}
\label{cor:tame-residual-vanishing}
Under the finite clean conically tame hypotheses,
\[
  \mathfrak W_A^k=0
  \qquad
  \text{for every }k\in\mathbb N_0.
\]
\end{corollary}

The smooth case uses the same patching with Corollary~\ref{cor:smooth-cutoff-flatness}.

\begin{theorem}[Conically tame smooth realisation]
\label{thm:conically-tame-smooth-realisation}
Under the finite clean conically tame hypotheses,
\[
  \boxed{
  \operatorname{Im}\operatorname{res}_1^\infty
  =
  \mathfrak E_A^\infty
  =
  \bigcap_{k\geq0}
  \ker\mathfrak d_A^k.
  }
\]
\end{theorem}

\begin{proof}
Universal necessity gives
\(\operatorname{Im}\operatorname{res}_1^\infty \subseteq \mathfrak E_A^\infty\).
Conversely, let
\(s_1\in\mathfrak E_A^\infty\).
For every multiple-incidence stratum and every incident branch,
the difference \(r_a[s_1] = \overline{\sigma}_a[s_1] - \overline{\sigma}_{a_0}[s_1]\) has zero jet of every finite order along the stratum, and is therefore
flat there. Smooth cutoff-flatness and the preceding local-to-global
construction produce a smooth companion on \(M_2\).
\end{proof}

\begin{remark}
Conical tameness is a sufficient realisation hypothesis and is not
claimed to be necessary. The construction proves the set-theoretic
restriction-image identities above. No boundedness or continuity of
the resulting extension operator in a Fr\'echet or Sobolev topology is
asserted.
\end{remark}

\section{Sobolev completion through incidence traces}
\label{sec:sobolev-closure}

Fix
\(m\in\mathbb N, \qquad 1<p<\infty\).
Choose a smooth positive density on \(M_1\), a bundle metric on
\(E_1\), and a connection used to define the Sobolev norm
\(W^{m,p}(M_1,E_1)\).
These ambient analytic data remain fixed throughout the section.

Write
\[
  \mathfrak R_A^\infty
  :=
  \operatorname{Im}\operatorname{res}_1^\infty
  \subseteq
  \Gamma^\infty(M_1,E_1)
\]
for the smooth restriction image and
\[
  \mathfrak E_A^\infty
  =
  \left\{
    s\in\Gamma^\infty(M_1,E_1):
    \mathfrak d_A^k(s)=0
    \text{ for every }k\geq0
  \right\}
\]
for the full smooth transported-jet equaliser. Set
\(\mathfrak R_A^{\infty;m,p} := \mathfrak R_A^\infty \cap W^{m,p}(M_1,E_1)\),
and
\(\mathfrak E_A^{\infty;m,p} := \mathfrak E_A^\infty \cap W^{m,p}(M_1,E_1)\).

The objective is to determine
\(\overline{ \mathfrak R_A^{\infty;m,p} }^{\,W^{m,p}(M_1,E_1)}\).

\subsection{Geometric source carriers}
\label{subsec:visible-incidence-traces}

Retain the incidence strata \(S_\lambda\), source-coincidence
partitions \(\mathcal P_\lambda\), and limiting source maps
\(\gamma_B: S_\lambda\longrightarrow M_1, \qquad B\in\mathcal P_\lambda\),
from Section~\ref{sec:transported-jets}. The finite-jet theory did not
require these maps to be embeddings. Sobolev trace theory does require
additional geometric hypotheses.

\begin{hypothesis}[Embedded source-carrier hypothesis]
\label{hyp:embedded-source-carriers}
For every pair
\[
  (\lambda,B),
  \qquad
  B\in\mathcal P_\lambda,
\]
the map
\[
  \gamma_B:
  S_\lambda\longrightarrow M_1
\]
is a proper smooth embedding onto a closed embedded submanifold
\[
  Z_{\lambda,B}
  :=
  \gamma_B(S_\lambda)
  \subseteq M_1.
\]
After identifying equal images, the resulting finite family of
distinct geometric carriers
\[
  \mathscr Z_A
  =
  \{Z\}
\]
is clean: every nonempty finite intersection
\[
  Z_1\cap\cdots\cap Z_q
\]
is an embedded submanifold, and at each intersection point,
\[
  T
  \left(
    Z_1\cap\cdots\cap Z_q
  \right)
  =
  TZ_1\cap\cdots\cap TZ_q.
\]
\end{hypothesis}

Set
\(Z_A := \bigcup_{Z\in\mathscr Z_A}Z\).
Because the family is finite and every carrier is closed,
\(Z_A\) is closed.

Different branch or stratum labels may determine the same geometric
carrier. Such repetitions do not produce independent source traces:
they correspond to the same jet of the same section along the same
submanifold. The transported compatibility equations may nevertheless
use that common source trace in different ways.

Let \(Z\in\mathscr Z_A\) have codimension
\(c_Z:=\operatorname{codim}_{M_1}Z\).
Define
\[
  r_Z(m,p)
  :=
  \max
  \left\{
    j\in\mathbb N_0:
    j<m-\frac{c_Z}{p}
  \right\},
\]
with the convention \(r_Z(m,p)=-1\) when the set is empty.

After choosing a metric, connection, and tubular splitting near \(Z\),
write \(\gamma_Z^{(j)}u\) for the symmetrised order-\(j\) normal covariant trace. In the standard
trace regime,
\[
  \gamma_Z^{(j)}:
  W^{m,p}(M_1,E_1)
  \longrightarrow
  B_{p,p}^{\,m-j-c_Z/p}
  \left(
    Z;
    E_1|_Z\otimes
    \operatorname{Sym}^jN_Z^*
  \right)
\]
is bounded whenever
\(0\leq j\leq r_Z(m,p)\).

For bookkeeping, define the ambient product of visible trace spaces by
\[
  \mathcal B_A^{m,p}
  :=
  \bigoplus_{Z\in\mathscr Z_A}
  \;
  \bigoplus_{j=0}^{r_Z(m,p)}
  B_{p,p}^{\,m-j-c_Z/p}
  \left(
    Z;
    E_1|_Z\otimes
    \operatorname{Sym}^jN_Z^*
  \right),
\]
where a carrier with \(r_Z(m,p)=-1\) contributes no summand.

When carriers intersect, arbitrary elements of this direct sum need
not be jointly realisable as traces of a single Sobolev section.
Accordingly, the complete trace target below is allowed to be a proper
closed subspace of \(\mathcal B_A^{m,p}\) encoding all intersection
compatibilities.

\begin{hypothesis}[Complete incidence trace hypothesis]
\label{hyp:complete-incidence-trace}
There exist:

\begin{enumerate}[label=\textup{(\roman*)}]
  \item a Banach space
        \[
          \mathcal T_A^{m,p}
          \subseteq
          \mathcal B_A^{m,p},
        \]
        continuously embedded as a closed subspace;

  \item a bounded surjective trace map
        \[
          \operatorname{Tr}_A^{m,p}:
          W^{m,p}(M_1,E_1)
          \longrightarrow
          \mathcal T_A^{m,p};
        \]

  \item a bounded linear right inverse
        \[
          \mathcal J_A^{m,p}:
          \mathcal T_A^{m,p}
          \longrightarrow
          W^{m,p}(M_1,E_1);
        \]
\end{enumerate}
such that:

\begin{enumerate}[label=\textup{(\alph*)}]
  \item \(\operatorname{Tr}_A^{m,p}\) contains exactly one copy of
        every visible normal derivative trace on every distinct
        geometric carrier;

  \item
        \[
          \operatorname{Tr}_A^{m,p}
          \mathcal J_A^{m,p}
          =
          \operatorname{id}_{\mathcal T_A^{m,p}};
        \]

  \item
        \[
          \ker\operatorname{Tr}_A^{m,p}
          =
          \overline{
            C_c^\infty(M_1\setminus Z_A,E_1)
          }^{\,W^{m,p}(M_1,E_1)}.
        \]
\end{enumerate}
\end{hypothesis}

The hypothesis includes both local trace data and the global approximation needed on noncompact \(M_1\).

\begin{hypothesis}[Ambient Sobolev core]
\label{hyp:ambient-sobolev-core}
The chosen global Sobolev realisation satisfies
\[
  W^{m,p}(M_1,E_1)
  =
  \overline{
    C_c^\infty(M_1,E_1)
  }^{\,W^{m,p}(M_1,E_1)}.
\]
\end{hypothesis}

The hypothesis is automatic when \(M_1\) is compact. On a noncompact
manifold it is a genuine global analytic assumption and is not a
formal consequence of completeness alone at arbitrary higher Sobolev
orders; see \citep{HondaMariRimoldiVeronelli2021}. It may instead be
verified directly for the Sobolev realisation used in a particular
application.

\begin{lemma}[Local-to-global complete trace package]
\label{lem:local-to-global-trace-package}
Assume Hypothesis~\ref{hyp:ambient-sobolev-core}. Let
\[
  Z_A=Z_1\sqcup\cdots\sqcup Z_N
\]
be a finite disjoint union of compact closed subsets. Suppose that, for
each \(i\), there exist:

\begin{enumerate}[label=\textup{(\roman*)}]
  \item a Banach space \(\mathcal T_i^{m,p}\) and a bounded trace map
        \[
          \operatorname{Tr}_i^{m,p}:
          W^{m,p}(M_1,E_1)
          \longrightarrow
          \mathcal T_i^{m,p};
        \]

  \item a precompact open neighbourhood \(U_i\) of \(Z_i\), with the
        sets \(U_i\) pairwise disjoint, and a bounded right inverse
        \[
          \mathcal J_i^{m,p}:
          \mathcal T_i^{m,p}
          \longrightarrow
          W^{m,p}(M_1,E_1)
        \]
        whose image is supported in \(U_i\);

  \item a cutoff
        \[
          \chi_i\in C_c^\infty(U_i,[0,1]),
          \qquad
          \chi_i=1
          \text{ on a neighbourhood of }Z_i,
        \]
        such that
        \[
          \chi_i\ker\operatorname{Tr}_i^{m,p}
          \subseteq
          \overline{
            C_c^\infty(U_i\setminus Z_i,E_1)
          }^{\,W^{m,p}(M_1,E_1)}.
        \]
\end{enumerate}

Assume also that
\[
  \operatorname{Tr}_j^{m,p}
  \mathcal J_i^{m,p}=0
  \qquad
  (i\neq j).
\]
Then the direct-sum trace map
\[
  \operatorname{Tr}_A^{m,p}
  :=
  \bigoplus_{i=1}^N
  \operatorname{Tr}_i^{m,p}
  :
  W^{m,p}(M_1,E_1)
  \longrightarrow
  \bigoplus_{i=1}^N\mathcal T_i^{m,p}
\]
is bounded and surjective, admits the bounded right inverse
\[
  \mathcal J_A^{m,p}
  :=
  \sum_{i=1}^N\mathcal J_i^{m,p},
\]
and satisfies
\[
  \boxed{
  \ker\operatorname{Tr}_A^{m,p}
  =
  \overline{
    C_c^\infty(M_1\setminus Z_A,E_1)
  }^{\,W^{m,p}(M_1,E_1)}.
  }
\]
\end{lemma}

\begin{proof}
Finiteness gives boundedness. The support and cross-trace assumptions imply
\(\operatorname{Tr}_A^{m,p}\mathcal J_A^{m,p}=I\), so the direct-sum trace is surjective with the stated right inverse. Since sections in \(C_c^\infty(M_1\setminus Z_A,E_1)\) have zero trace, their closure is contained in the kernel.

Conversely, let \(u\in\ker\operatorname{Tr}_A^{m,p}\) and \(\chi=\sum_i\chi_i\). For each \(i\), the local kernel hypothesis approximates \(\chi_i u\) by \(\psi_{i,n}\in C_c^\infty(U_i\setminus Z_i,E_1)\); the disjoint supports avoid all of \(Z_A\). By Hypothesis~\ref{hyp:ambient-sobolev-core}, choose \(\varphi_n\in C_c^\infty(M_1,E_1)\) with \(\varphi_n\to u\). Since \(1-\chi\) vanishes near \(Z_A\), \((1-\chi)\varphi_n\in C_c^\infty(M_1\setminus Z_A,E_1)\) and tends to \((1-\chi)u\). A diagonal sum with the \(\psi_{i,n}\) approximates \(u\) by sections supported off \(Z_A\).
\end{proof}

\begin{lemma}[Local zero-trace approximation]
\label{lem:local-zero-trace-approximation}
Let
\[
  U\Subset M_1
\]
be a precompact coordinate domain with smooth boundary, and let
\[
  W_0^{m,p}(U,E_1)
  :=
  \overline{
    C_c^\infty(U,E_1)
  }^{\,W^{m,p}(U,E_1)}.
\]
The following statements hold.

\begin{enumerate}[label=\textup{(\roman*)}]
  \item Let
        \[
          Z\subset U
        \]
        be a compact closed embedded hypersurface admitting a tubular
        neighbourhood with closure contained in \(U\). Then
        \[
          \boxed{
          \overline{
            C_c^\infty(U\setminus Z,E_1)
          }^{\,W^{m,p}(U,E_1)}
          =
          \left\{
            v\in W_0^{m,p}(U,E_1):
            \gamma_Z^{(j)}v=0,
            \ 0\leq j\leq m-1
          \right\}.
          }
        \]

  \item Let \(x\in U\), set
        \[
          \sigma:=m-\frac{n}{p},
          \qquad
          n:=\dim M_1,
        \]
        and assume
        \[
          \sigma\notin\mathbb N_0.
        \]
        Define
        \[
          r_x
          :=
          \max\{j\in\mathbb N_0:j<\sigma\},
        \]
        with \(r_x=-1\) when the set is empty. Then
        \[
          \boxed{
          \overline{
            C_c^\infty(U\setminus\{x\},E_1)
          }^{\,W^{m,p}(U,E_1)}
          =
          \left\{
            v\in W_0^{m,p}(U,E_1):
            \nabla^{(j)}v(x)=0,
            \ 0\leq j\leq r_x
          \right\},
          }
        \]
        where the point conditions are omitted when \(r_x=-1\). In
        that case the right-hand side is all of
        \(W_0^{m,p}(U,E_1)\).
\end{enumerate}

Let \(K\) denote either \(Z\) or \(\{x\}\), and let
\[
  \chi\in C_c^\infty(U)
\]
be equal to one on a neighbourhood of \(K\). If
\[
  u\in W^{m,p}(M_1,E_1)
\]
has vanishing visible traces along \(K\), then
\[
  \chi u
  \in
  \overline{
    C_c^\infty(U\setminus K,E_1)
  }^{\,W^{m,p}(M_1,E_1)}.
\]
\end{lemma}

\begin{proof}
For a hypersurface, the complete interior normal-trace theorem on a tubular collar gives precisely the stated kernel of the normal traces through order \(m-1\); the ambient space is \(W_0^{m,p}(U,E_1)\) because compact support in \(U\) also imposes the ordinary zero boundary condition at \(\partial U\) \citep{LionsMagenes1972,Triebel1992,GrosseSchneider2013,Leoni2017Sobolev}.

For a point, work in local coordinates and a frame and let
\[
 \mathcal Y=\overline{C_c^\infty(U\setminus\{x\},E_1)}^{\,W^{m,p}(U,E_1)}\subset W_0^{m,p}(U,E_1).
\]
The annihilator of \(\mathcal Y\) consists of distributions supported at \(x\), hence finite sums of derivatives \(D^\alpha\delta_x\). Such a derivative is continuous on \(W_0^{m,p}\) exactly when \(|\alpha|<m-n/p\): sufficiency is Sobolev--Morrey, and necessity follows by rescaling. Since the exponent is noncritical, the annihilator is generated exactly by \(|\alpha|\le r_x\). Hahn--Banach therefore yields \(\mathcal Y=\bigcap_{|\alpha|\le r_x}\ker D^\alpha|_x\); an invertible lower-triangular jet-coordinate change replaces coordinate derivatives by \(\nabla^{(j)}\). If \(r_x=-1\), the annihilator is zero \citep{AdamsFournier2003,Mazya2011Sobolev,AdamsHedberg1996}.

Finally, if \(\chi=1\) near \(K\), then \(\chi u\in W_0^{m,p}(U,E_1)\) and has the same vanishing visible traces. The appropriate local identity supplies approximants supported off \(K\); because all supports lie in a fixed compact subset of \(U\), local and ambient \(W^{m,p}\)-convergence agree.
\end{proof}

\begin{proposition}[Basic complete trace configurations]
\label{prop:basic-complete-trace-configurations}
Assume Hypothesis~\ref{hyp:ambient-sobolev-core}. Then
Hypothesis~\ref{hyp:complete-incidence-trace} holds in each of the
following situations.

\begin{enumerate}[label=\textup{(\roman*)}]
  \item \(Z_A=Z\) is a compact closed embedded hypersurface. In this
        case
        \[
          \mathcal T_A^{m,p}
          =
          \bigoplus_{j=0}^{m-1}
          B_{p,p}^{\,m-j-1/p}
          \left(
            Z;
            E_1|_Z\otimes\operatorname{Sym}^jN_Z^*
          \right).
        \]

  \item \(Z_A\) is a finite pairwise disjoint family of compact closed
        embedded hypersurfaces. The complete trace target is the direct
        sum of the hypersurface trace spaces in \textup{(i)}.

  \item \(Z_A=P\) is a finite set of points in an \(n\)-dimensional
        manifold and
        \[
          \sigma:=m-\frac{n}{p}
          \notin
          \mathbb N_0.
        \]
        Set
        \[
          r_P
          :=
          \max\{j\in\mathbb N_0:j<\sigma\},
        \]
        with \(r_P=-1\) when the set is empty. Then
        \[
          \mathcal T_A^{m,p}
          =
          \bigoplus_{x\in P}
          \bigoplus_{j=0}^{r_P}
          E_{1,x}\otimes\operatorname{Sym}^jT_x^*M_1,
        \]
        where the direct sum is zero when \(r_P=-1\).

  \item \(Z_A\) is a finite pairwise disjoint union of carriers of the
        types in \textup{(i)} and \textup{(iii)}.
\end{enumerate}
\end{proposition}

\begin{proof}
For a compact hypersurface, the classical normal-jet trace theorem gives a bounded surjection onto the displayed Besov product and a bounded right inverse; multiplying by a cutoff supported in a tubular neighbourhood makes the right inverse local. Lemma~\ref{lem:local-zero-trace-approximation}\textup{(i)} verifies the local kernel condition, so Lemma~\ref{lem:local-to-global-trace-package} gives the global package. Pairwise disjoint compact hypersurfaces are handled by disjoint tubular neighbourhoods and direct sums.

For an isolated point, Sobolev--Morrey gives the visible derivative evaluations through order \(r_P\). Arbitrary finite jets are realised by cutoff polynomials in a local frame, hence admit a bounded right inverse because the target is finite-dimensional. Lemma~\ref{lem:local-zero-trace-approximation}\textup{(ii)} supplies the local kernel identity; applying Lemma~\ref{lem:local-to-global-trace-package} at the finitely many points proves \textup{(iii)}. The case \(r_P=-1\) is covered by the same lemma together with the ambient-core hypothesis. Combining mutually disjoint neighbourhoods proves \textup{(iv)}.
\end{proof}

\begin{corollary}[Trace package for the relative-sign model]
\label{cor:relative-sign-complete-trace-package}
For the relative-sign bundle over \(M_1=\mathbb R\),
Hypothesis~\ref{hyp:complete-incidence-trace} holds for every
\[
  m\in\mathbb N,
  \qquad
  1<p<\infty,
\]
with
\[
  \mathcal T_A^{m,p}
  =
  \mathbb K^m,
  \qquad
  \operatorname{Tr}_A^{m,p}u
  =
  \left(
    u(0),u'(0),\ldots,u^{(m-1)}(0)
  \right).
\]
\end{corollary}

\begin{proof}
The ambient core hypothesis holds on \(\mathbb R\), and the carrier is
the single noncritical point \(\{0\}\). Apply
Proposition~\ref{prop:basic-complete-trace-configurations}\textup{(iii)}.
\end{proof}

Clean intersecting positive-dimensional carriers require a separate simultaneous trace theorem encoding intersection compatibilities; Proposition~\ref{prop:basic-complete-trace-configurations} is not a classification of such geometries.

\subsection{Changes of trace coordinates}
\label{subsec:trace-coordinate-invariance}

Derivative-trace coordinates depend on the metric, connection, tubular coordinates, and normal splitting. Smooth choices are boundedly equivalent on compact carriers; on noncompact carriers we impose the required uniform control.

\begin{definition}[Uniformly admissible change of trace data]
\label{def:uniformly-admissible-trace-change}
Let
\[
  \operatorname{Tr}_A^{m,p}:
  W^{m,p}(M_1,E_1)
  \longrightarrow
  \mathcal T_A^{m,p}
\]
and
\[
  \widetilde{\operatorname{Tr}}_A^{m,p}:
  W^{m,p}(M_1,E_1)
  \longrightarrow
  \widetilde{\mathcal T}_A^{m,p}
\]
be complete incidence-trace maps for the same geometric carrier
family. The change from the first auxiliary data to the second is
\emph{uniformly admissible at order \((m,p)\)} if, for every carrier
\(Z\) and every visible order \(j\), the two normal-trace systems
satisfy on smooth sections
\[
  \widetilde\gamma_Z^{(j)}u
  =
  A_{Z,j}\gamma_Z^{(j)}u
  +
  \sum_{\ell<j}
  P_{Z,j\ell}\gamma_Z^{(\ell)}u,
\]
where:

\begin{enumerate}[label=\textup{(\roman*)}]
  \item
        \[
          A_{Z,j}:
          E_1|_Z\otimes\operatorname{Sym}^jN_Z^*
          \longrightarrow
          E_1|_Z\otimes\operatorname{Sym}^jN_Z^*
        \]
        is a smooth bundle isomorphism whose induced action and inverse
        action are bounded on
        \[
          B_{p,p}^{\,m-j-c_Z/p}
          \left(
            Z;
            E_1|_Z\otimes\operatorname{Sym}^jN_Z^*
          \right);
        \]

  \item for \(\ell<j\), the tangential differential operator
        \(P_{Z,j\ell}\) extends boundedly as
        \[
          P_{Z,j\ell}:
          B_{p,p}^{\,m-\ell-c_Z/p}
          \left(
            Z;
            E_1|_Z\otimes\operatorname{Sym}^\ell N_Z^*
          \right)
          \longrightarrow
          B_{p,p}^{\,m-j-c_Z/p}
          \left(
            Z;
            E_1|_Z\otimes\operatorname{Sym}^j N_Z^*
          \right);
        \]

  \item the reverse change of auxiliary data admits relations of the
        same form with the same boundedness properties;

  \item the resulting lower-triangular operators on the ambient trace
        products map \(\mathcal T_A^{m,p}\) onto
        \(\widetilde{\mathcal T}_A^{m,p}\).
\end{enumerate}
\end{definition}

The last condition retains any compatibility conditions imposed at
intersections of carriers. It is automatic when both complete trace
targets are defined as the ranges of the corresponding simultaneous
trace maps and the lower-triangular transformation is boundedly
invertible on the ambient products.

\begin{proposition}[Bounded equivalence of uniformly admissible trace
coordinates]
\label{prop:trace-coordinate-invariance}
If two complete incidence-trace systems are related by a uniformly
admissible change at order \((m,p)\), then there exists a bounded
Banach-space isomorphism
\[
  L_A^{m,p}:
  \mathcal T_A^{m,p}
  \longrightarrow
  \widetilde{\mathcal T}_A^{m,p}
\]
such that
\[
  \widetilde{\operatorname{Tr}}_A^{m,p}
  =
  L_A^{m,p}\operatorname{Tr}_A^{m,p}.
\]
With the trace coordinates ordered by normal derivative degree,
\(L_A^{m,p}\) is lower triangular and has boundedly invertible
diagonal blocks.
\end{proposition}

\begin{proof}
For each carrier, Definition~\ref{def:uniformly-admissible-trace-change} gives a bounded lower-triangular operator on the visible Besov product whose diagonal blocks are boundedly invertible; the reverse coordinate change supplies its bounded inverse. The finite direct sum is therefore a bounded isomorphism of the ambient trace products and, by condition~\textup{(iv)}, restricts to \(\mathcal L_A^{m,p}:\mathcal T_A^{m,p}\to\widetilde{\mathcal T}_A^{m,p}\). The trace identities hold on smooth sections and extend to all of \(W^{m,p}\) by density and continuity.
\end{proof}

\begin{corollary}[Compact carriers]
\label{cor:compact-trace-coordinate-invariance}
If every carrier in \(\mathscr Z_A\) is compact, then any two smooth
choices of metric, connection, tubular coordinates, and normal
splitting are uniformly admissible at every fixed order \((m,p)\).
The associated complete trace systems are boundedly
equivalent whenever they encode the same intersection compatibility
conditions.
\end{corollary}

\begin{proof}
On a compact carrier, the coefficients and finitely many derivatives arising from changes of connection, tubular coordinates, and normal splitting are uniformly bounded. Multiplication by the leading bundle isomorphisms and their inverses is bounded on the relevant Besov spaces, and the lower-order tangential operators have the required Besov mapping properties. The same holds for the reverse change; finiteness of the carrier family gives a boundedly invertible direct-sum transformation.
\end{proof}

Hence the closed Sobolev subspaces are invariant under uniformly admissible trace-coordinate changes; no such global claim is made on noncompact carriers without uniform bounds.

\subsection{Zero-trace directions}
\label{subsec:zero-trace-directions}

Sections supported away from the source incidence locus are
automatically extendable.

\begin{lemma}[Off-incidence extension]
\label{lem:off-incidence-extension}
\[
  C_c^\infty(M_1\setminus Z_A,E_1)
  \subseteq
  \mathfrak R_A^\infty.
\]
\end{lemma}

\begin{proof}
Let \(s_1\in C_c^\infty(M_1\setminus Z_A,E_1)\) and choose a neighbourhood \(W\supset Z_A\) disjoint from \(\operatorname{supp}s_1\). Since \(g_a(S_a)=\gamma_a(S_a)\subset Z_A\), shrink each \(\Omega_a\) so that \(g_a(\Omega_a)\subset W\). Then every prolonged branch representative \(\Psi_a(g_a^*s_1)\) vanishes near \(S_a\). Finiteness of the branch system implies that the forced overlap section \(\sigma[s_1]\) vanishes near \(B_2\), so extension by zero to \(M_2\) is smooth and gives a compatible companion.
\end{proof}

\begin{corollary}
\label{cor:trace-kernel-in-closure}
Under
Hypothesis~\ref{hyp:complete-incidence-trace},
\[
  \ker\operatorname{Tr}_A^{m,p}
  \subseteq
  \overline{
    \mathfrak R_A^{\infty;m,p}
  }^{\,W^{m,p}(M_1,E_1)}.
\]
\end{corollary}

\begin{proof}
By
Hypothesis~\ref{hyp:complete-incidence-trace}\textup{(c)},
the trace kernel is the \(W^{m,p}\)-closure of
\(C_c^\infty(M_1\setminus Z_A,E_1)\).
Apply Lemma~\ref{lem:off-incidence-extension}.
\end{proof}

\subsection{Realisable and formal trace spaces}
\label{subsec:realizable-formal-traces}

Define the closed realisable trace space by
\[
  \mathcal X_A^{m,p}
  :=
  \overline{
    \operatorname{Tr}_A^{m,p}
    \left(
      \mathfrak R_A^{\infty;m,p}
    \right)
  }^{\,\mathcal T_A^{m,p}}.
\]
It consists of incidence traces approximable by traces of smooth
extendable sections.

Define the closed formal trace space by
\[
  \mathcal K_A^{m,p}
  :=
  \overline{
    \operatorname{Tr}_A^{m,p}
    \left(
      \mathfrak E_A^{\infty;m,p}
    \right)
  }^{\,\mathcal T_A^{m,p}}.
\]
Since
\(\mathfrak R_A^\infty \subseteq \mathfrak E_A^\infty\),
one has
\(\mathcal X_A^{m,p} \subseteq \mathcal K_A^{m,p}\).

The corresponding quotient-valued trace maps are
\[
  \mathfrak D_{A,\mathrm{real}}^{m,p}
  :=
  q_{\mathcal X}
  \circ
  \operatorname{Tr}_A^{m,p},
\]
and
\[
  \mathfrak D_{A,\mathrm{form}}^{m,p}
  :=
  q_{\mathcal K}
  \circ
  \operatorname{Tr}_A^{m,p},
\]
where
\[
  q_{\mathcal X}:
  \mathcal T_A^{m,p}
  \longrightarrow
  \mathcal T_A^{m,p}/\mathcal X_A^{m,p},
\]
and
\[
  q_{\mathcal K}:
  \mathcal T_A^{m,p}
  \longrightarrow
  \mathcal T_A^{m,p}/\mathcal K_A^{m,p}
\]
are the Banach quotient maps.

The quotient \(\mathcal W_A^{m,p} := \mathcal K_A^{m,p}/\mathcal X_A^{m,p}\) is the Sobolev residual Whitney defect. It measures the formal trace
directions approximable by smooth formally compatible sections but not
by smooth extendable sections.

Under a uniformly admissible change of trace coordinates, the
isomorphism \(L_A^{m,p}\) from
Proposition~\ref{prop:trace-coordinate-invariance}
maps the spaces \(\mathcal X_A^{m,p}, \qquad \mathcal K_A^{m,p}, \qquad \mathcal W_A^{m,p}\) onto their counterparts for the new trace system.

\subsection{An abstract trace-closure lemma}
\label{subsec:abstract-trace-closure}

\begin{lemma}[Trace-closure lemma]
\label{lem:abstract-trace-closure}
Let \(X\) and \(Y\) be Banach spaces, let
\[
  T:X\longrightarrow Y
\]
be a bounded surjection with a bounded linear right inverse
\[
  J:Y\longrightarrow X,
\]
and let \(R\subseteq X\) be a linear subspace satisfying
\[
  \ker T\subseteq\overline R^{\,X}.
\]
Then
\[
  \overline R^{\,X}
  =
  T^{-1}
  \left(
    \overline{T(R)}^{\,Y}
  \right).
\]
\end{lemma}

\begin{proof}
Continuity of \(T\) gives
\(\overline R^{\,X} \subseteq T^{-1} \left( \overline{T(R)}^{\,Y} \right)\).

Conversely, let
\(x\in T^{-1} \left( \overline{T(R)}^{\,Y} \right)\).
Choose \(r_n\in R\) such that
\(T r_n\longrightarrow Tx\).
Set \(v_n:=x-r_n\) and
\(w_n:=v_n-J T v_n\).
Then
\(T w_n=0\),
so
\(w_n\in\ker T \subseteq \overline R^{\,X}\).
Choose \(k_n\in R\) such that
\(\|k_n-w_n\|_X<\frac1n\).
Since
\(\|J T v_n\|_X\longrightarrow0\),
one has
\(x-(r_n+k_n) = w_n-k_n+J T v_n \longrightarrow0\).
As \(r_n+k_n\in R\), the reverse inclusion follows.
\end{proof}

\subsection{Exact Sobolev closure}
\label{subsec:exact-sobolev-closure}

\begin{theorem}[Exact Sobolev closure]
\label{thm:exact-sobolev-closure}
Assume
Hypotheses~\ref{hyp:embedded-source-carriers}
and \ref{hyp:complete-incidence-trace}. Then
\[
  \boxed{
  \overline{
    \mathfrak R_A^{\infty;m,p}
  }^{\,W^{m,p}(M_1,E_1)}
  =
  \left(
    \operatorname{Tr}_A^{m,p}
  \right)^{-1}
  \left(
    \mathcal X_A^{m,p}
  \right).
  }
\]
Equivalently,
\[
  \overline{
    \mathfrak R_A^{\infty;m,p}
  }^{\,W^{m,p}}
  =
  \ker
  \mathfrak D_{A,\mathrm{real}}^{m,p}.
\]
\end{theorem}

\begin{proof}
Apply
Lemma~\ref{lem:abstract-trace-closure}
with
\(X=W^{m,p}(M_1,E_1), \qquad Y=\mathcal T_A^{m,p}\),
\(T=\operatorname{Tr}_A^{m,p}, \qquad R=\mathfrak R_A^{\infty;m,p}\).
Corollary~\ref{cor:trace-kernel-in-closure}
supplies
\(\ker T\subseteq\overline R^{\,X}\),
and by definition,
\(\overline{T(R)}^{\,Y} = \mathcal X_A^{m,p}\).
\end{proof}

The remaining problem is therefore to identify the closed realisable trace space \(\mathcal X_A^{m,p}\).

\begin{theorem}[Analytic defect quotient]
\label{thm:analytic-defect-quotient}
Under the same hypotheses, the map
\[
  u
  \longmapsto
  \operatorname{Tr}_A^{m,p}u
  +
  \mathcal X_A^{m,p}
\]
induces a canonical Banach-space isomorphism
\[
  W^{m,p}(M_1,E_1)
  \Big/
  \overline{
    \mathfrak R_A^{\infty;m,p}
  }^{\,W^{m,p}}
  \cong
  \mathcal T_A^{m,p}/\mathcal X_A^{m,p}.
\]
\end{theorem}

\begin{proof}
The induced map is bounded and surjective. By
Theorem~\ref{thm:exact-sobolev-closure},
its kernel is
\[
  \left(
    \operatorname{Tr}_A^{m,p}
  \right)^{-1}
  \left(
    \mathcal X_A^{m,p}
  \right)
  =
  \overline{
    \mathfrak R_A^{\infty;m,p}
  }^{\,W^{m,p}}.
\]
The Banach-space first isomorphism theorem gives the result. A bounded
inverse on the quotient is induced by the right inverse
\(\mathcal J_A^{m,p}\).
\end{proof}

If every carrier is an isolated point, the trace target and analytic
defect quotient are finite-dimensional. For positive-dimensional
carriers, they are generally constructed from Besov spaces and need
not be finite-dimensional.

\subsection{Formal trace bounds}
\label{subsec:formal-trace-bounds}

The inclusion \(\mathcal X_A^{m,p} \subseteq \mathcal K_A^{m,p}\) gives a universal formal upper bound on the closure.

\begin{corollary}[Formal trace bound]
\label{cor:formal-trace-bound}
Under
Hypotheses~\ref{hyp:embedded-source-carriers}
and \ref{hyp:complete-incidence-trace},
\[
  \overline{
    \mathfrak R_A^{\infty;m,p}
  }^{\,W^{m,p}}
  \subseteq
  \left(
    \operatorname{Tr}_A^{m,p}
  \right)^{-1}
  \left(
    \mathcal K_A^{m,p}
  \right)
  =
  \ker
  \mathfrak D_{A,\mathrm{form}}^{m,p}.
\]
\end{corollary}

The possible strictness is measured by
\(\mathcal W_A^{m,p} = \mathcal K_A^{m,p}/\mathcal X_A^{m,p}\).

\begin{proposition}[Criterion for formal exactness]
\label{prop:criterion-formal-exactness}
The following are equivalent:

\begin{enumerate}[label=\textup{(\roman*)}]
  \item
        \[
          \mathcal W_A^{m,p}=0;
        \]

  \item
        \[
          \mathcal X_A^{m,p}
          =
          \mathcal K_A^{m,p};
        \]

  \item
        \[
          \overline{
            \mathfrak R_A^{\infty;m,p}
          }^{\,W^{m,p}}
          =
          \left(
            \operatorname{Tr}_A^{m,p}
          \right)^{-1}
          \left(
            \mathcal K_A^{m,p}
          \right).
        \]
\end{enumerate}
\end{proposition}

\begin{proof}
The equivalence of \textup{(i)} and \textup{(ii)} follows from the
definition of \(\mathcal W_A^{m,p}\). The exact Sobolev closure theorem
gives
\(\textup{(ii)} \Longrightarrow \textup{(iii)}\).

Conversely, assume \textup{(iii)}. The exact closure theorem gives
\[
  \left(
    \operatorname{Tr}_A^{m,p}
  \right)^{-1}
  \left(
    \mathcal X_A^{m,p}
  \right)
  =
  \left(
    \operatorname{Tr}_A^{m,p}
  \right)^{-1}
  \left(
    \mathcal K_A^{m,p}
  \right).
\]
Applying the surjective map
\(\operatorname{Tr}_A^{m,p}\) to both sides yields
\(\mathcal X_A^{m,p} = \mathcal K_A^{m,p}\).
\end{proof}

\subsection{Conically tame Sobolev closure}
\label{subsec:conically-tame-sobolev-closure}

Suppose that the multiple-incidence locus is a finite clean conically
tame system. By
Theorem~\ref{thm:conically-tame-smooth-realisation},
\(\mathfrak R_A^\infty = \mathfrak E_A^\infty\).
\(\mathfrak R_A^{\infty;m,p} = \mathfrak E_A^{\infty;m,p}\),
and hence
\(\mathcal X_A^{m,p} = \mathcal K_A^{m,p}\).

\begin{theorem}[Conically tame Sobolev closure]
\label{thm:conically-tame-sobolev-closure}
Assume
Hypotheses~\ref{hyp:embedded-source-carriers}
and \ref{hyp:complete-incidence-trace}, and suppose that the
multiple-incidence locus is a finite clean conically tame system. Then
\[
  \boxed{
  \overline{
    \mathfrak R_A^{\infty;m,p}
  }^{\,W^{m,p}(M_1,E_1)}
  =
  \left(
    \operatorname{Tr}_A^{m,p}
  \right)^{-1}
  \left(
    \mathcal K_A^{m,p}
  \right).
  }
\]
Equivalently,
\[
  \overline{
    \mathfrak R_A^{\infty;m,p}
  }^{\,W^{m,p}}
  =
  \ker
  \mathfrak D_{A,\mathrm{form}}^{m,p}.
\]
\end{theorem}

\begin{proof}
Apply
Theorem~\ref{thm:exact-sobolev-closure}
and
\(\mathcal X_A^{m,p} = \mathcal K_A^{m,p}\).
\end{proof}

\begin{corollary}[Formal analytic defect quotient]
\label{cor:formal-analytic-defect-quotient}
Under the hypotheses of
Theorem~\ref{thm:conically-tame-sobolev-closure},
\[
  W^{m,p}(M_1,E_1)
  \Big/
  \overline{
    \mathfrak R_A^{\infty;m,p}
  }^{\,W^{m,p}}
  \cong
  \mathcal T_A^{m,p}/\mathcal K_A^{m,p}.
\]
\end{corollary}

No bounded smooth extension operator is required. Conical tameness is
used only through the set-theoretic identity
\(\mathfrak R_A^\infty = \mathfrak E_A^\infty\).
Sobolev approximation is supplied by the complete trace package,
including the bounded right inverse
\(\mathcal J_A^{m,p}\),
rather than by continuity of the smooth conical construction.

\subsection{Full formal data and finite visible traces}
\label{subsec:full-formal-visible-traces}

The formal trace space is defined from
\(\mathfrak E_A^\infty = \bigcap_{k\geq0} \ker\mathfrak d_A^k\),
not from one finite transported-jet equaliser.

Let \(r := \max_{Z\in\mathscr Z_A} r_Z(m,p)\) when at least one visible trace exists. Even though the trace map
contains no normal derivative above order \(r\), one cannot generally
replace \(\mathfrak E_A^\infty\) by
\(\mathfrak E_A^r = \ker\mathfrak d_A^r\).
Higher transported-jet equations may impose conditions on coefficients
already represented in the finite trace target. Equivalently,
projection of the full formal equaliser to visible order need not equal
the naive equaliser obtained by truncating the compatibility equations
at that order.

Finite Sobolev visibility need not imply finite formal determination; the next section treats a controlled diagonal transport class.

\section{Visibility, explicit closures, and form domains}
\label{sec:visibility}

Finite Sobolev closure is governed by the trace threshold, with surviving traces constrained to the transport equaliser.

Throughout this section, the source carriers satisfy
Hypothesis~\ref{hyp:embedded-source-carriers}, and every invocation of
the abstract Sobolev closure theorem assumes
Hypothesis~\ref{hyp:complete-incidence-trace}.

\subsection{The noncritical visibility threshold}
\label{subsec:visibility-threshold}

Let \(Z\subseteq M_1\) be a closed embedded source carrier of codimension \(c\). Fix
\(m\in\mathbb N, \qquad 1<p<\infty\).
An order-\(j\) normal derivative has a bounded \(W^{m,p}\)-trace in
the noncritical range
\(j<m-\frac{c}{p}\).
The trace takes values in
\(B_{p,p}^{\,m-j-c/p} \left( Z; E_1|_Z\otimes\operatorname{Sym}^jN_Z^* \right)\)
\citep{Triebel1992,JonssonWallin1984,GrosseSchneider2013}.

\begin{proposition}[Closedness of visible constraints]
\label{prop:visibility-closed}
Let
\[
  j<m-\frac{c}{p},
\]
let
\[
  F\longrightarrow Z
\]
be a smooth vector bundle, and let
\[
  L:
  E_1|_Z\otimes\operatorname{Sym}^jN_Z^*
  \longrightarrow
  F
\]
be a smooth vector-bundle morphism. Assume that fibrewise application
of \(L\) induces a bounded operator
\[
  L_*:
  B_{p,p}^{\,m-j-c/p}
  \left(
    Z;
    E_1|_Z\otimes\operatorname{Sym}^jN_Z^*
  \right)
  \longrightarrow
  B_{p,p}^{\,m-j-c/p}(Z;F).
\]
Then
\[
  \left\{
    u\in W^{m,p}(M_1,E_1):
    L\gamma_Z^{(j)}u=0
  \right\}
\]
is a closed linear subspace of \(W^{m,p}(M_1,E_1)\).

More generally, the common kernel of any finite system of bounded
linear equations involving visible incidence traces is closed.
\end{proposition}

\begin{proof}
The composition \(L_*\gamma_Z^{(j)}\) is bounded. Its kernel is therefore closed. A finite family of such
conditions is the kernel of the corresponding finite direct-sum map.
\end{proof}

\begin{remark}
The bounded-multiplier assumption is automatic when \(Z\) is compact.
For noncompact carriers it follows, for example, from the usual
uniform multiplier hypotheses in a bounded-geometry Besov calculus.
Smoothness of \(L\) alone does not imply boundedness on a global Besov
space.
\end{remark}

Every transported compatibility equation that can be written
entirely in terms of bounded visible trace coordinates survives
Sobolev completion.

The opposite noncritical range is locally erasable. In tubular
coordinates
\((y,z) \in \mathbb R^{n-c}\times\mathbb R^c, \qquad Z=\{z=0\}\),
let
\[
  a
  \in
  C_c^\infty
  \left(
    \mathbb R^{n-c};
    F\otimes
    \operatorname{Sym}^j(\mathbb R^c)^*
  \right)
\]
be a prescribed order-\(j\) normal trace coefficient. The construction
of Appendix~\ref{app:sobolev-perturbations} gives smooth compactly
supported sections \(v_\varepsilon\) satisfying
\(\gamma_Z^{(\ell)}v_\varepsilon=0 \qquad (0\leq\ell<j)\),
\(\gamma_Z^{(j)}v_\varepsilon=a\),
and
\(\|v_\varepsilon\|_{W^{m,p}} \leq C \varepsilon^{\,j-m+c/p}\).

\begin{theorem}[Noncritical visibility and erasure]
\label{thm:visibility-erasure}
Let an incidence condition contain a pure order-\(j\) normal derivative
constraint along a codimension-\(c\) carrier.

\begin{enumerate}[label=\textup{(\roman*)}]
  \item If
        \[
          j<m-\frac{c}{p},
        \]
        the trace is bounded on \(W^{m,p}\), and every bounded linear
        constraint on that trace is closed.

  \item If
        \[
          j>m-\frac{c}{p},
        \]
        and the order-\(j\) coefficient can be varied independently of
        all lower visible trace data, then the constraint is locally
        erasable by smooth perturbations converging to zero in
        \(W^{m,p}\).
\end{enumerate}
\end{theorem}

\begin{proof}
Part \textup{(i)} is
Proposition~\ref{prop:visibility-closed}.
For part \textup{(ii)}, the inserted sections satisfy
\(\|v_\varepsilon\|_{W^{m,p}} \leq C \varepsilon^{\,j-m+c/p} \longrightarrow0\).
Their lower normal traces vanish, while the prescribed order-\(j\)
trace remains fixed. Localisation and bundle trivialisation give the
manifold statement. See
Lemma~\ref{lem:general-jet-insertion-scaling}
and
Corollary~\ref{cor:supercritical-erasure}.
\end{proof}

No universal assertion is made at
\(j=m-\frac{c}{p}\).
At the critical index, continuity and removability depend on the
precise trace scale, the geometry of the carrier, and capacity
\citep{Mazya2011Sobolev,AdamsHedberg1996}.

\subsection{A single diagonal full-jet transport carrier}
\label{subsec:constant-transport}

Assume now that the source incidence locus consists of one closed
embedded carrier \(Z\subseteq M_1\) of codimension \(c\), and that no carrier-intersection conditions are
present. Assume also that the complete trace target is the full
visible normal-trace product
\[
  \mathcal T_{A,Z}^{m,p}
  =
  \bigoplus_{j=0}^{r_Z}
  B_{p,p}^{\,m-j-c/p}
  \left(
    Z;
    E_1|_Z\otimes
    \operatorname{Sym}^jN_Z^*
  \right),
\]
where
\(r_Z=r_Z(m,p)\).
In this single-carrier setting, write
\(\mathcal K_{A,Z}^{m,p} := \mathcal K_A^{m,p}\).

Fix a tubular neighbourhood \(\mathcal U_Z\subseteq M_1\) of \(Z\) and a bundle trivialisation \(E_1|_{\mathcal U_Z} \cong \mathcal U_Z\times F\) for a finite-dimensional real or complex vector space \(F\). Suppose
that the zeroth-order branch transports are represented in this
trivialisation by fixed invertible maps
\(A_a\in\operatorname{GL}(F), \qquad a\in A_Z\).
Define
\[
  F_{\mathrm{eq}}
  :=
  \left\{
    v\in F:
    A_av=A_bv
    \text{ for all }a,b\in A_Z
  \right\},
\]
and let
\[
  \widetilde F_{\mathrm{eq}}
  :=
  \mathcal U_Z\times F_{\mathrm{eq}}
  \subseteq
  E_1|_{\mathcal U_Z}.
\]

Constancy of the maps \(A_a\) at order zero is not sufficient for the
formulas below. Higher transported jets may also involve derivatives
of the base reparametrisations, derivatives of the fibre transition,
and lower-order mixing. We therefore impose the complete condition
actually used.

\begin{hypothesis}[Diagonal full-jet transport]
\label{hyp:diagonal-full-jet-transport}
The tubular trivialisation and the connection defining the normal
trace system satisfy the following conditions.

\begin{enumerate}[label=\textup{(\roman*)}]
  \item The subbundle
        \[
          \widetilde F_{\mathrm{eq}}
          \subseteq
          E_1|_{\mathcal U_Z}
        \]
        is preserved by the connection.

  \item For every \(q\in\mathbb N_0\) and every smooth section \(u\),
        the order-\(q\) transported compatibility equations along
        \(Z\) are equivalent to
        \[
          \gamma_Z^{(j)}u
          \in
          F_{\mathrm{eq}}
          \otimes
          \operatorname{Sym}^jN_Z^*
          \qquad
          (0\leq j\leq q).
        \]
\end{enumerate}

Equivalently, under the chosen jet splitting, every branch acts on
normal-derivative level \(j\) by
\[
  A_a\otimes\operatorname{id}_{\operatorname{Sym}^jN_Z^*},
\]
with no coupling between distinct normal orders.
\end{hypothesis}

This hypothesis holds, for example, when the prolonged base maps agree
to all orders along \(Z\), the fibre transports differ by constant
matrices \(A_a\), and the chosen tubular trivialisation and connection
make these matrices parallel. It is strictly stronger than constancy
of the limiting fibre maps alone.

The reverse trace-space inclusion also requires a density statement
for the chosen global Besov model.

\begin{hypothesis}[Smooth equaliser-trace density]
\label{hyp:smooth-equaliser-trace-density}
For every
\[
  0\leq j\leq r_Z,
\]
the space
\[
  \Gamma_c^\infty
  \left(
    Z;
    F_{\mathrm{eq}}
    \otimes
    \operatorname{Sym}^jN_Z^*
  \right)
\]
is dense in
\[
  B_{p,p}^{\,m-j-c/p}
  \left(
    Z;
    F_{\mathrm{eq}}
    \otimes
    \operatorname{Sym}^jN_Z^*
  \right).
\]
\end{hypothesis}

The density hypothesis is automatic for isolated carriers, compact carriers with standard Besov spaces, and the usual bounded-geometry inhomogeneous Besov models \citep{Triebel1992,GrosseSchneider2013}.

\begin{lemma}[Finite normal-trace realisation in the equaliser]
\label{lem:equaliser-normal-trace-realisation}
Assume
Hypothesis~\ref{hyp:diagonal-full-jet-transport}. Let
\[
  r\in\mathbb N_0
\]
and let
\[
  \tau_j
  \in
  \Gamma_c^\infty
  \left(
    Z;
    F_{\mathrm{eq}}
    \otimes
    \operatorname{Sym}^jN_Z^*
  \right),
  \qquad
  0\leq j\leq r.
\]
Then there exists
\[
  u\in\Gamma_c^\infty(M_1,E_1)
\]
such that
\[
  \gamma_Z^{(j)}u=\tau_j
  \qquad
  (0\leq j\leq r),
\]
and such that \(u\) takes values in
\(\widetilde F_{\mathrm{eq}}\) on a neighbourhood of \(Z\). In
particular,
\[
  u\in\mathfrak E_A^\infty.
\]
\end{lemma}

\begin{proof}
In the fixed tubular trivialisation choose a cutoff equal to one near the zero section and, for coefficient fields \(b_j\) in \(F_{\rm eq}\), form
\[
 P_b(\tau(p,v))=\chi(p,v)\sum_{j=0}^r\frac1{j!}b_j(p)(v^{\otimes j}).
\]
Its normal covariant traces are triangular:
\(\gamma_Z^{(j)}P_b=b_j+Q_j(b_0,\ldots,b_{j-1})\), with each \(Q_j\) still \(F_{\rm eq}\)-valued because the connection preserves that subbundle. Recursively set \(b_0=\tau_0\) and \(b_j=\tau_j-Q_j(b_0,\ldots,b_{j-1})\). Then \(\gamma_Z^{(j)}P_b=\tau_j\) for \(j\le r\). The cutoff makes the section compactly supported and extendable by zero. Near \(Z\) it remains \(F_{\rm eq}\)-valued, so Hypothesis~\ref{hyp:diagonal-full-jet-transport} gives all transported-jet compatibility equations; hence \(u\in\mathfrak E_A^\infty\).
\end{proof}

\begin{proposition}[Formal trace space for diagonal transport]
\label{prop:constant-transport-traces}
Assume
Hypotheses~\ref{hyp:diagonal-full-jet-transport}
and \ref{hyp:smooth-equaliser-trace-density}. Then
\[
  \boxed{
  \mathcal K_{A,Z}^{m,p}
  =
  \bigoplus_{j=0}^{r_Z}
  B_{p,p}^{\,m-j-c/p}
  \left(
    Z;
    F_{\mathrm{eq}}
    \otimes
    \operatorname{Sym}^jN_Z^*
  \right).
  }
\]
\end{proposition}

\begin{proof}
If \(u\in\mathfrak E_A^{\infty;m,p}\), Hypothesis~\ref{hyp:diagonal-full-jet-transport} places every visible trace \(\gamma_Z^{(j)}u\) in \(F_{\rm eq}\otimes\operatorname{Sym}^jN_Z^*\); closedness of the displayed Besov sum gives one inclusion. Conversely, approximate any tuple in that sum by compactly supported smooth \(F_{\rm eq}\)-valued trace fields using Hypothesis~\ref{hyp:smooth-equaliser-trace-density}. Lemma~\ref{lem:equaliser-normal-trace-realisation} realises each approximating tuple by an element of \(\mathfrak E_A^\infty\cap C_c^\infty\), proving the reverse inclusion after closure.
\end{proof}

\begin{theorem}[Diagonal-transport Sobolev closure]
\label{thm:constant-transport-closure}
Assume:

\begin{enumerate}[label=\textup{(\roman*)}]
  \item Hypothesis~\ref{hyp:complete-incidence-trace};

  \item the finite clean conically tame hypotheses;

  \item Hypothesis~\ref{hyp:diagonal-full-jet-transport};

  \item Hypothesis~\ref{hyp:smooth-equaliser-trace-density}.
\end{enumerate}

Then
\[
  \boxed{
  \overline{
    \mathfrak R_A^{\infty;m,p}
  }^{\,W^{m,p}(M_1,E_1)}
  =
  \left\{
    u\in W^{m,p}(M_1,E_1):
    \gamma_Z^{(j)}u
    \in
    F_{\mathrm{eq}}
    \otimes
    \operatorname{Sym}^jN_Z^*
    \text{ for }0\leq j\leq r_Z
  \right\}.
  }
\]
\end{theorem}

\begin{proof}
Conically tame smooth realisation gives
\(\mathfrak R_A^\infty = \mathfrak E_A^\infty\).
Apply
Theorem~\ref{thm:conically-tame-sobolev-closure}
and
Proposition~\ref{prop:constant-transport-traces}.
\end{proof}

Under these hypotheses, the analytic defect quotient is
\[
  W^{m,p}(M_1,E_1)
  \Big/
  \overline{
    \mathfrak R_A^{\infty;m,p}
  }^{\,W^{m,p}}
  \cong
  \bigoplus_{j=0}^{r_Z}
  B_{p,p}^{\,m-j-c/p}
  \left(
    Z;
    (F/F_{\mathrm{eq}})
    \otimes
    \operatorname{Sym}^jN_Z^*
  \right).
\]
Indeed, choose a linear complement of \(F_{\mathrm{eq}}\) in \(F\).
The resulting bounded fibre projection identifies the quotient of each
\(F\)-valued Besov trace space by its closed
\(F_{\mathrm{eq}}\)-valued subspace with the corresponding
\(F/F_{\mathrm{eq}}\)-valued Besov space.

For positive-dimensional \(Z\) the quotient is generally infinite-dimensional; for an isolated point it is a finite sum of copies of \(F/F_{\rm eq}\).

\subsection{The relative-sign closure}
\label{subsec:relative-sign-closure}

Return to the relative-sign line bundle of
Section~\ref{subsec:relative-sign-bundle}. Here
\(M_1=\mathbb R, \qquad Z=\{0\}, \qquad c=1\),
and the two branch transports are
\(A_-=1, \qquad A_+=-1\).
They act by the same constants on every derivative order, so
Hypothesis~\ref{hyp:diagonal-full-jet-transport} holds and
\(F_{\mathrm{eq}} = \ker(A_--A_+) = \{0\}\).
Since \(Z\) is an isolated point,
Hypothesis~\ref{hyp:smooth-equaliser-trace-density} is automatic.

For \(0\leq j\leq m-1\) and \(p>1\),
\(j \leq m-1 < m-\frac{1}{p}\),
so all endpoint derivatives through order \(m-1\) are visible.

\begin{corollary}[Sobolev closure of the flat ideal]
\label{cor:flat-ideal-sobolev-closure}
For every
\[
  m\in\mathbb N,
  \qquad
  1<p<\infty,
\]
\[
  \boxed{
  \overline{
    \mathcal F_0
    \cap
    W^{m,p}(\mathbb R)
  }^{\,W^{m,p}(\mathbb R)}
  =
  \left\{
    u\in W^{m,p}(\mathbb R):
    u^{(j)}(0)=0,\
    0\leq j\leq m-1
  \right\}.
  }
\]
Under restriction to the two half-lines, this space is canonically
isomorphic to
\[
  W_0^{m,p}((-\infty,0))
  \oplus
  W_0^{m,p}((0,\infty)),
\]
where \(W_0^{m,p}\) denotes the closure of compactly supported smooth
functions and therefore has vanishing endpoint traces through order
\(m-1\).
\end{corollary}

\begin{proof}
Apply
Theorem~\ref{thm:constant-transport-closure}
with
\(F_{\mathrm{eq}}=\{0\}\).
The half-line description follows from the standard trace
characterisation of \(W_0^{m,p}\).
\end{proof}

The infinite-order smooth condition \(u\in\mathcal F_0\) collapses under \(W^{m,p}\)-completion to the finite family of endpoint
traces visible at order \(m\).

\subsection{A partial fibre defect}
\label{subsec:partial-fibre-defect}

Let
\(F=\mathbb R^2\),
and take
\[
  A_0=I,
  \qquad
  A_1
  =
  \begin{pmatrix}
    1&0\\
    0&-1
  \end{pmatrix}.
\]
Assume that these matrices act diagonally on the full normal-jet tower
as in
Hypothesis~\ref{hyp:diagonal-full-jet-transport}. Then
\(F_{\mathrm{eq}} = \operatorname{span}\{e_+\}, \qquad e_+=(1,0)\).

Write
\[
  u=
  \begin{pmatrix}
    u_+\\
    u_-
  \end{pmatrix}.
\]
The first fibre direction is invariant, while the second is
anti-invariant.

\begin{corollary}[Partial Sobolev defect]
\label{cor:partial-sobolev-defect}
Under the hypotheses of
Theorem~\ref{thm:constant-transport-closure},
\[
  \boxed{
  \overline{
    \mathfrak R_A^{\infty;m,p}
  }^{\,W^{m,p}}
  =
  \left\{
    (u_+,u_-)
    \in
    W^{m,p}(M_1,\mathbb R^2):
    \gamma_Z^{(j)}u_-=0,\
    0\leq j\leq r_Z
  \right\}.
  }
\]
No trace condition is imposed on \(u_+\).
\end{corollary}

This separates the compatible and defective fibre sectors. Smooth
compatibility requires flatness only in the anti-invariant direction,
while finite Sobolev completion retains only the visible traces of that
direction.

\subsection{Sobolev erasure of the exponential Whitney obstruction}
\label{subsec:exponential-sobolev-erasure}

Return to the exponentially tangent branch system of
Section~\ref{subsec:exponentially-tangent-branches}. Multiply the
section constructed there by a compactly supported smooth cutoff equal
to one near the multiple-incidence point, and retain the notation
\(s_1\). Then
\(s_1\in \mathfrak E_A^\infty \setminus \mathfrak R_A^\infty\).

Choose \(\chi\in C^\infty(\mathbb R,[0,1])\) such that
\(\chi(t)=0 \quad (t\leq1), \qquad \chi(t)=1 \quad (t\geq2)\),
and set
\(\chi_\varepsilon(x) := \chi\!\left(\frac{x}{\varepsilon}\right)\).
On the second source component, replace \(\rho(x)\) by
\(\chi_\varepsilon(x)\rho(x)\),
and leave the first component equal to zero. Denote the resulting
section by
\(s_{1,\varepsilon}\).

The two branch values now agree on a neighbourhood of the only
multiple-incidence point. At all other incidence points the assembled
field is locally the jet of a single prolonged branch section.
\(s_{1,\varepsilon} \in \mathfrak R_A^\infty\) by
Theorem~\ref{thm:exact-smooth-whitney}.

\begin{theorem}[Erasure of the exponential Whitney obstruction]
\label{thm:exponential-sobolev-erasure}
For every finite
\[
  m\in\mathbb N_0
\]
and every
\[
  1\leq p<\infty,
\]
\[
  s_{1,\varepsilon}
  \longrightarrow
  s_1
  \qquad
  \text{in }W^{m,p}(M_1,E_1).
\]
Hence
\[
  \boxed{
  s_1
  \in
  \overline{
    \mathfrak R_A^\infty
    \cap
    W^{m,p}(M_1,E_1)
  }^{\,W^{m,p}(M_1,E_1)}
  }
\]
for every finite Sobolev order.
\end{theorem}

\begin{proof}
The difference vanishes on the first source component. On the second
component it has the form
\(\eta(x,y) \left( 1-\chi_\varepsilon(x) \right) \rho(x)\),
where \(\eta\) is smooth and compactly supported. Since \(\rho\) is
flat at \(x=0\),
Lemma~\ref{lem:one-sided-flat-cutoff-estimate}
gives convergence in \(W^{m,p}\).
\end{proof}

This smooth Whitney obstruction is erased at every finite Sobolev order; no such claim is made for arbitrary residual obstructions.

For fixed \(m\) and \(p\), there is a natural map
\[
  \mathbf q_A^{m,p}:
  \frac{
    \mathfrak E_A^{\infty;m,p}
  }{
    \mathfrak R_A^{\infty;m,p}
  }
  \longrightarrow
  \frac{
    W^{m,p}(M_1,E_1)
  }{
    \overline{
      \mathfrak R_A^{\infty;m,p}
    }^{\,W^{m,p}}
  },
  \qquad
  [s]\longmapsto[s].
\]

\begin{corollary}[Noninjectivity of Sobolev completion]
\label{cor:noninjective-sobolev-completion}
For the exponentially tangent system,
\[
  [s_1]\neq0
\]
in the smooth residual quotient, whereas
\[
  \mathbf q_A^{m,p}([s_1])=0
\]
for every finite \(m\) and \(1<p<\infty\). Hence
\[
  \mathbf q_A^{m,p}
\]
is not injective.
\end{corollary}

Every derivative of \(s_1\) vanishes at the source
multiple-incidence carriers. Therefore \(\operatorname{Tr}_A^{m,p}s_1=0\) for every finite visible trace system. The obstruction is
cross-branch and Whitney-theoretic, not a nonzero finite incidence
trace.

\subsection{Quadratic-form closures}
\label{subsec:quadratic-form-closures}

Set \(p=2\) and let
\(\mathcal H=L^2(M_1,E_1)\).
Let \(D(\mathfrak q) = \mathfrak R_A^{\infty;m,2}\) and let \(\mathfrak q: D(\mathfrak q)\times D(\mathfrak q) \longrightarrow \mathbb C\) be a densely defined symmetric sesquilinear form. In the real case,
take a symmetric bilinear form. Assume that there exists \(\beta\geq0\) such that
\[
  \mathfrak q[u]
  :=
  \mathfrak q(u,u)
  \geq
  -\beta\|u\|_{\mathcal H}^2
  \qquad
  (u\in D(\mathfrak q)).
\]

Choose \(\lambda>\beta\) and define
\(\|u\|_{\mathfrak q,\lambda}^2 := \mathfrak q[u] + \lambda\|u\|_{\mathcal H}^2\).
Assume that this norm is equivalent on \(D(\mathfrak q)\) to the
\(H^m\)-norm:
\(c\|u\|_{H^m} \leq \|u\|_{\mathfrak q,\lambda} \leq C\|u\|_{H^m}\).

\begin{proposition}[Closability]
\label{prop:form-closability}
Under the norm-equivalence hypothesis, \(\mathfrak q\) is closable in
\(\mathcal H\).
\end{proposition}

\begin{proof}
Let \((u_n)\subseteq D(\mathfrak q)\) be Cauchy in the shifted form
norm and suppose that
\(u_n\longrightarrow0 \qquad \text{in }\mathcal H\).
Norm equivalence implies that \((u_n)\) is Cauchy in \(H^m\). Hence \(u_n\longrightarrow u \qquad \text{in }H^m\) for some \(u\in H^m(M_1,E_1)\).

The continuous embedding \(H^m(M_1,E_1) \hookrightarrow L^2(M_1,E_1)\) shows that \(u_n\to u\) in \(\mathcal H\). Since the
\(\mathcal H\)-limit is zero,
\(u=0\).
Therefore
\(u_n\longrightarrow0 \qquad \text{in }H^m\),
and norm equivalence gives
\(\|u_n\|_{\mathfrak q,\lambda} \longrightarrow0\).
This is the closability criterion.
\end{proof}

\begin{theorem}[Closed form domain]
\label{thm:closed-form-domain}
Under the preceding hypotheses,
\[
  \boxed{
  D(\overline{\mathfrak q})
  =
  \overline{
    \mathfrak R_A^{\infty;m,2}
  }^{\,H^m(M_1,E_1)}.
  }
\]
Under
Hypotheses~\ref{hyp:embedded-source-carriers}
and \ref{hyp:complete-incidence-trace},
\[
  \boxed{
  D(\overline{\mathfrak q})
  =
  \left(
    \operatorname{Tr}_A^{m,2}
  \right)^{-1}
  \left(
    \mathcal X_A^{m,2}
  \right).
  }
\]
If the adjunction is finite clean conically tame, then
\[
  \boxed{
  D(\overline{\mathfrak q})
  =
  \left(
    \operatorname{Tr}_A^{m,2}
  \right)^{-1}
  \left(
    \mathcal K_A^{m,2}
  \right).
  }
\]
\end{theorem}

\begin{proof}
By
Proposition~\ref{prop:form-closability},
the closed form domain is the completion of \(D(\mathfrak q)\) in
the shifted form norm. Since that norm is equivalent to the
\(H^m\)-norm, this completion is canonically the \(H^m\)-closure of
\(\mathfrak R_A^{\infty;m,2}\).
The remaining identities follow from
Theorems~\ref{thm:exact-sobolev-closure}
and
\ref{thm:conically-tame-sobolev-closure}
with \(p=2\).
\end{proof}

The result concerns only the closed form domain; the associated operator still depends on the differential expression, maximal domain, and Green pairing.

\section{Scalar closure on finite graph resolutions}
\label{sec:graph-resolutions}

For finite metric-graph Hausdorff reflections, scalar closure can be proved directly from the graph resolution: finite resolution and endpoint nondegeneracy reduce the problem to endpoint jets.

\subsection{Finite smooth graph resolutions}
\label{subsec:finite-smooth-graph-resolutions}

Let \(\Gamma\) be a connected finite metric graph, allowing compact
edges, external half-lines, parallel edges, and loops
\citep{BerkolaikoKuchment2013}. Denote its vertex and edge sets by
\(V(\Gamma), \qquad E(\Gamma)\).
Each edge \(e\) is represented by an interval
\(I_e=(0,\ell_e), \qquad 0<\ell_e\leq\infty\).

Let \(\mathcal E(\Gamma)\) be the finite set of edge ends. A compact edge contributes two edge
ends, including two distinct end labels when it is a loop; an external
half-line contributes one finite edge end. For
\(\epsilon\in\mathcal E(\Gamma)\),
write \(e(\epsilon)\in E(\Gamma), \qquad v(\epsilon)\in V(\Gamma)\) for the corresponding edge and incident vertex. Set
\[
  \mathcal E(v)
  :=
  \left\{
    \epsilon\in\mathcal E(\Gamma):
    v(\epsilon)=v
  \right\}.
\]

Each edge end carries an outward coordinate
\(t_\epsilon\geq0\).
At the left end of an interval this is the ordinary coordinate; at the
right end it is the distance from the right endpoint.

\begin{definition}[Finite smooth graph resolution]
\label{def:finite-smooth-graph-resolution}
A finite smooth graph resolution is a triple
\[
  (M,\Gamma,\pi),
  \qquad
  \pi:M\longrightarrow\Gamma,
\]
satisfying the following conditions.

\begin{enumerate}[label=\textup{(\roman*)}]
  \item \(M\) is a second-countable \(T_1\) smooth one-manifold
        without boundary, not assumed Hausdorff.

  \item \(\pi\) is the Hausdorff reflection of \(M\).

  \item The exceptional set
        \[
          B_\pi:=\pi^{-1}(V(\Gamma))
        \]
        is finite, and
        \[
          \pi:
          M\setminus B_\pi
          \longrightarrow
          \Gamma\setminus V(\Gamma)
        \]
        is a diffeomorphism.

  \item The map \(\pi\) is proper in the sense that
        \[
          \pi^{-1}(K)
        \]
        is compact for every compact
        \(K\subseteq\Gamma\).

  \item For every \(p\in B_\pi\), there exist a centred chart
        \[
          \phi_p:
          U_p\longrightarrow(-\delta_p,\delta_p),
          \qquad
          \phi_p(p)=0,
        \]
        and two distinct edge ends
        \[
          \epsilon_{p,+},
          \epsilon_{p,-}
          \in
          \mathcal E(\pi(p))
        \]
        such that the two components of \(U_p\setminus\{p\}\) are
        mapped by \(\pi\) diffeomorphically onto punctured collars of
        those edge ends.

        For \(\sigma\in\{+1,-1\}\), define
        \[
          \kappa_{p,\sigma}(x)
          :=
          t_{\epsilon_{p,\sigma}}
          \left(
            \pi\bigl(\phi_p^{-1}(\sigma x)\bigr)
          \right),
          \qquad
          0<x<\delta_p.
        \]
        Each \(\kappa_{p,\sigma}\) extends smoothly to \(x=0\) and
        satisfies
        \[
          \kappa_{p,\sigma}(0)=0,
          \qquad
          \kappa_{p,\sigma}'(0)>0.
        \]

  \item Every edge end is represented by at least one local side at
        an exceptional point.
\end{enumerate}
\end{definition}

The Hausdorff-reflection condition is natural for the
one-dimensional quotients considered here
\citep{HaefligerReeb1957,OConnell2023Adjunction,VlasenkoMaksymenko2026}.
The nondegeneracy condition \(\kappa_{p,\sigma}'(0)>0\) implies, by the inverse-function theorem for one-sided germs, that
\(\kappa_{p,\sigma}\) has a smooth inverse half-germ. Smoothness in the manifold coordinate is equivalent to ordinary
edgewise smoothness up to the corresponding finite edge end.

The bare resolution \((M,\Gamma,\pi)\) does not contain the multiplicity
weights used later in the quantum theory. Those weights arise from a labelled
Hausdorff presentation of \(M\), on which distinct resolving representatives
are counted separately even when they project to the same regular point of
\(M\).

\begin{definition}[Marked analytic graph resolution]
\label{def:marked-analytic-graph-resolution}
Let \((M,\Gamma,\pi)\) be a finite smooth graph resolution. A \emph{marked
analytic graph resolution} consists additionally of the following data.

\begin{enumerate}[label=\textup{(\roman*)}]
  \item A finite index set \(A\) and a disjoint union of Hausdorff smooth
        one-manifolds
        \[
          \widetilde M
          :=
          \bigsqcup_{\alpha\in A}M_\alpha,
        \]
        together with a surjective quotient map
        \[
          q:\widetilde M\longrightarrow M
        \]
        such that every restriction \(q_\alpha:=q|_{M_\alpha}\) is a smooth
        local diffeomorphism onto its image. The labels \(\alpha\) are retained
        as part of the mark.

  \item The metric structure on \(\Gamma\), together with smooth positive
        Radon measures \(\mu_\alpha\) on the resolving sheets \(M_\alpha\).
        Writing
        \[
          \widetilde\mu
          :=
          \bigoplus_{\alpha\in A}\mu_\alpha,
        \]
        the induced graph measure is
        \[
          \mu_\Gamma
          :=
          (\pi\circ q)_*\widetilde\mu.
        \]
        We require \(\mu_\Gamma\) to be locally finite and absolutely
        continuous with respect to edge arclength on every open graph edge.

  \item When constant-edge models are used, the Radon--Nikodym derivative of
        \(\mu_\Gamma\) is constant on each edge \(e\); its positive value is
        denoted by \(a_e\). More general variable densities are permitted by
        the definition but are not used in the flat models of Part~\ref{part:quantum-dynamics}.
\end{enumerate}

The measure \(\widetilde\mu\) is a measure on the labelled resolving
presentation, not an ordinary density on the quotient manifold \(M\). Two
labels may contribute separately to \(\mu_\Gamma\) even when their images in
\(M\) agree on an open set.
\end{definition}

\begin{definition}[Isomorphism and projective equivalence of marked resolutions]
\label{def:marked-resolution-isomorphism}
Let \(\mathcal R\) and \(\mathcal R'\) be marked analytic graph resolutions.
A \emph{marked isomorphism} consists of a diffeomorphism \(F:M\to M'\), a
metric-graph isomorphism \(\Psi:\Gamma\to\Gamma'\), and a diffeomorphism
\(\widetilde F:\widetilde M\to\widetilde M'\) carrying resolving sheets to
resolving sheets up to a permutation of their labels, such that
\[
  \Psi\circ\pi=\pi'\circ F,
  \qquad
  F\circ q=q'\circ\widetilde F,
  \qquad
  \widetilde F_*\widetilde\mu=\widetilde\mu'.
\]
It is a \emph{projective marked isomorphism} if the last equality is replaced
by \(\widetilde F_*\widetilde\mu=c\widetilde\mu'\) for one constant \(c>0\).
\end{definition}

\begin{proposition}[Invariance of the pushed-forward marked measure]
\label{prop:marked-measure-invariance}
Under a marked isomorphism,
\[
  \Psi_*\mu_\Gamma=\mu_{\Gamma'}.
\]
Under a projective marked isomorphism,
\[
  \Psi_*\mu_\Gamma=c\,\mu_{\Gamma'}.
\]
In the constant-edge case, the relative weight vector \([\mathbf a]\) is
therefore an invariant of the projective marked-isomorphism class.
\end{proposition}

\begin{proof}
Functoriality of pushforward gives
\[
  \Psi_*\mu_\Gamma
  =
  \Psi_*(\pi\circ q)_*\widetilde\mu
  =
  (\pi'\circ q'\circ\widetilde F)_*\widetilde\mu.
\]
The asserted identities follow from preservation, or common positive
rescaling, of \(\widetilde\mu\). On each open graph edge this identity is an
identity of Radon--Nikodym densities, so a common rescaling changes all edge
weights by the same factor and preserves their projective class.
\end{proof}

Bundle transport, connections, potentials, and quadratic forms are not part
of the bare triple \((M,\Gamma,\pi)\). When they are used below they are
additional quantum data placed over a specified marked analytic resolution.

\subsection{Endpoint transports and formal equalisers}
\label{subsec:graph-formal-equaliser}

For \(r\in\mathbb N_0\), let \(\mathbb J^r:=\mathbb C^{r+1}, \qquad \mathbb J^\infty:=\mathbb C^{\mathbb N_0}\) be the truncated and formal scalar jet spaces.

For a function smooth up to an edge end \(\epsilon\), write
\[
  J_\epsilon^r u
  :=
  \left(
    \partial_{t_\epsilon}^{\,j}
    u_{e(\epsilon)}(0)
  \right)_{j=0}^r,
\]
and
\[
  J_\epsilon^\infty u
  :=
  \left(
    \partial_{t_\epsilon}^{\,j}
    u_{e(\epsilon)}(0)
  \right)_{j\geq0}.
\]

Composition with a smooth nondegenerate half-germ \(\kappa:[0,\delta)\longrightarrow[0,\eta)\) induces an invertible triangular jet map
\[
  C_\kappa^r:
  \mathbb J^r\longrightarrow\mathbb J^r,
  \qquad
  C_\kappa^r(j_0^r u)
  :=
  j_0^r(u\circ\kappa).
\]
Its diagonal entries are
\(1, \kappa'(0), \ldots, \kappa'(0)^r\).
The same construction gives an automorphism \(C_\kappa^\infty\) of the formal jet space.

For
\(\sigma\in\{+1,-1\}\),
define \(S_\sigma^r(\xi_0,\ldots,\xi_r) := \left( \sigma^j\xi_j \right)_{j=0}^r\) and set
\(T_{p,\sigma}^r := S_\sigma^r C_{\kappa_{p,\sigma}}^r\).
Likewise, define
\(T_{p,\sigma}^\infty := S_\sigma^\infty C_{\kappa_{p,\sigma}}^\infty\).
The sign operator records that the two local sides are parameterised by
the signed manifold coordinates \(x\) and \(-x\).

Define the global endpoint-jet spaces
\[
  \mathbb J_\Gamma^r
  :=
  \prod_{\epsilon\in\mathcal E(\Gamma)}
  \mathbb J^r,
  \qquad
  \mathbb J_\Gamma^\infty
  :=
  \prod_{\epsilon\in\mathcal E(\Gamma)}
  \mathbb J^\infty.
\]

\begin{definition}[Truncated endpoint equaliser]
\label{def:truncated-graph-equaliser}
The order-\(r\) endpoint equaliser
\[
  \mathfrak E_\pi^r
  \subseteq
  \mathbb J_\Gamma^r
\]
consists of the families
\[
  \xi
  =
  (\xi_\epsilon)_{\epsilon\in\mathcal E(\Gamma)}
\]
satisfying:

\begin{enumerate}[label=\textup{(\alph*)}]
  \item for every \(v\in V(\Gamma)\),
        \[
          (\xi_\epsilon)_0
          =
          (\xi_{\epsilon'})_0
          \qquad
          \text{for all }
          \epsilon,\epsilon'\in\mathcal E(v);
        \]

  \item for every \(p\in B_\pi\),
        \[
          T_{p,+}^r\xi_{\epsilon_{p,+}}
          =
          T_{p,-}^r\xi_{\epsilon_{p,-}}.
        \]
\end{enumerate}
\end{definition}

Condition \textup{(a)} records continuity on the Hausdorff graph. It is
included separately because the exceptional points over a vertex need
not connect all incident edge ends into one transport component.

The full formal equaliser \(\mathfrak E_\pi^\infty \subseteq \mathbb J_\Gamma^\infty\) is defined by the corresponding relations at every finite order. Let
\[
  \operatorname{pr}_r:
  \mathbb J_\Gamma^\infty
  \longrightarrow
  \mathbb J_\Gamma^r
\]
denote truncation.

\begin{definition}[Prolongable endpoint jets]
\label{def:prolongable-endpoint-jets}
The order-\(r\) prolongation space is
\[
  \boxed{
  \mathfrak J_\pi^r
  :=
  \operatorname{pr}_r
  \left(
    \mathfrak E_\pi^\infty
  \right).
  }
\]
The resolution is \(r\)-prolongation complete when
\[
  \mathfrak J_\pi^r
  =
  \mathfrak E_\pi^r,
\]
and formally prolongation complete when this holds for every
\(r\in\mathbb N_0\).
\end{definition}

One always has
\(\mathfrak J_\pi^r \subseteq \mathfrak E_\pi^r\).
The inclusion may be strict because higher-order transport equations
can impose conditions on coefficients already present in a lower
truncation.

At order zero, no such distinction occurs.

\begin{proposition}[Zeroth-order prolongation]
\label{prop:zeroth-order-prolongation}
\[
  \boxed{
  \mathfrak J_\pi^0
  =
  \mathfrak E_\pi^0
  =
  \left\{
    (\xi_\epsilon):
    \xi_\epsilon=\xi_{\epsilon'}
    \text{ whenever }
    v(\epsilon)=v(\epsilon')
  \right\}.
  }
\]
\end{proposition}

\begin{proof}
Every composition transport and every sign operator fixes the
zero-order coefficient. Hence every formal transport fixes constant
jets. An assignment of one scalar value to each graph vertex therefore
extends to a compatible formal endpoint family by setting all
positive-order coefficients equal to zero.
\end{proof}

\subsection{The descended scalar core}
\label{subsec:descended-scalar-core}

Because \(\pi\) is the Hausdorff reflection, every continuous map from
\(M\) to a Hausdorff space is constant on the fibres of \(\pi\).
Every \(f\in C_c^\infty(M)\) factors uniquely as \(f=\widehat f\circ\pi\) for a continuous scalar function
\(\widehat f:\Gamma\longrightarrow\mathbb C\).

\begin{definition}[Descended scalar core]
\label{def:descended-scalar-core}
The descended compactly supported smooth core is
\[
  \mathscr A_M
  :=
  \left\{
    \widehat f:
    f\in C_c^\infty(M),\
    f=\widehat f\circ\pi
  \right\}.
\]
\end{definition}

Let \(C_{c,\oplus}^\infty(\Gamma)\) denote the compactly supported edge families that are smooth up to
every finite edge end, without initially imposing vertex
compatibility.

\begin{lemma}[Smooth gluing across an exceptional point]
\label{lem:graph-smooth-gluing}
Let
\[
  u\in C_{c,\oplus}^\infty(\Gamma)
\]
and let
\[
  p\in B_\pi.
\]
Then \(u\circ\pi\) is smooth near \(p\) if and only if
\[
  T_{p,+}^\infty
  J_{\epsilon_{p,+}}^\infty u
  =
  T_{p,-}^\infty
  J_{\epsilon_{p,-}}^\infty u.
\]
\end{lemma}

\begin{proof}
In the centred coordinate \(\phi_p\), the pullback \(u\circ\pi\) has
the one-sided representations
\(x\longmapsto u_{e(\epsilon_{p,+})} \bigl( \kappa_{p,+}(x) \bigr), \qquad x>0\),
and
\(x\longmapsto u_{e(\epsilon_{p,-})} \bigl( \kappa_{p,-}(-x) \bigr), \qquad x<0\).
A function that is smooth on the two half-intervals is smooth through
the origin precisely when all one-sided derivatives agree. Those
derivatives are exactly the two transported formal jets displayed
above.
\end{proof}

\begin{theorem}[Formal characterisation of the scalar core]
\label{thm:formal-characterization-scalar-core}
\[
  \boxed{
  \mathscr A_M
  =
  \left\{
    u\in C_{c,\oplus}^\infty(\Gamma):
    \operatorname{Tr}_\infty u
    \in
    \mathfrak E_\pi^\infty
  \right\}.
  }
\]
\end{theorem}

\begin{proof}
If \(u\in\mathcal A_M\), write \(f=u\circ\pi\in C_c^\infty(M)\). Smoothness on the regular locus gives edgewise smoothness, and the nondegenerate half-germs convert smoothness at each exceptional point into the formal transport relations of Lemma~\ref{lem:graph-smooth-gluing}; continuity gives common vertex values. Compactness follows from \(\operatorname{supp}u\subseteq\pi(\operatorname{supp}f)\).

Conversely, if \(u\in C_{c,\oplus}^\infty(\Gamma)\) has \(\operatorname{Tr}^\infty u\in\mathfrak E_\pi^\infty\), the zeroth-order equations make \(u\) continuous and Lemma~\ref{lem:graph-smooth-gluing} makes \(f=u\circ\pi\) smooth at every exceptional point. If \(K=\operatorname{supp}u\), then \(\operatorname{supp}f\subseteq\pi^{-1}(K)\), which is compact by properness of \(\pi\). \(f\in C_c^\infty(M)\).
\end{proof}

\subsection{Realisation of compatible formal jets}
\label{subsec:graph-formal-realisation}

The classical Borel theorem realises an arbitrary formal jet at one
point by a smooth function
\citep{Borel1895,Hormander2003AnalysisI}. Since the graph has finitely
many edge ends, all endpoint jets can be realised simultaneously.

\begin{lemma}[Edgewise Borel realisation]
\label{lem:edgewise-borel-realisation}
The formal endpoint trace map
\[
  \operatorname{Tr}_\infty:
  C_{c,\oplus}^\infty(\Gamma)
  \longrightarrow
  \mathbb J_\Gamma^\infty
\]
is surjective.
\end{lemma}

\begin{proof}
On a compact edge, choose disjoint collars of its two endpoint labels.
Apply Borel's theorem in each collar, multiply the resulting functions
by cutoffs supported in those collars and equal to one near the
corresponding endpoint, and add them. This realises the two formal jets
independently.

On an external half-line, realise the formal jet at its finite edge end
and multiply by a compactly supported cutoff equal to one near that
end. Since the graph has finitely many edges, the edgewise
constructions assemble into one element of
\(C_{c,\oplus}^\infty(\Gamma)\).
\end{proof}

\begin{theorem}[Formal trace realisation]
\label{thm:graph-formal-trace-realisation}
\[
  \boxed{
  \operatorname{Tr}_\infty(\mathscr A_M)
  =
  \mathfrak E_\pi^\infty.
  }
\]
For every
\[
  r\in\mathbb N_0,
\]
\[
  \boxed{
  \operatorname{Tr}_r(\mathscr A_M)
  =
  \mathfrak J_\pi^r.
  }
\]
\end{theorem}

\begin{proof}
The inclusion \(\operatorname{Tr}_\infty(\mathscr A_M) \subseteq \mathfrak E_\pi^\infty\) follows from
Theorem~\ref{thm:formal-characterization-scalar-core}.

Conversely, let
\(\xi\in\mathfrak E_\pi^\infty\).
By
Lemma~\ref{lem:edgewise-borel-realisation},
there exists \(u\in C_{c,\oplus}^\infty(\Gamma)\) such that
\(\operatorname{Tr}_\infty u=\xi\).
The formal characterisation theorem then gives
\(u\in\mathscr A_M\).
This proves the formal identity. Applying
\(\operatorname{pr}_r\) gives
\[
  \operatorname{Tr}_r(\mathscr A_M)
  =
  \operatorname{pr}_r
  \left(
    \mathfrak E_\pi^\infty
  \right)
  =
  \mathfrak J_\pi^r.
\]
\end{proof}

\subsection{The one-dimensional endpoint trace package}
\label{subsec:graph-endpoint-trace-package}

For
\(m\in\mathbb N, \qquad 1\leq p<\infty\),
set
\(W_\oplus^{m,p}(\Gamma) := \bigoplus_{e\in E(\Gamma)} W^{m,p}(I_e)\).

\begin{proposition}[Endpoint traces and zero-trace density]
\label{prop:one-dimensional-trace-package}
For every
\[
  m\in\mathbb N,
  \qquad
  1\leq p<\infty,
\]
the endpoint trace map
\[
  \operatorname{Tr}_{m-1}:
  W_\oplus^{m,p}(\Gamma)
  \longrightarrow
  \mathbb J_\Gamma^{m-1}
\]
is bounded and surjective. Moreover,
\[
  \ker\operatorname{Tr}_{m-1}
  =
  \overline{
    C_c^\infty
    \bigl(
      \Gamma\setminus V(\Gamma)
    \bigr)
  }^{\,W_\oplus^{m,p}(\Gamma)}.
\]
\end{proposition}

\begin{proof}
On each compact interval or half-line, \(W^{m,p}\) functions have well-defined endpoint derivatives through order \(m-1\), and the corresponding trace map is bounded also for \(p=1\). Disjoint endpoint collars and cutoff polynomials realise arbitrary trace data, giving surjectivity. If all endpoint traces vanish, extension by zero across the finite endpoints belongs to \(W^{m,p}\); mollification, with truncation at infinity on half-lines, approximates by smooth functions supported away from the endpoints. Taking the finite direct sum over edges gives the kernel identity.
\end{proof}

Write
\(W_{0,\oplus}^{m,p}(\Gamma) := \ker\operatorname{Tr}_{m-1}\).

\subsection{Exact integer-order Sobolev closure}
\label{subsec:graph-sobolev-closure}

\begin{theorem}[Exact graph-resolution closure]
\label{thm:exact-graph-resolution-closure}
Let
\[
  (M,\Gamma,\pi)
\]
be a finite smooth graph resolution. For every
\[
  m\in\mathbb N,
  \qquad
  1\leq p<\infty,
\]
\[
  \boxed{
  \overline{\mathscr A_M}^{\,W_\oplus^{m,p}(\Gamma)}
  =
  \left\{
    u\in W_\oplus^{m,p}(\Gamma):
    \operatorname{Tr}_{m-1}u
    \in
    \mathfrak J_\pi^{m-1}
  \right\}.
  }
\]
The same identity holds for every norm equivalent to the standard
direct-sum \(W^{m,p}\)-norm.
\end{theorem}

\begin{proof}
Because \(\mathfrak J_\pi^{m-1}\) is finite-dimensional, continuity of the endpoint trace and Theorem~\ref{thm:graph-formal-trace-realisation} give the inclusion from the closure of \(\mathcal A_M\). Conversely, if \(u\in W_\oplus^{m,p}(\Gamma)\) has \(\operatorname{Tr}^{m-1}u\in\mathfrak J_\pi^{m-1}\), choose \(\xi\in\mathfrak E_\pi^\infty\) with that truncation and, by Theorem~\ref{thm:graph-formal-trace-realisation}, \(w\in\mathcal A_M\) with \(\operatorname{Tr}^\infty w=\xi\). Then \(u-w\) has zero endpoint traces through order \(m-1\), so Proposition~\ref{prop:one-dimensional-trace-package} approximates it by \(\phi_n\in C_c^\infty(\Gamma\setminus V(\Gamma))\subset\mathcal A_M\). Hence \(w+\phi_n\to u\). Equivalent norms have the same closure.
\end{proof}

The closure argument separates ordinary zero-trace approximation from the finite-dimensional prolongable endpoint data.

\subsection{First-order universality}
\label{subsec:graph-first-order-universality}

Define
\[
  W_{\mathrm{cont}}^{1,p}(\Gamma)
  :=
  \left\{
    u\in W_\oplus^{1,p}(\Gamma):
    u_\epsilon(v)
    =
    u_{\epsilon'}(v)
    \text{ for every }
    v\in V(\Gamma)
    \text{ and }
    \epsilon,\epsilon'\in\mathcal E(v)
  \right\}.
\]

\begin{corollary}[First-order scalar closure]
\label{cor:universal-first-order-graph-closure}
For every finite smooth graph resolution and every
\[
  1\leq p<\infty,
\]
\[
  \boxed{
  \overline{\mathscr A_M}^{\,W_\oplus^{1,p}(\Gamma)}
  =
  W_{\mathrm{cont}}^{1,p}(\Gamma).
  }
\]
\end{corollary}

\begin{proof}
At \(m=1\), the endpoint trace contains only vertex values. Apply
Proposition~\ref{prop:zeroth-order-prolongation}
and
Theorem~\ref{thm:exact-graph-resolution-closure}.
\end{proof}

First-order scalar closure therefore depends only on the Hausdorff
metric graph. The higher smooth compatibility data of the marked
resolution become visible only when derivative traces enter the
Sobolev topology.

This conclusion is specific to scalar functions and fixed trivial
coefficient spaces. A nontrivially transported vector bundle may
impose a proper fibre equaliser already at order zero, as in
Sections~\ref{sec:extension-failure} and \ref{sec:visibility}.

\subsection{Higher-order sensitivity and recovery}
\label{subsec:graph-recovery}

Let \((M_1,\Gamma,\pi_1), \qquad (M_2,\Gamma,\pi_2)\) be finite smooth graph resolutions of the same metric graph.

\begin{proposition}[Classification of smooth-resolution Sobolev closures]
\label{prop:classification-marked-closures}
For fixed
\[
  m\in\mathbb N,
  \qquad
  1\leq p<\infty,
\]
the following are equivalent:

\begin{enumerate}[label=\textup{(\roman*)}]
  \item
        \[
          \overline{\mathscr A_{M_1}}^{\,W_\oplus^{m,p}(\Gamma)}
          =
          \overline{\mathscr A_{M_2}}^{\,W_\oplus^{m,p}(\Gamma)};
        \]

  \item
        \[
          \mathfrak J_{\pi_1}^{m-1}
          =
          \mathfrak J_{\pi_2}^{m-1}.
        \]
\end{enumerate}
\end{proposition}

\begin{proof}
The implication \(\textup{(ii)} \Longrightarrow \textup{(i)}\) follows from
Theorem~\ref{thm:exact-graph-resolution-closure}.

Conversely, the endpoint trace maps either closed space surjectively
onto its corresponding prolongable trace space. Indeed, every element
of \(\mathfrak J_{\pi_i}^{m-1}\) is realised by an element of
\(\mathscr A_{M_i}\). Equality of the closed spaces therefore implies
equality of their endpoint-trace images.
\end{proof}

\begin{theorem}[Recovery of the formal equaliser]
\label{thm:recovery-formal-equaliser}
Natural truncation induces a linear isomorphism
\[
  \boxed{
  \mathfrak E_\pi^\infty
  \cong
  \varprojlim_{r\in\mathbb N_0}
  \mathfrak J_\pi^r.
  }
\]
\end{theorem}

\begin{proof}
Every \(\eta\in\mathfrak E_\pi^\infty\) determines the compatible family
\(\left( \operatorname{pr}_r\eta \right)_{r\geq0}\),
giving an injective linear map into the inverse limit.

Conversely, let
\((\xi^r)_{r\geq0} \in \varprojlim_r\mathfrak J_\pi^r\).
Compatibility under truncation determines a unique formal endpoint
family \(\xi\in\mathbb J_\Gamma^\infty\) such that
\(\operatorname{pr}_r\xi=\xi^r \qquad (r\geq0)\).
For every \(r\),
\(\xi^r \in \mathfrak J_\pi^r \subseteq \mathfrak E_\pi^r\).
Every defining equation of
\(\mathfrak E_\pi^\infty\) involves only finitely many jet
coefficients, and hence occurs in one of the finite systems
\(\mathfrak E_\pi^r\). It therefore holds for \(\xi\), so
\(\xi\in\mathfrak E_\pi^\infty\).
\end{proof}

\begin{corollary}[Recovery from the closure tower]
\label{cor:recovery-from-graph-closure-tower}
For two finite smooth resolutions of the same metric graph, the
following are equivalent:

\begin{enumerate}[label=\textup{(\roman*)}]
  \item for every \(m\in\mathbb N\) and one fixed
        \(p\in[1,\infty)\),
        \[
          \overline{\mathscr A_{M_1}}^{\,W_\oplus^{m,p}(\Gamma)}
          =
          \overline{\mathscr A_{M_2}}^{\,W_\oplus^{m,p}(\Gamma)};
        \]

  \item the equality in \textup{(i)} holds for every
        \(m\in\mathbb N\) and every \(1\leq p<\infty\);

  \item for every \(r\geq0\),
        \[
          \mathfrak J_{\pi_1}^r
          =
          \mathfrak J_{\pi_2}^r;
        \]

  \item
        \[
          \mathfrak E_{\pi_1}^\infty
          =
          \mathfrak E_{\pi_2}^\infty.
        \]
\end{enumerate}
\end{corollary}

\begin{proof}
For any fixed \(p\), applying
Proposition~\ref{prop:classification-marked-closures}
at every \(m=r+1\) gives
\(\textup{(i)} \Longleftrightarrow \textup{(iii)}\).
The exact closure theorem shows that \textup{(iii)} implies
\textup{(ii)}. Finally,
Theorem~\ref{thm:recovery-formal-equaliser}
gives
\(\textup{(iii)} \Longleftrightarrow \textup{(iv)}\).
\end{proof}

The full integer-order closure tower recovers the formal endpoint relation over the fixed graph, but not necessarily the underlying non-Hausdorff manifold or unmarked resolution.

\begin{remark}[Scope]
\label{rem:graph-resolution-scope}
The results above use:

\begin{enumerate}[label=\textup{(\roman*)}]
  \item finitely many edges and exceptional points;

  \item properness of the reflection map over compact graph sets;

  \item nondegenerate endpoint-coordinate germs;

  \item scalar functions, or a fixed trivial finite-dimensional
        coefficient space with componentwise transport;

  \item integer Sobolev order and
        \(1\leq p<\infty\).
\end{enumerate}

For locally finite infinite graphs, simultaneous jet realisation and
closedness require additional control of edge lengths, vertex degrees,
weights, support, and endpoint trace norms. Degenerate endpoint
coordinates may destroy ordinary edgewise Sobolev regularity before any
compatibility equation is imposed. Neither case is treated here.

Signed transport graphs organise the local endpoint relations and yield the
strict inclusions \(\mathfrak J_\pi^r \subsetneq \mathfrak E_\pi^r\) produced by tangent-to-identity return holonomy.
\end{remark}
\section{Signed formal holonomy and prolongation defects}
\label{sec:formal-holonomy}

The endpoint equations reduce vertexwise to formal return holonomy; this section classifies the prolongation defects for tangent-to-identity return germs.

\subsection{Signed relative endpoint transport}
\label{subsec:signed-endpoint-transport}

Fix a finite smooth graph resolution \((M,\Gamma,\pi)\) and an exceptional point
\(p\in B_\pi\).
Recall the centred coordinate \(\phi_p: U_p\longrightarrow(-\delta_p,\delta_p)\) and the nondegenerate endpoint half-germs
\(\kappa_{p,\sigma}, \qquad \sigma\in\{+1,-1\}\).

For a scalar function smooth up to an edge end \(\epsilon\), write its
formal Taylor jet as
\[
  \widehat J_\epsilon^\infty u
  :=
  \sum_{j\geq0}
  \frac{
    \partial_{t_\epsilon}^{\,j}u(0)
  }{
    j!
  }
  T^j.
\]

For
\(\sigma,\sigma'\in\{+1,-1\}\),
define the signed relative formal germ
\[
  \rho_{p,\sigma\to\sigma'}(T)
  :=
  \kappa_{p,\sigma}
  \left(
    \sigma\sigma'
    \kappa_{p,\sigma'}^{-1}(T)
  \right).
\]
The expression is interpreted formally. Each
\(\kappa_{p,\sigma}\) has zero constant term and nonzero linear term,
so its Taylor series is formally invertible.

\begin{proposition}[Signed relative endpoint transport]
\label{prop:correct-relative-endpoint-transport}
The endpoint jets of a function smooth near \(p\) satisfy
\[
  \boxed{
  \widehat J_{\epsilon_{p,\sigma'}}^\infty u
  =
  \left(
    \widehat J_{\epsilon_{p,\sigma}}^\infty u
  \right)
  \circ
  \rho_{p,\sigma\to\sigma'}.
  }
\]
Moreover,
\[
  \rho_{p,\sigma\to\sigma'}'(0)
  =
  \sigma\sigma'
  \frac{
    \kappa_{p,\sigma}'(0)
  }{
    \kappa_{p,\sigma'}'(0)
  }.
\]
In particular,
\[
  \rho_{p,+\to-}'(0)<0,
  \qquad
  \rho_{p,-\to+}'(0)<0.
\]
\end{proposition}

\begin{proof}
Let \(g:=u\circ\phi_p^{-1}\) be the local coordinate representative. On the side labelled by
\(\sigma\),
\(u_{\epsilon_{p,\sigma}}(t) = g \left( \sigma\kappa_{p,\sigma}^{-1}(t) \right)\).
If
\[
  t
  =
  \kappa_{p,\sigma}
  \left(
    \sigma\sigma'
    \kappa_{p,\sigma'}^{-1}(t')
  \right),
\]
then
\(\sigma\kappa_{p,\sigma}^{-1}(t) = \sigma' \kappa_{p,\sigma'}^{-1}(t')\).
\(u_{\epsilon_{p,\sigma}}(t) = u_{\epsilon_{p,\sigma'}}(t')\).
Passing to Taylor series gives the transport identity. The derivative
formula follows from the formal chain rule.
\end{proof}

Every elementary endpoint transport is therefore orientation reversing; the factor \(\sigma\sigma'\) is essential to the linear multiplier.

\subsection{Transport graphs and orientation character}
\label{subsec:transport-graphs}

Fix a vertex
\(v\in V(\Gamma)\).

\begin{definition}[Transport graph]
\label{def:transport-graph}
The transport graph
\[
  \mathcal G_v
\]
is the finite multigraph whose vertex set is
\[
  \mathcal E(v).
\]
Every exceptional point
\[
  p\in\pi^{-1}(v)
\]
gives an edge joining
\[
  \epsilon_{p,+}
  \quad\text{and}\quad
  \epsilon_{p,-}.
\]
An orientation
\[
  a:\epsilon\longrightarrow\epsilon'
\]
of this edge is labelled by the corresponding formal germ
\[
  \rho_a
  \in
  \widehat{\operatorname{Diff}}(\mathbb R,0).
\]
The reverse orientation is labelled by
\[
  \rho_{a^{-1}}=\rho_a^{-1}.
\]
\end{definition}

For a formal germ \(\rho\), let \(\mathcal T_\rho(f):=f\circ\rho\) denote its substitution action.

Let \(P=(a_1,\ldots,a_N)\) be an oriented path, traversed successively in the displayed order.
Define
\(\rho_P := \rho_{a_1} \circ\cdots\circ \rho_{a_N}\).
With this convention, successive transport along \(P\) sends
\(f \longmapsto f\circ\rho_{a_1} \longmapsto \cdots \longmapsto f\circ\rho_P\).
If \(P\cdot Q\) denotes traversal of \(P\) followed by \(Q\), then
\(\rho_{P\cdot Q} = \rho_P\circ\rho_Q\).

Since every elementary label has negative derivative,
\(\operatorname{sgn}\rho_P'(0) = (-1)^{|P|}\),
where \(|P|\) is the path length counted with multiplicity.

\begin{definition}[Orientation character]
\label{def:orientation-character}
Let \(C\) be a connected component of \(\mathcal G_v\), and choose a
root
\[
  \epsilon_C\in C.
\]
The orientation character is
\[
  \omega_C:
  \pi_1(C,\epsilon_C)
  \longrightarrow
  \{\pm1\},
  \qquad
  [P]\longmapsto
  \operatorname{sgn}\rho_P'(0).
\]
\end{definition}

This is well defined because cancelling an edge immediately followed
by its reverse removes two orientation-reversing factors. A closed
path of odd length has orientation-reversing return germ. In
particular, a transport component containing an odd cycle cannot have
all its return holonomy tangent to the identity.

\subsection{Formal holonomy and spanning-tree reduction}
\label{subsec:holonomy-reduction}

Let \(C\) be a connected component of a transport graph, with root
\(\epsilon_C\). Its based formal holonomy group is
\[
  H_C
  :=
  \operatorname{Hol}_\pi(C)
  :=
  \left\{
    \rho_P:
    P
    \text{ is a closed path based at }
    \epsilon_C
  \right\}.
\]
It is a subgroup of
\(\widehat{\operatorname{Diff}}(\mathbb R,0)\).
Changing the root conjugates \(H_C\) by the transport germ of a path
between the two roots.

Let
\[
  \mathfrak m=(T)\subseteq\mathbb K[[T]],
  \qquad
  \mathbb K\in\{\mathbb R,\mathbb C\},
\]
and define the reduced order-\(r\) jet space
\(\mathcal J_r := \mathfrak m/\mathfrak m^{r+1}\).
The constant coefficient is omitted because it is fixed by every
coordinate substitution and is already controlled by graph
continuity.

Choose a spanning tree
\(\mathcal T_C\subseteq C\).
For every edge end \(\epsilon\in C\), let \(P_\epsilon: \epsilon_C\longrightarrow\epsilon\) be the unique path in \(\mathcal T_C\).

\begin{proposition}[Holonomy reduction]
\label{prop:graph-holonomy-reduction}
A reduced order-\(r\) endpoint family
\[
  (\xi_\epsilon)_{\epsilon\in C}
\]
satisfies all transport relations in \(C\) if and only if there exists
\[
  \eta_C\in\mathcal J_r^{H_C}
\]
such that
\[
  \boxed{
  \xi_\epsilon
  =
  \mathcal T_{\rho_{P_\epsilon}}^r\eta_C
  \qquad
  (\epsilon\in C).
  }
\]
The corresponding statement holds for formal reduced jets.
\end{proposition}

\begin{proof}
Given a compatible family, set \(\eta_C=\xi_{\epsilon_C}\); transport along the spanning-tree path \(P_\epsilon\) gives \(\xi_\epsilon=T_{\rho_{P_\epsilon}}^r\eta_C\). Transport around any closed path fixes the root jet, so \(\eta_C\in\mathcal J_r^{H_C}\). Conversely, start with such an invariant root jet and define the family by tree transport. For an oriented edge \(a:\epsilon\to\epsilon'\), the loop \(P_\epsilon aP_{\epsilon'}^{-1}\) fixes \(\eta_C\); composing its fixedness relation by \(\rho_{P_{\epsilon'}}\) gives exactly \(\xi_{\epsilon'}=\xi_\epsilon\circ\rho_a\). The formal case is identical.
\end{proof}

\begin{corollary}[Tree components]
\label{cor:tree-components}
If \(C\) is a tree, then
\[
  H_C=\{\operatorname{id}\},
\]
and \(C\) is formally prolongation complete at every order.
\end{corollary}

\begin{proof}
A tree has no nontrivial reduced closed paths, so its holonomy is
trivial. Every reduced finite root jet may then be extended to an
arbitrary formal root jet and transported along the tree.
\end{proof}

A prolongation defect requires a cycle with nontrivial return holonomy.

\begin{definition}[Tangent and active components]
\label{def:tangent-active-component}
A connected transport component \(C\) is
\emph{tangent to the identity} if
\[
  H_C
  \subseteq
  \widehat{\operatorname{Diff}}_1(\mathbb R,0),
\]
where
\[
  \widehat{\operatorname{Diff}}_1(\mathbb R,0)
  :=
  \left\{
    h:
    h(0)=0,\
    h'(0)=1
  \right\}.
\]
A tangent component is \emph{active} if
\[
  H_C\neq\{\operatorname{id}\}.
\]
\end{definition}

For a transport component to be tangent to the identity, two
independent linear conditions are required: \(\omega_C\equiv1\) and
\(|\rho_P'(0)|=1 \qquad \text{for every closed path }P\).
The first excludes orientation-reversing returns, while the second
excludes orientation-preserving returns with nonunit multiplier.
Together they give \(\rho_P'(0)=1\) for every closed path.

For complex-valued scalar fields, the real substitution action is
extended \(\mathbb C\)-linearly.

\subsection{Fixed jets of tangent-to-identity holonomy}
\label{subsec:fixed-jet-classification}

Let
\(H \subseteq \widehat{\operatorname{Diff}}_1(\mathbb K,0)\).
Composition acts on
\(\mathcal J_r = \mathfrak m/\mathfrak m^{r+1}\).
Define \(F_H^r := \mathcal J_r^H\) and
\[
  F_H^\infty
  :=
  \mathfrak m^H
  =
  \left\{
    f\in\mathfrak m:
    f\circ h=f
    \text{ for every }h\in H
  \right\}.
\]
The finite projection of the formal fixed space is
\(I_H^r := \operatorname{pr}_r \left( F_H^\infty \right) \subseteq F_H^r\),
and the corresponding prolongation quotient is
\(P_H^r := F_H^r/I_H^r\).

For a nonidentity germ \(h\), set
\(q(h) := \operatorname{ord} \left( h-\operatorname{id} \right)-1\).
For notational convenience, set
\(q(\operatorname{id}):=\infty\).
If \(H\) is nontrivial, define
\(q(H) := \min_{h\in H}q(h) = \min_{h\in H\setminus\{\operatorname{id}\}}q(h)\).

The local dynamical theory of tangent-to-identity germs is classical
\citep{Ecalle1975,IlyashenkoYakovenko2008,Abate2010,EynardBontempsNavas2024}.
The result required here concerns their substitution action on
truncated scalar jets.

\begin{lemma}[Leading substitution term]
\label{lem:leading-substitution-term}
Suppose
\[
  h(T)
  =
  T+aT^{q+1}
  +
  O(T^{q+2}),
  \qquad
  a\neq0,
\]
and
\[
  f(T)
  =
  c_nT^n
  +
  O(T^{n+1}),
  \qquad
  c_n\neq0.
\]
Then
\[
  \boxed{
  f\circ h-f
  =
  nac_nT^{n+q}
  +
  O(T^{n+q+1}).
  }
\]
\end{lemma}

\begin{proof}
One has
\[
  h(T)^n
  =
  T^n
  \left(
    1+aT^q+O(T^{q+1})
  \right)^n
  =
  T^n
  +
  naT^{n+q}
  +
  O(T^{n+q+1}).
\]
For every \(m>n\),
\(h(T)^m-T^m = O(T^{m+q}) = O(T^{n+q+1})\),
so no higher term of \(f\) contributes at degree \(n+q\).
\end{proof}

\begin{theorem}[Fixed jets of one germ]
\label{thm:fixed-jets-single-germ}
Let
\[
  h(T)
  =
  T+aT^{q+1}
  +
  O(T^{q+2}),
  \qquad
  a\neq0.
\]
Then, for every \(r\geq1\),
\[
  \boxed{
  F_{\langle h\rangle}^r
  =
  \frac{
    \mathfrak m^{\max\{1,r-q+1\}}
  }{
    \mathfrak m^{r+1}
  }.
  }
\]
\[
  \dim F_{\langle h\rangle}^r
  =
  \min\{q,r\},
\]
and
\[
  F_{\langle h\rangle}^\infty=0.
\]
\end{theorem}

\begin{proof}
Let a nonzero \(r\)-jet have lowest term \(c_nT^n\). By
Lemma~\ref{lem:leading-substitution-term}, \(f\circ h-f\) has lowest possible term of degree \(n+q\). The jet is fixed modulo
\(\mathfrak m^{r+1}\) exactly when
\(n+q\geq r+1\),
equivalently,
\(n\geq r-q+1\).
This gives the finite formula.

If \(f\in\mathfrak m\) is a nonzero formal series, its lowest nonzero
term produces a nonzero lowest term in \(f\circ h-f\). Hence no
nonzero reduced formal series is fixed.
\end{proof}

\begin{lemma}[Tangency filtration]
\label{lem:tangency-filtration}
For
\[
  g,h
  \in
  \widehat{\operatorname{Diff}}_1(\mathbb K,0),
\]
\[
  q(g^{-1})=q(g)
\]
and
\[
  q(g\circ h)
  \geq
  \min\{q(g),q(h)\}.
\]
If a set \(S\) generates \(H\), then
\[
  q(H)
  =
  \min_{h\in S}q(h).
\]
\end{lemma}

\begin{proof}
The inverse of \(T+aT^{q+1}+O(T^{q+2})\) begins
\(T-aT^{q+1}+O(T^{q+2})\),
which proves the inverse identity. In a composition, no nonlinear term
can occur below the minimum of the first nonlinear orders of the two
factors. Cancellation at that minimum may increase the order but
cannot decrease it.

Every word in the generators and their inverses therefore has tangency
order at least
\(\min_{h\in S}q(h)\).
Since a nonidentity generator attaining this minimum belongs to \(H\),
equality follows.
\end{proof}

\begin{theorem}[Fixed jets for tangent-to-identity holonomy]
\label{thm:complete-fixed-jet-classification}
Let
\[
  H
  \subseteq
  \widehat{\operatorname{Diff}}_1(\mathbb K,0)
\]
be nontrivial, and set
\[
  q=q(H).
\]
Then, for every \(r\geq1\),
\[
  \boxed{
  F_H^r
  =
  \frac{
    \mathfrak m^{\max\{1,r-q+1\}}
  }{
    \mathfrak m^{r+1}
  }.
  }
\]
Moreover,
\[
  \boxed{
  F_H^\infty=0,
  \qquad
  I_H^r=0,
  \qquad
  P_H^r=F_H^r.
  }
\]
In particular,
\[
  \boxed{
  \dim P_H^r
  =
  \min\{q(H),r\}.
  }
\]
\end{theorem}

\begin{proof}
Choose \(h_0\in H\) with \(q(h_0)=q(H)=q\). Theorem~\ref{thm:fixed-jets-single-germ} gives
\(F_H^r\subseteq F_{\langle h_0\rangle}^r=\mathfrak m^{\max\{1,r-q+1\}}/\mathfrak m^{r+1}\). Conversely, if \([f]\) lies in this terminal monomial subspace and \(g\in H\), then \(q(g)\ge q\); if \(n\) is the lowest degree of \(f\), Lemma~\ref{lem:leading-substitution-term} gives \(f\circ g-f\in\mathfrak m^{r+1}\). Hence equality holds. Any formal \(H\)-invariant series is in particular \(h_0\)-invariant, so Theorem~\ref{thm:fixed-jets-single-germ} gives \(F_H^\infty=0\); the remaining assertions follow.
\end{proof}

The first finite constraint appears at order \(q(H)+1\): \(F_H^r=\mathcal J_r\) for \(r\le q(H)\), while \(F_H^{q(H)+1}=\operatorname{span}\{T^2,\ldots,T^{q(H)+1}\}\).

\subsection{Finite horizons and interval decomposition}
\label{subsec:finite-horizon}

For
\(r\geq1, \qquad s\geq0\),
define the finite-horizon image
\[
  F_H^{r\langle s\rangle}
  :=
  \operatorname{Im}
  \left(
    F_H^{r+s}
    \longrightarrow
    F_H^r
  \right).
\]

\begin{theorem}[Finite-horizon prolongation]
\label{thm:finite-horizon-prolongation}
For every \(r\geq1\) and \(s\geq0\),
\[
  \boxed{
  F_H^{r\langle s\rangle}
  =
  \frac{
    \mathfrak m^{
      \max\{1,r+s-q(H)+1\}
    }
  }{
    \mathfrak m^{r+1}
  }.
  }
\]
The space is interpreted as zero when the lower exponent exceeds \(r\).
\[
  \boxed{
  \dim F_H^{r\langle s\rangle}
  =
  \min
  \left\{
    r,
    (q(H)-s)_+
  \right\}.
  }
\]
\end{theorem}

\begin{proof}
At order \(r+s\),
Theorem~\ref{thm:complete-fixed-jet-classification}
gives the terminal monomial space beginning in degree
\(\max\{1,r+s-q(H)+1\}\).
Truncating to order \(r\) gives the displayed formula.
\end{proof}

\begin{corollary}[Uniform determination horizon]
\label{cor:uniform-determination-horizon}
For every \(r\geq1\),
\[
  F_H^{r\langle q(H)\rangle}=0.
\]
For every \(r\geq q(H)\),
\[
  F_H^{r\langle q(H)-1\rangle}
  =
  \operatorname{span}\{T^r\}
  \neq0.
\]
\(q(H)\) is the least uniform number of additional jet levels
needed to eliminate all nonprolongable finite fixed directions.
\end{corollary}

For integers
\(1\leq a\leq b\),
let \(\mathbb I_{[a,b]}\) denote the inverse system equal to \(\mathbb K\) at orders
\(a,\ldots,b\), zero at all other orders, and with identity transition
maps inside the interval.

\begin{corollary}[Interval decomposition]
\label{cor:interval-decomposition}
The fixed-jet inverse system
\[
  \mathbf F_H
  :=
  \left(
    F_H^r
  \right)_{r\geq1}
\]
has the decomposition
\[
  \boxed{
  \mathbf F_H
  \cong
  \bigoplus_{n\geq1}
  \mathbb I_{[n,n+q(H)-1]}.
  }
\]
\end{corollary}

\begin{proof}
The monomial direction \(T^n\) occurs in \(F_H^r\) exactly when \(n\leq r\) and
\(n\geq r-q(H)+1\).
Equivalently,
\(n\leq r\leq n+q(H)-1\).
Jet truncation acts diagonally on the monomial basis.
\end{proof}

\subsection{The graph prolongation profile}
\label{subsec:graph-prolongation-profile}

Assume throughout this subsection that every connected transport
component has either trivial holonomy or nontrivial
tangent-to-identity holonomy. Components with orientation-reversing or
orientation-preserving nonunit-linear return holonomy are excluded
from the formulas below.

For a vertex \(v\), let \(\mathscr C_v\) be the set of connected components of \(\mathcal G_v\). The scalar
endpoint value common to all edge ends at \(v\) is denoted by \(c_v\).
After subtracting this constant coefficient, each component carries an
independent reduced root jet.

\begin{proposition}[Rooted-tree component decomposition]
\label{prop:rooted-tree-component-decomposition}
Choose one root and one spanning tree in every connected transport
component. For every \(r\geq0\), these choices induce linear
isomorphisms
\[
  \mathfrak E_\pi^r
  \cong
  \mathbb K^{V(\Gamma)}
  \oplus
  \bigoplus_{v\in V(\Gamma)}
  \;
  \bigoplus_{C\in\mathscr C_v}
  F_{H_C}^r,
\]
and
\[
  \mathfrak J_\pi^r
  \cong
  \mathbb K^{V(\Gamma)}
  \oplus
  \bigoplus_{v\in V(\Gamma)}
  \;
  \bigoplus_{C\in\mathscr C_v}
  I_{H_C}^r.
\]
Here, at order zero the reduced summands are zero.
\end{proposition}

\begin{proof}
Graph continuity contributes one scalar value \(c_v\) at each vertex. Since every substitution fixes constants, the endpoint equations split as \(\mathbb K[[T]]=\mathbb K\oplus\mathfrak m\). After removing \(c_v\), different connected components of each transport graph decouple; Proposition~\ref{prop:graph-holonomy-reduction} identifies the reduced finite solution space on \(C\) with \(F_{H_C}^r\), and the finite projection of the formal solution space with \(I_{H_C}^r\). Summing over vertices and components gives the displayed decompositions.
\end{proof}

For a component \(C\), define its reduced prolongation quotient by
\(\mathfrak P_{\pi,C}^r := F_{H_C}^r/I_{H_C}^r\).
If \(C\) is active, set
\(q_{\pi,C}:=q(H_C)\).

\begin{theorem}[Componentwise prolongation profile]
\label{thm:componentwise-prolongation-profile}
If \(C\) has trivial holonomy, then
\[
  \mathfrak P_{\pi,C}^r=0
  \qquad
  \text{for every }r.
\]

If \(C\) is active, then
\[
  I_{H_C}^r=0,
\]
\[
  F_{H_C}^r
  =
  \frac{
    \mathfrak m^{
      \max\{1,r-q_{\pi,C}+1\}
    }
  }{
    \mathfrak m^{r+1}
  },
\]
and
\[
  \boxed{
  \dim\mathfrak P_{\pi,C}^r
  =
  \min\{q_{\pi,C},r\}.
  }
\]
Moreover,
\[
  \dim
  \operatorname{Im}
  \left(
    F_{H_C}^{r+s}
    \longrightarrow
    F_{H_C}^r
  \right)
  =
  \min
  \left\{
    r,
    (q_{\pi,C}-s)_+
  \right\}.
\]
\end{theorem}

\begin{proof}
If \(H_C\) is trivial, every finite reduced root jet has an arbitrary
formal prolongation, so
\(F_{H_C}^r=I_{H_C}^r=\mathcal J_r\).

For an active component, apply
Theorems~\ref{thm:complete-fixed-jet-classification}
and \ref{thm:finite-horizon-prolongation}
to \(H_C\).
\end{proof}

Define the global order-\(r\) prolongation defect space by
\(\mathfrak D_\pi^r := \mathfrak E_\pi^r/ \mathfrak J_\pi^r\),
and set
\(d_\pi^r := \dim\mathfrak D_\pi^r\).
Let \(\mathscr C_\pi^{\mathrm{act}}\) be the finite set of active transport components.

\begin{corollary}[Global defect formula]
\label{cor:global-defect-formula}
After the rooted-tree choices of
Proposition~\ref{prop:rooted-tree-component-decomposition},
there is an isomorphism
\[
  \boxed{
  \mathfrak D_\pi^r
  \cong
  \bigoplus_{
    C\in\mathscr C_\pi^{\mathrm{act}}
  }
  \mathfrak P_{\pi,C}^r.
  }
\]
\[
  \boxed{
  d_\pi^r
  =
  \sum_{
    C\in\mathscr C_\pi^{\mathrm{act}}
  }
  \min\{q_{\pi,C},r\}.
  }
\]
\end{corollary}

\begin{proof}
Apply
Proposition~\ref{prop:rooted-tree-component-decomposition}.
The common vertex-value summand occurs identically in
\(\mathfrak E_\pi^r\) and \(\mathfrak J_\pi^r\), and hence contributes
nothing to the quotient. Trivial transport components also contribute
zero. The remaining quotient is the direct sum of the active
component quotients.
\end{proof}

The vertex-indexed multiset
\[
  \mathbf q_\pi
  :=
  \left(
    \left\{
      q_{\pi,C}:
      C\subseteq\mathcal G_v
      \text{ is active}
    \right\}
  \right)_{v\in V(\Gamma)}
\]
is called the prolongation profile.

Set
\(d_\pi^0:=0\).
Then
\[
  \boxed{
  d_\pi^r-d_\pi^{r-1}
  =
  \#\left\{
    C\in\mathscr C_\pi^{\mathrm{act}}:
    q_{\pi,C}\geq r
  \right\}.
  }
\]
The aggregate sequence \((d_\pi^r)_{r\geq1}\) therefore recovers the global multiset of active tangency orders, although it
does not recover their distribution among graph vertices. The
vertexwise family of local defect sequences recovers the full
vertex-indexed profile.

Let
\(Q_\pi := \max_{ C\in\mathscr C_\pi^{\mathrm{act}} } q_{\pi,C}\),
with \(Q_\pi=0\) when no active component exists.

\begin{corollary}[Finite determination of prolongable jets]
\label{cor:finite-determination-prolongable-jets}
For every \(r\geq0\),
\[
  \boxed{
  \mathfrak J_\pi^r
  =
  \operatorname{Im}
  \left(
    \mathfrak E_\pi^{r+Q_\pi}
    \longrightarrow
    \mathfrak E_\pi^r
  \right).
  }
\]
\end{corollary}

\begin{proof}
On an active component, every nonprolongable reduced direction is
eliminated after at most \(q_{\pi,C}\leq Q_\pi\) additional levels. On a trivial component, finite truncation is
surjective at every order. The common vertex-value components also
truncate surjectively.
\end{proof}

Combining this with
Theorem~\ref{thm:exact-graph-resolution-closure}
shows that the order-\(m\) Sobolev closure is determined by endpoint
compatibility through the finite order
\(m-1+Q_\pi\).

\subsection{Invariance}
\label{subsec:holonomy-invariance}

\begin{proposition}[Conjugacy invariance]
\label{prop:formal-conjugacy-invariance}
Let
\[
  \widetilde H
  =
  \psi^{-1}H\psi,
  \qquad
  \psi
  \in
  \widehat{\operatorname{Diff}}(\mathbb K,0).
\]
Substitution by \(\psi\) induces an isomorphism
\[
  \mathbf F_H
  \cong
  \mathbf F_{\widetilde H}
\]
of fixed-jet inverse systems and identifies their formal images and
prolongation quotients.

If
\[
  H
  \subseteq
  \widehat{\operatorname{Diff}}_1(\mathbb K,0),
\]
then
\[
  q(\widetilde H)=q(H).
\]
\end{proposition}

\begin{proof}
Let
\(\widetilde h=\psi^{-1}h\psi\).
If \(f\) is fixed by \(h\), then
\((f\circ\psi)\circ\widetilde h = f\circ h\circ\psi = f\circ\psi\).
Substitution by \(\psi\) intertwines the fixed spaces, at both
finite and formal order. Formal conjugation preserves
\(\operatorname{ord} \left( h-\operatorname{id} \right)\),
which gives the tangency-order statement.
\end{proof}

An admissible coordinate change is a smooth or formal
reparametrisation with nonzero derivative that preserves the marked
endpoint incidence data. A reversal of a centred manifold coordinate
may interchange its two side labels.

\begin{theorem}[Geometric auxiliary-choice invariance]
\label{thm:geometric-holonomy-invariance}
For a fixed finite smooth graph resolution, the componentwise
fixed-jet inverse systems, prolongation quotients, orientation
characters, and tangency orders are unchanged, up to canonical
isomorphism and permutation of components, under:

\begin{enumerate}[label=\textup{(\roman*)}]
  \item change of root;

  \item change of spanning tree;

  \item reversal of a chosen transport-edge orientation;

  \item admissible changes of centred manifold coordinates and formal
        endpoint-jet coordinates.
\end{enumerate}
\end{theorem}

\begin{proof}
Changing the root conjugates the based holonomy group by the transport
germ of a connecting path. Changing the spanning tree alters only the
chosen root parametrisation and generating loops. Reversing an edge
replaces its label by its inverse without changing the represented
transport groupoid.

Under a coordinate change, the intermediate reparametrisations cancel
along every closed path. The return holonomy is therefore conjugated
by the coordinate change at the root. Conjugacy preserves the
orientation sign, the linear multiplier, the fixed-jet inverse system,
and the tangency order.
\end{proof}

The same statement can be formulated at the presentation level using
finite labelled groupoids
\citep{MoerdijkMrcun2003}.

\begin{proposition}[Labelled-groupoid presentation invariance]
\label{prop:labelled-groupoid-invariance}
Suppose that two finite labelled transport-groupoid presentations are
isomorphic, including their objects, formal transport morphisms, and
composition. Then corresponding connected components have conjugate
based holonomy groups and isomorphic fixed-jet inverse systems.

This applies, in particular, to subdivision of a presentation by
composable morphisms, insertion or deletion of identity morphisms, and
insertion or deletion of redundant generators, provided that the
represented labelled groupoid is unchanged.
\end{proposition}

\begin{proof}
A labelled-groupoid isomorphism identifies closed-path morphisms after
a possible change of base object. The corresponding based holonomy
groups are therefore conjugate.
\end{proof}

These are presentation moves, not geometric insertions into a fixed resolution; a genuine elementary endpoint transport is orientation reversing and cannot be the identity germ.

\subsection{Realisation of prescribed return holonomy}
\label{subsec:two-cycle-realisation}

\begin{theorem}[Two-cycle realisation]
\label{thm:two-cycle-realisation}
Let
\[
  h
  \in
  \widehat{\operatorname{Diff}}_1(\mathbb R,0).
\]
There exists a finite smooth graph resolution whose transport graph at
one vertex has two vertices joined by two parallel edges and whose
based formal return holonomy is generated by \(h\).

If
\[
  h\neq\operatorname{id},
\]
the resulting active component has tangency order
\[
  q(h).
\]
The same real transport acts on complex-valued scalar jets by complex
linear extension.
\end{theorem}

\begin{proof}
By Borel realisation choose an orientation-preserving smooth germ \(\widetilde h\) with Taylor series \(h\), after shrinking so that \(\widetilde h\) is a local diffeomorphism \citep{Borel1895,Hormander2003AnalysisI}. Let \(M\) be the topological line with two origins, with Hausdorff reflection \(\pi([t,i])=t\) onto \(\Gamma=\mathbb R\). Near the first origin use the centred coordinate \(\phi_1([t])=t\). Near the second use
\[
 \phi_2([t])=
 \begin{cases}
 t,&t<0,\\
 \widetilde h(t),&t>0,
 \end{cases}
 \qquad \phi_2(0_2)=0,
\]
with domains chosen so that both charts map onto an interval about \(0\). Their negative overlap transition is the identity and their positive transition is \(\widetilde h\); hence, together with the ordinary charts away from the origins, they define a smooth second-countable \(T_1\) one-manifold without boundary. The regular locus maps diffeomorphically to \(\mathbb R\setminus\{0\}\), and \(\pi\) is proper because the preimage of a compact set is the union of its two compact copies. \((M,\Gamma,\pi)\) is a finite smooth graph resolution.

Let \(\epsilon_\pm\) be the two edge ends at the graph vertex. At the first origin \(\kappa_{1,+}=\kappa_{1,-}=\operatorname{id}\), so \(\rho_1=-\operatorname{id}\). At the second, \(\kappa_{2,-}=\operatorname{id}\) and \(\kappa_{2,+}=\widetilde h^{-1}\), hence \(\rho_2=h^{-1}\circ(-\operatorname{id})\). Traversing the first transport edge and returning along the second gives
\[
 \rho_1\circ\rho_2^{-1}=(-\operatorname{id})\circ(-\operatorname{id})\circ h=h.
\]
The two-parallel-edge transport component has fundamental group \(\mathbb Z\), so its based formal holonomy is \(\langle h\rangle\); if \(h\ne\operatorname{id}\), its tangency order is \(q(h)\). Complex-valued jets are obtained by complex-linear extension.
\end{proof}

\subsection{Sharpness and limitations}
\label{subsec:holonomy-sharpness}

For every
\(q\geq1\),
the germ \(h_q(T) = T+T^{q+1}\) satisfies
\(q(h_q)=q\).
Theorem~\ref{thm:two-cycle-realisation}
therefore shows that every tangency order occurring in the
classification is geometrically realisable.

Moreover, for every
\(r\geq q\),
\[
  F_{\langle h_q\rangle}^{r\langle q-1\rangle}
  =
  \operatorname{span}\{T^r\}
  \neq0,
\]
whereas
\(F_{\langle h_q\rangle}^{r\langle q\rangle}=0\).

The invariant has the following limitations.

\begin{enumerate}[label=\textup{(\roman*)}]
  \item The fixed-jet inverse system depends only on \(q(H)\). It does
        not detect the leading nonlinear coefficient or any later
        coefficient of the return germs.

  \item Equal tangency order does not imply formal conjugacy. The
        prolongation profile is not a complete invariant of formal
        holonomy.

  \item Orientation-reversing and orientation-preserving
        nonunit-linear return groups are excluded from
        Theorem~\ref{thm:componentwise-prolongation-profile}. Their
        finite and formal fixed spaces require a separate linear and
        resonant analysis.

  \item The classification is scalar and one-dimensional.
        Vector-valued, fibre-coupled, and multivariable formal
        holonomy require separate fixed-space calculations.

  \item The profile measures the difference between compatibility at
        a finite order and compatibility through every formal order.
        It does not determine differential operators preserving those
        spaces or their spectral and dynamical consequences.
\end{enumerate}

Within the stated class of finite smooth graph resolutions,
Theorem~\ref{thm:graph-formal-trace-realisation}
identifies the full formal equaliser with the formal endpoint traces of
the descended smooth core, while
Theorem~\ref{thm:exact-graph-resolution-closure}
identifies its finite projections with the exact integer-order
Sobolev closures. The prolongation profile is therefore an invariant
of the smooth scalar graph resolution and its Sobolev closure tower, rather
than merely a count attached to an abstract formal transport system.

\clearpage
\part{Quantum dynamics on finite graph resolutions}
\label{part:quantum-dynamics}

\section{Quantum operators on graph-resolved domains}
\label{sec:analytic-framework}

For the scalar graph-resolved models, first-order completion gives \(H^1_{\mathrm{cont}}(\Gamma)\), the marked presentation supplies edge weights, and the flat derivative form gives the weighted Kirchhoff Hamiltonian. The same reduced data determine vertex scattering and the boundary form of the oriented derivative.

\subsection{Pushforward measure and edge weights}
\label{subsec:graph-measure}

Theorem~\ref{thm:exact-graph-resolution-closure} and Corollary~\ref{cor:universal-first-order-graph-closure} fix the scalar \(H^1\)-domain; the Hamiltonian additionally requires a measure.

The examples below are equipped with marked analytic graph resolutions in
the sense of Definition~\ref{def:marked-analytic-graph-resolution}. The
measure is specified on the disjoint labelled presentation
\[
  q:\widetilde M=\bigsqcup_{\alpha\in A}M_\alpha\longrightarrow M,
  \qquad
  \widetilde\mu=\bigoplus_{\alpha\in A}\mu_\alpha,
\]
and the graph measure is
\(\mu_\Gamma=(\pi\circ q)_*\widetilde\mu\).
When several labelled resolving representatives contribute to the same
propagation edge, they are counted separately in this pushforward. In the
flat models used below, the resulting Radon--Nikodym density is constant on
each edge; write
\[
  d\mu_\Gamma|_{I_e}=a_e\,dx,
  \qquad
  a_e>0,
  \qquad
  \mathbf a=(a_e)_{e\in E(\Gamma)}.
\]
The weights belong to the marked presentation, not the bare quotient. Write \(\mu_{\mathbf a}:=\mu_\Gamma\)
in the constant-edge case. The scalar Hilbert space is
\[
  \mathcal H_{\mathbf a}
  :=
  L^2(\Gamma,\mu_{\mathbf a})
  =
  \bigoplus_{e\in E(\Gamma)}
  L^2(I_e,a_e\,dx),
\]
with inner product
\[
  \langle u,v\rangle_{\mathbf a}
  =
  \sum_{e}
  a_e
  \int_{I_e}
  \overline{u_e(x)}\,v_e(x)\,dx.
\]

Because \(\Gamma\) has finitely many edges and every weight is strictly
positive, the weighted and unweighted direct-sum \(H^1\)-norms are
equivalent. Hence
\(\overline{\mathscr A_M}^{\,H^1_{\mathbf a}} = H^1_{\mathrm{cont}}(\Gamma)\).

Only the projective weight class matters for the flux laws and scattering amplitudes below.

\subsection{Resolution-induced kinetic form and weighted Kirchhoff Hamiltonian}
\label{subsec:weighted-form}

We use units
\(\hbar=2m_{\mathrm{particle}}=1\).

On \(\mathcal V_{\mathbf a} := H^1_{\mathrm{cont}}(\Gamma)\) define
\[
  \mathfrak q_{\mathbf a}[u,v]
  :=
  \sum_{e}
  a_e
  \int_{I_e}
  \overline{u_e'(x)}\,v_e'(x)\,dx.
\]

\begin{proposition}
\label{prop:weighted-form-closed}
The form
\[
  \mathfrak q_{\mathbf a}
\]
is densely defined, nonnegative, symmetric, and closed in
\(\mathcal H_{\mathbf a}\).
\end{proposition}

\begin{proof}
Its form norm is the weighted \(H^1\)-norm on
\(H^1_{\mathrm{cont}}(\Gamma)\). Vertex evaluation is continuous on each
one-dimensional \(H^1\)-edge space, so the finite system of continuity
conditions defining the form domain is closed.
\end{proof}

Let \(H_{\mathbf a}\) be the self-adjoint operator associated with this form.

\begin{theorem}[Weighted Kirchhoff realisation]
\label{thm:weighted-kirchhoff}
The operator \(H_{\mathbf a}\) acts edgewise as
\[
  (H_{\mathbf a}u)_e=-u_e''.
\]
Its domain consists of
\[
  u\in H^2_\oplus(\Gamma)
\]
that are continuous at every vertex and satisfy
\[
  \sum_{\epsilon\in\mathcal E(v)}
  a_{e(\epsilon)}
  \partial_{\nu_\epsilon}u(v)
  =
  0
\]
for every
\[
  v\in V(\Gamma).
\]
\end{theorem}

\begin{proof}
Testing the form identity against functions supported in the interior of an
edge gives
\((H_{\mathbf a}u)_e=-u_e''\).
For
\(\phi\in H^1_{\mathrm{cont}}(\Gamma)\),
edgewise integration by parts gives
\[
  \mathfrak q_{\mathbf a}[\phi,u]
  =
  \langle \phi,-u''\rangle_{\mathbf a}
  -
  \sum_{v\in V(\Gamma)}
  \overline{\phi(v)}
  \sum_{\epsilon\in\mathcal E(v)}
  a_{e(\epsilon)}
  \partial_{\nu_\epsilon}u(v).
\]
The vertex values \(\phi(v)\) may be varied independently, so the boundary
term vanishes for every \(\phi\) precisely under the stated weighted
Kirchhoff conditions. The converse follows from the same identity.
\end{proof}

The weights therefore modify the vertex flux law while leaving the local
free differential expression unchanged. When \(a_e=1\) on every edge, one recovers the standard Kirchhoff quantum-graph Laplacian
\citep{Kuchment04,BerkolaikoKuchment2013}.

\subsection{Vertex scattering and reflectionless balance}
\label{subsec:vertex-scattering}

Consider a vertex \(v\) of degree \(d\) with incident edge weights
\(a_1,\ldots,a_d>0\).
Represent the incident edges locally by half-lines \(x_j\geq0\) with \(x_j=0\) at \(v\), and set
\(A_v := \sum_{j=1}^{d}a_j\).

The weighted Kirchhoff conditions determine an energy-independent local
scattering matrix in the flux-normalized channel basis.

\begin{theorem}[Flux-normalized weighted Kirchhoff scattering matrix]
\label{thm:weighted-vertex-S}
Let
\[
  \mathbf v_{\mathbf a}
  :=
  \begin{pmatrix}
    \sqrt{a_1}\\
    \vdots\\
    \sqrt{a_d}
  \end{pmatrix}.
\]
Then
\[
  S_v
  =
  \frac{2}{A_v}
  \mathbf v_{\mathbf a}
  \mathbf v_{\mathbf a}^{*}
  -
  I_d.
\]
Equivalently,
\[
  (S_v)_{ji}
  =
  \frac{2\sqrt{a_ja_i}}{A_v}
  -
  \delta_{ji}.
\]
Moreover,
\[
  S_v^*=S_v,
  \qquad
  S_v^2=I_d.
\]
\end{theorem}

\begin{proof}
For incidence on channel \(i\), continuity gives one common vertex
amplitude. The weighted Kirchhoff condition yields the unnormalized
coefficients
\(r_i=\frac{2a_i}{A_v}-1, \qquad t_{ji}=\frac{2a_i}{A_v} \quad (j\neq i)\).
Flux normalization multiplies transmission from \(i\) to \(j\) by
\(\sqrt{\frac{a_j}{a_i}}\),
giving the stated matrix.

Since
\(\mathbf v_{\mathbf a}^{*}\mathbf v_{\mathbf a}=A_v\),
the matrix \(\frac{1}{A_v} \mathbf v_{\mathbf a}\mathbf v_{\mathbf a}^{*}\) is the orthogonal projection onto
\(\operatorname{span}\{\mathbf v_{\mathbf a}\}\). Hence \(S_v=2P-I_d\) is a self-adjoint unitary.
\end{proof}

\begin{corollary}[Reflectionless channel]
\label{cor:reflectionless-channel}
Channel \(i\) is reflectionless if and only if
\[
  a_i
  =
  \sum_{j\neq i}a_j.
\]
Under this condition,
\[
  |(S_v)_{ji}|^2
  =
  \frac{a_j}{a_i},
  \qquad
  j\neq i.
\]
\end{corollary}

In the branching models this balance follows from the pushed-forward presentation measure.

\subsection{The oriented first-order operator}
\label{subsec:scalar-first-order-framework}

After orienting the edges, the same graph and weights define \(D_{\mathbf a}=-i\,d/dx\). Momentum and Dirac-type graph operators have an independent extension theory \citep{FullingKuchmentWilson07,Post09,ExnerMomentum13,Richtsfeld24}; here the resolution selects the orientation, weights, and closed trace domain.

On every oriented edge set
\(D_{\mathbf a}u := -i\,u'\),
with domain
\(\mathcal D(D_{\mathbf a}) = H^1_{\mathrm{cont}}(\Gamma)\).

\begin{proposition}[Density and closedness of the oriented derivative]
\label{prop:first-order-closed-dense}
The operator
\[
  D_{\mathbf a}:
  H^1_{\mathrm{cont}}(\Gamma)
  \longrightarrow
  \mathcal H_{\mathbf a},
  \qquad
  D_{\mathbf a}u=-iu',
\]
is densely defined and closed. Its graph norm is equivalent to the weighted
direct-sum \(H^1\)-norm on \(H^1_{\mathrm{cont}}(\Gamma)\).
\end{proposition}

\begin{proof}
The subspace \(C_c^\infty(\Gamma\setminus V(\Gamma))\subset H^1_{\mathrm{cont}}(\Gamma)\) is dense in \(\mathcal H_{\mathbf a}\), so \(D_{\mathbf a}\) is densely defined. Moreover,
\[
 \|u\|_{\mathbf a}^2+\|D_{\mathbf a}u\|_{\mathbf a}^2
 =\sum_e a_e\bigl(\|u_e\|_{L^2}^2+\|u_e'\|_{L^2}^2\bigr),
\]
which is the weighted direct-sum \(H^1\)-norm. The finitely many vertex-continuity conditions are kernels of bounded endpoint traces, hence define a closed subspace of \(H_\oplus^1(\Gamma)\). Completeness in this graph norm proves closedness.
\end{proof}

For a vertex \(v\), let \(E_{\mathrm{in}}(v)\) be the incident edges whose orientation terminates at \(v\), and \(E_{\mathrm{out}}(v)\) those whose orientation begins at \(v\).

\begin{proposition}[First-order boundary form]
\label{prop:first-order-boundary-form-general}
For
\[
  u,v\in H^1_{\mathrm{cont}}(\Gamma),
\]
\[
  \langle D_{\mathbf a}u,v\rangle_{\mathbf a}
  -
  \langle u,D_{\mathbf a}v\rangle_{\mathbf a}
  =
  i
  \sum_{w\in V(\Gamma)}
  \left(
    \sum_{e\in E_{\mathrm{in}}(w)}a_e
    -
    \sum_{e\in E_{\mathrm{out}}(w)}a_e
  \right)
  \overline{u(w)}\,v(w).
\]
\(D_{\mathbf a}\) is symmetric if and only if
\[
  \sum_{e\in E_{\mathrm{in}}(v)}a_e
  =
  \sum_{e\in E_{\mathrm{out}}(v)}a_e
\]
at every vertex.
\end{proposition}

\begin{proof}
On an oriented compact edge \(e=(0,\ell_e)\),
\[
  \langle -iu_e',v_e\rangle_{a_e}
  -
  \langle u_e,-iv_e'\rangle_{a_e}
  =
  i a_e
  \left[
    \overline{u_e(x)}v_e(x)
  \right]_{0}^{\ell_e}.
\]
The corresponding term at infinity vanishes on an external half-line.
Summing over the graph and using continuity of the vertex traces gives the
stated formula.
\end{proof}

First- and second-order balance are distinct: weighted in--out balance makes the oriented derivative symmetric, whereas the analogous Kirchhoff balance can make a chosen channel reflectionless; neither implies first-order self-adjointness.

\subsection{Geometry-selected and additional operator data}
\label{subsec:geometry-versus-vertex-data}

Quantum graphs admit many other self-adjoint vertex couplings \citep{KostrykinSchrader99,BerkolaikoKuchment2013}; none is assumed here.

For scalars, the resolution supplies the propagation graph and smooth core; first-order closure gives the continuous graph domain, the marked measure supplies weights, and the positive derivative form fixes the weighted Kirchhoff Hamiltonian without an additional vertex parameter.

Bundles can modify the \(H^1\)-trace domain, while equal nonzero first-order deficiency indices admit additional self-adjoint extension parameters. Such parameters are not part of the resolution unless selected by further geometry.

The subsequent quantum-graph formulas are used only after identifying which reduced data the non-Hausdorff construction selects.

\section{Return to the line with two origins: Sobolev closure and dynamics}
\label{sec:L2-return}

\subsection{Sobolev closure of the relative-sign core}
\label{subsec:L2-H1-return}

At first Sobolev order this infinite system collapses to one trace
condition.

\begin{theorem}[\(H^1\)-closure of the relative-sign core]
\label{thm:L2-H1-core}
If
\[
  \mathcal C_-
  :=
  \mathcal F^\infty_0(\mathbb R)
  \cap
  C_c^\infty(\mathbb R),
\]
then
\[
  \overline{\mathcal C_-}^{\,H^1(\mathbb R)}
  =
  H^1_{\mathrm{tr},0}(\mathbb R)
  :=
  \left\{
    \psi\in H^1(\mathbb R):
    \psi(0)=0
  \right\}.
\]
Equivalently,
\[
  H^1_{\mathrm{tr},0}(\mathbb R)
  =
  H^1_0((-\infty,0))
  \oplus
  H^1_0((0,\infty)).
\]
\end{theorem}

\begin{proof}
Continuity of the point trace on \(H^1(\mathbb R)\) gives the inclusion into
\(H^1_{\mathrm{tr},0}(\mathbb R)\). Conversely, a function with zero trace
restricts to \(H^1_0\) on each half-line and is approximable there by smooth
functions compactly supported away from the origin. Their direct sums belong
to \(\mathcal C_-\). The half-line decomposition follows from the same trace
characterisation.
\end{proof}

Here infinite-order smooth compatibility reduces, at the kinetic-form level, to a single value-trace condition.

\subsection{First- and second-order dynamics}
\label{subsec:L2-dynamics}

On the relative-sign smooth core define
\(D_0=-i\frac{d}{dx}\).
Its closure is the minimal derivative
\[
  D_{\min}\psi=-i\psi',
  \qquad
  \mathcal D(D_{\min})
  =
  H^1_{\mathrm{tr},0}(\mathbb R),
\]
while
\(\mathcal D(D_{\min}^*) = H^1((-\infty,0)) \oplus H^1((0,\infty))\).

\begin{theorem}[Deficiency indices]
\label{thm:L2-deficiency}
The deficiency indices of \(D_{\min}\) are
\[
  (n_+,n_-)=(1,1).
\]
\end{theorem}

\begin{proof}
For
\((D_{\min}^*-i)\psi=0\),
the only square-integrable mode is proportional to \(e^{-x}\) on the
positive half-line. For
\((D_{\min}^*+i)\psi=0\),
the only square-integrable mode is proportional to \(e^x\) on the negative
half-line.
\end{proof}

The self-adjoint extensions are therefore
\(\psi(0^+)=e^{i\alpha}\psi(0^-), \qquad \alpha\in[0,2\pi)\),
and the phase \(\alpha\) is additional operator-domain data rather than a
parameter fixed by the spin structure.

\begin{proposition}[Self-adjointness versus global twisted sections]
\label{prop:L2-SA-not-global}
No self-adjoint extension of \(D_{\min}\) has domain contained in the space
of continuous global sections of the relative-sign spinor line.
\end{proposition}

\begin{proof}
Every continuous global relative-sign section has both one-sided traces
equal to zero. A self-adjoint extension requires a maximal isotropic
one-dimensional boundary relation and therefore admits nonzero traces.
Restricting both traces to zero recovers the minimal symmetric domain,
which is not self-adjoint because its deficiency indices are \((1,1)\).
\end{proof}

The positive derivative form on the same geometric core closes on \(H^1_{\mathrm{tr},0}(\mathbb R)\) and therefore selects a unique Friedrichs Hamiltonian.

\begin{theorem}[Friedrichs realisation]
\label{thm:L2-Friedrichs}
The resolution-induced positive Hamiltonian in the relative-sign sector is
\[
  H_{\mathrm F}
  =
  H_{D,-}\oplus H_{D,+},
\]
the direct sum of the Dirichlet Laplacians on the two half-lines.
Equivalently,
\[
  H_{\mathrm F}=D_{\min}^*D_{\min}.
\]
\end{theorem}

\begin{proof}
The closed form is the orthogonal sum of
\(\int_{-\infty}^0|u'|^2\,dx\) and \(\int_0^\infty|u'|^2\,dx\) on the two
half-line \(H^1_0\)-spaces. The representation theorem for closed
nonnegative forms gives the two Dirichlet Laplacians. The factorisation by
the closed minimal derivative gives \(H_{\mathrm F}=D_{\min}^*D_{\min}\).
\end{proof}

\(R(k)=-1\) and \(T(k)=0\): the scalar sector is blind to the doubled origin, while the relative-sign sector survives through its closed trace domain.

\section{\texorpdfstring{The line with \(k\) origins}{The line with k origins}}
\label{sec:Lk}

For \(k>2\), the exceptional fibre carries more rank-one transport data while the regular scalar propagation geometry remains unchanged.

\subsection{Geometry and scalar reduction}
\label{subsec:Lk-geometry}

Fix \(k\geq2\) and define
\[
  \mathbb L_k
  :=
  \left(
    \bigsqcup_{i=1}^{k}\mathbb R_i
  \right)\big/\!\sim,
  \qquad
  (x,i)\sim(x,j)
  \quad
  (x\neq0).
\]
The \(k\) origins \(o_i=[(0,i)]\) remain distinct. With
\(U_i := \mathbb L_k\setminus\{o_j:j\neq i\}\),
one has \(U_i\cong\mathbb R\) and
\[
  U_i\cap U_j
  \cong
  \mathbb R_-\sqcup\mathbb R_+,
  \qquad
  \mathbb R_-:=(-\infty,0),
  \quad
  \mathbb R_+:=(0,\infty).
\]
We use the identity-glued smooth structure on both overlap components.

The Hausdorff reflection
\(q_k:\mathbb L_k\longrightarrow\mathbb R, \qquad q_k([x])=x, \quad q_k(o_i)=0\),
collapses precisely the exceptional fibre. Pairwise inseparability of the
origins implies the universal Hausdorff factorisation property, so \(q_k\)
is the Hausdorff reflection. For the identity-glued smooth structure,
\[
  C^\infty(\mathbb L_k,\mathbb C)
  \cong
  C^\infty(\mathbb R,\mathbb C),
  \qquad
  L^2(\mathbb L_k)
  \cong
  L^2(\mathbb R)
\]
for the normalized measure transported from the regular locus. Hence the
resolution-induced scalar Hamiltonian is the ordinary free Laplacian on
\(H^2(\mathbb R)\), independently of \(k\).

\(\mathbb L_k^\circ \cong \mathbb R_-\sqcup\mathbb R_+\) has exactly two scalar propagation regions for every finite \(k\):
exceptional-fibre multiplicity does not by itself create additional scalar
channels.

\subsection{Rank-one transport and classification}
\label{subsec:Lk-labelled-classification}
\label{subsec:Lk-unitary-lines}

Equip \(\mathbb L_k\) with the orientation induced by increasing \(x\) and
the flat metric \(dx^2\). Since
\(\operatorname{Spin}(1)\cong\mathbb Z_2\),
a spin structure is represented by locally constant transition signs on the
two components of each overlap. The negative-side cocycle can be
trivialised by local frame changes. The remaining positive-side cocycle is
therefore represented by a sign vector
\(\boldsymbol{\varepsilon} = (\varepsilon_1,\ldots,\varepsilon_k) \in\{\pm1\}^k\),
defined modulo simultaneous multiplication by a common sign.

\begin{theorem}[Rank-one spin classification]
\label{thm:Lk-labelled-spin}
The labelled spin structures on \(\mathbb L_k\) are classified by
\[
  \mathscr S(\mathbb L_k)
  \cong
  \{\pm1\}^k/\{\pm1\}
  \cong
  (\mathbb Z_2)^{k-1}.
\]
Hence
\[
  \#\mathscr S(\mathbb L_k)=2^{k-1}.
\]
Under orientation-preserving diffeomorphisms, the classes are determined by
\[
  m(\boldsymbol{\varepsilon})
  =
  \min\left\{
    \#\{i:\varepsilon_i=+1\},
    \#\{i:\varepsilon_i=-1\}
  \right\},
\]
and therefore form
\[
  \left\lfloor\frac{k}{2}\right\rfloor+1
\]
orbits.
\end{theorem}

\begin{proof}
After negative-side trivialisation, the cocycle relation gives
\(g_{ij}^+ = \varepsilon_i\varepsilon_j^{-1}\).
A common sign leaves every relative transition unchanged, giving the
labelled quotient. Every permutation of the origins is realised by an
orientation-preserving diffeomorphism acting identically on the regular
coordinate, while the common-sign identification exchanges \(p\) positive
entries with \(k-p\). The stated value of \(m\) is therefore the complete
orbit invariant.
\end{proof}

Locally constant \(U(1)\)-transitions similarly give \(U(1)^{k-1}\), with the spin structures as its \(2\)-torsion. The \(H^1\)-reduction below is much coarser than this cocycle classification.

\subsection{Smooth compatibility and Sobolev reduction}
\label{subsec:Lk-Sobolev-collapse}

Fix a normalized sign vector
\(\boldsymbol{\varepsilon}\). On the regular locus, a spinor is represented
by
\(\psi_-\oplus\psi_+\),
and smoothness in the \(i\)-th exceptional chart requires
\(\psi_-^{(r)}(0^-) = \varepsilon_i\psi_+^{(r)}(0^+) \qquad (r\geq0)\).
For the trivial class this is ordinary full-jet matching. If the class is
nontrivial, choose \(i,j\) with
\(\varepsilon_i\neq\varepsilon_j\); comparison of the two relations forces
\(\psi_-^{(r)}(0^-) = \psi_+^{(r)}(0^+) = 0 \qquad (r\geq0)\).
Exactly the same argument applies to a nonconstant vector of \(U(1)\)
transport phases.

\begin{theorem}[Sobolev reduction of nontrivial rank-one sectors]
\label{thm:Lk-H1-collapse}
The trivial rank-one sector closes to
\[
  H^1(\mathbb R).
\]
Every nontrivial spin sector, and every nonconstant locally constant
\(U(1)\)-transport class considered above, closes to
\[
  H^1_0(\mathbb R_-)
  \oplus
  H^1_0(\mathbb R_+).
\]
\end{theorem}

\begin{proof}
The nontrivial smooth core has zero one-sided traces, so its closure is
contained in the displayed space. Conversely, \(C_c^\infty(\mathbb R_-) \oplus C_c^\infty(\mathbb R_+)\) is contained in every nontrivial smooth core and is dense in that direct
sum.
\end{proof}

\(2^{k-1}\) labelled spin classes and a continuous
\(U(1)^{k-1}\)-family reduce, at first Sobolev order, to the same two trace
domains distinguished only by whether the local rank-one transport is
coherent.

\subsection{First- and second-order consequences}
\label{subsec:Lk-first-order}
\label{subsec:Lk-positive-dynamics}

In the trivial sector, \(-i\frac{d}{dx} \quad\text{on }H^1(\mathbb R)\) is the ordinary self-adjoint momentum operator, and the positive derivative
form gives the free Laplacian.

Every nontrivial sector yields the same closed first-order operator
\[
  D_{\mathrm{ref}}
  =
  -i\frac{d}{dx},
  \qquad
  \mathcal D(D_{\mathrm{ref}})
  =
  H^1_0(\mathbb R_-)
  \oplus
  H^1_0(\mathbb R_+).
\]
This is exactly the minimal two-half-line derivative already encountered on
\(\mathbb L_2\). \((n_+,n_-)=(1,1)\),
and its self-adjoint extensions are
\(\psi(0^+)=e^{i\theta}\psi(0^-), \qquad \theta\in[0,2\pi)\).
The phase \(\theta\) is additional operator-domain data; it is not selected
by the rank-one transport class.

The associated positive form has domain \(H^1_0(\mathbb R_-) \oplus H^1_0(\mathbb R_+)\) and therefore produces
\(H_{\mathrm{ref}} = H_{D,-}\oplus H_{D,+}\).
Hence every nontrivial sector has
\(R(p)=-1, \qquad T(p)=0, \qquad p>0\),
while the trivial sector is the free line.

The \(2^{k-1}\) labelled spin classes therefore produce only two
resolution-induced Friedrichs regimes. Increasing the multiplicity of a
single exceptional fibre enlarges the bundle classification without
enlarging the scalar propagation geometry or the number of resulting
\(H^1\)-domains.

\section{Finite multiple-origin defects on the line}
\label{sec:multisplit-line}

With finitely many split coordinates, active nodal defects can partition the line and isolate bounded reducing components.

\subsection{Geometry and rank-one data}
\label{subsec:multisplit-definition}
\label{subsec:multisplit-spin-classification}

Fix
\(x_1<\cdots<x_N, \qquad k_j\geq2\),
and replace \(x_j\) by \(k_j\) origins
\(o_{j,1},\ldots,o_{j,k_j}\).
Denote the resulting identity-glued one-manifold by
\[
  \mathbb L_{\mathbf k}(\mathbf x),
  \qquad
  \mathbf x=(x_1,\ldots,x_N),
  \quad
  \mathbf k=(k_1,\ldots,k_N),
\]
with natural projection
\[
  q:
  \mathbb L_{\mathbf k}(\mathbf x)
  \longrightarrow
  \mathbb R,
  \qquad
  q(o_{j,\alpha})=x_j.
\]

\begin{proposition}
\label{prop:multisplit-basic}
The space \(\mathbb L_{\mathbf k}(\mathbf x)\) is a connected,
second-countable, \(T_1\), locally Euclidean one-manifold. Its
non-Hausdorff pairs are precisely the distinct origins in a common fibre
\(q^{-1}(x_j)\), and \(q\) is its Hausdorff reflection.
\end{proposition}

\begin{proof}
The standard exceptional charts give local Euclideanity and a countable
basis. Origins in a common fibre have intersecting punctured
neighbourhoods; fibres over distinct coordinates admit disjoint coordinate
neighbourhoods. Pairwise inseparability forces every Hausdorff-valued
continuous map to be constant on each exceptional fibre, and the quotient
property gives the universal factorisation through \(q\).
\end{proof}

Accordingly,
\[
  C^\infty(
    \mathbb L_{\mathbf k}(\mathbf x)
  )
  \cong
  C^\infty(\mathbb R),
  \qquad
  L^2(
    \mathbb L_{\mathbf k}(\mathbf x)
  )
  \cong
  L^2(\mathbb R)
\]
for the normalized regular-locus measure, so the scalar free Hamiltonian is
the ordinary Laplacian.

The regular locus has components
\[
  I_0=(-\infty,x_1),
  \qquad
  I_j=(x_j,x_{j+1}),
  \quad
  1\leq j\leq N-1,
  \qquad
  I_N=(x_N,\infty).
\]

For rank-one spin transport, the fibre over \(x_j\) is represented after
left-side normalization by
\(\boldsymbol{\varepsilon}_j \in \{\pm1\}^{k_j}/\{\pm1\}\).
Because the Hausdorff incidence graph is a line, the local data at distinct
coordinates are independent.

\begin{theorem}[Classification of labelled multi-split spin structures]
\label{thm:multisplit-spin-classification}
The labelled spin structures are classified by
\[
  \mathscr S(
    \mathbb L_{\mathbf k}(\mathbf x)
  )
  \cong
  \prod_{j=1}^{N}
  (\mathbb Z_2)^{k_j-1},
\]
and hence
\[
  \#\mathscr S
  =
  2^{\sum_{j=1}^{N}(k_j-1)}.
\]
\end{theorem}

\begin{proof}
Apply Theorem~\ref{thm:Lk-labelled-spin} independently at each split
coordinate. Since the Hausdorff incidence graph is a line, there is no cycle
relation coupling the local cocycles, so the labelled classification is the
direct product of the local factors.
\end{proof}

At \(x_j\), call the local transport \emph{transparent} when the normalized
sign vector is constant and \emph{active} otherwise. Let
\(D := \{j:x_j\text{ is active}\}, \qquad m:=|D|\).
The finer orientation-preserving classification remembers, at each
coordinate, the number
\(m_j = \min\{p_j,k_j-p_j\}\),
where \(p_j\) is the number of positive signs. Hence the geometric classes
are indexed by
\((m_1,\ldots,m_N), \qquad 0\leq m_j\leq \left\lfloor\frac{k_j}{2}\right\rfloor\),
and their number is
\(\prod_{j=1}^{N} \left( \left\lfloor\frac{k_j}{2}\right\rfloor+1 \right)\).
Only the dichotomy \(m_j=0\) versus \(m_j>0\) survives the
\(H^1\)-reduction below.

\subsection{\texorpdfstring{Compatibility and \(H^1\)-reduction}{Compatibility and H1-reduction}}
\label{subsec:multisplit-H1}

Write a regular-locus section as
\(\psi=(\psi_0,\ldots,\psi_N)\).
Smoothness through \(o_{j,\alpha}\) requires
\[
  \psi_{j-1}^{(r)}(x_j^-)
  =
  \varepsilon_{j,\alpha}
  \psi_j^{(r)}(x_j^+)
  \qquad
  (r\geq0).
\]

\begin{theorem}[Local compatibility dichotomy]
\label{thm:multisplit-local-dichotomy}
At a transparent coordinate,
\[
  \psi_{j-1}^{(r)}(x_j^-)
  =
  \psi_j^{(r)}(x_j^+)
  \qquad
  (r\geq0).
\]
At an active coordinate,
\[
  \psi_{j-1}^{(r)}(x_j^-)
  =
  \psi_j^{(r)}(x_j^+)
  =
  0
  \qquad
  (r\geq0).
\]
\end{theorem}

\begin{proof}
Only the active case is nontrivial. Two different local signs give two
expressions for the same left jet; their difference forces the right jet to
vanish, and hence the left jet also vanishes.
\end{proof}

Define
\[
  \mathcal V_D
  :=
  \left\{
    \psi\in\bigoplus_{r=0}^{N}H^1(I_r):
    \begin{array}{ll}
      \psi_{j-1}(x_j^-)=\psi_j(x_j^+),
        &j\notin D,\\[1mm]
      \psi_{j-1}(x_j^-)=\psi_j(x_j^+)=0,
        &j\in D
    \end{array}
  \right\}.
\]

\begin{theorem}[\(H^1\)-reduction by the active set]
\label{thm:multisplit-H1-closure}
The smooth spinor core with active set \(D\) has \(H^1\)-closure
\[
  \mathcal V_D.
\]
The complete local sign vectors disappear at first Sobolev
order once the active subset is fixed.
\end{theorem}

\begin{proof}
Trace continuity gives one inclusion. For the converse, glue across
transparent coordinates and approximate independently on the components
separated by active nodes, where the traces vanish.
\end{proof}

The smooth regular-locus cores already collapse to the
\(2^N\) possibilities indexed by the active subset
\(D\subseteq\{1,\ldots,N\}\).

The same reduction holds for locally constant \(U(1)\)-transport: the
complete parameter space is
\(\prod_{j=1}^{N}U(1)^{k_j-1}\),
a local phase vector is transparent when it is constant modulo a common
phase, and it is active otherwise.

If the active coordinates are
\(a_1<\cdots<a_m\),
then transparent split points disappear under gluing and
\(\mathcal V_D \cong \bigoplus_{r=0}^{m}H^1_0(J_r)\),
where
\[
  J_0=(-\infty,a_1),
  \quad
  J_r=(a_r,a_{r+1}),
  \quad
  1\leq r\leq m-1,
  \quad
  J_m=(a_m,\infty).
\]
For \(D=\varnothing\), one recovers \(H^1(\mathbb R)\).

\subsection{Friedrichs dynamics and embedded cavities}
\label{subsec:multisplit-Friedrichs}
\label{subsec:multisplit-cavities}

The positive derivative form on \(\mathcal V_D\) gives
\(H_D = \bigoplus_{r=0}^{m} H_{D,J_r}\),
with the convention that \(D=\varnothing\) yields the free line. The
resulting second-order operator depends on the rank-one transport only
through the positions of the active defects.

If \(m\geq2\), the bounded components
\(J_r=(a_r,a_{r+1}), \qquad 1\leq r\leq m-1\),
have lengths \(L_r=a_{r+1}-a_r\) and standard Dirichlet eigenvalues
\(\lambda_{r,n} = \left( \frac{\pi n}{L_r} \right)^2, \qquad n\geq1\).
The corresponding eigenfunctions, extended by zero, are eigenfunctions of
the full operator. Since the two exterior components are Dirichlet
half-lines,
\(\sigma_{\mathrm{ess}}(H_D)=[0,\infty)\),
so these compactly supported cavity levels are embedded eigenvalues. Their
multiplicities are determined by coincidences among the interval Dirichlet
spectra.

Exterior scattering contains no further structure. For nonempty \(D\), a
wave incident from the left is reflected at \(a_1\):
\(R_L(p)=-e^{2ipa_1}, \qquad T_{L\to R}(p)=0\).
Additional active defects alter only the internal reducing components.

\subsection{First-order dynamics}
\label{subsec:multisplit-first-order}

For nonempty \(D\), the resolution-induced first-order operator is
\[
  D_{\min}
  =
  \bigoplus_{r=0}^{m}
  \left(
    -i\frac{d}{dx}
  \right),
  \qquad
  \mathcal D(D_{\min})
  =
  \bigoplus_{r=0}^{m}H^1_0(J_r).
\]
The left half-line contributes \((0,1)\), the right half-line contributes
\((1,0)\), and the \(m-1\) bounded components each contribute \((1,1)\).
Hence
\((n_+,n_-)=(m,m)\).
For \(m=0\), the closed derivative is the ordinary self-adjoint momentum.

Writing
\[
  \Psi_-
  =
  \bigl(
    \psi(a_1^-),\ldots,\psi(a_m^-)
  \bigr)^{\mathsf T},
  \qquad
  \Psi_+
  =
  \bigl(
    \psi(a_1^+),\ldots,\psi(a_m^+)
  \bigr)^{\mathsf T},
\]
the self-adjoint extensions are
\(\Psi_+=U\Psi_-, \qquad U\in U(m)\).
Diagonal \(U\) gives local phase reconnection at each defect; non-diagonal
\(U\) can couple traces belonging to different coordinates. These unitary
parameters are additional first-order operator data, not information fixed
by the original non-Hausdorff rank-one transport.

\section{Multiple-origin intervals}
\label{sec:multiple-origin-intervals}

On a bounded interval, the interior-supported smooth core imposes Dirichlet traces at \(0,L\), while active interior defects subdivide the spectrum.

\subsection{Closed domain and active-coordinate normal form}
\label{subsec:multiinterval-H1}

Fix
\(0<x_1<\cdots<x_N<L, \qquad k_j\geq2\),
and let \(\mathbb I_{\mathbf k}(\mathbf x;L)\) be obtained by replacing \(x_j\) with \(k_j\) inseparable origins. Its
Hausdorff reflection is \([0,L]\), and the local rank-one classification is
the restriction of
Theorem~\ref{thm:multisplit-spin-classification}. Let
\(D := \{j:x_j\text{ is active}\}\).

By Theorem~\ref{thm:multisplit-local-dichotomy}, transparent coordinates
impose ordinary matching and active coordinates impose zero one-sided
traces. The chosen boundary subcore adds
\(\psi(0)=\psi(L)=0\).
Hence
\[
  \overline{\mathcal C_D}^{\,H^1_\oplus}
  =
  \left\{
    \psi\in\bigoplus_{r=0}^{N}H^1(I_r):
    \begin{array}{l}
      \psi(0)=\psi(L)=0,\\
      \psi_{j-1}(x_j^-)=\psi_j(x_j^+),\ j\notin D,\\
      \psi_{j-1}(x_j^-)=\psi_j(x_j^+)=0,\ j\in D
    \end{array}
  \right\}.
\]

List the active coordinates as
\(0<a_1<\cdots<a_m<L, \qquad m=|D|\),
and put
\(a_0=0, \qquad a_{m+1}=L, \qquad J_r=(a_r,a_{r+1}), \qquad \ell_r=a_{r+1}-a_r\).
After gluing across transparent points,
\(\overline{\mathcal C_D}^{\,H^1_\oplus} \cong \bigoplus_{r=0}^{m}H^1_0(J_r)\).
The entire second-order problem is reduced to a Dirichlet subdivision
with
\(\ell_r>0, \qquad \sum_{r=0}^{m}\ell_r=L\).

\subsection{Spectral consequences of Dirichlet subdivision}
\label{subsec:multiinterval-spectrum}

The associated Friedrichs operator is
\(H_D = \bigoplus_{r=0}^{m} \left( -\frac{d^2}{dx^2} \right)_{D,J_r}\),
with domain
\(\bigoplus_{r=0}^{m} \left( H^2(J_r)\cap H^1_0(J_r) \right)\).
Accordingly,
\[
  \sigma(H_D)
  =
  \biguplus_{r=0}^{m}
  \left\{
    \left(
      \frac{\pi n}{\ell_r}
    \right)^2:
    n\in\mathbb N
  \right\}.
\]
The spatial ordering of the cells is absent from this direct-sum spectrum:
permuting the lengths gives a unitarily equivalent operator.

Two distinct cells contribute a common eigenvalue precisely when their
length ratio is rational. In particular, pairwise irrational ratios give
simple spectrum. Equal spacing,
\(\ell_0=\cdots=\ell_m=\frac{L}{m+1}\),
gives \(\lambda_n = \left( \frac{\pi n(m+1)}{L} \right)^2\) with multiplicity \(m+1\).

The ground-state energy is
\(E_0(D) = \frac{\pi^2}{\ell_{\max}^2}, \qquad \ell_{\max}:=\max_r\ell_r\).
Since
\(\ell_{\max}\geq\frac{L}{m+1}\),
one obtains
\(E_0(D) \leq \frac{\pi^2(m+1)^2}{L^2}\),
with equality exactly for equal spacing. Among configurations with
\(m\) active defects, equal spacing uniquely maximizes the ground-state
energy.

Insertion of additional active defects imposes further Dirichlet trace
constraints. If \(D_1\subseteq D_2\), the min--max principle gives
\(\lambda_n(D_1) \leq \lambda_n(D_2)\).
If \(D_2\) adds \(r\) active coordinates, the form-domain inclusion has
codimension \(r\), yielding the interlacing bound
\(\lambda_n(D_1) \leq \lambda_n(D_2) \leq \lambda_{n+r}(D_1)\).

\subsection{Counting function and heat trace}
\label{subsec:multiinterval-counting}
\label{subsec:multiinterval-heat}

The direct-sum form gives
\[
  N_D(\Lambda)
  =
  \sum_{r=0}^{m}
  \left\lfloor
    \frac{\ell_r\sqrt{\Lambda}}{\pi}
  \right\rfloor.
\]
Comparing with
\[
  N_\varnothing(\Lambda)
  =
  \left\lfloor
    \frac{L\sqrt{\Lambda}}{\pi}
  \right\rfloor
\]
and using the elementary floor inequality gives
\(0 \leq N_\varnothing(\Lambda)-N_D(\Lambda) \leq m\),
hence
\(N_D(\Lambda) = \frac{L}{\pi}\sqrt{\Lambda} + O(1)\).

Likewise, summing the standard Dirichlet interval heat expansions yields
\[
  \operatorname{Tr}(e^{-tH_D})
  =
  \frac{L}{\sqrt{4\pi t}}
  -
  \frac{m+1}{2}
  +
  O(e^{-c/t})
  \qquad
  (t\downarrow0).
\]
The leading Weyl coefficient retains only total length, whereas the constant
heat coefficient records the number of dynamically separated cells.

\subsection{First-order operator}
\label{subsec:multiinterval-first-order}

On the same closed trace domain, the resolution-induced first-order
operator is
\[
  D_{\min}
  =
  \bigoplus_{r=0}^{m}
  \left(
    -i\frac{d}{dx}
  \right),
  \qquad
  \mathcal D(D_{\min})
  =
  \bigoplus_{r=0}^{m}H^1_0(J_r).
\]
Each bounded component contributes deficiency indices \((1,1)\), so
\((n_+,n_-)=(m+1,m+1)\).
Even the unsplit interval has a one-parameter first-order
self-adjointness problem, and each active interior defect increases both
deficiency dimensions by one.

With left- and right-endpoint trace vectors
\[
  \Psi_L
  =
  \bigl(
    \psi(a_0^+),\ldots,\psi(a_m^+)
  \bigr)^{\mathsf T},
  \qquad
  \Psi_R
  =
  \bigl(
    \psi(a_1^-),\ldots,\psi(a_{m+1}^-)
  \bigr)^{\mathsf T},
\]
the self-adjoint extensions are
\(\Psi_R=U\Psi_L, \qquad U\in U(m+1)\).
These unitary parameters belong to the first-order completion, whereas the
positive Friedrichs operator above requires no additional extension choice.

\section{\texorpdfstring{The circle with \(k\) origins}{The circle with k origins}}
\label{sec:circle-k}

Assume throughout this section that
\(k\geq2\).
For \(k=1\) the construction is the ordinary circle and there is no mixed
origin sector.

Splitting a point on a circle differs from splitting a point on the line
because the regular locus remains connected. The local relative transport
data are accompanied by one global cycle holonomy. The resulting rank-one
theory has two coherent spin sectors and a large family of mixed sectors
that collapse to the same nodal interval after \(H^1\)-completion.

\subsection{Geometry and scalar reduction}
\label{subsec:circle-k-geometry}

Let \(S^1_L=\mathbb R/L\mathbb Z\) and replace one point \(p\) by \(k\) inseparable origins
\(p_1,\ldots,p_k\). Denote the resulting identity-glued manifold by
\(\mathbb C_k(L)\),
with natural projection
\(q:\mathbb C_k(L)\longrightarrow S^1_L\).
Its regular locus is
\(\mathbb C_k(L)^\circ\cong(0,L)\).

The standard exceptional charts show that \(\mathbb C_k(L)\) is compact,
connected, second-countable, \(T_1\), and locally Euclidean. Pairwise
inseparability of the \(p_i\) gives the Hausdorff universal property of
\(q\). \[
  C^\infty(\mathbb C_k(L))
  \cong
  C^\infty(S^1_L),
  \qquad
  L^2(\mathbb C_k(L))
  \cong
  L^2(S^1_L)
\]
for the normalized regular-locus measure. The scalar Hamiltonian is therefore
the ordinary periodic Laplacian on the circle for every finite \(k\geq2\).

\subsection{Rank-one transport and Sobolev trichotomy}
\label{subsec:circle-k-spin-classification}
\label{subsec:circle-k-H1}

Choose a spin frame on the regular interval and one near each origin. The
two components of the overlap at \(p_i\) carry signs
\(a_i^0,a_i^L\in\{\pm1\}\).
The gauge-invariant ratio \(\eta_i := (a_i^L)^{-1}a_i^0 \in\{\pm1\}\) records the transport through the \(i\)-th origin.

\begin{theorem}[Spin classification and Sobolev reduction]
\label{thm:circle-k-spin-classification}
The labelled spin structures are classified by
\[
  \boldsymbol{\eta}
  =
  (\eta_1,\ldots,\eta_k)
  \in
  \{\pm1\}^k,
\]
so
\[
  \mathscr S(\mathbb C_k(L))
  \cong
  (\mathbb Z_2)^k.
\]
Under orientation-preserving diffeomorphisms, the classes are determined by
the number of negative entries, giving \(k+1\) geometric orbits.

Their \(H^1\)-closures have exactly three forms:
\[
  H^1_+(0,L)
  :=
  \{\psi\in H^1(0,L):\psi(L)=\psi(0)\},
\]
\[
  H^1_-(0,L)
  :=
  \{\psi\in H^1(0,L):\psi(L)=-\psi(0)\},
\]
and
\[
  H^1_0(0,L).
\]
The first two occur when all \(\eta_i\) are respectively \(+1\) or
\(-1\); every mixed sign vector gives \(H^1_0(0,L)\).
\end{theorem}

\begin{proof}
The transition ratios are invariant under local frame changes, and every
choice of \(\boldsymbol{\eta}\) is realised. Permuting the origins is the
only orientation-preserving action on these labelled ratios, giving the
orbit count.

For a smooth regular-locus spinor, \(\psi^{(r)}(L) = \eta_i\psi^{(r)}(0) \qquad (r\geq0)\) for every \(i\). If all \(\eta_i\) agree, this gives periodic or
antiperiodic full-jet matching. If two ratios differ, both endpoint jets
vanish to every order. Trace continuity and standard smooth density then
give the three stated \(H^1\)-closures.
\end{proof}

The same argument for locally constant \(U(1)\)-valued transport replaces
\(\eta_i\) by phases \(z_i\in U(1)\). If all phases agree,
\(z_i=e^{i\theta}\),
the closed domain is the ordinary quasiperiodic space
\(H^1_\theta(0,L) := \{ \psi\in H^1(0,L): \psi(L)=e^{i\theta}\psi(0) \}\).
If the phase vector is nonconstant, the closure is \(H^1_0(0,L)\).
The specifically non-Hausdorff feature is the coexistence of incompatible local holonomies at one graph point through inseparable origins; circle holonomy itself is standard.

\subsection{First-order dynamics}
\label{subsec:circle-k-Dirac}

On the coherent domains,
\(D_\pm=-i\frac{d}{dx}, \qquad \mathcal D(D_\pm)=H^1_\pm(0,L)\),
are the standard self-adjoint periodic and antiperiodic momentum operators,
with spectra
\[
  \sigma(D_+)
  =
  \left\{
    \frac{2\pi n}{L}:n\in\mathbb Z
  \right\},
  \qquad
  \sigma(D_-)
  =
  \left\{
    \frac{(2n+1)\pi}{L}:n\in\mathbb Z
  \right\}.
\]
More generally, coherent \(U(1)\) holonomy \(e^{i\theta}\) gives
\[
  \sigma(D_\theta)
  =
  \left\{
    \frac{2\pi n+\theta}{L}:n\in\mathbb Z
  \right\}.
\]

Every mixed sector selects instead
\[
  D_{\mathrm{mix}}
  =
  -i\frac{d}{dx},
  \qquad
  \mathcal D(D_{\mathrm{mix}})=H^1_0(0,L).
\]
Its adjoint is the maximal derivative on \(H^1(0,L)\), and
\((n_+,n_-)=(1,1)\).
The self-adjoint extensions \(\psi(L)=e^{i\alpha}\psi(0)\) therefore introduce a phase \(\alpha\) not fixed by the mixed bundle
sector.

\subsection{Positive Hamiltonians}
\label{subsec:circle-k-positive}

Closing \(\mathfrak q[\psi] = \int_0^L|\psi'(x)|^2\,dx\) on the three spin domains yields the periodic, antiperiodic, and Dirichlet
realisations of \(-d^2/dx^2\). Their spectra are, respectively,
\[
  \left\{
    \left(
      \frac{2\pi n}{L}
    \right)^2:
    n\in\mathbb Z
  \right\},
\]
\[
  \left\{
    \left(
      \frac{(2n+1)\pi}{L}
    \right)^2:
    n\in\mathbb Z
  \right\},
\]
and
\[
  \left\{
    \left(
      \frac{\pi n}{L}
    \right)^2:
    n\in\mathbb N
  \right\}.
\]
The first two are the standard circle spectra: the periodic zero mode is
simple, every positive periodic eigenvalue has multiplicity two, and every
antiperiodic eigenvalue has multiplicity two. The Dirichlet spectrum is
simple. The third realisation is the non-Hausdorff mixed-sector reduction
produced by a vanishing endpoint equaliser.

For coherent flat \(U(1)\) transport,
\(\psi(L)=e^{i\theta}\psi(0), \qquad \psi'(L)=e^{i\theta}\psi'(0)\),
and
\(\lambda_n(\theta) = \left( \frac{2\pi n+\theta}{L} \right)^2\).
Noncoherent phase transport again yields the Dirichlet interval
realisation. Coherent transport preserves the cycle and its holonomy,
whereas incompatible local transport cuts the cycle at the level of the
closed form domain.

\section{Several split points on a circle}
\label{sec:multisplit-circle}

Several split fibres combine local nodal constraints with one global circle holonomy; any active fibre cuts the \(H^1\)-cycle and removes that global phase.

\subsection{Geometry and spin classification}
\label{subsec:multisplit-circle-spin-data}

Let \(S^1_L=\mathbb R/L\mathbb Z\) and choose split coordinates \(p_1,\ldots,p_N\) in cyclic order. Replace \(p_j\) by \(k_j\geq2\) inseparable origins and
denote the resulting identity-glued manifold by
\(\mathbb C_{\mathbf k}(\mathbf p;L)\).
Its Hausdorff reflection is \(S^1_L\), and its regular locus consists of the
\(N\) open arcs \(A_1,\ldots,A_N\) between successive split points. As in the line case, the standard
exceptional charts make the space compact, connected, second-countable,
\(T_1\), and locally Euclidean; scalar smooth functions and the normalized
scalar \(L^2\)-space reduce to those of the ordinary circle.

Choose a reference origin \(p_{j,1}\) in each exceptional fibre. If
\(\varepsilon_{j,\alpha}\in\{\pm1\}\) denotes the transition from the
incoming to the outgoing adjacent arc through \(p_{j,\alpha}\), define
\[
  \delta_{j,\alpha}
  :=
  \varepsilon_{j,\alpha}
  \varepsilon_{j,1}^{-1},
  \qquad
  \alpha=2,\ldots,k_j,
\]
and
\(h := \prod_{j=1}^{N} \varepsilon_{j,1} \in\{\pm1\}\).
The \(\delta_{j,\alpha}\) are the local relative signs and \(h\) is the
remaining cyclic invariant.

\begin{theorem}[Spin classification on a multi-split circle]
\label{thm:multisplit-circle-spin-classification}
The labelled spin structures are classified by
\[
  \mathbb Z_2
  \times
  \prod_{j=1}^{N}
  (\mathbb Z_2)^{k_j-1},
\]
and therefore
\[
  \#\mathscr S
  =
  2^{
    1+\sum_{j=1}^{N}(k_j-1)
  }.
\]
\end{theorem}

\begin{proof}
The local ratios and the product \(h\) are invariant under changes of arc
frames. Conversely, two transition systems with the same invariants differ
by factors \(\lambda_j = \widetilde{\varepsilon}_{j,1} \varepsilon_{j,1}^{-1}\) whose product is one; the arc gauges may therefore be chosen recursively
around the circle to identify the two systems.
\end{proof}

Call the fibre over \(p_j\) \emph{active} when at least one
\(\delta_{j,\alpha}\neq1\), and let
\(A := \{j:p_j\text{ is active}\}, \qquad m:=|A|\).

\subsection{Cycle cutting and Sobolev reduction}
\label{subsec:multisplit-circle-H1}

Let \(\psi_j\) denote the spinor on \(A_j\). Smooth compatibility through
\(p_{j,\alpha}\) requires
\[
  \psi_j^{(r)}(p_j^+)
  =
  \varepsilon_{j,\alpha}
  \psi_{j-1}^{(r)}(p_j^-)
  \qquad
  (r\geq0).
\]
If the fibre is coherent, all local transports agree and give one sign
relation. If it is active, two incompatible relations force both one-sided
jets to vanish.

If \(A=\varnothing\), all local coherent signs can be gauged away except
their product \(h\). The resulting \(H^1\)-domain is therefore
\(H^1_+(0,L) \quad\text{for }h=+1, \qquad H^1_-(0,L) \quad\text{for }h=-1\).

The labelled spin structures produce exactly \(2^N+1\) smooth and \(H^1\)-core types: two coherent sectors and one type for each
nonempty active subset.

If \(A\neq\varnothing\), list the active points in cyclic order as
\(a_1,\ldots,a_m, \qquad a_{m+1}:=a_1+L\),
and define
\(J_r=(a_r,a_{r+1}), \qquad \ell_r=a_{r+1}-a_r\).
Then
\(\overline{\mathcal C_A}^{\,H^1} \cong \bigoplus_{r=1}^{m}H^1_0(J_r)\).
All coherent transition signs on the resulting path components are gauge
removable, so neither the global holonomy nor the detailed active sign
patterns survive after a nodal cut.

Accordingly, the positive derivative form gives the periodic or
antiperiodic circle Laplacian when \(A=\varnothing\), and \(H_A = \bigoplus_{r=1}^{m}H_{D,J_r}\) when \(A\neq\varnothing\).

The same statement holds for locally constant \(U(1)\)-transport. With one
reference phase per fibre, the complete family is
\(U(1)^{ 1+\sum_{j=1}^{N}(k_j-1) }\).
If every fibre is locally coherent, only the total holonomy \(Z=e^{i\theta}\) survives and gives the quasiperiodic circle domain. If at least one fibre is
active, the zero trace cuts the cycle and \(Z\) disappears from the
resulting \(H^1\)-domain.

\subsection{Spectral and first-order consequences}
\label{subsec:multisplit-circle-spectrum}
\label{subsec:multisplit-circle-first-order}

Assume \(A\neq\varnothing\) and \(|A|=m\). The Friedrichs spectrum is
\[
  \sigma(H_A)
  =
  \biguplus_{r=1}^{m}
  \left\{
    \left(
      \frac{\pi n}{\ell_r}
    \right)^2:
    n\in\mathbb N
  \right\}.
\]
Pairwise irrational length ratios give simple spectrum, whereas equal
spacing,
\(\ell_r=\frac{L}{m}\),
gives \(\lambda_n = \left( \frac{\pi nm}{L} \right)^2\) with multiplicity \(m\).

Writing
\(\ell_{\max}:=\max_r\ell_r\),
the ground-state energy is
\(E_0(A)=\frac{\pi^2}{\ell_{\max}^2} \leq \left( \frac{\pi m}{L} \right)^2\),
with equality exactly for equal spacing. The heat trace satisfies
\[
  \operatorname{Tr}(e^{-tH_A})
  =
  \frac{L}{\sqrt{4\pi t}}
  -
  \frac{m}{2}
  +
  O(e^{-c/t}),
\]
while the coherent periodic or antiperiodic circle has no constant boundary
term.

For the first-order operator, the coherent periodic and antiperiodic domains
are self-adjoint. An active set gives
\[
  D_{\min}
  =
  \bigoplus_{r=1}^{m}
  \left(
    -i\frac{d}{dx}
  \right),
  \qquad
  \mathcal D(D_{\min})
  =
  \bigoplus_{r=1}^{m}H^1_0(J_r),
\]
and therefore
\((n_+,n_-)=(m,m)\).
Its self-adjoint extensions are parametrized by \(U(m)\). Those extension
parameters may reconnect the cuts, but they are additional first-order
operator choices and are not fixed by the original spin or flat-unitary
transport.

The circle models therefore isolate the global effect of a surviving cycle:
rank-one holonomy remains visible precisely while the closed domain contains
an uncut cycle. Branching changes the Hausdorff propagation graph itself.

\section{The non-Hausdorff branching line}
\label{sec:Bk}

For the branching line, the Hausdorff reflection is a star graph, so the non-Hausdorff incidence becomes scalar channel structure rather than only exceptional-fibre multiplicity.

A second feature is supplied by measure. If the resolving sheets carry
positive densities \(c_j\,dx\), their pushforward gives the common edge the
sum of the sheet weights while retaining the individual weights on the
outgoing arms. The resulting balance law makes the common channel exactly
reflectionless for the resolution-induced weighted Kirchhoff Hamiltonian. The same
balance makes the naturally oriented first-order derivative symmetric, but
its deficiency indices are
\((k-1,0)\).
Genuine branching separates second-order flux conservation from
first-order self-adjoint completeness.

\subsection{Geometry, Hausdorff reflection, and scalar closure}
\label{subsec:Bk-definition}
\label{subsec:Bk-reflection}
\label{subsec:Bk-smooth-scalars}
\label{subsec:Bk-versus-Lk}

Fix
\(k\geq2\).
Let \(\widetilde{\mathbb B}_k := \bigsqcup_{j=1}^{k}\mathbb R_j\) be the disjoint union of \(k\) labelled copies of the real line. Define
\(\mathbb B_k := \widetilde{\mathbb B}_k/\!\sim\),
let \(\varrho_k:\widetilde{\mathbb B}_k\longrightarrow\mathbb B_k\) be the quotient map, and impose the equivalence relation
\((x,i)\sim(x,j) \qquad \text{for }x<0\).
No identifications are imposed for \(x\geq0\).

For \(x<0\), write \([x]\) for the common equivalence class. At \(x=0\)
there are \(k\) distinct points
\(b_j:=[(0,j)], \qquad j=1,\ldots,k\),
and for \(x>0\) write \([x]_j\) for the point on the \(j\)-th positive branch.

A neighbourhood of \(b_j\) is
\[
  U_j(\varepsilon)
  =
  \{b_j\}
  \cup
  \{[x]:-\varepsilon<x<0\}
  \cup
  \{[x]_j:0<x<\varepsilon\},
\]
with coordinate
\(\varphi_j([x])=x, \qquad \varphi_j(b_j)=0, \qquad \varphi_j([x]_j)=x\).
Together with the ordinary coordinates away from the exceptional fibre,
these charts define the standard smooth structure.

\begin{proposition}
\label{prop:Bk-topology}
The space \(\mathbb B_k\) is a connected, second-countable, \(T_1\), smooth
one-manifold. Its only non-Hausdorff pairs are the distinct points
\[
  b_1,\ldots,b_k,
\]
which are pairwise nonseparable.
\end{proposition}

\begin{proof}
The displayed charts give local Euclideanity, and second countability follows
from the finite exceptional set and a countable basis on the regular locus.
Any two branch-point neighbourhoods contain a common negative interval, so
the \(b_j\) cannot be separated. Away from the exceptional fibre the space
is locally Hausdorff. Every singleton is closed: this is immediate on the
regular locus, while \(\mathbb B_k\setminus\{b_j\}\) is open by the displayed
exceptional charts. A regular point can be separated from every \(b_j\) by
choosing the exceptional collar shorter than its distance from the origin,
so no further non-Hausdorff pairs occur. Connectedness follows from the common
negative branch.
\end{proof}

The regular locus is
\(\mathbb B_k^\circ = (-\infty,0) \sqcup \bigsqcup_{j=1}^{k}(0,\infty)_j\).
Hence
\(\#\pi_0(\mathbb B_k^\circ)=k+1\).

Let \(\Gamma_k\) be the metric star graph consisting of one negative
half-line \(e_0\) and \(k\) positive half-lines \(e_1,\ldots,e_k\) joined at a vertex \(v\). Define \(\pi:\mathbb B_k\longrightarrow\Gamma_k\) by sending the common negative branch to \(e_0\), every \(b_j\) to \(v\),
and the \(j\)-th positive branch to \(e_j\).

\begin{theorem}[Hausdorff reflection of the branching line]
\label{thm:Bk-Hausdorff-reflection}
The map
\[
  \pi:\mathbb B_k\longrightarrow\Gamma_k
\]
is the Hausdorff reflection.
\end{theorem}

\begin{proof}
The map is continuous and surjective. It is also a quotient map. Indeed, let
\(A\subseteq\Gamma_k\) and suppose \(\pi^{-1}(A)\) is open. Openness of
\(A\) away from \(v\) follows from the edgewise homeomorphism. If \(v\in A\),
openness at the finitely many points \(b_j\) supplies positive collar radii;
their minimum gives a full star neighbourhood of \(v\) contained in \(A\).

Every continuous map from \(\mathbb B_k\) to a Hausdorff space must identify
the pairwise inseparable points \(b_1,\ldots,b_k\). These are precisely the
nontrivial fibres of \(\pi\), so the required factorisation follows from the
quotient property.
\end{proof}

Unlike the line with \(k\) origins,
\(\Gamma_k\not\cong\mathbb R \qquad (k>1)\).
The same exceptional-fibre cardinality can therefore correspond either to
two scalar propagation regions, as for \(\mathbb L_k\), or to \(k+1\)
regions, as here.

A smooth scalar function on \(\mathbb B_k\) is represented by functions \(f_j\in C^\infty(\mathbb R_j)\) that agree for \(x<0\). Smoothness through the branch points gives
\(f_0^{(r)}(0^-) = f_1^{(r)}(0^+) = \cdots = f_k^{(r)}(0^+) \qquad (r\geq0)\).
The descended smooth scalar core has full jet matching at the star
vertex.

At first Sobolev order only the common value remains.

\begin{corollary}[First-order scalar closure]
\label{cor:Bk-H1-closure}
The \(H^1\)-closure of the descended smooth scalar core is
\[
  H^1_{\mathrm{cont}}(\Gamma_k)
  =
  \left\{
    u\in
    H^1(e_0)\oplus
    \bigoplus_{j=1}^{k}H^1(e_j):
    u_0(0)=u_1(0)=\cdots=u_k(0)
  \right\}.
\]
\end{corollary}

This is the star-graph specialization of
Corollary~\ref{cor:universal-first-order-graph-closure}.

\subsection{Pushforward measure and the resolution-induced Hamiltonian}
\label{subsec:Bk-measure}
\label{subsec:Bk-Hamiltonian}
\label{subsec:Bk-topology-measure}

Equip the labelled resolving presentation
\(\varrho_k:\widetilde{\mathbb B}_k\to\mathbb B_k\) with the positive sheet
measures
\(d\mu_j=c_j\,dx, \qquad c_j>0\).
Together with \(\pi:\mathbb B_k\to\Gamma_k\), this is a marked analytic graph
resolution in the sense of
Definition~\ref{def:marked-analytic-graph-resolution}. Its graph measure is
\(\mu_{\Gamma_k} = (\pi\circ\varrho_k)_* \left( \bigoplus_{j=1}^k\mu_j \right)\).
Pushforward to \(\Gamma_k\) gives the common negative edge the weight
\(a_0 = \sum_{j=1}^{k}c_j\),
while the \(j\)-th positive arm has
\(a_j=c_j\).

Hence
\(a_0=\sum_{j=1}^{k}a_j\).

\begin{theorem}[Resolution-induced balance]
\label{thm:Bk-weight-balance}
For every positive choice of resolving-sheet densities,
\[
  a_0
  =
  \sum_{j=1}^{k}a_j.
\]
\end{theorem}

The equality is a pushforward identity. It is not an independently imposed
vertex condition. A common positive rescaling of the \(c_j\) leaves the
relative quantum dynamics unchanged.

The scalar Hilbert space is
\(\mathcal H_k = L^2(e_0,a_0\,dx) \oplus \bigoplus_{j=1}^{k} L^2(e_j,a_j\,dx)\),
and the positive kinetic form on
\(H^1_{\mathrm{cont}}(\Gamma_k)\) is
\[
  \mathfrak q_k[u]
  =
  a_0
  \int_{-\infty}^{0}|u_0'(x)|^2\,dx
  +
  \sum_{j=1}^{k}
  a_j
  \int_0^\infty|u_j'(x)|^2\,dx.
\]

The general result of
Theorem~\ref{thm:weighted-kirchhoff} gives:

\begin{theorem}[Resolution-induced branching Hamiltonian]
\label{thm:Bk-canonical-Hamiltonian}
The Friedrichs Hamiltonian \(H_k\) acts as
\[
  (H_ku)_e=-u_e''
\]
on every edge and has domain
\[
  \mathcal D(H_k)
  =
  \left\{
    u\in H^2_\oplus(\Gamma_k):
    \begin{array}{l}
      u_0(0^-)
      =
      u_1(0^+)
      =
      \cdots
      =
      u_k(0^+),\\[2mm]
      \displaystyle
      -a_0u_0'(0^-)
      +
      \sum_{j=1}^{k}a_ju_j'(0^+)
      =
      0
    \end{array}
  \right\}.
\]
\end{theorem}

For equal resolving-sheet densities, the projective weight vector is
\((a_0,a_1,\ldots,a_k) \sim (k,1,\ldots,1)\).

Once the graph and weights have been obtained, \(H_k\) is a standard
weighted Kirchhoff quantum-graph Hamiltonian
\citep{ExnerSeba89,KostrykinSchrader99,BerkolaikoKuchment2013}.
The geometric content is that both the star graph and its balanced relative
weights arise from the resolution.

\subsection{Scattering and the bright--dark decomposition}
\label{subsec:Bk-forward-scattering}
\label{subsec:Bk-S-matrix}
\label{subsec:Bk-bright-dark}
\label{subsec:Bk-reverse-scattering}

For incidence from the common branch at energy
\(E=p^2, \qquad p>0\),
write
\(u_0(x) = e^{ipx} + R_0(p)e^{-ipx}, \qquad x<0\),
and
\(u_j(x) = T_{j0}(p)e^{ipx}, \qquad x>0\).

Continuity gives a common vertex amplitude,
\(1+R_0 = T_{10} = \cdots = T_{k0}\).
The weighted Kirchhoff condition gives
\(a_0(1-R_0) = \sum_{j=1}^{k}a_jT_{j0}\).
Using \(a_0=\sum_j a_j\) yields \(R_0(p)=0, \qquad T_{j0}(p)=1\) in the unnormalised edge-amplitude convention.

Since the flux carried by an amplitude \(A\) on edge \(e\) is proportional
to
\(a_ep|A|^2\),
the outgoing probabilities are

\begin{theorem}[Reflectionless branching]
\label{thm:Bk-reflectionless}
For incidence from the common branch,
\[
  R_0(p)=0
  \qquad
  (p>0),
\]
and
\[
  P_{j\leftarrow0}
  =
  \frac{a_j}{a_0}
  =
  \frac{c_j}{\sum_{\ell=1}^{k}c_\ell}.
\]
\end{theorem}

For equal resolving-sheet densities,
\(P_{j\leftarrow0}=\frac1k\).

The complete scattering matrix is most transparent in the flux-normalized
basis. Define
\[
  \mathbf s
  :=
  \frac{1}{\sqrt{a_0}}
  \begin{pmatrix}
    \sqrt{a_1}\\
    \vdots\\
    \sqrt{a_k}
  \end{pmatrix},
\]
so that
\(\|\mathbf s\|=1\).

\begin{theorem}[Branching scattering matrix]
\label{thm:Bk-S}
In the decomposition into the common channel and the \(k\) outgoing arms,
the resulting scattering matrix is
\[
  S_k
  =
  \begin{pmatrix}
    0 & \mathbf s^*\\
    \mathbf s & \mathbf s\mathbf s^*-I_k
  \end{pmatrix}.
\]
It is real-symmetric, Hermitian, and unitary.
\end{theorem}

\begin{proof}
The general weighted Kirchhoff formula of
Theorem~\ref{thm:weighted-vertex-S} has total incident weight
\(a_0+\sum_{j=1}^{k}a_j=2a_0\).
Substitution gives
\((S_k)_{00}=0\),
\((S_k)_{j0} = (S_k)_{0j} = \sqrt{\frac{a_j}{a_0}}\),
and
\((S_k)_{ji} = \frac{\sqrt{a_ja_i}}{a_0} - \delta_{ji}\).
The displayed block form follows. Unitarity also follows directly from
\(\mathbf s^*\mathbf s=1\).
\end{proof}

Define the normalized bright arm state \(|s\rangle := \sum_{j=1}^{k} \sqrt{\frac{a_j}{a_0}}\, |j\rangle\) and its orthogonal complement
\[
  \mathcal D_{\mathrm{dark}}
  :=
  \{
    w:
    \langle s,w\rangle=0
  \}
  \subset
  \operatorname{span}\{|1\rangle,\ldots,|k\rangle\}.
\]
Then
\(\dim\mathcal D_{\mathrm{dark}}=k-1\).

\begin{theorem}[Bright--dark decomposition]
\label{thm:Bk-bright-dark}
The scattering matrix satisfies
\[
  S_k|0\rangle=|s\rangle,
  \qquad
  S_k|s\rangle=|0\rangle,
\]
and
\[
  S_kw=-w
  \qquad
  \text{for every }
  w\in\mathcal D_{\mathrm{dark}}.
\]
\end{theorem}

\begin{proof}
The first two identities follow immediately from the block form and
\(\mathbf s^*\mathbf s=1\). If
\(\mathbf s^*w=0\),
then
\((\mathbf s\mathbf s^*-I_k)w=-w\).
\end{proof}

The common branch couples exactly to one weighted superposition of the
arms. For equal densities this is
\(|s\rangle = \frac1{\sqrt{k}} \sum_{j=1}^{k}|j\rangle\).

Incidence from a single arm \(i\) gives
\(P_{0\leftarrow i} = \frac{a_i}{a_0}\),
\(P_{i\leftarrow i} = \left( 1-\frac{a_i}{a_0} \right)^2\),
and, for \(j\neq i\),
\(P_{j\leftarrow i} = \frac{a_ia_j}{a_0^2}\).
An individual arm therefore contains both bright and dark components; only
the common channel and the bright arm state are perfectly exchanged.

\subsection{The oriented first-order operator}
\label{subsec:Bk-first-order}
\label{subsec:Bk-first-order-adjoint}
\label{subsec:Bk-deficiency}
\label{subsec:Bk-punctured-derivative}
\label{subsec:Bk-comparison}
\label{subsec:Bk-interpretation}

The resolving-sheet coordinate increases from the common negative branch
through the exceptional fibre and along each positive branch. Accordingly,
orient \(e_0\) toward the vertex and every \(e_j\), \(j\geq1\), away from
it. This orientation is compatible with the original sheetwise expression
\(-i\frac{d}{dx}\).

Define \(\mathsf D_k u := -iu'\) edgewise on
\(\mathcal D(\mathsf D_k) = H^1_{\mathrm{cont}}(\Gamma_k)\).

By Proposition~\ref{prop:first-order-boundary-form-general},
\[
  \langle\mathsf D_ku,v\rangle
  -
  \langle u,\mathsf D_kv\rangle
  =
  i
  \left(
    a_0-\sum_{j=1}^{k}a_j
  \right)
  \overline{u(v)}\,v(v).
\]

\begin{proposition}[Automatic first-order symmetry]
\label{prop:Bk-first-order-symmetry}
The resolution-induced weights make
\[
  \mathsf D_k
\]
symmetric.
\end{proposition}

\begin{proof}
The coefficient in the boundary form vanishes by
Theorem~\ref{thm:Bk-weight-balance}.
\end{proof}

The adjoint does not retain vertex continuity.

\begin{proposition}[Adjoint of the branching derivative]
\label{prop:Bk-adjoint}
The adjoint acts edgewise as
\[
  \mathsf D_k^*\psi=-i\psi'
\]
and has domain
\[
  \mathcal D(\mathsf D_k^*)
  =
  \left\{
    \psi\in H^1_\oplus(\Gamma_k):
    a_0\psi_0(0^-)
    =
    \sum_{j=1}^{k}a_j\psi_j(0^+)
  \right\}.
\]
\end{proposition}

\begin{proof}
For a maximal piecewise \(H^1\)-function \(\psi\), integration by parts
against \(\phi\in H^1_{\mathrm{cont}}(\Gamma_k)\) leaves the vertex term proportional to
\(a_0\psi_0(0^-) - \sum_{j=1}^{k}a_j\psi_j(0^+)\).
Since the common value \(\phi(v)\) is arbitrary, this term must vanish.
\end{proof}

The deficiency spaces can now be calculated explicitly.

\begin{theorem}[First-order deficiency obstruction]
\label{thm:Bk-deficiency}
The resolution-induced oriented derivative on the branching line has deficiency
indices
\[
  \bigl(
    n_+(\mathsf D_k),
    n_-(\mathsf D_k)
  \bigr)
  =
  (k-1,0).
\]
\end{theorem}

\begin{proof}
For
\(\psi\in\ker(\mathsf D_k^*-i)\),
the equation \(\psi'=-\psi\) has no nonzero square-integrable solution on the negative half-line. On the
positive arms,
\(\psi_j(x)=C_je^{-x}\).
The adjoint boundary condition becomes
\(\sum_{j=1}^{k}a_jC_j=0\),
which leaves a \((k-1)\)-dimensional solution space. Hence
\(n_+=k-1\).

For
\(\psi\in\ker(\mathsf D_k^*+i)\),
one has
\(\psi'=\psi\).
Square integrability forces all positive-arm components to vanish, while the
negative half-line permits
\(\psi_0(x)=C_0e^x\).
The adjoint boundary condition then gives
\(a_0C_0=0\),
so \(C_0=0\). Hence
\(n_-=0\).
\end{proof}

\begin{corollary}[No self-adjoint first-order completion]
\label{cor:Bk-no-SA}
For every \(k>1\), the geometrically selected first-order operator
\(\mathsf D_k\) admits no self-adjoint extension in the same Hilbert space.
\end{corollary}

The deficiency space has a direct relation to the scattering decomposition.
A positive-deficiency vector is specified by \((C_1,\ldots,C_k)\) subject to
\(\sum_{j=1}^{k}a_jC_j=0\).
After flux normalization this is precisely the \((k-1)\)-dimensional arm
sector orthogonal to the bright vector. The channel directions that are dark
to forward second-order transport therefore also account for the unmatched
positive first-order deficiency directions.

For comparison, the completely punctured minimal derivative on one negative
and \(k\) positive half-lines has deficiency indices
\((k,1)\).
Passing to the descended non-Hausdorff scalar core imposes the common vertex
trace and, together with the pushforward balance, reduces these to
\((k-1,0)\).
The descended smooth structure selects a closed symmetric operator that
is not obtained as a self-adjoint extension of the punctured derivative.

The branching line therefore exhibits two distinct consequences of the same
resolution-induced balance. The second-order Friedrichs Hamiltonian is
self-adjoint and exchanges the common channel with the bright arm mode
without reflection. The naturally oriented first-order derivative is
symmetric but has a one-sided deficiency space and cannot be made
self-adjoint. Both statements extend from a \(1\)-to-\(k\) splitter to
general finite \(m\)-to-\(n\) junctions.
\section{\texorpdfstring{General \(m\)-to-\(n\) non-Hausdorff junctions}{General m-to-n non-Hausdorff junctions}}
\label{sec:Bmn}

The branching line \(\mathbb B_k\) has one incoming channel and \(k\)
outgoing channels. A general finite junction has \(m\) incoming and \(n\) outgoing propagation channels.

Unlike the \(1\)-to-\(k\) case, the channel counts do not determine the
non-Hausdorff resolution. One must also specify which incoming and outgoing
branches belong to the same smooth resolving sheets. This information is
encoded by a finite bipartite multigraph.

The reduction separates sharply into three levels. The non-Hausdorff
resolution remembers the complete sheet incidence and the sheet densities.
The resolution-induced scalar \(H^1\)-theory retains only the induced weights on the
incoming and outgoing channels. The deficiency dimensions of the naturally
oriented first-order derivative retain only the channel counts:
\((n_+,n_-) = (n-1,m-1)\).
Self-adjoint first-order completion is possible precisely when
\(m=n\).

\subsection{Resolution data and Hausdorff reflection}
\label{subsec:Bmn-resolution-data}
\label{subsec:Bmn-smooth-structure}
\label{subsec:Bmn-reflection}

Let \(I=\{1,\ldots,m\}, \qquad O=\{1,\ldots,n\}\) be the incoming and outgoing channel labels. Let \(\mathcal S\) be a finite set of resolving sheets, equipped with incidence maps
\(\ell:\mathcal S\longrightarrow I, \qquad r:\mathcal S\longrightarrow O\).
Each \(\alpha\in\mathcal S\) connects the incoming channel \(\ell(\alpha)\) to the outgoing channel
\(r(\alpha)\).

Equivalently, \(\mathcal S\) is the edge set of a bipartite multigraph \(\mathcal R\) with vertex set
\(I\sqcup O\).
We assume throughout that \(\mathcal R\) is connected.

For every \(\alpha\in\mathcal S\), let \(\mathbb R_\alpha\) be a labelled copy of the real line, and set
\[
  \widetilde{\mathbb B}_{\mathcal R}
  :=
  \bigsqcup_{\alpha\in\mathcal S}
  \mathbb R_\alpha.
\]

Define an equivalence relation by
\[
  (x,\alpha)\sim(x,\beta)
  \quad\Longleftrightarrow\quad
  \ell(\alpha)=\ell(\beta),
  \qquad
  x<0,
\]
and
\[
  (x,\alpha)\sim(x,\beta)
  \quad\Longleftrightarrow\quad
  r(\alpha)=r(\beta),
  \qquad
  x>0.
\]
At \(x=0\), distinct sheets remain distinct.

Define
\(\mathbb B_{\mathcal R} := \widetilde{\mathbb B}_{\mathcal R}/\!\sim\),
let
\[
  \varrho_{\mathcal R}:
  \widetilde{\mathbb B}_{\mathcal R}
  \longrightarrow
  \mathbb B_{\mathcal R}
\]
be the quotient map, and denote the origin of sheet \(\alpha\) by
\(b_\alpha=[(0,\alpha)]\).

For \(i\in I\), the negative half-lines with left incidence \(i\) form one
incoming branch
\(e_i^-\).
For \(j\in O\), the positive half-lines with right incidence \(j\) form one
outgoing branch
\(e_j^+\).
Hence
\[
  \mathbb B_{\mathcal R}^{\circ}
  =
  \bigsqcup_{i=1}^{m}e_i^-
  \sqcup
  \bigsqcup_{j=1}^{n}e_j^+.
\]

A neighbourhood of \(b_\alpha\) contains a negative interval in
\(e_{\ell(\alpha)}^-\) and a positive interval in \(e_{r(\alpha)}^+\):
\[
  U_\alpha(\varepsilon)
  =
  \{b_\alpha\}
  \cup
  \{x\in e_{\ell(\alpha)}^-:-\varepsilon<x<0\}
  \cup
  \{x\in e_{r(\alpha)}^+:0<x<\varepsilon\}.
\]
Using the ordinary coordinate \(x\) on this neighbourhood gives the
identity-glued smooth structure.

\begin{proposition}
\label{prop:Bmn-basic}
If the incidence graph \(\mathcal R\) is connected, then
\[
  \mathbb B_{\mathcal R}
\]
is a connected, second-countable, \(T_1\), locally Euclidean smooth
one-manifold.
\end{proposition}

\begin{proof}
The sets \(U_\alpha(\varepsilon)\), together with ordinary intervals on the
regular branches, form a smooth one-dimensional atlas. Finiteness of
\(\mathcal S\) and of the channel sets gives second countability. Regular
singletons are closed, and the complement of an origin \(b_\alpha\) is open
in every exceptional chart; hence the space is \(T_1\). Finally, every sheet
joins its incident incoming and outgoing branches through \(b_\alpha\).
Connectedness of the incidence graph propagates these joins across all
channels and origins, so \(\mathbb B_{\mathcal R}\) is connected.
\end{proof}

Distinct origins need not be directly nonseparable. If two sheets share an
incoming or outgoing channel, however, their origins are nonseparable.
Connectedness of \(\mathcal R\) propagates this relation through chains of
such incidences.

Let \(\Gamma_{m,n}\) be the metric star graph with one vertex \(v\), incoming half-lines
\(e_1^-,\ldots,e_m^-\),
and outgoing half-lines
\(e_1^+,\ldots,e_n^+\).
Define \(\pi: \mathbb B_{\mathcal R} \longrightarrow \Gamma_{m,n}\) to be the identity on the regular branches and to send every \(b_\alpha\)
to \(v\).

\begin{theorem}[Hausdorff reflection]
\label{thm:Bmn-Hausdorff-reflection}
If \(\mathcal R\) is connected, then
\[
  \pi:
  \mathbb B_{\mathcal R}
  \longrightarrow
  \Gamma_{m,n}
\]
is the Hausdorff reflection.
\end{theorem}

\begin{proof}
The map \(\pi\) is continuous, surjective and quotient. For the last claim,
let \(A\subseteq\Gamma_{m,n}\) with \(\pi^{-1}(A)\) open. Away from the vertex
the map is an edgewise homeomorphism. If \(v\in A\), openness at every point
of the finite fibre \(\{b_\alpha\}\) supplies collars on the incident incoming
and outgoing channels. Taking the minimum of their radii gives a star
neighbourhood of \(v\) contained in \(A\), because every channel is incident
to a sheet in the connected graph \(\mathcal R\).

Let \(F:\mathbb B_{\mathcal R}\longrightarrow Y\) be continuous with \(Y\) Hausdorff. Whenever two sheets share an incoming
or outgoing channel, their origins are nonseparable and hence have the same
image under \(F\). Since \(\mathcal R\) is connected, equality propagates
along a finite incidence chain, so \(F(b_\alpha)\) is independent of \(\alpha\).

\(F\) is constant on the unique nontrivial fibre of \(\pi\), and the
quotient property gives the required factorisation.
\end{proof}

The number of origins \(|\mathcal S|\) and the number of propagation channels \(m+n\) are therefore independent geometric quantities.

\subsection{Scalar closure and resolution-induced weights}
\label{subsec:Bmn-scalar-core}
\label{subsec:Bmn-weight-matrix}
\label{subsec:Bmn-matrix-loss}
\label{subsec:Bmn-information-loss}

A smooth scalar function is represented on each resolving sheet by
\(f_\alpha\in C^\infty(\mathbb R)\).
For \(x<0\), sheets with the same left incidence give one function
\(f_i^-\),
while for \(x>0\), sheets with the same right incidence give one function
\(f_j^+\).

Smoothness through \(b_\alpha\) gives
\[
  (f_{\ell(\alpha)}^-)^{(r)}(0^-)
  =
  (f_{r(\alpha)}^+)^{(r)}(0^+)
  \qquad
  (r\geq0).
\]
Connectedness of \(\mathcal R\) therefore implies:

\begin{proposition}[Full scalar jet matching]
\label{prop:Bmn-full-jet-matching}
For every \(r\geq0\),
\[
  (f_1^-)^{(r)}(0^-)
  =
  \cdots
  =
  (f_m^-)^{(r)}(0^-)
  =
  (f_1^+)^{(r)}(0^+)
  =
  \cdots
  =
  (f_n^+)^{(r)}(0^+).
\]
\end{proposition}

At first Sobolev order only the common value remains.

\begin{corollary}[Scalar \(H^1\)-closure]
\label{cor:Bmn-H1}
The descended scalar smooth core closes to
\[
  H^1_{\mathrm{cont}}(\Gamma_{m,n}).
\]
Equivalently,
\[
  u_1^-(0)
  =
  \cdots
  =
  u_m^-(0)
  =
  u_1^+(0)
  =
  \cdots
  =
  u_n^+(0).
\]
\end{corollary}

Equip the labelled resolving presentation
\[
  \varrho_{\mathcal R}:
  \widetilde{\mathbb B}_{\mathcal R}
  \longrightarrow
  \mathbb B_{\mathcal R}
\]
with sheet measures
\(d\mu_\alpha=c_\alpha\,dx, \qquad c_\alpha>0\).
Together with the Hausdorff reflection
\(\pi:\mathbb B_{\mathcal R}\to\Gamma_{m,n}\), these data form a marked
analytic graph resolution. Its graph measure is
\[
  \mu_{\Gamma_{m,n}}
  =
  (\pi\circ\varrho_{\mathcal R})_*
  \left(
    \bigoplus_{\alpha\in\mathcal S}\mu_\alpha
  \right).
\]
Collect the resulting incidence weights in the nonnegative matrix
\(C=(C_{ij})\),
where
\[
  C_{ij}
  :=
  \sum_{\substack{
    \alpha\in\mathcal S\\
    \ell(\alpha)=i,\ r(\alpha)=j
  }}
  c_\alpha.
\]

Define the incoming marginal weights
\(a_i := \sum_{j=1}^{n}C_{ij}, \qquad i=1,\ldots,m\),
and the outgoing marginal weights
\(b_j := \sum_{i=1}^{m}C_{ij}, \qquad j=1,\ldots,n\).
Finally set
\[
  \mathcal C
  :=
  \sum_{i=1}^{m}a_i
  =
  \sum_{j=1}^{n}b_j
  =
  \sum_{\alpha\in\mathcal S}c_\alpha.
\]

\begin{theorem}[Resolution-induced global balance]
\label{thm:Bmn-global-balance}
The pushforward weights satisfy
\[
  \sum_{i=1}^{m}a_i
  =
  \sum_{j=1}^{n}b_j
  =
  \mathcal C.
\]
\end{theorem}

The matrix \(C\) contains strictly finer information than its row and column
sums. The scalar \(H^1\)-domain, however, depends only on the connected
Hausdorff star, and the kinetic form depends only on the induced edge
weights.

\begin{proposition}[Marginal-weight principle]
\label{prop:Bmn-marginal-principle}
For the resolution-induced scalar \(H^1\)-theory, the resolving density matrix
\[
  C=(C_{ij})
\]
enters the Hamiltonian only through
\[
  \mathbf a=(a_1,\ldots,a_m),
  \qquad
  \mathbf b=(b_1,\ldots,b_n).
\]
Connected resolutions with the same incoming and outgoing
marginal weights have identical scalar \(H^1\)-domains, canonical
Hamiltonians, and vertex scattering matrices.
\end{proposition}

The ordinary scalar theory does not in general recover which resolving
sheets connect particular incoming and outgoing channels.

\subsection{Resolution-induced Hamiltonian and bright--dark scattering}
\label{subsec:Bmn-Hamiltonian}
\label{subsec:Bmn-S-matrix}
\label{subsec:Bmn-bright-dark}
\label{subsec:Bmn-individual-scattering}
\label{subsec:Bmn-complete-example}

The weighted Hilbert space is
\[
  \mathcal H_{m,n}
  =
  \bigoplus_{i=1}^{m}
  L^2(e_i^-,a_i\,dx)
  \oplus
  \bigoplus_{j=1}^{n}
  L^2(e_j^+,b_j\,dx).
\]

On
\(H^1_{\mathrm{cont}}(\Gamma_{m,n})\),
consider
\[
  \mathfrak q_{m,n}[u]
  =
  \sum_{i=1}^{m}
  a_i
  \int_{-\infty}^{0}
  |(u_i^-)'(x)|^2\,dx
  +
  \sum_{j=1}^{n}
  b_j
  \int_0^\infty
  |(u_j^+)'(x)|^2\,dx.
\]

By Theorem~\ref{thm:weighted-kirchhoff}:

\begin{theorem}[Resolution-induced \(m\)-to-\(n\) Hamiltonian]
\label{thm:Bmn-Hamiltonian}
The associated self-adjoint Hamiltonian acts edgewise as
\[
  -\frac{d^2}{dx^2}
\]
and has domain
\[
  \mathcal D(H_{m,n})
  =
  \left\{
    u\in H^2_\oplus(\Gamma_{m,n}):
    \begin{array}{l}
      u_i^-(0)=u_j^+(0)
      \text{ for all }i,j,\\[2mm]
      \displaystyle
      -\sum_{i=1}^{m}
      a_i(u_i^-)'(0^-)
      +
      \sum_{j=1}^{n}
      b_j(u_j^+)'(0^+)
      =
      0
    \end{array}
  \right\}.
\]
\end{theorem}

Define
\[
  |s_{\mathrm{in}}\rangle
  :=
  \frac{1}{\sqrt{\mathcal C}}
  \begin{pmatrix}
    \sqrt{a_1}\\
    \vdots\\
    \sqrt{a_m}
  \end{pmatrix},
\]
and
\[
  |s_{\mathrm{out}}\rangle
  :=
  \frac{1}{\sqrt{\mathcal C}}
  \begin{pmatrix}
    \sqrt{b_1}\\
    \vdots\\
    \sqrt{b_n}
  \end{pmatrix}.
\]
By the global balance theorem, both vectors have unit norm.

Since the total weight incident at the star vertex is
\(2\mathcal C\),
Theorem~\ref{thm:weighted-vertex-S} gives:

\begin{theorem}[Flux-normalized scattering matrix]
\label{thm:Bmn-S}
The resulting vertex scattering matrix is
\[
  S_{m,n}
  =
  \begin{pmatrix}
  |s_{\mathrm{in}}\rangle
  \langle s_{\mathrm{in}}|-I_m
  &
  |s_{\mathrm{in}}\rangle
  \langle s_{\mathrm{out}}|
  \\[2mm]
  |s_{\mathrm{out}}\rangle
  \langle s_{\mathrm{in}}|
  &
  |s_{\mathrm{out}}\rangle
  \langle s_{\mathrm{out}}|-I_n
  \end{pmatrix}.
\]
It is Hermitian and unitary.
\end{theorem}

Define
\[
  \mathcal D_{\mathrm{in}}
  :=
  \{
    u\in\mathbb C^m:
    \langle s_{\mathrm{in}},u\rangle=0
  \},
\]
and
\[
  \mathcal D_{\mathrm{out}}
  :=
  \{
    v\in\mathbb C^n:
    \langle s_{\mathrm{out}},v\rangle=0
  \}.
\]
Then
\(\dim\mathcal D_{\mathrm{in}}=m-1, \qquad \dim\mathcal D_{\mathrm{out}}=n-1\).

\begin{theorem}[Bright-channel transfer]
\label{thm:Bmn-bright-transfer}
The scattering matrix satisfies
\[
  S_{m,n}|s_{\mathrm{in}}\rangle
  =
  |s_{\mathrm{out}}\rangle,
\]
\[
  S_{m,n}|s_{\mathrm{out}}\rangle
  =
  |s_{\mathrm{in}}\rangle,
\]
and
\[
  S_{m,n}u=-u
  \qquad
  (u\in\mathcal D_{\mathrm{in}}),
\]
\[
  S_{m,n}v=-v
  \qquad
  (v\in\mathcal D_{\mathrm{out}}).
\]
Hence
\[
  S_{m,n}
  \cong
  \begin{pmatrix}
    0&1\\
    1&0
  \end{pmatrix}
  \oplus
  (-I_{m-1})
  \oplus
  (-I_{n-1}).
\]
\end{theorem}

The resolution-induced junction transfers a single collective mode
between the incoming and outgoing channel spaces. The orthogonal
combinations are dark.

For a flux-normalized wave incident on a single incoming channel \(i\),
\(P_{\mathrm{out}\leftarrow i} = \frac{a_i}{\mathcal C}\),
and
\(P_{j\leftarrow i} = \frac{a_ib_j}{\mathcal C^2}\).

\begin{proposition}
\label{prop:Bmn-perfect-individual}
An individual incoming channel is transmitted with total probability one
if and only if
\[
  m=1.
\]
\end{proposition}

\begin{proof}
Perfect transmission from channel \(i\) requires
\(a_i=\mathcal C\).
Since every incoming weight is positive and
\(\mathcal C=\sum_{\ell=1}^{m}a_\ell\),
this is possible exactly when there is only one incoming channel.
\end{proof}

The \(1\)-to-\(k\) branching line is therefore the special case in which
the incoming channel space contains no dark sector.

As a symmetric example, suppose there is one equal-density resolving sheet
for every pair
\((i,j)\in I\times O\).
Then
\(a_i=nc, \qquad b_j=mc, \qquad \mathcal C=mnc\),
and
\[
  |s_{\mathrm{in}}\rangle
  =
  \frac1{\sqrt m}
  \sum_{i=1}^{m}|i\rangle,
  \qquad
  |s_{\mathrm{out}}\rangle
  =
  \frac1{\sqrt n}
  \sum_{j=1}^{n}|j\rangle.
\]
The uniform incoming superposition is transferred exactly to the uniform
outgoing superposition, while incidence on a single incoming channel
transfers total probability \(1/m\).

\subsection{First-order operator and the channel-index theorem}
\label{subsec:Bmn-first-order}
\label{subsec:Bmn-adjoint}
\label{subsec:Bmn-deficiency}
\label{subsec:Bmn-SA}
\label{subsec:Bmn-dark-deficiency}
\label{subsec:Bmn-orientation-reversal}
\label{subsec:Bmn-Bk-comparison}
\label{subsec:Bmn-interpretation}

Orient every incoming edge toward the junction and every outgoing edge away
from it. This orientation agrees with increasing sheet coordinate on the
resolving copies of \(\mathbb R\).

Define \(\mathsf D_{m,n} = -i\frac{d}{dx}\) edgewise on
\(\mathcal D(\mathsf D_{m,n}) = H^1_{\mathrm{cont}}(\Gamma_{m,n})\).

By Proposition~\ref{prop:first-order-boundary-form-general},
\[
  \langle\mathsf D_{m,n}u,v\rangle
  -
  \langle u,\mathsf D_{m,n}v\rangle
  =
  i
  \left(
    \sum_{i=1}^{m}a_i
    -
    \sum_{j=1}^{n}b_j
  \right)
  \overline{u(v)}\,v(v).
\]

\begin{proposition}[Automatic first-order symmetry]
\label{prop:Bmn-first-order-symmetry}
The resolution-induced global balance makes
\[
  \mathsf D_{m,n}
\]
symmetric.
\end{proposition}

The adjoint imposes one weighted boundary relation rather than continuity.

\begin{proposition}[Adjoint boundary condition]
\label{prop:Bmn-adjoint}
The adjoint acts edgewise as
\[
  \mathsf D_{m,n}^*\psi=-i\psi'
\]
with domain
\[
  \mathcal D(\mathsf D_{m,n}^*)
  =
  \left\{
    \psi\in H^1_\oplus(\Gamma_{m,n}):
    \sum_{i=1}^{m}
    a_i\psi_i^-(0)
    =
    \sum_{j=1}^{n}
    b_j\psi_j^+(0)
  \right\}.
\]
\end{proposition}

\begin{proof}
Integration by parts against an element of
\(H^1_{\mathrm{cont}}(\Gamma_{m,n})\) leaves one boundary term proportional
to
\(\sum_{i=1}^{m} a_i\psi_i^-(0) - \sum_{j=1}^{n} b_j\psi_j^+(0)\).
The common test-function value at the vertex is arbitrary, giving the stated
condition.
\end{proof}

The calculation below concerns the particular weighted continuity domain
selected by the non-Hausdorff resolution. General first-order and index
theory on metric graphs is well established
\citep{FullingKuchmentWilson07,Post09,ExnerMomentum13,Richtsfeld24};
the point here is the reduction of the present resolution to this explicit
domain and the resulting channel-count formula.

The deficiency spaces follow directly from the adjoint boundary relation.

For
\(\psi\in\ker(\mathsf D_{m,n}^*-i)\),
one has
\(\psi'=-\psi\).
There is no nonzero square-integrable solution on an incoming half-line,
whereas on the outgoing edges
\(\psi_j^+(x)=A_je^{-x}\).
The adjoint condition becomes
\(\sum_{j=1}^{n}b_jA_j=0\).
Hence
\(n_+=n-1\).

For
\(\psi\in\ker(\mathsf D_{m,n}^*+i)\),
one has
\(\psi'=\psi\).
The outgoing components vanish by square integrability, while \(\psi_i^-(x)=B_ie^x\) on the incoming edges. The boundary relation gives
\(\sum_{i=1}^{m}a_iB_i=0\),
and therefore
\(n_-=m-1\).

\begin{theorem}[Resolution-induced channel-index theorem]
\label{thm:Bmn-deficiency}
For every connected \(m\)-to-\(n\) junction resolution,
\[
  \left(
    n_+(\mathsf D_{m,n}),
    n_-(\mathsf D_{m,n})
  \right)
  =
  (n-1,m-1).
\]
In particular,
\[
  n_+-n_-
  =
  n-m.
\]
\end{theorem}

The result is independent of the number of resolving sheets, the detailed
bipartite incidence matrix, and the positive sheet densities.

\begin{corollary}[Self-adjointness criterion]
\label{cor:Bmn-SA-criterion}
The resolution-induced oriented derivative admits self-adjoint extensions if and only
if
\[
  m=n.
\]
If
\[
  m=n=r,
\]
the extensions are parametrized by
\[
  U(r-1).
\]
For
\[
  r=1,
\]
the operator is already self-adjoint.
\end{corollary}

The deficiency spaces have the same weighted orthogonality structure as the
dark channel sectors. At \(+i\),
\(\sum_{j=1}^{n}b_jA_j=0\),
which becomes orthogonality to
\(|s_{\mathrm{out}}\rangle\) after flux normalization. Similarly, the
negative deficiency condition becomes orthogonality to
\(|s_{\mathrm{in}}\rangle\).

\begin{proposition}[Deficiency as dark-channel imbalance]
\label{prop:Bmn-dark-deficiency}
There are natural identifications
\[
  \ker(\mathsf D_{m,n}^*-i)
  \cong
  \mathcal D_{\mathrm{out}},
\]
and
\[
  \ker(\mathsf D_{m,n}^*+i)
  \cong
  \mathcal D_{\mathrm{in}}.
\]
\[
  n_+
  =
  \dim\mathcal D_{\mathrm{out}}
  =
  n-1,
\]
and
\[
  n_-
  =
  \dim\mathcal D_{\mathrm{in}}
  =
  m-1.
\]
\end{proposition}

The same channel decomposition appears in both operator orders. The
second-order Hamiltonian transfers the unique bright mode and reflects the
dark sectors, whereas the first-order deficiency spaces are precisely those
incoming and outgoing dark sectors. Self-adjoint first-order completion is
possible exactly when their dimensions agree.

Reversing every edge orientation exchanges the two deficiency indices. The
absolute imbalance \(|n_+-n_-| = |n-m|\) is therefore orientation independent, while the signed index records the
oriented excess of outgoing over incoming channels.

For
\(m=1, \qquad n=k\),
the results reduce to those of the branching line:
\((n_+,n_-) = (k-1,0)\).

Finite branching trees compose these junctions. Local pushforward balance
then holds at every resolution-generated vertex, allowing global propagation
and first-order index calculations on the resulting acyclic network.
\section{Iterated branching and finite non-Hausdorff trees}
\label{sec:trees}

The single-junction bright--dark law extends directly to finite acyclic networks.

For a rooted tree, resolution-induced measure is recursively additive: the
weight of every parent edge equals the sum of the weights of its children.
A wave incident through the unique root channel encounters a
reflectionless channel at every successive vertex. The resulting
root-to-leaf probabilities depend only on the terminal sheet weights, not on
how the branching is factored into intermediate stages.

The first-order theory admits a similarly global description. For a finite
oriented tree with positive weights satisfying the incoming--outgoing balance
condition at every internal vertex, the deficiency indices are
\((n_+,n_-) = (N_{\mathrm{out}}-1,N_{\mathrm{in}}-1)\),
where \(N_{\mathrm{in}}\) and \(N_{\mathrm{out}}\) count the incoming and
outgoing external channels. The proof requires only acyclicity, local
balance, and a rank argument for the vertex equations.

\subsection{Rooted resolving trees and descendant weights}
\label{subsec:trees-rooted}
\label{subsec:trees-descendant-weights}
\label{subsec:trees-H1}

Let \(\Gamma\) be a finite rooted metric tree with one external root edge \(e_{\mathrm{root}}\) and a finite set of external leaf edges
\(\mathcal L = \{\lambda_1,\ldots,\lambda_N\}\).

Orient the root edge toward the finite tree and every other edge away from
the root. Every internal vertex \(v\) has a unique incoming parent edge \(p(v)\) and a finite set of outgoing child edges
\(\operatorname{Ch}(v)\).

A non-Hausdorff resolution of this tree is obtained by assigning one
one-dimensional resolving sheet to each terminal leaf. The sheet associated
with \(\lambda\) follows the unique path from the root to that leaf. Regular
points belonging to sheets that traverse the same open tree edge are
identified, while the corresponding sheet origins remain distinct at
branching vertices.

The resulting non-Hausdorff one-manifold will be denoted by
\(M_\Gamma\),
with Hausdorff reflection
\(\pi: M_\Gamma \longrightarrow \Gamma\).
Locally, a vertex with \(d\) outgoing children is the branching model
\(\mathbb B_d\).

Assign to the terminal sheet \(\lambda\) a positive density
\(c_\lambda>0\).
For an edge \(e\), let \(\mathcal L(e)\) be the set of terminal leaves descending from \(e\), and define
\(a_e := \sum_{\lambda\in\mathcal L(e)} c_\lambda\).
\[
  a_{\mathrm{root}}
  =
  \sum_{\lambda\in\mathcal L}c_\lambda,
  \qquad
  a_\lambda=c_\lambda
\]
for every terminal edge.

At an internal vertex,
\(\mathcal L(p(v)) = \bigsqcup_{e\in\operatorname{Ch}(v)} \mathcal L(e)\),
and therefore:

\begin{theorem}[Recursive weight conservation]
\label{thm:trees-weight-conservation}
At every internal vertex,
\[
  a_{p(v)}
  =
  \sum_{e\in\operatorname{Ch}(v)}
  a_e.
\]
\end{theorem}

The identity is simply additivity of the resolving-sheet measure. In
particular, the parent edge is a reflectionless channel for the weighted
Kirchhoff vertex.

The scalar smooth core of \(M_\Gamma\) imposes the transported smooth
compatibility inherited from the resolving sheets. Its first-order closure
is, by Corollary~\ref{cor:universal-first-order-graph-closure},
\(H^1_{\mathrm{cont}}(\Gamma)\).

On
\(\mathcal H_\Gamma = \bigoplus_{e\in E(\Gamma)} L^2(e,a_e\,dx)\),
consider \(\mathfrak q_\Gamma[u] = \sum_{e\in E(\Gamma)} a_e \int_e|u_e'(x)|^2\,dx\) with domain
\(H^1_{\mathrm{cont}}(\Gamma)\).
Its associated Hamiltonian is the weighted Kirchhoff Laplacian. In
coordinates increasing with the rooted orientation, the vertex condition is
\(a_{p(v)}u_{p(v)}'(v^-) = \sum_{e\in\operatorname{Ch}(v)} a_eu_e'(v^+)\).

\subsection{Root-to-leaf scattering and associativity}
\label{subsec:trees-local-bright}
\label{subsec:trees-global-reflectionless}
\label{subsec:trees-associativity}
\label{subsec:trees-equal-leaves}

At an internal vertex \(v\), define the normalized child bright state
\[
  |s_v\rangle
  :=
  \sum_{e\in\operatorname{Ch}(v)}
  \sqrt{\frac{a_e}{a_{p(v)}}}\,
  |e\rangle.
\]
Theorem~\ref{thm:trees-weight-conservation} gives
\(\langle s_v,s_v\rangle=1\).

The local scattering result of
Theorem~\ref{thm:Bk-bright-dark} therefore gives \(|p(v)\rangle \longmapsto |s_v\rangle\) with no reflected amplitude in the parent edge.

Because every non-root internal edge is the unique parent edge at the next
vertex, this reflectionless transfer iterates down the tree.

For a terminal leaf
\(\lambda\in\mathcal L\),
let \(d_\lambda\) denote the total finite metric length along the unique path from the root
vertex to \(\lambda\).

\begin{theorem}[Root-to-leaf scattering law]
\label{thm:trees-root-to-leaf}
For a unit flux-normalized wave of momentum \(p>0\) incident through the
root edge,
\[
  R_{\mathrm{root}}(p)=0.
\]
The amplitude in terminal channel \(\lambda\) is
\[
  T_{\lambda,\mathrm{root}}(p)
  =
  e^{ipd_\lambda}
  \sqrt{
    \frac{a_\lambda}{a_{\mathrm{root}}}
  },
\]
and hence
\[
  P_{\lambda\leftarrow\mathrm{root}}
  =
  \frac{a_\lambda}{a_{\mathrm{root}}}
  =
  \frac{c_\lambda}
  {\displaystyle\sum_{\mu\in\mathcal L}c_\mu}.
\]
\end{theorem}

\begin{proof}
At every vertex on the unique path to \(\lambda\), the incoming parent
channel is mapped into the local bright state with zero parent reflection.
The transmission factor from a parent edge \(e\) to the selected child
\(f\) is
\(\sqrt{\frac{a_f}{a_e}}\).
Their product along the path telescopes:
\[
  \prod_{e\to f}
  \sqrt{\frac{a_f}{a_e}}
  =
  \sqrt{
    \frac{a_\lambda}{a_{\mathrm{root}}}
  }.
\]
Free propagation along the finite edges contributes the phase
\(e^{ipd_\lambda}\).
Since
\(\sum_{\lambda\in\mathcal L}a_\lambda = a_{\mathrm{root}}\),
the terminal probabilities sum to one.
\end{proof}

The theorem gives an immediate factorisation invariance.

\begin{corollary}[Associativity of resolution-induced splitting]
\label{cor:trees-associativity}
Let two rooted resolving trees have the same terminal channels and the same
terminal sheet weights \(c_\lambda\). Then their root-to-leaf transmission
probabilities agree:
\[
  P_{\lambda\leftarrow\mathrm{root}}
  =
  \frac{c_\lambda}
  {\sum_\mu c_\mu}.
\]
The forward probabilities are therefore independent of how the total split is
factored into intermediate branching vertices.
\end{corollary}

The internal metric structure is not completely lost. Root-to-leaf
amplitudes retain the path phases
\(e^{ipd_\lambda}\).
Hence two trees with the same terminal weights but different
root-to-leaf lengths have identical forward probabilities but need not have
identical amplitudes.

If every terminal sheet has equal density and there are \(N\) leaves, then
\(P_{\lambda\leftarrow\mathrm{root}} = \frac1N\),
and
\(T_{\lambda,\mathrm{root}}(p) = \frac{e^{ipd_\lambda}}{\sqrt N}\).
Equal path lengths therefore produce, up to a common phase, the uniform
terminal superposition.

\subsection{Dark sectors and the limits of forward equivalence}
\label{subsec:trees-dark-modes}

At an internal vertex \(v\), the child-channel dark space is
\[
  \mathcal D_v
  :=
  \left\{
    w\in
    \mathbb C^{|\operatorname{Ch}(v)|}:
    \langle s_v,w\rangle=0
  \right\},
\]
with
\(\dim\mathcal D_v = |\operatorname{Ch}(v)|-1\).

A wave incident through the root never excites these sectors: at every
successive vertex it arrives through the unique parent channel and is
mapped directly into the local bright state.

This is special to root incidence. A wave arriving from a terminal branch
generally enters a branching vertex through one child basis vector and
therefore contains both bright and dark components. The dark component is
reflected within the child sector, allowing reverse and leaf-to-leaf
scattering to probe internal branching structure.

The associativity statement above applies to the
root-incidence scattering column, not to the complete multi-channel
scattering matrix.

\subsection{General balanced oriented trees}
\label{subsec:trees-general-oriented}
\label{subsec:trees-global-first-order}

The rooted orientation is no longer assumed.

Let \(\Gamma\) be a connected finite metric tree with a finite set of external half-lines.
Choose an orientation on every edge. An external edge is called incoming if
its orientation points from infinity toward the finite tree and outgoing if
its orientation points from the tree toward infinity.

Let \(\mathcal E_{\mathrm{in}}, \qquad \mathcal E_{\mathrm{out}}\) be the two external sets, and write
\[
  N_{\mathrm{in}}
  :=
  |\mathcal E_{\mathrm{in}}|,
  \qquad
  N_{\mathrm{out}}
  :=
  |\mathcal E_{\mathrm{out}}|.
\]

Assume positive edge weights \(a_e\) satisfy \(\sum_{e\in E_{\mathrm{in}}(v)}a_e = \sum_{e\in E_{\mathrm{out}}(v)}a_e\) at every finite vertex \(v\). Resolution-induced tree weights of the type
constructed above satisfy this condition.

Summing over the finite vertices cancels every finite internal edge once on
each side and gives
\[
  \sum_{e\in\mathcal E_{\mathrm{in}}}a_e
  =
  \sum_{e\in\mathcal E_{\mathrm{out}}}a_e.
\]

\begin{lemma}[External channels are present in both directions]
\label{lem:trees-external-directions}
Under the positive local balance hypothesis,
\[
  N_{\mathrm{in}}\geq1,
  \qquad
  N_{\mathrm{out}}\geq1.
\]
\end{lemma}

\begin{proof}
If the finite core consists of a single vertex, local balance with positive
incident weights already requires at least one incoming and one outgoing
external edge.

Assume instead that the finite core contains an internal edge. A finite tree
has a leaf vertex \(v\) in its finite core. If no external half-line were
incident at \(v\), the unique finite edge meeting \(v\) would contribute a
strictly positive weight to exactly one side of the local balance equation,
while the other side would be zero, a contradiction. Hence at least one
external half-line is present. The global identity above equates the total
incoming and outgoing external weights. Since all weights are positive,
neither external orientation class can then be empty.
\end{proof}

On
\(\mathcal H_{\mathbf a} = \bigoplus_eL^2(e,a_e\,dx)\),
define \(\mathsf D_\Gamma = -i\frac{d}{dx}\) edgewise with
\(\mathcal D(\mathsf D_\Gamma) = H^1_{\mathrm{cont}}(\Gamma)\).

By Proposition~\ref{prop:first-order-boundary-form-general}, the local
balance identities make \(\mathsf D_\Gamma\) symmetric.

Its adjoint acts edgewise by the same differential expression. At every
finite vertex,
\[
  \sum_{e\in E_{\mathrm{in}}(v)}
  a_e\psi_e(v)
  =
  \sum_{e\in E_{\mathrm{out}}(v)}
  a_e\psi_e(v).
\]

\subsection{Rank of the deficiency boundary system}
\label{subsec:trees-deficiency-rank}

The global deficiency calculation reduces to the rank of these vertex
relations.

Let \(V\) be the number of finite vertices. Since the finite core is a tree, it has \(V-1\) finite internal edges.

\begin{lemma}[Independence of the deficiency vertex relations]
\label{lem:trees-deficiency-rank}
For either deficiency equation
\[
  (\mathsf D_\Gamma^*\mp i)\psi=0,
\]
the \(V\) adjoint vertex relations are linearly independent after the
non-\(L^2\) external coefficients have been removed.
\end{lemma}

\begin{proof}
For \((\mathsf D_\Gamma^*-i)\psi=0\), only outgoing external half-lines carry nonzero \(L^2\)-modes; each internal edge has one free coefficient. Suppose a linear combination of the vertex relations, with coefficients \(\lambda_v\), vanishes identically. An outgoing external coefficient occurs in only its incident vertex relation, so \(\lambda_v=0\) there. Lemma~\ref{lem:trees-external-directions} supplies such a vertex. On an internal edge oriented from \(u\) to \(v\), the endpoint values differ by the nonzero factor \(e^{-\ell_e}\), so cancellation forces \(\lambda_u=e^{-\ell_e}\lambda_v\). Connectivity of the finite tree propagates \(\lambda=0\) to every vertex. For \((\mathsf D_\Gamma^*+i)\psi=0\), start from an incoming external edge and use the nonzero factor \(e^{\ell_e}\). The relations are independent in both deficiency equations.
\end{proof}

\subsection{Global channel-index theorem}
\label{subsec:trees-global-deficiency}
\label{subsec:trees-SA}

We can now count the deficiency dimensions.

For
\(\ker(\mathsf D_\Gamma^*-i)\),
the equation is
\(\psi'=-\psi\).
Incoming external half-lines carry no nonzero square-integrable solution.
Each outgoing external half-line contributes one coefficient, as does each
finite internal edge. Before imposing the vertex relations there are
therefore \((V-1)+N_{\mathrm{out}}\) free coefficients.

By Lemma~\ref{lem:trees-deficiency-rank}, the \(V\) vertex equations have
full rank. Hence
\(n_+ = (V-1)+N_{\mathrm{out}}-V = N_{\mathrm{out}}-1\).

Similarly, \(\ker(\mathsf D_\Gamma^*+i)\) contains one coefficient for each incoming external edge and for each
finite internal edge. Therefore
\(n_- = (V-1)+N_{\mathrm{in}}-V = N_{\mathrm{in}}-1\).

\begin{theorem}[Channel-index formula for finite trees]
\label{thm:trees-global-index}
Let \(\Gamma\) be a connected finite oriented metric tree with positive
weights satisfying
\[
  \sum_{e\in E_{\mathrm{in}}(v)}a_e
  =
  \sum_{e\in E_{\mathrm{out}}(v)}a_e
\]
at every finite vertex. Then
\[
  \left(
    n_+(\mathsf D_\Gamma),
    n_-(\mathsf D_\Gamma)
  \right)
  =
  \left(
    N_{\mathrm{out}}-1,\,
    N_{\mathrm{in}}-1
  \right).
\]
\[
  n_+-n_-
  =
  N_{\mathrm{out}}
  -
  N_{\mathrm{in}}.
\]
\end{theorem}

The numbers and degrees of the internal vertices, their metric separations,
and the multiplicities of the underlying resolving sheets do not enter the
deficiency indices.

\begin{corollary}[Global channel-balance criterion]
\label{cor:trees-SA}
The resolution-induced oriented first-order operator admits self-adjoint extensions
if and only if
\[
  N_{\mathrm{in}}
  =
  N_{\mathrm{out}}.
\]
If
\[
  N_{\mathrm{in}}
  =
  N_{\mathrm{out}}
  =
  r,
\]
then
\[
  (n_+,n_-)
  =
  (r-1,r-1),
\]
and the self-adjoint extensions are parametrized by
\[
  U(r-1).
\]
\end{corollary}

For a rooted tree with one incoming root and \(N\) outgoing terminal
channels,
\((n_+,n_-) = (N-1,0)\).
Every genuinely branching rooted tree has no self-adjoint extension of
its geometrically selected first-order derivative in the same Hilbert
space.

The branching depth does not alter this result: a direct \(1\)-to-\(N\)
junction and any finite rooted factorisation with the same terminal channel
count have identical first-order deficiency indices.

The two operator orders therefore exhibit complementary forms of
compositional stability. Resolution-induced balance makes the
second-order Hamiltonian exactly reflectionless from the root, while the
first-order deficiency imbalance reduces to the asymptotic channel count
\(N_{\mathrm{out}}-N_{\mathrm{in}}\).

A tree cannot, however, produce genuine root-to-terminal path interference:
each terminal channel is reached through a unique route. Recombination
provides several coherent propagation paths between the same exterior
channels. The resulting cycles convert relative path phases into observable
interference.
\section{Split-and-rejoin geometries and quantum interference}
\label{sec:split-rejoin}

Recombination introduces interference between distinct routes to the same terminal channel, absent in a tree.

We consider \(k\) resolving sheets that coincide before a splitting point,
remain distinct through a finite intermediate region, and coincide again
after a recombination point. The Hausdorff reflection has two external leads
connected by \(k\) parallel internal edges. Locally, the two vertices are
the balanced branching junctions studied above. Globally, propagation along
the distinct internal paths rotates the bright mode into dark directions,
producing repeated internal scattering, transmission resonances, zeros, and
compactly supported embedded eigenstates.

\subsection{Geometry, weights, and the resolution-induced Hamiltonian}
\label{subsec:split-rejoin-geometry}
\label{subsec:split-rejoin-metric}
\label{subsec:split-rejoin-vertex-conditions}

Take \(k\geq1\) labelled copies \(\mathbb R_1,\ldots,\mathbb R_k\) of the real line. Define an equivalence relation on their disjoint union by \((x,i)\sim(x,j)\) for
\(x<0 \qquad\text{or}\qquad x>1\).
For
\(0<x<1\),
the sheets remain distinct, and the sheet origins at \(x=0\) and \(x=1\)
remain distinct.

Denote the resulting non-Hausdorff one-manifold by
\(M_{\Theta,k}\).
Its Hausdorff reflection is the metric graph \(\Theta_k\) with two vertices \(v_L,v_R\), one external half-line attached to each
vertex, and \(k\) parallel internal edges \(e_1,\ldots,e_k\) joining \(v_L\) to \(v_R\).

Its first Betti number is
\(\beta_1(\Theta_k)=k-1\).

Choose on the resolving sheets smooth positive one-dimensional metrics that
agree on the identified exterior regions. Let \(\ell_j>0\) denote the induced length of the resolved segment on the \(j\)-th sheet;
after arc-length parametrization, the corresponding internal edge is
identified with \((0,\ell_j)\). The metric lengths used below are part
of the source resolution rather than data inserted only after Hausdorff
reduction.

Assign the corresponding resolving sheet a positive density
\(a_j>0\),
and set
\(A:=\sum_{j=1}^{k}a_j\).

Pushforward of the resolving-sheet measure gives weight \(A\) on each
external lead and weight \(a_j\) on the \(j\)-th internal edge. Both
vertices satisfy the resolution-induced balance relation
\(A=\sum_{j=1}^{k}a_j\).

The scalar Hilbert space is
\[
  \mathcal H_\Theta
  =
  L^2(e_L,A\,dx)
  \oplus
  \left(
    \bigoplus_{j=1}^{k}
    L^2(0,\ell_j;a_j\,dx)
  \right)
  \oplus
  L^2(e_R,A\,dx).
\]

By Corollary~\ref{cor:universal-first-order-graph-closure}, the descended scalar
smooth core closes to
\(H^1_{\mathrm{cont}}(\Theta_k)\).
The positive kinetic form is
\[
  \mathfrak q_\Theta[u]
  =
  A\int_{e_L}|u_L'|^2\,dx
  +
  \sum_{j=1}^{k}
  a_j\int_0^{\ell_j}|u_j'|^2\,dx
  +
  A\int_{e_R}|u_R'|^2\,dx.
\]

Its associated operator is the weighted Kirchhoff Laplacian. If
\(U:=u(v_L), \qquad V:=u(v_R)\),
then continuity gives \(u_j(0)=U, \qquad u_j(\ell_j)=V\) for every internal arm, together with the corresponding external traces.
The weighted Kirchhoff conditions are
\(A\,\partial_{\nu_L}u_L(v_L) + \sum_{j=1}^{k}a_ju_j'(0) = 0\),
and
\(A\,\partial_{\nu_R}u_R(v_R) - \sum_{j=1}^{k}a_ju_j'(\ell_j) = 0\).

\subsection{Bright--dark propagation and multiple scattering}
\label{subsec:split-rejoin-bright-dark}
\label{subsec:split-rejoin-propagation}
\label{subsec:split-rejoin-multiple-scattering}

Define
\(w_j:=\frac{a_j}{A}, \qquad \sum_{j=1}^{k}w_j=1\),
and the normalized internal bright state
\(|s\rangle := \sum_{j=1}^{k}\sqrt{w_j}\,|j\rangle\).
Let
\(P:=|s\rangle\langle s|, \qquad Q:=I_k-P\).

At either vertex the local flux-normalized scattering matrix has the form
\[
  S_v
  =
  \begin{pmatrix}
    0 & \langle s|\\[1mm]
    |s\rangle & -Q
  \end{pmatrix}.
\]
The external channel is exchanged exactly with the bright internal
mode, while every vector in \(\operatorname{Ran}Q\) is reflected with phase \(-1\).

At momentum \(p>0\), propagation through the internal arms is
\(D(p) := \operatorname{diag} \left( e^{ip\ell_1}, \ldots, e^{ip\ell_k} \right)\).
After one traversal, a bright state becomes
\(D(p)|s\rangle\).
Its bright overlap at the second vertex is
\(F(p) := \langle s,D(p)s\rangle = \sum_{j=1}^{k}w_je^{ip\ell_j}\).
Accordingly,
\(\|QD(p)s\|^2 = 1-|F(p)|^2\).

Since every \(w_j>0\), \(|F(p)|=1\) if and only if
\(e^{ip\ell_1} = \cdots = e^{ip\ell_k}\).

The single-pass overlap \(F(p)\) is not, in general, the full transmission
amplitude. If a dark component reaches \(v_R\), it is reflected, returns to
\(v_L\), and may be converted back into the bright sector after further
propagation.

Let \(a\in\mathbb C^k\) be the right-moving internal amplitude immediately
after \(v_L\). For unit incidence from the left and no incidence from the
right,
\(a = |s\rangle-Qb\),
where the returning amplitude is
\(b=-DQDa\).
Hence
\(a = |s\rangle+QDQD\,a\).

Whenever the inverse exists,
\(a = (I_k-QDQD)^{-1}|s\rangle\),
and therefore
\(t(p) = \langle s| D(I_k-QDQD)^{-1} |s\rangle\),
\(r(p) = - \langle s| DQD(I_k-QDQD)^{-1} |s\rangle\).

For \(\operatorname{Im}p>0\), the inverse has the convergent expansion
\((I_k-QDQD)^{-1} = \sum_{q=0}^{\infty}(QDQD)^q\),
whose successive terms describe repeated dark-sector round trips.

\subsection{Exact two-port scattering and resonant energies}
\label{subsec:split-rejoin-exact-scattering}
\label{subsec:split-rejoin-perfect-transmission}
\label{subsec:split-rejoin-zeros}

Away from internal Dirichlet energies, the scattering problem reduces to
two scalar vertex values. Assume first that
\(\sin(p\ell_j)\neq0 \qquad \text{for every }j\),
and define
\(K(p) := \sum_{j=1}^{k} a_j\cot(p\ell_j)\),
\(J(p) := \sum_{j=1}^{k} a_j\csc(p\ell_j)\).

For fixed endpoint values \(U,V\), the solution on the \(j\)-th arm is
\[
  u_j(x)
  =
  U
  \frac{\sin(p(\ell_j-x))}
       {\sin(p\ell_j)}
  +
  V
  \frac{\sin(px)}
       {\sin(p\ell_j)}.
\]

For incidence from the left,
\(u_L(x) = e^{ipx}+r(p)e^{-ipx}, \qquad x<0\),
and
\(u_R(x) = t(p)e^{ipx}, \qquad x>0\).
Hence
\(U=1+r, \qquad V=t\).

\begin{theorem}[Two-port scattering law]
\label{thm:split-rejoin-exact-S}
Suppose
\[
  \sin(p\ell_j)\neq0
  \qquad
  \text{for all }j.
\]
Set
\[
  \Delta(p)
  :=
  A^2+J(p)^2-K(p)^2+2iAK(p).
\]
Then
\[
  t(p)
  =
  \frac{2iA\,J(p)}{\Delta(p)},
\]
and
\[
  r(p)
  =
  \frac{A^2+K(p)^2-J(p)^2}{\Delta(p)}.
\]
For real \(p\),
\[
  |r(p)|^2+|t(p)|^2=1.
\]
\end{theorem}

\begin{proof}
The arm solutions give
\(\sum_j a_ju_j'(0)=p(JV-KU)\) and \(-\sum_j a_ju_j'(\ell_j)=p(JU-KV)\). With \(U=1+r\), \(V=t\), and exterior outward derivatives \(ip(r-1)\) and \(ipt\), the two Kirchhoff equations are
\[
 (iA-K)U+JV=2iA,\qquad JU+(iA-K)V=0.
\]
Their determinant is \((iA-K)^2-J^2=-\Delta\); Cramer's rule yields the displayed \(t\) and \(r\). For real regular \(p\), \(J,K\in\mathbb R\), and direct substitution gives
\( |\Delta|^2=(A^2+J^2-K^2)^2+4A^2K^2 \) and \(|r|^2+|t|^2=1\).
\end{proof}

The corresponding transmission probability is
\(T(p) = \frac{ 4A^2J(p)^2 }{ \left[ A^2+J(p)^2-K(p)^2 \right]^2 + 4A^2K(p)^2 }\).

Away from the arm poles, \(r(p)=0\) if and only if
\(J(p)^2-K(p)^2=A^2\).
This yields the following criterion.

\begin{proposition}[Perfect-transmission criterion]
\label{prop:split-rejoin-perfect}
At every regular energy,
\[
  T(p)=1
\]
if and only if
\[
  J(p)^2-K(p)^2=A^2.
\]
\end{proposition}

\begin{proof}
At a regular real energy the denominator \(\Delta(p)\) is nonzero. By
Theorem~\ref{thm:split-rejoin-exact-S}, perfect transmission is equivalent
to \(r(p)=0\), which is precisely the displayed identity.
\end{proof}

Complete phase alignment is sufficient.

\begin{proposition}[Phase-aligned transparency]
\label{prop:split-rejoin-phase-alignment}
If
\[
  e^{ip\ell_1}
  =
  \cdots
  =
  e^{ip\ell_k}
  =
  e^{i\phi},
\]
then
\[
  r(p)=0,
  \qquad
  t(p)=e^{i\phi}.
\]
\end{proposition}

\begin{proof}
At regular phase-aligned energies,
\(K=A\cot\phi, \qquad J=A\csc\phi\),
so
\(J^2-K^2=A^2\).
Substitution into
Theorem~\ref{thm:split-rejoin-exact-S} gives the result. The remaining
phase-aligned Dirichlet energies follow from the all-energy matching system
of Proposition~\ref{prop:split-rejoin-pole-matching}, or equivalently by
continuity of the exterior scattering amplitudes.
\end{proof}

In particular, if
\(\ell_1=\cdots=\ell_k=\ell\),
then \(r(p)=0, \qquad t(p)=e^{ip\ell}\) for every \(p>0\).

Transmission zeros at regular energies are equally explicit.

\begin{proposition}[Transmission zeros]
\label{prop:split-rejoin-transmission-zero}
If
\[
  \sin(p\ell_j)\neq0
  \qquad
  \text{for every }j,
\]
then
\[
  t(p)=0
\]
if and only if
\[
  J(p)=0.
\]
At such an energy,
\[
  |r(p)|=1.
\]
\end{proposition}

\begin{proof}
For real regular \(p\), the imaginary part of \(\Delta(p)\) is \(2AK(p)\).
If \(\Delta(p)=0\), then \(K(p)=0\), whereas its real part would be
\(A^2+J(p)^2>0\), a contradiction. The formula for \(t(p)\) therefore
vanishes exactly when \(J(p)=0\); unitarity then gives \(|r(p)|=1\).
\end{proof}

These zeros require at least two internal paths. Unlike rooted trees, the
energy-dependent external scattering can therefore detect relative internal
path lengths.

The functions \(J\) and \(K\) are singular when one or more arms satisfy
\(p\ell_j\in\pi\mathbb Z\).
These are singularities of the Dirichlet-to-Neumann coordinates, not of the
underlying matching problem. To treat them without division by
\(\sin(p\ell_j)\), write on every internal arm
\(u_j(x) = U\cos(px)+B_j\sin(px)\),
and set
\(c_j:=\cos(p\ell_j), \qquad s_j:=\sin(p\ell_j)\).

\begin{proposition}[Matching system at arbitrary energy]
\label{prop:split-rejoin-pole-matching}
At every \(p>0\), including energies for which some \(s_j=0\), the
left-incidence scattering problem is equivalent to
\[
  U=1+r,
  \qquad
  V=t,
\]
together with
\[
  V=Uc_j+B_js_j,
  \qquad
  j=1,\ldots,k,
\]
\[
  \sum_{j=1}^{k}a_jB_j
  =
  iA(2-U),
\]
and
\[
  \sum_{j=1}^{k}a_jc_jB_j
  =
  U\sum_{j=1}^{k}a_js_j
  +
  iAV.
\]
When all \(s_j\neq0\), elimination of the \(B_j\) reduces this system to
Theorem~\ref{thm:split-rejoin-exact-S}.
\end{proposition}

\begin{proof}
The representation \(u_j(x)=U\cos(px)+B_j\sin(px)\) automatically gives \(u_j(0)=U\), while evaluation at \(x=\ell_j\) gives
the first displayed relation. The remaining two equations are precisely the
weighted Kirchhoff conditions at \(v_L\) and \(v_R\), after substituting
the exterior scattering waves.
\end{proof}

Pole energies require no separate singular boundary condition. The
finite matching system remains well defined there. If its homogeneous
internal problem has nonzero solutions, the internal coefficients are
determined only modulo the corresponding compact dark eigenspace; the
external scattering amplitudes are unchanged by adding such a state.

\subsection{Compact dark eigenstates}
\label{subsec:split-rejoin-dark-eigenstates}
\label{subsec:split-rejoin-equal-dark}
\label{subsec:split-rejoin-generic-dark}

The homogeneous internal solutions at resonant arm energies are classified
by the two weighted Kirchhoff constraints.

Suppose an eigenfunction vanishes on both external leads. Continuity then
forces
\(U=V=0\).
On arm \(j\), a nonzero solution is possible only when
\(p\ell_j\in\pi\mathbb Z\),
in which case
\(u_j(x)=B_j\sin(px)\).

Define
\[
  \mathcal R(p)
  :=
  \left\{
    j:
    p\ell_j\in\pi\mathbb Z
  \right\},
\]
and, for \(j\in\mathcal R(p)\),
\(\sigma_j := \cos(p\ell_j) \in\{\pm1\}\).

The two Kirchhoff equations become
\(\sum_{j\in\mathcal R(p)} a_jB_j=0\),
and
\(\sum_{j\in\mathcal R(p)} a_j\sigma_jB_j=0\).

\begin{theorem}[Dark-state multiplicity]
\label{thm:split-rejoin-dark-state-multiplicity}
The compactly supported eigenspace at energy \(p^2\) has dimension
\[
  \dim\mathcal E_{\mathrm{dark}}(p)
  =
  |\mathcal R(p)|
  -
  \operatorname{rank}
  \begin{pmatrix}
    (a_j)_{j\in\mathcal R(p)}\\
    (a_j\sigma_j)_{j\in\mathcal R(p)}
  \end{pmatrix}.
\]

If \(\mathcal R(p)\neq\varnothing\) and all resonant arms have the same
parity, then
\[
  \dim\mathcal E_{\mathrm{dark}}(p)
  =
  |\mathcal R(p)|-1.
\]
If both signs \(+1\) and \(-1\) occur among the resonant arms, then
\[
  \dim\mathcal E_{\mathrm{dark}}(p)
  =
  |\mathcal R(p)|-2.
\]
\end{theorem}

\begin{proof}
The coefficient vector \((B_j)_{j\in\mathcal R(p)}\) lies in the kernel of
the displayed two-row matrix, so the first formula is rank--nullity. The two
rows are proportional exactly when all \(\sigma_j\) agree. Positivity of the
weights makes the nonzero row rank one in that case and two when both signs
occur.
\end{proof}

These eigenfunctions are supported entirely on the compact parallel-edge
region. Since the external leads contribute continuous spectrum
\([0,\infty)\),
their eigenvalues are embedded in the essential spectrum. They are the
standard compactly supported bound states in the continuum familiar from
parallel-path quantum graphs; here the graph and weighted coupling are
selected by the non-Hausdorff resolution
\citep{BerkolaikoKuchment2013,ColinDeVerdiereTruc18}.

For equal arm lengths,
\(\ell_1=\cdots=\ell_k=\ell\),
all arms resonate simultaneously at
\(p_n=\frac{\pi n}{\ell}\).
Their parity signs are equal, and therefore:

\begin{corollary}[Equal-arm dark eigenspaces]
\label{cor:split-rejoin-equal-dark}
At every energy
\[
  E_n
  =
  \left(
    \frac{\pi n}{\ell}
  \right)^2,
  \qquad
  n\geq1,
\]
the equal-arm system has a compactly supported dark eigenspace of dimension
\[
  k-1.
\]
\end{corollary}

Hence exact transparency of the equal-arm network coexists with a
\((k-1)\)-dimensional internal eigenspace at every common Dirichlet level.

At the opposite extreme:

\begin{proposition}[Absence of compact dark states]
\label{prop:split-rejoin-generic-no-dark}
If
\[
  \frac{\ell_i}{\ell_j}\notin\mathbb Q
  \qquad
  (i\neq j),
\]
then the resolution-induced Hamiltonian has no compactly supported internal
eigenstates.
\end{proposition}

\begin{proof}
Two distinct arms cannot then satisfy \(p\ell_i\in\pi\mathbb Z, \qquad p\ell_j\in\pi\mathbb Z\) at the same positive momentum. A single resonant arm cannot satisfy the
first Kirchhoff constraint with nonzero coefficient.
\end{proof}

Interference therefore persists generically, whereas exact compact dark
eigenstates require metric commensurability.

\subsection{The oriented first-order operator}
\label{subsec:split-rejoin-first-order}
\label{subsec:split-rejoin-deficiency}
\label{subsec:split-rejoin-SA}
\label{subsec:split-rejoin-tree-comparison}

Orient the left external lead toward \(v_L\), every internal arm from
\(v_L\) to \(v_R\), and the right external lead away from \(v_R\). This
orientation is compatible with the sheetwise coordinate used in the
resolution.

Define \(\mathsf D_\Theta = -i\frac{d}{dx}\) edgewise on
\(\mathcal D(\mathsf D_\Theta) = H^1_{\mathrm{cont}}(\Theta_k)\).
The balance relation \(A=\sum_{j=1}^{k}a_j\) holds at both vertices, so
\(\mathsf D_\Theta\) is symmetric by
Proposition~\ref{prop:first-order-boundary-form-general}.

Its adjoint acts by the same edgewise differential expression and satisfies
the two boundary relations
\(A\psi_L(v_L) = \sum_{j=1}^{k}a_j\psi_j(0)\),
\(\sum_{j=1}^{k}a_j\psi_j(\ell_j) = A\psi_R(v_R)\).

\begin{theorem}[First-order index of the generalized theta resolution]
\label{thm:split-rejoin-deficiency}
The resolution-induced oriented first-order operator satisfies
\[
  (n_+,n_-)
  =
  (k-1,k-1).
\]
For this generalized-theta family,
\[
  n_+=n_-=\beta_1(\Theta_k).
\]
\end{theorem}

\begin{proof}
For
\(\psi\in\ker(\mathsf D_\Theta^*-i)\),
the equation is
\(\psi'=-\psi\).
The left external solution is not square-integrable and therefore vanishes.
On the internal edges,
\(\psi_j(x)=C_je^{-x}\).
The left boundary relation gives
\(\sum_{j=1}^{k}a_jC_j=0\),
leaving \(k-1\) free coefficients. The right external coefficient is then
uniquely determined by
\(A\,C_R = \sum_{j=1}^{k} a_jC_je^{-\ell_j}\).
\(n_+=k-1\).

For
\(\psi\in\ker(\mathsf D_\Theta^*+i)\),
one has
\(\psi'=\psi\).
The right external component vanishes. On the internal edges,
\(\psi_j(x)=C_je^x\),
and the right boundary relation imposes
\(\sum_{j=1}^{k} a_jC_je^{\ell_j}=0\).
This again leaves \(k-1\) free coefficients, while the left external
coefficient is uniquely determined by the remaining boundary relation.
Hence
\(n_-=k-1\).
\end{proof}

For \(k>1\), the geometrically selected derivative is not self-adjoint, but
its self-adjoint extensions exist and are parametrized by
\(U(k-1)\).

The equality with \(\beta_1(\Theta_k)\) is a property of this particular
two-vertex parallel-path family and will not be used as a general
first-order index formula for arbitrary cyclic graphs.

The result nevertheless marks the structural difference from a rooted
branching tree. Splitting alone produces the one-sided deficiency indices
\((k-1,0)\), whereas recombination supplies the opposite deficiency sector
and gives \((k-1,k-1)\). At the same time, the parallel paths create the
phase-sensitive scattering and compact dark states derived above.

Closing the external geometry gives compact parallel-path networks in which
the same bright--dark structure is encoded entirely in the discrete spectrum.
\section{Compact split-and-rejoin networks}
\label{sec:compact-networks}

Closing the split-and-rejoin geometry produces parallel compact paths between two vertices; interference then appears spectrally rather than through a two-port matrix.

The compactification also gives a useful comparison for the first-order
operator. The return path restores a globally balanced orientation, but the
relative parallel-path degrees of freedom remain. For \(k\) resolved paths,
the geometrically selected derivative has deficiency indices
\((k-1,k-1)\).
The first Betti number of the compact graph is instead \(k\), showing already
that the equality between deficiency dimension and cycle rank found for the
open generalized theta graph is not universal.

\subsection{Compact resolution and resolution-induced Hamiltonian}
\label{subsec:compact-theta-resolution}
\label{subsec:compact-theta-metric}
\label{subsec:compact-theta-Hamiltonian}

Take \(k\geq1\) labelled circles
\(S_1,\ldots,S_k\).
On each \(S_j\), choose two marked points
\(p_j, \qquad q_j\),
which divide the circle into two open arcs. Denote by \(C_j\) the arc from
\(q_j\) to \(p_j\), and by \(I_j\) the complementary arc from \(p_j\) to
\(q_j\).

Identify corresponding regular points of \(C_1,\ldots,C_k\) through a fixed orientation-preserving smooth identification, while leaving
the arcs \(I_1,\ldots,I_k\) distinct. The marked endpoints themselves remain distinct. Denote the
resulting non-Hausdorff one-manifold by
\(\widehat M_{\Theta,k}\).

The points \(p_1,\ldots,p_k\) form one pairwise nonseparable exceptional fibre, and \(q_1,\ldots,q_k\) form another. The two exceptional fibres are Hausdorff separated.

The Hausdorff reflection \(\widehat\pi: \widehat M_{\Theta,k} \longrightarrow \widehat\Theta_k\) has two vertices
\(v_L, \qquad v_R\),
and \(k+1\) parallel edges connecting them: one common return edge \(e_0\) and \(k\) resolved edges
\(e_1,\ldots,e_k\).
Hence
\(|E(\widehat\Theta_k)|=k+1, \qquad |V(\widehat\Theta_k)|=2\),
and
\(\beta_1(\widehat\Theta_k)=k\).

For \(k=1\), the Hausdorff reflection is an ordinary circle. For \(k=2\),
it is the standard three-edge theta graph.

Equip the resolving circles with smooth positive one-dimensional metrics
that agree under the identification of the common arcs
\(C_1,\ldots,C_k\). Let \(\ell_0>0\) be the resulting common length of the identified return arc and let
\(\ell_j>0, \qquad j=1,\ldots,k\),
be the length of the complementary resolved arc \(I_j\) on the \(j\)-th
circle. These lengths therefore descend from the source metrics on
\(\widehat M_{\Theta,k}\).

Assign the \(j\)-th resolving circle a positive density
\(a_j>0\),
and set
\(A:=\sum_{j=1}^{k}a_j\).
Pushforward gives the common edge weight \(w_0=A\) and the resolved-edge weights
\(w_j=a_j, \qquad 1\leq j\leq k\).
\(w_0=\sum_{j=1}^{k}w_j\).

The scalar Hilbert space is
\[
  \mathcal H_{\widehat\Theta}
  =
  L^2(0,\ell_0;A\,dx)
  \oplus
  \bigoplus_{j=1}^{k}
  L^2(0,\ell_j;a_j\,dx).
\]
By the scalar closure theorem,
\[
  \overline{\mathcal A_{\widehat M_{\Theta,k}}}^{\,H^1}
  =
  H^1_{\mathrm{cont}}(\widehat\Theta_k).
\]

The resolution-induced kinetic form is
\[
  \mathfrak q_{\widehat\Theta}[u]
  =
  \sum_{e=0}^{k}
  w_e
  \int_0^{\ell_e}|u_e'(x)|^2\,dx
\]
on
\(H^1_{\mathrm{cont}}(\widehat\Theta_k)\).

For the second-order calculation, orient every edge coordinate from
\(v_L\) to \(v_R\), and write
\(U:=u(v_L), \qquad V:=u(v_R)\).

\begin{theorem}[Compact split-and-rejoin Hamiltonian]
\label{thm:compact-theta-Hamiltonian}
The resolution determines the weighted Kirchhoff Laplacian
\[
  H_{\widehat\Theta}
\]
acting as
\[
  (H_{\widehat\Theta}u)_e=-u_e''
\]
on each edge, with domain
\[
  \mathcal D(H_{\widehat\Theta})
  =
  \left\{
    u\in H^2_\oplus(\widehat\Theta_k):
    \begin{array}{l}
      u_e(0)=U
      \text{ for every }e,\\[1mm]
      u_e(\ell_e)=V
      \text{ for every }e,\\[1mm]
      \displaystyle
      \sum_{e=0}^{k}w_eu_e'(0)=0,\\[3mm]
      \displaystyle
      \sum_{e=0}^{k}w_eu_e'(\ell_e)=0
    \end{array}
  \right\}.
\]
It is nonnegative, self-adjoint, and has compact resolvent.
\end{theorem}

\begin{proof}
The domain and differential expression follow from
Theorem~\ref{thm:weighted-kirchhoff} with the stated orientations and
weights. The closed kinetic form is nonnegative, and the embedding of its
form domain into \(\mathcal H_{\widehat\Theta}\) is compact because the graph
has finitely many compact edges. The associated self-adjoint operator
therefore has compact resolvent.
\end{proof}

The final Kirchhoff equation is written after multiplication by \(-1\) from
the outward-normal convention at \(v_R\).

\subsection{Spectral matching at arbitrary energy}
\label{subsec:compact-theta-DN}

Let
\(E=p^2, \qquad p>0\).
Write
\(c_e:=\cos(p\ell_e), \qquad s_e:=\sin(p\ell_e)\).

Instead of dividing initially by \(s_e\), represent the solution on every
edge as
\(u_e(x) = U\cos(px)+B_e\sin(px)\).
The value at \(v_R\) gives
\(V=Uc_e+B_es_e\).

The two Kirchhoff conditions become \(\sum_{e=0}^{k}w_eB_e=0\) and
\(\sum_{e=0}^{k}w_ec_eB_e = U\sum_{e=0}^{k}w_es_e\).

\begin{proposition}[All-energy spectral matching system]
\label{prop:compact-theta-all-energy}
For every \(p>0\), including values satisfying
\[
  \sin(p\ell_e)=0
\]
on one or more edges, \(p^2\) is an eigenvalue of
\(H_{\widehat\Theta}\) if and only if the finite homogeneous system
\[
  V=Uc_e+B_es_e,
  \qquad
  e=0,\ldots,k,
\]
\[
  \sum_{e=0}^{k}w_eB_e=0,
\]
\[
  \sum_{e=0}^{k}w_ec_eB_e
  =
  U\sum_{e=0}^{k}w_es_e
\]
has a nonzero solution
\[
  (U,V,B_0,\ldots,B_k).
\]
\end{proposition}

\begin{proof}
The edge representation solves \(-u_e''=p^2u_e\) and automatically gives the common value \(U\) at \(v_L\). The first family
of equations is continuity at \(v_R\). Since \(u_e'(0)=pB_e\) and
\(u_e'(\ell_e) = p(-Us_e+B_ec_e)\),
the remaining two equations are precisely the weighted Kirchhoff
conditions.
\end{proof}

\subsection{Regular secular equations}
\label{subsec:compact-theta-half-angle}

Suppose now that
\(s_e\neq0 \qquad \text{for every }e\).
Then
\(B_e = \frac{V-Uc_e}{s_e}\).

Define
\(K_c(p) := \sum_{e=0}^{k} w_e\cot(p\ell_e)\),
and
\(J_c(p) := \sum_{e=0}^{k} w_e\csc(p\ell_e)\).

Eliminating the \(B_e\) from
Proposition~\ref{prop:compact-theta-all-energy} gives
\(-K_c(p)U+J_c(p)V=0\),
\(J_c(p)U-K_c(p)V=0\).

\begin{theorem}[Regular secular equation]
\label{thm:compact-theta-secular}
At energies satisfying
\[
  \sin(p\ell_e)\neq0
  \qquad
  \text{for every }e,
\]
one has
\[
  p^2\in\sigma(H_{\widehat\Theta})
\]
if and only if
\[
  K_c(p)^2=J_c(p)^2.
\]
More precisely:
\begin{enumerate}[label=\textup{(\roman*)}]
  \item if
        \[
          K_c(p)=J_c(p)\neq0,
        \]
        then the endpoint kernel is one-dimensional and is spanned by
        \[
          (U,V)=(1,1);
        \]

  \item if
        \[
          K_c(p)=-J_c(p)\neq0,
        \]
        then the endpoint kernel is one-dimensional and is spanned by
        \[
          (U,V)=(1,-1);
        \]

  \item if
        \[
          K_c(p)=J_c(p)=0,
        \]
        then the endpoint matrix vanishes and its kernel is
        \[
          \mathbb C^2.
        \]
        In this simultaneous-zero case the regular eigenspace is
        two-dimensional, with one vertex-symmetric and one
        vertex-antisymmetric direction.
\end{enumerate}
\end{theorem}

\begin{proof}
At regular energy the unique solution on each compact path with endpoint values \(U,V\) has outgoing derivatives \(p(J_cV-K_cU)\) at the left vertex and \(p(J_cU-K_cV)\) at the right. Kirchhoff matching is therefore
\[
 \begin{pmatrix}-K_c&J_c\\J_c&-K_c\end{pmatrix}
 \binom UV=0.
\]
A nonzero endpoint pair exists exactly when \(K_c^2=J_c^2\). If \(K_c=J_c\ne0\) the kernel is \(\operatorname{span}(1,1)\); if \(K_c=-J_c\ne0\) it is \(\operatorname{span}(1,-1)\); and if \(K_c=J_c=0\), the matrix vanishes and the endpoint kernel is \(\mathbb C^2\). Regularity of every path then determines the corresponding edge functions uniquely from \((U,V)\).
\end{proof}

Using
\(\cot x-\csc x = -\tan\frac{x}{2}\),
and
\(\cot x+\csc x = \cot\frac{x}{2}\),
the regular secular equations become \(\sum_{e=0}^{k} w_e \tan \left( \frac{p\ell_e}{2} \right) = 0\) in the vertex-symmetric sector, and \(\sum_{e=0}^{k} w_e \cot \left( \frac{p\ell_e}{2} \right) = 0\) in the vertex-antisymmetric sector.

These are the regular nonzero-vertex secular equations. If \(K_c=J_c=0\), both sectors occur; pole energies are governed by Proposition~\ref{prop:compact-theta-all-energy}.

\subsection{Zero-vertex dark states and commensurability}
\label{subsec:compact-theta-dark}
\label{subsec:compact-theta-perturbation}

The all-energy system contains a distinguished class with
\(U=V=0\).
Such an eigenfunction can be nonzero only on edges satisfying
\(p\ell_e\in\pi\mathbb Z\).

Define
\[
  \mathcal R_c(p)
  :=
  \left\{
    e\in\{0,\ldots,k\}:
    p\ell_e\in\pi\mathbb Z
  \right\},
\]
and, for \(e\in\mathcal R_c(p)\),
\(\sigma_e := \cos(p\ell_e) \in\{\pm1\}\).

On a resonant edge,
\(u_e(x)=C_e\sin(px)\).
The two Kirchhoff conditions reduce to
\(\sum_{e\in\mathcal R_c(p)} w_eC_e=0\),
and
\(\sum_{e\in\mathcal R_c(p)} w_e\sigma_eC_e=0\).

\begin{theorem}[Compact dark-state multiplicity]
\label{thm:compact-theta-dark}
The zero-vertex eigenspace at energy \(p^2\) has dimension
\[
  |\mathcal R_c(p)|
  -
  \operatorname{rank}
  \begin{pmatrix}
    (w_e)_{e\in\mathcal R_c(p)}
    \\[1mm]
    (w_e\sigma_e)_{e\in\mathcal R_c(p)}
  \end{pmatrix}.
\]

If \(\mathcal R_c(p)\neq\varnothing\) and all resonant edges have the same
parity, then
\[
  \dim\mathcal E_{\mathrm{dark}}(p)
  =
  |\mathcal R_c(p)|-1.
\]
If both parities occur, then
\[
  \dim\mathcal E_{\mathrm{dark}}(p)
  =
  |\mathcal R_c(p)|-2.
\]
\end{theorem}

\begin{proof}
The zero-vertex coefficient vector belongs to the kernel of the displayed
matrix. Rank--nullity gives the first formula. If the resonant set is
nonempty and all \(\sigma_e\) agree, the two nonzero rows are proportional;
if both signs occur, positivity of the weights makes them independent.
\end{proof}

The theorem gives only the zero-vertex part of a pole eigenspace; Proposition~\ref{prop:compact-theta-all-energy} also allows solutions with nonzero endpoint values.

Simultaneous Dirichlet resonance is therefore not synonymous with a purely dark eigenvalue.

If the lengths are pairwise rationally incommensurate, no two distinct
edges can resonate simultaneously at positive momentum. A single resonant
edge cannot satisfy the first dark Kirchhoff condition with nonzero
coefficient. Hence the zero-vertex dark sector is absent in that case.

\subsection{Equal-length spectrum}
\label{subsec:compact-theta-equal}

Assume
\(\ell_0 = \ell_1 = \cdots = \ell_k = \ell\).

The spectrum can then be determined directly from the all-energy matching
system.

At
\(p_n=\frac{\pi n}{\ell}, \qquad n\geq1\),
one has \(s_e=0, \qquad c_e=(-1)^n\) for every edge. Continuity gives
\(V=(-1)^nU\).
The two Kirchhoff equations both reduce to the single constraint
\(\sum_{e=0}^{k}w_eB_e=0\).

\(U\) contributes one coherent direction, while the coefficient vector \((B_0,\ldots,B_k)\) has a \(k\)-dimensional weighted-orthogonal subspace.

At \(p=0\), the only eigenfunctions are constants.

\begin{theorem}[Equal-length spectrum]
\label{thm:compact-theta-equal-spectrum}
If all \(k+1\) paths have common length \(\ell\), then
\[
  E_0=0
\]
is simple.

For every \(n\geq1\),
\[
  E_n
  =
  \left(
    \frac{\pi n}{\ell}
  \right)^2
\]
has multiplicity
\[
  k+1.
\]
The eigenspace decomposes as
\[
  \mathcal E_n
  =
  \mathcal E_n^{\mathrm{coh}}
  \oplus
  \mathcal E_n^{\mathrm{dark}},
\]
where
\[
  \dim\mathcal E_n^{\mathrm{coh}}=1,
  \qquad
  \dim\mathcal E_n^{\mathrm{dark}}=k.
\]
\end{theorem}

\begin{proof}
At \(p=0\), continuity and the Kirchhoff conditions leave only the constant
mode. For \(p_n=\pi n/\ell\), the all-energy system above gives one free
coherent endpoint value and the \(k\)-dimensional kernel of the nonzero
functional \((B_e)\mapsto\sum_e w_eB_e\). If
\(\sin(p\ell)\neq0\), then
\(K_c^2-J_c^2=-(\sum_e w_e)^2\neq0\), so the regular secular equation has no
solutions. The listed pole energies therefore exhaust the positive spectrum
and have the stated multiplicities.
\end{proof}

The coherent direction may be represented by \(u_e(x) = \cos \left( \frac{\pi n x}{\ell} \right)\) on every edge. The dark directions have zero vertex values and coefficients
satisfying
\(\sum_{e=0}^{k}w_eC_e=0\).

The multiplicity \(k+1\) is an exact equal-length degeneracy and generically splits under length perturbation.

\subsection{Weyl asymptotics and relation to the open interferometer}
\label{subsec:compact-theta-Weyl}
\label{subsec:compact-theta-cutting}
\label{subsec:compact-theta-open-comparison}

Let
\(L_{\mathrm{tot}} := \sum_{e=0}^{k}\ell_e\).
Since \(H_{\widehat\Theta}\) is a compact quantum-graph Laplacian,
\[
  N(\Lambda)
  =
  \frac{L_{\mathrm{tot}}}{\pi}
  \sqrt{\Lambda}
  +
  O(1)
  \qquad
  (\Lambda\to\infty).
\]
The leading coefficient therefore records the total metric length of the
Hausdorff graph. The resolution-induced weights affect the matching and the
detailed spectrum but not this leading one-dimensional Weyl term.

Cutting the common return edge and replacing the cut directions by leads recovers the open parallel-path core of Section~\ref{sec:split-rejoin}; compact phase quantization is then replaced by scattering.

This comparison is structural rather than an operator identification, because compact closure and asymptotic opening impose different global problems.

\subsection{Balanced first-order orientation}
\label{subsec:compact-theta-first-order-orientation}
\label{subsec:compact-theta-adjoint}

For the first-order operator, orient the resolved edges \(e_1,\ldots,e_k\) from \(v_L\) to \(v_R\), and orient the common return edge \(e_0\) from \(v_R\) to \(v_L\). This cyclic orientation is inherited from the
source circles.

At \(v_L\), the incoming weight is \(A\) and the total outgoing weight is
\(\sum_{j=1}^{k}a_j=A\).
At \(v_R\), the incoming resolved weight is \(A\) and the outgoing return
weight is again \(A\).

\begin{proposition}[Balanced orientation]
\label{prop:compact-theta-balanced-orientation}
With this orientation,
\[
  \mathsf D_{\widehat\Theta}
  =
  -i\frac{d}{dx}
\]
on
\[
  H^1_{\mathrm{cont}}(\widehat\Theta_k)
\]
is symmetric.
\end{proposition}

Use coordinates
\(x_j\in(0,\ell_j), \qquad j=1,\ldots,k\),
increasing from \(v_L\) to \(v_R\), and \(x_0\in(0,\ell_0)\) increasing from \(v_R\) to \(v_L\).

The adjoint acts edgewise as the same differential expression and has the
boundary relations \(A\psi_0(\ell_0) = \sum_{j=1}^{k}a_j\psi_j(0)\) at \(v_L\), and \(\sum_{j=1}^{k}a_j\psi_j(\ell_j) = A\psi_0(0)\) at \(v_R\).

\subsection{First-order deficiency and self-adjoint completion}
\label{subsec:compact-theta-deficiency}
\label{subsec:compact-theta-SA}

Consider
\(\psi\in \ker( \mathsf D_{\widehat\Theta}^*-i )\).
On every oriented edge,
\(\psi'=-\psi\).
Write
\(\psi_0(x)=C_0e^{-x}, \qquad \psi_j(x)=C_je^{-x}\).
The two adjoint relations become
\(Ae^{-\ell_0}C_0 = \sum_{j=1}^{k}a_jC_j\),
and
\(\sum_{j=1}^{k} a_je^{-\ell_j}C_j = AC_0\).

These are two independent relations among the \(k+1\) coefficients. Indeed,
if the second were proportional to the first, comparison of the \(C_0\)
coefficients would force the proportionality factor to be \(-e^{\ell_0}\) after both relations are written on one side, while comparison with any
\(C_j\) would require
\(e^{-\ell_j}=e^{\ell_0}\),
which is impossible for positive lengths.

Hence
\(n_+=k-1\).

For
\(\ker( \mathsf D_{\widehat\Theta}^*+i )\),
one has
\(\psi'=\psi\).
The corresponding equations are
\(Ae^{\ell_0}C_0 = \sum_{j=1}^{k}a_jC_j\),
and
\(\sum_{j=1}^{k} a_je^{\ell_j}C_j = AC_0\).
The same argument shows that they are independent, so
\(n_-=k-1\).

\begin{theorem}[First-order compact deficiency]
\label{thm:compact-theta-deficiency}
For the compact split-and-rejoin resolution,
\[
  (n_+,n_-)
  =
  (k-1,k-1).
\]
Since
\[
  \beta_1(\widehat\Theta_k)=k,
\]
this family satisfies
\[
  n_+=n_-=\beta_1(\widehat\Theta_k)-1.
\]
\end{theorem}

For \(k=1\),
\((n_+,n_-)=(0,0)\),
and the operator is the ordinary self-adjoint momentum operator on the
corresponding circle.

For \(k>1\), self-adjoint extensions exist and are parametrized by
\(U(k-1)\).

Both the open and compact generalized theta models have \((n_+,n_-)=(k-1,k-1)\), although compact closure increases the first Betti number by one. Hence no universal identity between deficiency and cycle rank holds; the ordinary circulation cycle already exists for \(k=1\) without deficiency.

The split-and-rejoin models complete the scalar recombination analysis; finite-rank unitary transport is treated in Part~\ref{part:unitary-gluing}.

\clearpage
\part{Unitary gluing and invariant-sector scattering}
\label{part:unitary-gluing}

For finite-rank gluing, the common fixed space is the transmitting sector. Part~\ref{part:unitary-gluing} identifies this space with compact-group invariants and studies the number of relative gluings required for invariant completion.

\section{Finite-rank unitary gluing on the line with two origins}
\label{sec:unitary-L2}

Retain the standard line \(\Ltwo\), its origins \(o_1,o_2\), and the
charts \(U_1,U_2\) from Section~\ref{sec:L2}. Their overlap is
\(U_1\cap U_2\cong\R\setminus\{0\} =(-\infty,0)\sqcup(0,\infty)\).
A rank-\(r\) Hermitian bundle in this part means a locally trivial complex
bundle on the identity-glued atlas with smooth \(U(r)\)-valued transitions.

Let \(E\to \Ltwo\) be such a bundle with transition function constant on each component of the overlap,
\[
g_{21}(x)=
\begin{cases}
A_-,&x<0,\\
A_+,&x>0,
\end{cases}
\qquad A_\pm\in U(r).
\]

\subsection{Normalization and smooth sections}

A constant unitary change of frame sends
\(A_\pm\longmapsto C_2A_\pm C_1^{-1}\).
Every pair is equivalent to
\[
g_{21}(x)=
\begin{cases}
I,&x<0,\\
W,&x>0,
\end{cases}
\qquad
W=A_-^{-1}A_+\in U(r).
\]

\begin{proposition}[Normalized classification]
\label{prop:unitary-L2-classification}
Two normalized bundles \(E_W\) and \(E_{W'}\) are unitarily isomorphic by a bundle map covering the identity on \(\Ltwo\) if and only if
\[
W'=CWC^{-1}
\]
for some \(C\in U(r)\).
\end{proposition}

\begin{proof}
If \(h_a:U_a\to U(r)\) intertwine the normalized transition functions, then
\(h_2=h_1\) on \(x<0\), so smoothness gives
\(h_1(0)=h_2(0)=:C\). On \(x>0\),
\(h_2(x)W=W'h_1(x)\),
and the limit \(x\to0^+\) gives \(CW=W'C\). The converse is obtained with the constant gauge transformation \(h_1=h_2=C\).
\end{proof}

Write
\(V_W:=\ker(W-I), \qquad P_W:=P_{V_W}, \qquad Q_W:=I-P_W\).
A section is determined by its representative \(\psi=\psi_1\), since
\[
\psi_2(x)=
\begin{cases}
\psi(x),&x<0,\\
W\psi(x),&x>0.
\end{cases}
\]

\begin{theorem}[Fixed-space compatibility]
\label{thm:unitary-L2-fixed}
A function \(\psi\in C^\infty(\R;\C^r)\) represents a smooth global section of \(E_W\) if and only if
\[
\psi^{(n)}(0)\in V_W
\qquad(n\geq0).
\]
The compactly supported smooth section core is
\[
\mathcal C_W
=
C_c^\infty(\R;V_W)
\oplus
C_{c,\mathrm{flat}}^\infty(\R;V_W^\perp),
\]
where the second summand consists of functions whose derivatives of every order vanish at \(0\).
\end{theorem}

\begin{proof}
Smoothness of the second representative at its origin requires the one-sided jets of
\(\psi\) and \(W\psi\) to agree. Since \(W\) is constant, this is \(W\psi^{(n)}(0)=\psi^{(n)}(0)\) for every \(n\). The converse follows by defining the second representative with the displayed transition rule; the equalities of all one-sided jets give a smooth extension through \(0\). The decomposition follows by applying \(P_W\) and \(Q_W\).
\end{proof}

\subsection{Sobolev closure and Friedrichs realisation}

Let \(q:\Ltwo\longrightarrow\R\) be the collapse map that is the identity on the common regular locus and sends both origins to \(0\). Write
\(\Ltwo^\circ:=\Ltwo\setminus\{o_1,o_2\}\).
We use the Borel measure \(\mu\) characterised by \(\mu(\{o_1\})=\mu(\{o_2\})=0\) and by transport of Lebesgue measure through
\(q|_{\Ltwo^\circ}:\Ltwo^\circ\longrightarrow\R\setminus\{0\}\).
In a unitary reference trivialisation,
\(L^2(\Ltwo,E_W;\mu)\cong L^2(\R;\C^r)\).
Because the transition function is locally constant,
\(d(g_{21}\psi)=g_{21}\,d\psi\),
so the ordinary derivative in the two unitary frames defines a flat Hermitian connection. On \(\mathcal C_W\) consider
\(q_W^{(0)}[\psi] = \int_\R\|\psi'(x)\|^2\,dx\).

Point evaluation is continuous on \(H^1(\R)\), and the zero-trace subspace on either half-line is \(H_0^1\)
\citep{AdamsFournier2003}. These facts give the following closure directly.

\begin{theorem}[Form domain and operator decomposition]
\label{thm:unitary-L2-Friedrichs}
The \(H^1\)-closure of \(\mathcal C_W\) is
\[
\mathcal H_W^1
=
\{\psi\in H^1(\R;\C^r):\psi(0)\in V_W\},
\]
with orthogonal decomposition
\[
\mathcal H_W^1
=
H^1(\R;V_W)
\oplus
H_0^1((-\infty,0);V_W^\perp)
\oplus
H_0^1((0,\infty);V_W^\perp).
\]
The form \(q_W^{(0)}\) closes on \(\mathcal H_W^1\), and its associated nonnegative self-adjoint operator is
\[
H_W
\cong
H_{\mathrm{free}}\otimes I_{V_W}
\oplus
(H_{D,-}\oplus H_{D,+})\otimes I_{V_W^\perp}.
\]
Its domain is
\[
\mathcal D(H_W)
=
H^2(\R;V_W)
\oplus
\{u\in H^2(\R\setminus\{0\};V_W^\perp):
u(0^-)=u(0^+)=0\},
\]
where
\[
H^2(\R\setminus\{0\};V_W^\perp)
:=
H^2((-\infty,0);V_W^\perp)
\oplus
H^2((0,\infty);V_W^\perp).
\]
\end{theorem}

\begin{proof}
The trace condition is necessary by continuity of
\(\psi\mapsto Q_W\psi(0)\). Conversely, write
\(\psi=P_W\psi+Q_W\psi\). The invariant component is approximated in \(H^1\) by \(C_c^\infty(\R;V_W)\). The complementary component has zero trace, hence belongs to the direct sum of the two half-line \(H_0^1\) spaces and is approximated there by smooth compactly supported functions away from the origin. This proves the form-domain decomposition.

The representation theorem for closed nonnegative forms
\citep{Kato1995,Schmudgen2012}
then gives the orthogonal direct sum of the free-line Laplacian and the two Dirichlet half-line Laplacians.
\end{proof}

The smooth condition retains every defect jet, the \(H^1\) form domain retains only the value trace, and the underlying \(L^2\) space retains no pointwise condition. Accordingly,
\(\overline{\mathcal C_W}^{\,L^2}=L^2(\R;\C^r)\),
because \(C_c^\infty(\R\setminus\{0\};\C^r)\subset\mathcal C_W\) is \(L^2\)-dense.

\subsection{Stationary scattering}

Fix an energy \(E=p^2>0\) and use the standard left/right incoming-outgoing normalization
\citep{KostrykinSchrader99,BerkolaikoKuchment2013}. On \(V_W\), the free-line component has transmission amplitude \(1\) and reflection amplitude \(0\). On \(V_W^\perp\), each Dirichlet half-line has transmission amplitude \(0\) and reflection amplitude \(-1\). Hence the internal transmission and reflection amplitude operators are
\(T_W=P_W, \qquad R_W=-Q_W\),
and
\[
S_W=
\begin{pmatrix}
-Q_W&P_W\\
P_W&-Q_W
\end{pmatrix}.
\]
The matrix satisfies
\(S_W^\ast S_W=I, \qquad S_W^2=I\).
For \(\|\xi\|=1\),
\(\mathcal T_W(\xi)=\|P_W\xi\|^2, \qquad \mathcal R_W(\xi)=\|Q_W\xi\|^2\).
The scattering data therefore depend on \(W\) only through its \(1\)-eigenspace. If
\(\ker(W-I)=\ker(W'-I)\),
then the closed forms, Friedrichs operators, and scattering matrices coincide in the common reference realisation.

\section{Finite-rank unitary gluing on lines with finitely many origins}
\label{sec:unitary-Lk}

Retain the line \(\mathbb L_k\), its origins \(o_1,\ldots,o_k\), and the standard
charts \(U_a\) from Section~\ref{sec:Lk}. Each \(U_a\) is identified with
\(\R\), and every pairwise overlap is \(\R\setminus\{0\}\).

Let \(E\to \mathbb L_k\) be a rank-\(r\) Hermitian bundle whose transition functions are unitary and constant on each component of each overlap. Choose \(U_1\) as a reference chart. Let \(q_k:\mathbb L_k\to\R\) send the common nonzero point represented by \((x,a)\) to \(x\) and each origin to \(0\). The measure used below assigns zero mass to the origins and is the transport of Lebesgue measure through
\(q_k: \mathbb L_k\setminus\{o_1,\ldots,o_k\} \longrightarrow \R\setminus\{0\}\).

\subsection{Normalized cocycle data}

\begin{proposition}[Normalization and simultaneous conjugacy]
\label{prop:unitary-Lk-normalization}
After constant unitary changes of frame, the transitions from \(U_1\) to \(U_a\) may be written
\[
g_{a1}(x)=
\begin{cases}
I,&x<0,\\
W_a,&x>0,
\end{cases}
\qquad
a=2,\ldots,k,
\]
for arbitrary \(W_a\in U(r)\). The remaining transitions are forced by the cocycle:
\[
g_{ba}(x)=
\begin{cases}
I,&x<0,\\
W_bW_a^{-1},&x>0.
\end{cases}
\]
Two normalized tuples
\(\mathbf W=(W_2,\ldots,W_k)\) and \(\mathbf W'=(W_2',\ldots,W_k')\) define bundles isomorphic by a unitary bundle map covering the identity if and only if
\[
W_a'=CW_aC^{-1}
\qquad(a=2,\ldots,k)
\]
for a single \(C\in U(r)\).
\end{proposition}

\begin{proof}
If the original transition from \(U_1\) to \(U_a\) is
\((A_{a1,-},A_{a1,+})\), the frame change
\(C_a=A_{a1,-}^{-1}\), \(C_1=I\), gives the normalized pair
\((I,W_a)\) with \(W_a=A_{a1,-}^{-1}A_{a1,+}\).
The pairwise formula follows from
\(g_{ba}=g_{b1}g_{a1}^{-1}\).

For an isomorphism between normalized tuples, the negative-overlap relation gives
\(h_a(0)=h_1(0)=C\), while the positive-overlap relation gives
\(CW_a=W_a'C\). Constant \(h_a=C\) proves the converse.
\end{proof}

Set
\[
V_{\mathbf W}:=
\bigcap_{a=2}^{k}\ker(W_a-I),
\qquad
P_{\mathbf W}:=P_{V_{\mathbf W}},
\qquad
Q_{\mathbf W}:=I-P_{\mathbf W}.
\]

\subsection{Common fixed-space theorem}

\begin{theorem}[Finite-origin fixed-space dynamics]
\label{thm:unitary-Lk-main}
With the measure defined above and the same flat Dirichlet form as in the two-origin case, the normalized tuple
\(\mathbf W=(W_2,\ldots,W_k)\) has the following properties.

\begin{enumerate}
\item A reference representative \(\psi\in C^\infty(\R;\C^r)\) defines a smooth global section if and only if
\[
\psi^{(n)}(0)\in V_{\mathbf W}
\qquad(n\geq0).
\]

\item The \(H^1\)-closure of the compactly supported smooth section core is
\[
\mathcal H_{\mathbf W}^1
=
\{\psi\in H^1(\R;\C^r):\psi(0)\in V_{\mathbf W}\}.
\]

\item The Friedrichs operator is
\[
H_{\mathbf W}
\cong
H_{\mathrm{free}}\otimes I_{V_{\mathbf W}}
\oplus
(H_{D,-}\oplus H_{D,+})\otimes I_{V_{\mathbf W}^\perp}.
\]

\item The stationary transmission and reflection amplitude operators are
\[
T_{\mathbf W}=P_{\mathbf W},
\qquad
R_{\mathbf W}=-Q_{\mathbf W}.
\]
\end{enumerate}
\end{theorem}

\begin{proof}
For each \(a\geq2\), the local representative is
\[
\psi_a(x)=
\begin{cases}
\psi(x),&x<0,\\
W_a\psi(x),&x>0.
\end{cases}
\]
Smoothness at \(o_a\) therefore requires
\(W_a\psi^{(n)}(0)=\psi^{(n)}(0)\) for every \(n\), simultaneously for all \(a\), which gives the first assertion.

Once \(V_W\) in the proof of Theorem~\ref{thm:unitary-L2-Friedrichs} is replaced by the closed subspace \(V_{\mathbf W}\), the trace, density, form-closure, and direct-sum arguments are unchanged. The scattering statement follows from the resulting free--Dirichlet decomposition.
\end{proof}

The additional origins contribute the nested internal constraints
\[
\ker(W_2-I)
\supseteq
\ker(W_2-I)\cap\ker(W_3-I)
\supseteq\cdots\supseteq
V_{\mathbf W}.
\]

\section{The generated compact group}
\label{sec:generated-group}

Define
\[
\Gamma_{\mathbf W}:=\langle W_2,\ldots,W_k\rangle,
\qquad
G_{\mathbf W}:=\overline{\Gamma_{\mathbf W}}\subseteq U(r).
\]
Let \(d_{\mathrm{top}}(G)\) denote the least cardinality of a subset whose generated subgroup is dense in the compact group \(G\).

\begin{proposition}[Origin bound on gluing-group generation]
\label{prop:generator-bound}
For every normalized \(\mathbb L_k\) gluing tuple,
\[
d_{\mathrm{top}}(G_{\mathbf W})\leq k-1.
\]
In particular, for \(k=2\),
\[
G_W=\overline{\langle W\rangle}
\]
is abelian. Hence a non-Abelian compact gluing group cannot occur on \(\Ltwo\) within the present locally constant model.
\end{proposition}

\begin{proof}
By definition, \(G_{\mathbf W}\) is the closure of the subgroup generated by the \(k-1\) matrices
\(W_2,\ldots,W_k\). For \(k=2\) the generating subgroup is cyclic and therefore abelian; its closure in \(U(r)\) remains abelian.
\end{proof}

\begin{proposition}[Invariant-sector formulation]
\label{prop:group-invariant}
One has
\[
V_{\mathbf W}
=
(\C^r)^{\Gamma_{\mathbf W}}
=
(\C^r)^{G_{\mathbf W}}.
\]
\end{proposition}

\begin{proof}
A vector fixed by the generators is fixed by every word in them and their inverses. Invariance extends from \(\Gamma_{\mathbf W}\) to its closure by continuity.
\end{proof}

\begin{theorem}[Fibrewise-unitary classification of the flat dynamics]
\label{thm:dimension-classification}
Let \(\mathbf W\) and \(\mathbf W'\) be rank-\(r\) normalized gluing tuples, possibly on lines with different numbers of origins. The associated flat Friedrichs operators are intertwined by a constant fibrewise unitary
\[
I_{L^2(\R)}\otimes C,
\qquad C\in U(r),
\]
if and only if
\[
\dim V_{\mathbf W}=\dim V_{\mathbf W'}.
\]
For fixed rank \(r\), there are exactly \(r+1\) classes under constant fibrewise unitary equivalence, indexed by
\[
m=\dim V_{\mathbf W}\in\{0,1,\ldots,r\}.
\]
\end{theorem}

\begin{proof}
If \(\dim V_{\mathbf W}=\dim V_{\mathbf W'}\), choose \(C\in U(r)\) with
\(CV_{\mathbf W}=V_{\mathbf W'}\).
Then \(C\) maps \(V_{\mathbf W}^{\perp}\) onto \(V_{\mathbf W'}^{\perp}\), and
\(I\otimes C\) intertwines the direct-sum realisations in Theorem~\ref{thm:unitary-Lk-main}.

Conversely, suppose \(U=I\otimes C\) intertwines the two nonnegative self-adjoint operators. By functional calculus,
\(U H_{\mathbf W}^{1/2} = H_{\mathbf W'}^{1/2}U\),
and hence \(U\) maps the square-root form domain of \(H_{\mathbf W}\) onto that of \(H_{\mathbf W'}\):
\(U\mathcal H_{\mathbf W}^1 = \mathcal H_{\mathbf W'}^1\).
The sets of admissible traces are
\[
\{\psi(0):\psi\in\mathcal H_{\mathbf W}^1\}=V_{\mathbf W},
\qquad
\{\psi(0):\psi\in\mathcal H_{\mathbf W'}^1\}=V_{\mathbf W'}.
\]
Since
\((U\psi)(0)=C\psi(0)\),
it follows that
\(CV_{\mathbf W}=V_{\mathbf W'}\).
Hence the fixed spaces have equal dimension. Every \(m\in\{0,\ldots,r\}\) occurs, for example by taking on \(\Ltwo\) a diagonal unitary with \(1\)-eigenspace of dimension \(m\).
\end{proof}

\begin{remark}[Equivalence notion]
\label{rem:fibrewise-equivalence}
Theorem~\ref{thm:dimension-classification} is a classification under
\emph{constant fibrewise} unitaries \(I_{L^2(\R)}\otimes C\), which preserve
the spatial/fibre factorisation and the free--Dirichlet sector
identification. No claim is made that
\(\dim V_{\mathbf W}\) is an invariant under an arbitrary unitary operator
on \(L^2(\R;\C^r)\), which need not preserve the trace domain, the spatial
factorisation, or the left--right scattering identification. All
subsequent statements about the \(r+1\) dynamical classes use this
fibrewise notion of equivalence.
\end{remark}

Hence Theorem~\ref{thm:unitary-Lk-main} may be written
\[
H_{\mathbf W}
\cong
H_{\mathrm{free}}\otimes I_{(\C^r)^{G_{\mathbf W}}}
\oplus
(H_{D,-}\oplus H_{D,+})
\otimes I_{((\C^r)^{G_{\mathbf W}})^\perp}.
\]

\subsection{Finite and permutation gluing groups}

For a finite group \(G\) with a unitary representation
\(\rho:G\to U(V)\), the standard averaging operator \(P_G=\frac1{|G|}\sum_{g\in G}\rho(g)\) is the orthogonal projector onto \(V^G\), and
\(\dim V^G=\frac1{|G|}\sum_{g\in G}\chi_\rho(g)\)
\citep{BrockerTomDieck1985,Hall2015}.

For the natural permutation representation of
\(G\leq S_r\) on \(\C^r\), invariant vectors are exactly those whose coordinates are constant on \(G\)-orbits. If \(\mathcal O_1,\ldots,\mathcal O_m\) are the orbits, then \(v_j= \frac1{\sqrt{|\mathcal O_j|}} \sum_{i\in\mathcal O_j}e_i\) form an orthonormal basis of \((\C^r)^G\). Therefore
\(\dim(\C^r)^G = \#(\{1,\ldots,r\}/G)\).
For \(k=2\), \(G=\langle\sigma\rangle\), and the orbits are the cycles of \(\sigma\).

\subsection{Compact groups}

Let \(G\) be compact and \(\rho:G\to U(V)\) a finite-dimensional unitary representation. If the relative gluing elements generate a dense subgroup of \(G\), Theorem~\ref{thm:unitary-Lk-main} selects \(V^G\).

For normalized Haar measure \(\mu_G\), \(P_G = \int_G\rho(g)\,d\mu_G(g)\) is the orthogonal projector onto \(V^G\), and
\(\dim V^G = \int_G\chi_\rho(g)\,d\mu_G(g)\).
If
\(V\cong \bigoplus_{\lambda\in\widehat G} M_\lambda\otimes V_\lambda\),
then
\(V^G\cong M_{\mathbf1}\),
so its dimension is the multiplicity of the trivial representation
\citep{Hall2015,BrockerTomDieck1985}. For the invariant-complete flat model, the transmission and reflection amplitude operators are \(T=P_G, \qquad R=-(I-P_G)\) for the corresponding non-Hausdorff Friedrichs problem.

\section{Connected groups and infinitesimal invariants}
\label{sec:lie}

Let \(G\) be a connected compact Lie group with Lie algebra \(\mathfrak g\), and let \(\rho:G\to U(V)\) be a finite-dimensional continuous unitary representation with derived representation
\(d\rho:\mathfrak g\to\mathfrak u(V)\).
Standard Lie theory gives
\[
V^G
=
\{v\in V:d\rho(X)v=0\ \text{for all }X\in\mathfrak g\}
=
\bigcap_{X\in\mathfrak g}\ker d\rho(X)
\]
\citep{Hall2015}. If \(X_1,\ldots,X_m\) generate \(\mathfrak g\) as a Lie algebra, then
\(V^G = \bigcap_{a=1}^{m}\ker d\rho(X_a)\),
because the common kernel of the \(d\rho(X_a)\) is invariant under iterated Lie brackets.

For a single \(X\in\mathfrak g\), \(\ker(\rho(\exp X)-I)\) can strictly contain \(\ker d\rho(X)\), since a nonzero infinitesimal eigenvalue may exponentiate to \(1\). The one-parameter fixed-space identity is
\[
\rho(\exp(tX))v=v\quad\text{for all }t\in\R
\quad\Longleftrightarrow\quad
d\rho(X)v=0.
\]
\section{Fixed-space filtrations and invariant completion}
\label{sec:invariant-completion}

Let \(\rho:G\to GL(V)\) be a representation on a finite-dimensional complex vector space. For \(g_1,\ldots,g_m\in G\), set
\[
V^{(0)}:=V,
\qquad
V^{(j)}
:=
\bigcap_{a=1}^{j}\ker(\rho(g_a)-I),
\qquad
1\leq j\leq m.
\]
Then
\(V=V^{(0)} \supseteq V^{(1)} \supseteq\cdots\supseteq V^{(m)} \supseteq V^G\).
When \(\rho(G)\subseteq U(V)\), the matrices \(W_{a+1}=\rho(g_a)\) may be used as relative transitions on \(\mathbb L_k\) for \(m\leq k-1\). The transmission rank after the first \(j\) chosen constraints is
\(s_j(g_1,\ldots,g_j) := \dim V^{(j)} = \operatorname{rank}T_j\).
The sequence \(s_j\) is attached to one nested gluing tuple.

\begin{definition}[Effective and redundant gluing]
\label{def:effective-gluing}
For a fixed sequence \(g_1,\ldots,g_m\), the element \(g_j\) is \emph{effective relative to} \(g_1,\ldots,g_{j-1}\) if
\[
V^{(j)}\subsetneq V^{(j-1)};
\]
it is \emph{redundant} if
\[
V^{(j)}=V^{(j-1)}.
\]
\end{definition}

\begin{proposition}[Redundancy criterion]
\label{prop:redundancy}
The \(j\)-th gluing is redundant if and only if
\[
V^{(j-1)}
\subseteq
\ker(\rho(g_j)-I).
\]
\end{proposition}

\begin{proof}
By definition,
\(V^{(j)} = V^{(j-1)}\cap\ker(\rho(g_j)-I)\),
so equality with \(V^{(j-1)}\) is equivalent to the displayed inclusion.
\end{proof}

\subsection{Invariant-completion number}

\begin{definition}[Invariant-completion number]
\label{def:completion-number}
Let \(\rho:G\to GL(V)\) be finite dimensional. Its \emph{invariant-completion number} is
\[
\nu(\rho)
:=
\min
\left\{
m\geq0:
\exists\,g_1,\ldots,g_m\in G,\ 
\bigcap_{a=1}^{m}\ker(\rho(g_a)-I)=V^G
\right\},
\]
where the intersection for \(m=0\) is \(V\).
\end{definition}

The equality \(\nu(\rho)=0\) holds exactly for the trivial representation. Common fixed spaces and finite-subfamily questions for matrix families have been studied independently
\citep{BernikEtAl2005CommonFixed}. The following dimension bound is used here to convert finite-dimensional fixed-space completion into an origin bound.

\begin{theorem}[Finite invariant completion]
\label{thm:finite-invariant-completion}
For every finite-dimensional representation \(\rho:G\to GL(V)\),
\[
\nu(\rho)
\leq
\dim V-\dim V^G.
\]
\end{theorem}

\begin{proof}
Set \(U_0=V\). If \(U_j\neq V^G\), choose \(v_j\in U_j\setminus V^G\). There exists \(g_{j+1}\in G\) with
\(\rho(g_{j+1})v_j\neq v_j\).
Then \(U_{j+1} = U_j\cap\ker(\rho(g_{j+1})-I)\) is a proper subspace of \(U_j\) containing \(V^G\). Each such step reduces the dimension by at least one, so the process reaches \(V^G\) after at most \(\dim V-\dim V^G\) steps.
\end{proof}

For a unitary representation \(\rho:G\to U(V)\), if
\(k-1\geq\nu(\rho)\),
Corollary~\ref{cor:origin-completion} below identifies the corresponding threshold on \(\mathbb L_k\).

\begin{corollary}[Origin bound for invariant completion]
\label{cor:origin-completion}
Let \(\rho:G\to U(V)\) be finite dimensional. If
\[
k-1\geq\nu(\rho),
\]
there is a locally constant gluing system on \(\mathbb L_k\) with relative matrices in \(\rho(G)\) and transmitting sector \(V^G\).
\end{corollary}

\begin{proof}
Choose \(g_1,\ldots,g_{\nu(\rho)}\) realising Definition~\ref{def:completion-number} and set
\(W_{a+1}=\rho(g_a)\).
If \(k-1>\nu(\rho)\), take the remaining relative transitions to be the identity. Theorem~\ref{thm:unitary-Lk-main} then gives
\(V_{\mathbf W}=V^G\).
\end{proof}

If the base is required to be non-Hausdorff, the least admissible origin number for invariant completion is
\(k_{\mathrm{inv}}(\rho) = \max\{2,\nu(\rho)+1\}\).

\subsection{Completion and topological generation}

If
\(\Gamma=\langle g_1,\ldots,g_m\rangle\),
then
\(\bigcap_{a=1}^{m}\ker(\rho(g_a)-I)=V^\Gamma\).
For a continuous representation, density of \(\Gamma\) in \(G\) implies
\(V^\Gamma=V^G\).

\begin{proposition}[Comparison with topological generation]
\label{prop:completion-generation}
Let \(G\) be a topologically finitely generated compact group and let
\(\rho:G\to GL(V)\) be a continuous finite-dimensional representation. Then
\[
\nu(\rho)
\leq
d_{\mathrm{top}}(\rho(G))
\leq
d_{\mathrm{top}}(G).
\]
\end{proposition}

\begin{proof}
Since \(G\) is compact, \(\rho(G)\) is compact. A dense generating
\(m\)-tuple in \(\rho(G)\), with
\(m=d_{\mathrm{top}}(\rho(G))\), has common fixed space
\(V^{\rho(G)}=V^G\),
so
\(\nu(\rho)\leq d_{\mathrm{top}}(\rho(G))\).
If \(g_1,\ldots,g_d\) generate a dense subgroup of \(G\), continuity of
\(\rho\) implies that \(\rho(g_1),\ldots,\rho(g_d)\) generate a dense subgroup of \(\rho(G)\). Hence
\(d_{\mathrm{top}}(\rho(G))\leq d_{\mathrm{top}}(G)\).
\end{proof}

\begin{corollary}[Two-element completion for compact semisimple groups]
\label{cor:semisimple-two-completion}
Let \(G\) be a connected compact semisimple Lie group and let \(\rho:G\to U(V)\) be a finite-dimensional continuous unitary representation. Then
\[
\nu(\rho)\leq2.
\]
If \(G\) is nontrivial, then
\[
d_{\mathrm{top}}(G)=2.
\]
For \(G=SU(N)\), \(N\geq2\), \(\mathbb L_3\) has enough relative-gluing slots to realise \(V^{SU(N)}\) for every finite-dimensional continuous unitary representation.
\end{corollary}

\begin{proof}
If \(G\) is trivial, then \(\nu(\rho)=0\). Assume that \(G\) is nontrivial.
The dense-two-generation theorem for connected semisimple real Lie groups
\citep{BreuillardGelander2003,BreuillardGelander2026Revisited}, applied to \(G\) as a dense subgroup of
itself, supplies \(g_1,g_2\in G\) with dense generated subgroup. Hence
\(V^{g_1}\cap V^{g_2}=V^G\),
which gives \(\nu(\rho)\leq2\). A nontrivial connected compact semisimple group is non-Abelian, while the closure of a cyclic subgroup is abelian, so one element cannot generate \(G\) densely.
\end{proof}

For nontrivial connected compact semisimple \(G\), dense generating pairs
form a Zariski-open dense subset of \(G^2\)
\citep{Chirvasitu2021GeneratingTuples}.

Both inequalities \(\nu(\rho) \leq d_{\mathrm{top}}(\rho(G)) \leq d_{\mathrm{top}}(G)\) can be strict. The first can be strict because a proper subgroup of
\(\rho(G)\) may have the same fixed space as \(\rho(G)\); the second can be
strict when \(\rho\) has nontrivial kernel or otherwise lowers the
topological generator number of the image.

\subsection{Optimal fixed-space ranks}

For a finite-dimensional unitary representation of a compact group, define
\begin{equation}
\label{eq:optimal-rank}
r_m(\rho)
:=
\min_{g_1,\ldots,g_m\in G}
\dim
\left(
\bigcap_{a=1}^{m}\ker(\rho(g_a)-I)
\right),
\qquad
r_0(\rho):=\dim V.
\end{equation}
Identity elements may be inserted, so this is also the minimum attainable with at most \(m\) nontrivial constraints.

\begin{proposition}[Optimal-rank envelope]
\label{prop:rank-profile-properties}
The sequence \(r_m(\rho)\) satisfies
\[
\dim V=r_0(\rho)
\geq r_1(\rho)
\geq r_2(\rho)
\geq\cdots
\geq\dim V^G,
\]
and
\[
\nu(\rho)
=
\min\{m:r_m(\rho)=\dim V^G\}.
\]
For \(m\geq\nu(\rho)\),
\[
r_m(\rho)=\dim V^G.
\]
\end{proposition}

\begin{proof}
Appending an identity element shows
\(r_{m+1}(\rho)\leq r_m(\rho)\).
Every common fixed space contains \(V^G\), which gives the lower bound. The remaining statements are equivalent to Definition~\ref{def:completion-number}.
\end{proof}

The minima \(r_m(\rho)\) are optimised independently in \(m\); a fixed nested tuple has the actual ranks \(s_j(g_1,\ldots,g_j)\).

\begin{theorem}[Origin--rank correspondence]
\label{thm:origin-rank-correspondence}
Let \(\rho:G\to U(V)\) be a finite-dimensional unitary representation of a compact group. Among locally constant gluing systems on \(\mathbb L_k\) with relative matrices in \(\rho(G)\),
\[
\min_{\mathbf W}\operatorname{rank}T_{\mathbf W}
=
r_{k-1}(\rho).
\]
Exact \(G\)-invariant transmission is realizable if and only if
\[
k-1\geq\nu(\rho).
\]
\end{theorem}

\begin{proof}
A normalized \(\mathbb L_k\) model has \(k-1\) relative matrices. Writing
\(W_{a+1}=\rho(g_a)\),
Theorem~\ref{thm:unitary-Lk-main} gives
\(\operatorname{rank}T_{\mathbf W} = \dim\bigcap_{a=1}^{k-1}\ker(\rho(g_a)-I)\).
Minimization gives the first identity; the second follows from the definition of \(\nu(\rho)\).
\end{proof}
\section{Invariant completion for compact symmetry groups}
\label{sec:compact-completion}

\subsection{One-element completion}

\begin{proposition}[Criterion for one-element completion]
\label{prop:one-element-completion}
Let \(\rho:G\to U(V)\) be a finite-dimensional unitary representation of a compact group. Then
\[
  \nu(\rho)\leq1
\]
if and only if there exists \(g\in G\) such that
\[
  \ker(\rho(g)-I)=V^G.
\]
Equivalently, there exists \(g\in G\) such that the restriction
\[
  \rho(g)|_{(V^G)^\perp}
\]
has no eigenvalue \(1\). Moreover,
\[
  \nu(\rho)=1
\]
if and only if \(V^G\neq V\) and such an element \(g\) exists.
\end{proposition}

\begin{proof}
Since \(V^G\) is fixed pointwise and the representation is unitary,
\(\ker(\rho(g)-I) = V^G \oplus \ker\!\left( \rho(g)|_{(V^G)^\perp}-I \right)\).
Hence a single group element has fixed space exactly \(V^G\) if and only if
the second summand vanishes. This is equivalent to \(\nu(\rho)\leq1\). If
\(V^G=V\), then \(\nu(\rho)=0\); otherwise one-element completion gives
\(\nu(\rho)=1\).
\end{proof}

\subsection{\texorpdfstring{\(U(1)\)}{U(1)} representations}

Let
\(V=\bigoplus_{q\in\mathbb Z}V_q, \qquad \rho(e^{i\theta})|_{V_q}=e^{iq\theta}I\),
be a finite-dimensional unitary representation of \(U(1)\). Then \(V^{U(1)}=V_0\) and
\[
\ker(\rho(e^{i\theta})-I)
=
\bigoplus_{\substack{q\in\mathbb Z\\q\theta\in2\pi\mathbb Z}}V_q.
\]

\begin{proposition}[\texorpdfstring{\(U(1)\)}{U(1)} completion]
\label{prop:U1-completion}
Every finite-dimensional nontrivial unitary representation of \(U(1)\) satisfies
\[
\nu(\rho)=1.
\]
\end{proposition}

\begin{proof}
Only finitely many nonzero weights occur. Choose \(\theta\) outside the finite union of sets on which \(e^{iq\theta}=1\) for an occurring nonzero weight \(q\). Then
\(\ker(\rho(e^{i\theta})-I)=V_0\).
\end{proof}

\subsection{Regular torus elements}

Let \(G\) be compact and connected, \(T\leq G\) a maximal torus, and \(V=\bigoplus_\lambda V_\lambda\) the corresponding weight decomposition. For \(t\in T\),
\(V^t = \bigoplus_{\lambda(t)=1}V_\lambda\).
Outside the finite union of resonance loci determined by the nonzero occurring weights,
\(V^t=V_0\).
If
\(V_0\supsetneq V^G\),
no element of \(T\) is invariant-complete, because every \(t\in T\) fixes
\(V_0\). For regular \(t\), its centralizer is \(T\); any further element
that removes residual zero-weight vectors must therefore lie outside the
centralizer of \(t\).

\subsection{\texorpdfstring{\(SU(2)\)}{SU(2)} representations}

Let \(V_j\) be the spin-\(j\) irreducible representation. Relative to a maximal torus, its weights are
\(-j,-j+1,\ldots,j\).
A generic torus element has fixed space
\[
V_j^t
=
\begin{cases}
\{0\}, & j\in\frac12+\mathbb Z_{\geq0},\\[1mm]
\mathbb C|j,0\rangle, & j\in\mathbb Z_{\geq0},
\end{cases}
\]
whereas
\[
V_j^{SU(2)}
=
\begin{cases}
\mathbb C, & j=0,\\
\{0\}, & j>0.
\end{cases}
\]
These are standard consequences of the weight decomposition
\citep{Hall2015,BrockerTomDieck1985}.

\begin{proposition}[Irreducible \texorpdfstring{\(SU(2)\)}{SU(2)} completion]
\label{prop:SU2-irrep-completion}
For \(j>0\),
\[
\nu(V_j)
=
\begin{cases}
1, & j\in\frac12+\mathbb Z_{\geq0},\\
2, & j\in\mathbb Z_{>0}.
\end{cases}
\]
The trivial representation satisfies \(\nu(V_0)=0\).
\end{proposition}

\begin{proof}
For half-integral \(j\), a generic torus element has no fixed vector, so \(\nu(V_j)=1\). For positive integral \(j\), every element of \(SU(2)\) lies in a maximal torus and fixes the corresponding zero-weight line; hence \(\nu(V_j)\geq2\). Corollary~\ref{cor:semisimple-two-completion} gives \(\nu(V_j)\leq2\).
\end{proof}

For a general finite-dimensional representation
\(V\cong\bigoplus_j M_j\otimes V_j, \qquad V^{SU(2)}\cong M_0\).

\begin{proposition}[Reducible \texorpdfstring{\(SU(2)\)}{SU(2)} completion]
\label{prop:SU2-reducible-completion}
If \(M_j=0\) for every positive integral \(j\), then
\[
\nu(\rho)=
\begin{cases}
0,&V=V^{SU(2)},\\
1,&V\neq V^{SU(2)}.
\end{cases}
\]
If \(M_j\neq0\) for at least one positive integral \(j\), then
\[
\nu(\rho)=2.
\]
\end{proposition}

\begin{proof}
When no positive integral-spin summand occurs, a generic torus element fixes exactly \(M_0\). If a positive integral-spin summand occurs, every group element fixes a zero-weight line in that summand, so its fixed space strictly contains \(M_0\). One-element completion is therefore impossible. Corollary~\ref{cor:semisimple-two-completion} gives the upper bound \(2\).
\end{proof}

For the diagonal action on two spin-\(\frac12\) factors,
\(V_{\frac12}\otimes V_{\frac12} \cong V_0\oplus V_1\).
The invariant line is
\[
V^{SU(2)}
=
\operatorname{span}\{|\Psi^-\rangle\},
\qquad
|\Psi^-\rangle
=
\frac{|\uparrow\downarrow\rangle-|\downarrow\uparrow\rangle}{\sqrt2}.
\]
A generic torus element fixes the singlet and the zero-weight triplet state, so
\(r_1=2\).
Proposition~\ref{prop:SU2-reducible-completion} gives
\(\nu=2, \qquad r_2=1\).
An invariant-complete two-gluing realisation therefore has
\(T=|\Psi^-\rangle\langle\Psi^-|\).
The tensor product is used as an internal representation over a single spatial coordinate; no two-particle configuration-space dynamics is assumed.
\section{\texorpdfstring{\(SU(N)\)}{SU(N)} tensor invariants and transmission multiplicities}
\label{sec:SUN-tensors}

Let \(F:=\mathbb C^N\) denote the defining representation of \(SU(N)\), and let \(\overline F\)
denote its complex-conjugate representation. We consider the mixed tensor
spaces
\begin{equation}
\label{eq:Vpq}
V_{p,q}
:=
F^{\otimes p}\otimes\overline F^{\otimes q},
\qquad
p,q\geq0.
\end{equation}
The invariant sectors of these representations are standard objects of
classical invariant theory and Schur--Weyl duality
\citep{Weyl1997ClassicalGroups,FultonHarris1991,GoodmanWallach2009,
Howe1989ClassicalInvariantTheory}.
For invariant-complete gluing, Theorem~\ref{thm:unitary-Lk-main} identifies these invariant multiplicities with transmission ranks.

\subsection{The center obstruction}
\label{subsec:center-obstruction}

The center of \(SU(N)\) is
\[
Z(SU(N))
=
\left\{
\zeta I_N:\zeta^N=1
\right\}
\cong
\mathbb Z_N.
\]
On the defining representation,
\(\rho(\zeta I_N)=\zeta I_F\),
whereas on the conjugate representation,
\(\overline\rho(\zeta I_N) = \zeta^{-1}I_{\overline F}\).
The central element \(\zeta I_N\) acts on \(V_{p,q}\) by the scalar
\begin{equation}
\label{eq:center-action}
\zeta^{p-q}.
\end{equation}

\begin{proposition}[Center selection rule]
\label{prop:center-selection}
For
\[
V_{p,q}=F^{\otimes p}\otimes\overline F^{\otimes q},
\]
one has
\[
V_{p,q}^{SU(N)}\neq\{0\}
\quad\Longleftrightarrow\quad
p-q\equiv0\pmod N.
\]
\end{proposition}

\begin{proof}
If \(v\in V_{p,q}^{SU(N)}\), every central element \(\zeta I_N\) fixes \(v\), so \(\zeta^{p-q}v=v\) for every \(N\)-th root of unity \(\zeta\). A nonzero invariant therefore requires
\(p-q\equiv0\pmod N\).

Conversely, suppose \(p-q=\ell N\). Let
\(\delta=\sum_{i=1}^{N}e_i\otimes\overline e_i\in F\otimes\overline F\).
If \(\ell\geq0\), a permutation of the factors identifies a nonzero invariant tensor of the form \(\delta^{\otimes q}\otimes\varepsilon_N^{\otimes\ell}\) inside \(V_{p,q}\). If \(\ell<0\), use
\(\delta^{\otimes p}\otimes\overline{\varepsilon_N}^{\otimes(-\ell)}\).
Hence the congruence is sufficient.
\end{proof}

The congruence determines existence but not invariant multiplicity.

For invariant-complete gluing with relative matrices in \(\rho(SU(N))\),
Proposition~\ref{prop:center-selection} gives the following scattering
consequence.

\begin{corollary}[Complete reflection from the center obstruction]
\label{cor:center-reflection}
Suppose the relative gluings on \(\mathbb L_k\) are invariant-complete for the
representation \(V_{p,q}\). If
\[
p-q\not\equiv0\pmod N,
\]
then
\[
T=0,
\qquad
R=-I_{V_{p,q}}.
\]
\end{corollary}

The center obstruction therefore forces complete reflection in these tensor sectors.

\subsection{Balanced tensor powers}
\label{subsec:balanced-tensors}

For the balanced sectors
\(V_{m,m} = F^{\otimes m}\otimes\overline F^{\otimes m}\),
the central action is trivial.

The Hermitian structure on \(F\) identifies \(\overline F\cong F^\ast\) as \(SU(N)\)-representations. Hence there is a natural equivariant
identification
\begin{equation}
\label{eq:balanced-End}
F^{\otimes m}\otimes\overline F^{\otimes m}
\cong
\operatorname{End}(F^{\otimes m}),
\end{equation}
under which the \(SU(N)\) action becomes conjugation. Therefore
\begin{equation}
\label{eq:balanced-commutant}
\left(
F^{\otimes m}\otimes\overline F^{\otimes m}
\right)^{SU(N)}
\cong
\operatorname{End}_{SU(N)}(F^{\otimes m}).
\end{equation}

In the balanced tensor representation the central
\(U(1)\subset U(N)\) acts trivially. Since every element of \(U(N)\) can
be written as a central phase times an element of \(SU(N)\), the
\(SU(N)\)- and \(U(N)\)-invariant subspaces of \(V_{m,m}\) coincide.
Schur--Weyl duality then identifies the commutant of the \(U(N)\)-action
on \(F^{\otimes m}\) with the image of the place-permutation action of the
symmetric group \(S_m\)
\citep{FultonHarris1991,GoodmanWallach2009}.

For
\(\sigma\in S_m\),
let \(P_\sigma: F^{\otimes m}\longrightarrow F^{\otimes m}\) be the permutation operator
\begin{equation}
\label{eq:permutation-operator}
P_\sigma
(v_1\otimes\cdots\otimes v_m)
=
v_{\sigma^{-1}(1)}
\otimes\cdots\otimes
v_{\sigma^{-1}(m)}.
\end{equation}

\begin{theorem}[Balanced invariant multiplicity]
\label{thm:balanced-invariant-multiplicity}
Let \(N\geq m\). Then
\begin{equation}
\label{eq:balanced-invariant-space}
\left(
F^{\otimes m}\otimes\overline F^{\otimes m}
\right)^{SU(N)}
\cong
\operatorname{span}
\left\{
P_\sigma:\sigma\in S_m
\right\},
\end{equation}
and the operators \(P_\sigma\) are linearly independent. Hence
\begin{equation}
\label{eq:factorial-multiplicity}
\dim
\left(
F^{\otimes m}\otimes\overline F^{\otimes m}
\right)^{SU(N)}
=
m!.
\end{equation}
\end{theorem}

\begin{proof}
By \eqref{eq:balanced-commutant}, the invariant space is the commutant of
the diagonal \(SU(N)\) action on \(F^{\otimes m}\). As observed above,
this commutant is the same as that for \(U(N)\). Schur--Weyl duality
identifies it with the image of the group algebra \(\mathbb C[S_m]\) under the permutation action \(\sigma\mapsto P_\sigma\).

It remains to show that the \(m!\) permutation operators are linearly
independent when \(N\geq m\). Choose linearly independent basis vectors
\(e_1,\ldots,e_m\in F\).
The vectors
\(P_\sigma (e_1\otimes\cdots\otimes e_m), \qquad \sigma\in S_m\),
are the \(m!\) distinct tensors obtained by permuting the \(m\) distinct
basis labels and are therefore linearly independent. If
\(\sum_{\sigma\in S_m}a_\sigma P_\sigma=0\),
application to \(e_1\otimes\cdots\otimes e_m\) forces every \(a_\sigma=0\). Hence the image of \(\mathbb C[S_m]\) has
dimension \(m!\).
\end{proof}

The equality with \(m!\) is the stable-range statement \(N\ge m\); for \(m>N\), permutation contractions satisfy linear relations.

\begin{corollary}[Factorial transmission multiplicity]
\label{cor:factorial-transmission}
Let \(N\geq m\), and suppose the relative gluings on \(\mathbb L_k\) are
invariant-complete for
\[
V_{m,m}
=
F^{\otimes m}\otimes\overline F^{\otimes m}.
\]
Then the transmission amplitude is the invariant projector
\[
T=P_{\mathrm{inv}},
\]
and
\begin{equation}
\label{eq:factorial-transmission-rank}
\operatorname{rank}T=m!.
\end{equation}
The complementary reflected sector has dimension
\[
N^{2m}-m!.
\]
\end{corollary}

Theorem~\ref{thm:unitary-Lk-main} realises this standard invariant multiplicity as a transmission rank.

\subsection{Permutation-contraction tensors}
\label{subsec:permutation-contractions}

The invariant space may equivalently be written in tensor form. For
\(\sigma\in S_m\),
define the contraction tensor \(D_\sigma\) by
\begin{equation}
\label{eq:D-sigma-components}
(D_\sigma)_{i_1\cdots i_m}^{\phantom{i_1\cdots i_m}j_1\cdots j_m}
=
\prod_{a=1}^{m}
\delta_{i_a}^{j_{\sigma(a)}}.
\end{equation}
Under the identification \eqref{eq:balanced-End}, \(D_\sigma\) corresponds
to the permutation operator \(P_\sigma\). For \(N\geq m\),
\[
\left\{
D_\sigma:\sigma\in S_m
\right\}
\]
is a basis of the balanced invariant sector.

Their Hermitian Gram matrix is determined by permutation cycle structure.

\begin{proposition}[Gram matrix of contraction invariants]
\label{prop:contraction-Gram}
For \(\sigma,\tau\in S_m\),
\begin{equation}
\label{eq:contraction-Gram}
\langle D_\sigma,D_\tau\rangle
=
N^{c(\sigma^{-1}\tau)},
\end{equation}
where \(c(\pi)\) denotes the number of cycles of the permutation
\(\pi\in S_m\), including one-cycles.
\end{proposition}

\begin{proof}
Under \eqref{eq:balanced-End}, the tensor inner product becomes the
Hilbert--Schmidt inner product. Hence
\[
\langle D_\sigma,D_\tau\rangle
=
\operatorname{Tr}(P_\sigma^\ast P_\tau)
=
\operatorname{Tr}(P_{\sigma^{-1}\tau}).
\]
A basis tensor \(e_{i_1}\otimes\cdots\otimes e_{i_m}\) is fixed by \(P_\pi\) precisely when its labels are constant along every
cycle of \(\pi\). There are \(N\) independent choices for the common label
on each cycle, so
\(\operatorname{Tr}(P_\pi) = N^{c(\pi)}\).
Taking \(\pi=\sigma^{-1}\tau\) gives \eqref{eq:contraction-Gram}.
\end{proof}

\subsection{The quadratic balanced sector}
\label{subsec:quadratic-balanced}

For \(m=2\), there are exactly two permutation invariants:
\(D_e, \qquad D_s\),
where \(s=(12)\). In components,
\[
(D_e)_{ij}^{kl}
=
\delta_i^k\delta_j^l,
\qquad
(D_s)_{ij}^{kl}
=
\delta_i^l\delta_j^k.
\]
Their Gram matrix is
\begin{equation}
\label{eq:m2-Gram}
\begin{pmatrix}
\langle D_e,D_e\rangle
&
\langle D_e,D_s\rangle
\\
\langle D_s,D_e\rangle
&
\langle D_s,D_s\rangle
\end{pmatrix}
=
\begin{pmatrix}
N^2&N\\
N&N^2
\end{pmatrix}.
\end{equation}

The symmetric and antisymmetric combinations \(D_+:=D_e+D_s, \qquad D_-:=D_e-D_s\) are orthogonal and satisfy
\begin{equation}
\label{eq:Dpm-norms}
\lVert D_+\rVert^2
=
2N(N+1),
\qquad
\lVert D_-\rVert^2
=
2N(N-1).
\end{equation}
Therefore
\begin{equation}
\label{eq:Dpm-normalized}
\widehat D_+
=
\frac{D_+}{\sqrt{2N(N+1)}},
\qquad
\widehat D_-
=
\frac{D_-}{\sqrt{2N(N-1)}}
\end{equation}
form an orthonormal basis of \(\left( F^{\otimes2}\otimes\overline F^{\otimes2} \right)^{SU(N)}\) for \(N\geq2\).

The corresponding invariant projector is therefore
\begin{equation}
\label{eq:m2-invariant-projector}
P_{\mathrm{inv}}^{(2)}
=
|\widehat D_+\rangle\langle\widehat D_+|
+
|\widehat D_-\rangle\langle\widehat D_-|.
\end{equation}

For \(N=3\), the norms become
\(\lVert D_+\rVert^2=24, \qquad \lVert D_-\rVert^2=12\).
For \(N=3\), these normalized directions span the quartic invariant sector used below.

\subsection{The determinant invariant}
\label{subsec:determinant-invariant}

Balanced tensors are not the only sectors allowed by the center rule. The
first purely fundamental tensor power for which the center obstruction
disappears is
\(F^{\otimes N}\).
The alternating tensor
\begin{equation}
\label{eq:epsilon-general}
\varepsilon_N
=
\sum_{\sigma\in S_N}
\operatorname{sgn}(\sigma)\,
e_{\sigma(1)}\otimes\cdots\otimes e_{\sigma(N)}
\end{equation}
is \(SU(N)\)-invariant because
\(g^{\otimes N}\varepsilon_N = \det(g)\,\varepsilon_N = \varepsilon_N\).

Standard tensor representation theory gives
\begin{equation}
\label{eq:determinant-invariant-space}
\left(F^{\otimes N}\right)^{SU(N)}
=
\mathbb C\varepsilon_N
\end{equation}
\citep{FultonHarris1991,GoodmanWallach2009}. Hence an invariant-complete
gluing on this tensor sector has
\(\operatorname{rank}T=1\).

For \(N=3\),
\((F^{\otimes3})^{SU(3)}=\mathbb C\varepsilon_3\).
Together with
\[
\dim(F\otimes\overline F)^{SU(3)}=1,
\qquad
\dim\left(F^{\otimes2}\otimes\overline F^{\otimes2}\right)^{SU(3)}=2,
\]
this fixes the terminal invariant dimensions used in Section~\ref{sec:SU3-fixed-spaces}.

\section{Fixed-space ranks in low \texorpdfstring{\(SU(3)\)}{SU(3)} tensor representations}
\label{sec:SU3-fixed-spaces}

Let
\(G=SU(3), \qquad F=\mathbb C^3\),
and let
\[
t=\operatorname{diag}(\lambda_1,\lambda_2,\lambda_3)\in SU(3),
\qquad
\lambda_1\lambda_2\lambda_3=1.
\]
For the finite set of tensor weights occurring below, call \(t\) \emph{generic} if its eigenvalues are distinct and no nontrivial occurring weight evaluates to \(1\). The excluded set is a finite union of proper closed character-resonance loci in the maximal torus.

\subsection{The defining representation}

For \(F\), a generic \(t\) has no eigenvalue \(1\), so
\(F^t=\{0\}=F^{SU(3)}\).

\begin{proposition}[Defining representation]
\label{prop:SU3-fundamental-completion}
For the defining representation,
\[
\nu(F)=1,
\qquad
r_1(F)=0.
\]
\end{proposition}

\begin{proof}
Choose \(t\in SU(3)\) with no eigenvalue \(1\). Then \(F^t=F^{SU(3)}=\{0\}\).
\end{proof}

\subsection{The quadratic mixed tensor}

For \(F\otimes\overline F\), the basis tensor \(e_i\otimes\overline e_j\) has weight \(\lambda_i\lambda_j^{-1}\). Genericity gives
\[
(F\otimes\overline F)^t
=
\operatorname{span}
\{e_1\otimes\overline e_1,
e_2\otimes\overline e_2,
e_3\otimes\overline e_3\},
\]
so
\(\dim(F\otimes\overline F)^t=3\).
Under
\(F\otimes\overline F\cong\operatorname{End}(F)\),
the \(SU(3)\)-invariant subspace is \(\mathbb C I_F\).

\begin{proposition}[Quadratic fixed-space ranks]
\label{prop:SU3-quadratic-nu}
For the conjugation representation,
\[
r_1(F\otimes\overline F)=3,
\qquad
\nu(F\otimes\overline F)=2,
\qquad
r_2(F\otimes\overline F)=1.
\]
\end{proposition}

\begin{proof}
For \(g\in SU(3)\), the fixed space in \(\operatorname{End}(F)\) is the commutant of \(g\). If the eigenvalue multiplicities of \(g\) are \(m_1,\ldots,m_s\), its commutant has dimension
\(\sum_{\alpha=1}^{s}m_\alpha^2\geq3\).
A generic regular element realises equality, hence \(r_1=3\). Since
\(\dim(F\otimes\overline F)^{SU(3)}=1\),
one element cannot be invariant-complete. Corollary~\ref{cor:semisimple-two-completion} gives \(\nu=2\), and therefore \(r_2=1\).
\end{proof}

\subsection{The cubic fundamental tensor}

A basis tensor \(e_i\otimes e_j\otimes e_k\) in \(F^{\otimes3}\) has weight
\(\lambda_i\lambda_j\lambda_k\).
For generic \(t\), this weight is \(1\) exactly when \((i,j,k)\) is a permutation of \((1,2,3)\). Hence
\(\dim(F^{\otimes3})^t=6\).
The full invariant sector is
\[
(F^{\otimes3})^{SU(3)}
=
\mathbb C\varepsilon,
\qquad
\varepsilon
=
\sum_{\sigma\in S_3}
\operatorname{sgn}(\sigma)\,
e_{\sigma(1)}\otimes e_{\sigma(2)}\otimes e_{\sigma(3)}.
\]

\begin{lemma}[Uniform cubic lower bound]
\label{lem:cubic-single-lower-bound}
For every \(g\in SU(3)\),
\[
\dim(F^{\otimes3})^g\geq6.
\]
\end{lemma}

\begin{proof}
Let \(v_1,v_2,v_3\) be an eigenbasis of \(g\), with eigenvalues \(\mu_1,\mu_2,\mu_3\). Since \(\mu_1\mu_2\mu_3=1\), each of the six tensors
\(v_{\sigma(1)}\otimes v_{\sigma(2)}\otimes v_{\sigma(3)}, \qquad \sigma\in S_3\),
is fixed by \(g^{\otimes3}\). They are linearly independent.
\end{proof}

\begin{proposition}[Cubic fixed-space ranks]
\label{prop:SU3-cubic-nu}
For \(F^{\otimes3}\),
\[
r_1(F^{\otimes3})=6,
\qquad
\nu(F^{\otimes3})=2,
\qquad
r_2(F^{\otimes3})=1.
\]
\end{proposition}

\begin{proof}
The lemma gives \(r_1\geq6\), and a generic torus element realises equality. Since the invariant sector is one-dimensional, \(\nu\geq2\). Corollary~\ref{cor:semisimple-two-completion} gives \(\nu=2\), hence \(r_2=1\).
\end{proof}

\subsection{The quartic balanced tensor}

Set
\(V_{\mathrm{quart}} = F^{\otimes2}\otimes\overline F^{\otimes2}\).
The basis tensor \(e_i\otimes e_j\otimes\overline e_k\otimes\overline e_l\) has weight
\(\lambda_i\lambda_j\lambda_k^{-1}\lambda_l^{-1}\).
For generic \(t\), this weight is \(1\) exactly when the multisets \(\{i,j\}\) and \(\{k,l\}\) coincide. Three diagonal multisets contribute one fixed basis tensor each, and the three unordered pairs with distinct indices contribute four each. Therefore
\(\dim V_{\mathrm{quart}}^t = 3+3\cdot4 = 15\).
By Theorem~\ref{thm:balanced-invariant-multiplicity},
\(\dim V_{\mathrm{quart}}^{SU(3)}=2\).

\begin{lemma}[Uniform quartic lower bound]
\label{lem:quartic-lower-bound}
For every \(g\in SU(3)\),
\[
\dim V_{\mathrm{quart}}^g\geq15.
\]
A generic regular element realises equality.
\end{lemma}

\begin{proof}
Using
\(V_{\mathrm{quart}} \cong \operatorname{End}(F^{\otimes2})\),
the fixed space is the commutant of \(g^{\otimes2}\). For an eigenbasis \(v_1,v_2,v_3\) of \(g\), the generic eigenspace multiplicities of \(g^{\otimes2}\) are
\(1,1,1,2,2,2\):
the diagonal tensors \(v_i\otimes v_i\) are one-dimensional, while \(v_i\otimes v_j\) and \(v_j\otimes v_i\) share an eigenvalue for \(i\neq j\). The corresponding commutant dimension is
\(3(1^2)+3(2^2)=15\).
Further eigenvalue coincidences merge eigenspaces and increase the sum of squared multiplicities.
\end{proof}

\begin{proposition}[Quartic fixed-space ranks]
\label{prop:SU3-quartic-nu}
For \(V_{\mathrm{quart}}\),
\[
r_1(V_{\mathrm{quart}})=15,
\qquad
\nu(V_{\mathrm{quart}})=2,
\qquad
r_2(V_{\mathrm{quart}})=2.
\]
\end{proposition}

\begin{proof}
Lemma~\ref{lem:quartic-lower-bound} and the generic calculation give \(r_1=15\). Since the invariant sector has dimension \(2\), one element cannot be invariant-complete. Corollary~\ref{cor:semisimple-two-completion} gives \(\nu=2\), and hence \(r_2=2\).
\end{proof}

\subsection{Summary}

\begin{theorem}[Optimal ranks for low-degree \texorpdfstring{\(SU(3)\)}{SU(3)} tensors]
\label{thm:SU3-rank-collapse}
For the four representations considered above,
\[
\begin{array}{c|c|c|c|c}
V
&
r_0(V)
&
r_1(V)
&
r_2(V)
&
\nu(V)
\\
\hline
F
&
3&0&0&1
\\[1mm]
F\otimes\overline F
&
9&3&1&2
\\[1mm]
F^{\otimes3}
&
27&6&1&2
\\[1mm]
F^{\otimes2}\otimes\overline F^{\otimes2}
&
81&15&2&2
\end{array}
\]
A generic torus element realises each displayed value of \(r_1\). The values \(r_2\) for the last three rows are realised by invariant-complete two-element tuples.
\end{theorem}

The displayed \(r_0,r_1,r_2\) are independently optimised tuple-length minima, not necessarily successive prefixes of one gluing tuple.

\section{The quartic \texorpdfstring{\(SU(3)\)}{SU(3)} invariant projector}
\label{sec:SU3-quartic-projector}

Let
\[
V_{\mathrm{quart}}
=
F^{\otimes2}\otimes\overline F^{\otimes2},
\qquad
F=\mathbb C^3.
\]
Theorem~\ref{thm:balanced-invariant-multiplicity} gives
\(\dim V_{\mathrm{quart}}^{SU(3)}=2\).
Equip \(F\) with its standard Hermitian inner product and \(V_{\mathrm{quart}}\) with the induced tensor-product inner product.

\subsection{Invariant basis}

Define
\[
(D_1)_{ij}^{\phantom{ij}kl}
=
\delta_i^k\delta_j^l,
\qquad
(D_2)_{ij}^{\phantom{ij}kl}
=
\delta_i^l\delta_j^k.
\]
These are the two permutation-contraction invariants corresponding to the identity and transposition in \(S_2\). Their inner products are
\[
\langle D_1,D_1\rangle
=
\langle D_2,D_2\rangle
=
9,
\qquad
\langle D_1,D_2\rangle=3.
\]
Set
\(S=D_1+D_2, \qquad A=D_1-D_2\).
Then
\(\langle S,A\rangle=0, \qquad \|S\|^2=24, \qquad \|A\|^2=12\),
so
\[
\widehat S=\frac{D_1+D_2}{\sqrt{24}},
\qquad
\widehat A=\frac{D_1-D_2}{\sqrt{12}}
\]
form an orthonormal basis of \(V_{\mathrm{quart}}^{SU(3)}\).

\begin{proposition}[Quartic invariant basis]
\label{prop:quartic-invariant-basis}
\[
V_{\mathrm{quart}}^{SU(3)}
=
\operatorname{span}\{\widehat S,\widehat A\}.
\]
\end{proposition}

\begin{proof}
The tensors \(D_1,D_2\) are linearly independent \(SU(3)\)-invariants, and the invariant space has dimension \(2\).
\end{proof}

Under
\(F\otimes F = \operatorname{Sym}^2F\oplus\Lambda^2F\),
the vector \(\widehat S\) lies in the invariant line of
\(\operatorname{Sym}^2F\otimes\overline{\operatorname{Sym}^2F}\),
and \(\widehat A\) lies in the invariant line of
\(\Lambda^2F\otimes\overline{\Lambda^2F}\).

\subsection{Invariant projector}

The orthogonal projector onto the invariant sector is
\begin{equation}
\label{eq:quartic-projector}
P_{\mathrm{quart}}
=
|\widehat S\rangle\langle\widehat S|
+
|\widehat A\rangle\langle\widehat A|.
\end{equation}
Equivalently,
\[
P_{\mathrm{quart}}
=
\frac{1}{24}
|D_1+D_2\rangle\langle D_1+D_2|
+
\frac{1}{12}
|D_1-D_2\rangle\langle D_1-D_2|.
\]
In the nonorthogonal basis \((D_1,D_2)\),
\[
G
=
\begin{pmatrix}
9&3\\
3&9
\end{pmatrix},
\qquad
G^{-1}
=
\frac1{72}
\begin{pmatrix}
9&-3\\
-3&9
\end{pmatrix},
\]
and therefore
\(P_{\mathrm{quart}} = \sum_{a,b=1}^{2} |D_a\rangle(G^{-1})_{ab}\langle D_b|\).

\subsection{Invariant-complete scattering}

Choose \(g_1,g_2\in SU(3)\) with
\(V^{g_1}\cap V^{g_2} = V_{\mathrm{quart}}^{SU(3)}\);
a dense generating pair is sufficient
\citep{BreuillardGelander2003,BreuillardGelander2026Revisited}. Set
\(W_2=\rho(g_1), \qquad W_3=\rho(g_2)\).
Then
\(V_{\mathbf W} = \operatorname{span}\{\widehat S,\widehat A\}\),
and
\[
H_{\mathrm{quart}}
\cong
H_{\mathrm{free}}\otimes
I_{\operatorname{span}\{\widehat S,\widehat A\}}
\oplus
(H_{D,-}\oplus H_{D,+})
\otimes
I_{\operatorname{span}\{\widehat S,\widehat A\}^{\perp}}.
\]

\begin{theorem}[Quartic invariant-sector scattering]
\label{thm:quartic-scattering}
For invariant-complete gluing with relative matrices in the quartic representation of \(SU(3)\),
\[
T_{\mathrm{quart}}
=
|\widehat S\rangle\langle\widehat S|
+
|\widehat A\rangle\langle\widehat A|,
\qquad
R_{\mathrm{quart}}
=
-\left(I-T_{\mathrm{quart}}\right),
\]
and
\[
\operatorname{rank}T_{\mathrm{quart}}=2.
\]
\end{theorem}

\begin{proof}
The common fixed space is \(V_{\mathrm{quart}}^{SU(3)}\); Theorem~\ref{thm:unitary-Lk-main} identifies its orthogonal projector with the transmission amplitude.
\end{proof}

For
\[
\xi=\alpha\widehat S+\beta\widehat A+\xi_\perp,
\qquad
\xi_\perp\perp V_{\mathrm{quart}}^{SU(3)},
\qquad
\|\xi\|=1,
\]
the transmission and reflection probabilities are
\(\mathcal T(\xi)=|\alpha|^2+|\beta|^2, \qquad \mathcal R(\xi)=\|\xi_\perp\|^2\).

\subsection{One- and two-gluing optima}

Proposition~\ref{prop:SU3-quartic-nu} gives
\(r_1(V_{\mathrm{quart}})=15, \qquad r_2(V_{\mathrm{quart}})=2\).
A generic torus element realises the one-gluing minimum, while an invariant-complete pair realises the two-gluing minimum. These are independently optimised values; no assertion is made that a fixed generic first element necessarily extends to a pair realising \(r_2\).

The quartic representation is the first \(SU(3)\) tensor example considered here with invariant multiplicity greater than one.

\subsection{Invariant couplings}

The bases \((D_1,D_2)\) and \((\mathcal S,\mathcal A)\) span the same invariant space. Since the flat transmission operator is the identity on that space, distinguishing internal coupling bases requires an additional operator acting nontrivially on \(V_{\mathrm{quart}}^{SU(3)}\).

\subsection{Relation to tetraquark color space}

In particle-physics notation,
\[
V_{\mathrm{quart}}
=
\mathbf3\otimes\mathbf3\otimes\overline{\mathbf3}\otimes\overline{\mathbf3}.
\]
The representation is the color space for two fundamental and two antifundamental color factors, and its two-dimensional invariant subspace is the standard two-singlet sector used in tetraquark color decompositions
\citep{Slansky1981,ParkLee2014}. The present model contains a single spatial coordinate and no Yang--Mills field, gluon exchange, confinement, spin, flavor, or four-body configuration-space dynamics. The tetraquark connection is therefore confined to the internal \(SU(3)\) representation and its invariant projector.
\section{Optimal ranks and origin thresholds}
\label{sec:rank-profiles}

Recall that \(r_m(\rho)\) is the least transmission rank attainable with \(m\) relative gluing matrices, optimised independently for each \(m\).

For the examples above,
\[
\begin{array}{c|c|c|c|c}
G
&
V
&
r_0
&
r_1
&
r_2
\\
\hline
SU(2)
&
V_j,\ j\in\frac12+\mathbb Z_{\geq0}
&
2j+1&0&0
\\[1mm]
SU(2)
&
V_j,\ j\in\mathbb Z_{>0}
&
2j+1&1&0
\\[1mm]
SU(2)
&
V_{\frac12}\otimes V_{\frac12}
&
4&2&1
\\[1mm]
SU(3)
&
F
&
3&0&0
\\[1mm]
SU(3)
&
F\otimes\overline F
&
9&3&1
\\[1mm]
SU(3)
&
F^{\otimes3}
&
27&6&1
\\[1mm]
SU(3)
&
F^{\otimes2}\otimes\overline F^{\otimes2}
&
81&15&2
\end{array}
\]
For the \(SU(3)\) rows, generic torus elements realise \(r_1\). For the rows with \(\nu=2\), invariant-complete pairs realise \(r_2=\dim V^G\).

Define the optimal-rank decrement
\(\Delta_m(\rho) = r_{m-1}(\rho)-r_m(\rho)\).
Then \(\Delta_m(\rho)\geq0\) and
\(\sum_{m=1}^{\nu(\rho)}\Delta_m(\rho) = \dim V-\dim V^G\).
Because the minima defining \(r_{m-1}\) and \(r_m\) may be attained by different tuples, \(\Delta_m\) is a difference between optimal envelopes, not necessarily the rank removed by appending one gluing to a fixed tuple.

For the three nontrivial \(SU(3)\) tensor examples with \(\nu=2\),
\[
\begin{array}{c|c|c}
V & \Delta_1 & \Delta_2\\
\hline
F\otimes\overline F & 6 & 2\\
F^{\otimes3} & 21 & 5\\
F^{\otimes2}\otimes\overline F^{\otimes2} & 66 & 13
\end{array}
\]

The origin threshold for invariant completion is
\(k_{\mathrm{inv}}(\rho) = \max\{2,\nu(\rho)+1\}\).
The intrinsic matrix-group threshold for dense generation of the gluing
image is
\(k_{\mathrm{grp}}(\rho) = \max\{2,d_{\mathrm{top}}(\rho(G))+1\}\).
Proposition~\ref{prop:completion-generation} gives
\(\nu(\rho) \leq d_{\mathrm{top}}(\rho(G)) \leq d_{\mathrm{top}}(G)\),
so
\(k_{\mathrm{inv}}(\rho) \leq k_{\mathrm{grp}}(\rho)\).
The ambient quantity \(d_{\mathrm{top}}(G)+1\) remains a sufficient bound
for dense generation of the image but need not be intrinsic when
\(\rho\) has nontrivial kernel.

For a nontrivial connected compact semisimple group \(G\),
\(d_{\mathrm{top}}(G)=2\),
hence every finite-dimensional continuous unitary representation satisfies
\(d_{\mathrm{top}}(\rho(G))\leq2\).
Accordingly, \(\mathbb L_3\) suffices for dense generation of the matrix image and
for invariant completion. Individual representations can satisfy
\(\nu(\rho)=1\), in which case invariant completion is already realizable
on \(\Ltwo\).

For a fixed nested gluing tuple, the ranks \(s_j(g_1,\ldots,g_j) = \dim\bigcap_{a=1}^{j}\ker(\rho(g_a)-I)\) record the actual origin-by-origin filtration. The optimised quantities \(r_j(\rho)\) record the best rank attainable at each tuple length.
\section{Representation-theoretic input and model dependence}
\label{sec:unitary-scope}

The representation-theoretic inputs above--Haar averaging, infinitesimal invariants, the \(SU(N)\) center rule, Schur--Weyl duality, and low-degree singlet multiplicities--are standard \citep{Hall2015,BrockerTomDieck1985,Weyl1997ClassicalGroups,FultonHarris1991,GoodmanWallach2009,Howe1989ClassicalInvariantTheory}; common fixed-space questions are likewise classical \citep{BernikEtAl2005CommonFixed}.

The finite-origin model realises the chain from a finite gluing tuple to its common fixed space, then by \(H^1\)-closure to the transmitting sector of the flat Hamiltonian. The \(k-1\) relative transition slots bound how far this chain can proceed on \(\mathbb L_k\).

Once the free--Dirichlet boundary decomposition has been obtained, its two-lead scattering matrix is standard point-interaction and quantum-graph scattering theory
\citep{KostrykinSchrader99,BerkolaikoKuchment2013,AlbeverioEtAl05}.
The non-Hausdorff step is the derivation of that decomposition from bundle compatibility and Sobolev trace closure.

The law \(T=P_{V_{\mathbf W}},\ R=-(I-P_{V_{\mathbf W}})\) is specific to the collapse measure, locally constant unitary gluing, and flat Dirichlet form; other connections, potentials, or defect interactions can retain further representation data.

For
\(F^{\otimes2}\otimes\overline F^{\otimes2}, \qquad F=\C^3\),
the invariant sector is the standard two-dimensional singlet color space used in tetraquark decompositions
\citep{Slansky1981,ParkLee2014}. No claim is made about four-body dynamics or QCD phenomenology.

Extensions that retain holonomy, resolve invariant multiplicity spaces, or introduce genuine multiparticle spatial degrees of freedom require additional geometric or operator-theoretic structure.

Part~\ref{part:detectability} compares the information retained by the propagation graph, measure, bundle domain, and operator realisation.

\clearpage
\part{Reduction and quantum detectability}
\label{part:detectability}

\section{Weighted Hausdorff reduction: topology and measure}
\label{sec:topology-versus-channels}
\label{sec:weighted-hausdorff-reduction}
\label{sec:measure-dependence}

Exceptional-fibre multiplicity is not itself a scalar propagation invariant: the boundaryless first-order reduction retains the metric Hausdorff graph and projective edge weights, while boundary variants also retain the closed scalar trace domain.

\subsection{Propagation graph and distinct multiplicities}
\label{subsec:propagation-graph}
\label{subsec:three-multiplicities}
\label{subsec:local-incidence-partitions}

Let \(q:X\longrightarrow\Gamma\) be a finite graph resolution. For a vertex \(v\) over which \(q\) has a
nontrivial fibre, define \(\kappa(v):=|q^{-1}(v)|\) and let \(d(v)\) be the number of incident graph germs. For a noncompact graph write \(N_{\mathrm{ext}}\) for the number of external half-lines, and for a connected finite graph set
\(\beta_1(\Gamma)=|E|-|V|+1\).

Exceptional multiplicity and graph valence are distinct. If the local quotient is one graph vertex, \(\mu(v)=|\mathcal S_v|\) whereas \(d(v)=|\mathcal P_-(v)|+|\mathcal P_+(v)|\); channel multiplicity is controlled by punctured-side partitions. \(\mathbb L_k\) has \((\mu,d)=(k,2)\), \(\mathbb B_k\) has \((k,k+1)\), and a two-lead \(k\)-arm theta graph has \(N_{\rm ext}=2\) and \(\beta_1=k-1\).

\begin{proposition}[Exceptional multiplicity does not determine channels]
\label{prop:multiple-origin-channel-invariance}
Within the multiple-origin lines, intervals, and circles considered in
Part~\ref{part:quantum-dynamics}, increasing the multiplicity of a fixed exceptional fibre does not
increase the local valence of the Hausdorff propagation graph. In the
branching resolutions of Part~\ref{part:quantum-dynamics}, the resolving incidence can change that
valence.
\end{proposition}

\subsection{Weighted reduction theorem}
\label{subsec:weighted-Hausdorff-reduction}
\label{subsec:topology-through-measure}

Let \(\mathcal A_X\) be the descended scalar smooth core and write
\(\mathcal V_X := \overline{\mathcal A_X}^{\,H^1_\oplus(\Gamma)}\).
Assume that \(\mathcal V_X\) is a closed trace domain containing
\(C_c^\infty(\Gamma\setminus V(\Gamma))\). For the boundaryless finite smooth
graph resolutions of Section~\ref{sec:graph-resolutions},
Corollary~\ref{cor:universal-first-order-graph-closure} gives
\(\mathcal V_X=H^1_{\mathrm{cont}}(\Gamma)\).
Boundary models may impose additional ordinary endpoint trace conditions.

Let the marked analytic presentation of
Definition~\ref{def:marked-analytic-graph-resolution} carry measure
\(\widetilde\mu\), and assume that \(\mu_\Gamma=(\pi\circ q)_*\widetilde\mu\) has constant positive edge densities \(a_e>0\). Set \(\mathcal H_{\Gamma,\mathbf a} = \bigoplus_e L^2(e,a_e\,dx)\) and define
\[
  \mathfrak q_{\Gamma,\mathbf a,\mathcal V_X}[u]
  =
  \sum_e a_e\int_e|u_e'|^2\,dx,
  \qquad
  \mathcal D(\mathfrak q_{\Gamma,\mathbf a,\mathcal V_X})=\mathcal V_X.
\]

\begin{theorem}[Weighted Hausdorff reduction]
\label{thm:weighted-Hausdorff-reduction}
Under these hypotheses, the resolution-induced scalar kinetic form and its
associated Friedrichs Hamiltonian are completely determined by the reduced
datum
\[
  (\Gamma,\mathbf a,\mathcal V_X).
\]
If \(\mathcal V_X=H^1_{\mathrm{cont}}(\Gamma)\), the trace-domain entry is
universal and the reduced datum simplifies to the weighted metric graph
\((\Gamma,\mathbf a)\); the associated operator is then the weighted
Kirchhoff Laplacian of Theorem~\ref{thm:weighted-kirchhoff}.
\end{theorem}

\begin{proof}
The completed scalar core fixes the closed form domain \(\mathcal V_X\). The
marked resolving-presentation measure fixes the Hilbert norm and derivative
form through its pushforward \(\mu_\Gamma\). The representation theorem for
closed nonnegative forms therefore determines the associated self-adjoint
operator from \((\Gamma,\mathbf a,\mathcal V_X)\). When
\(\mathcal V_X=H^1_{\mathrm{cont}}(\Gamma)\), integration by parts gives
exactly the weighted Kirchhoff realisation of
Theorem~\ref{thm:weighted-kirchhoff}. No higher scalar smooth data enter this
second-order operator once the reduced datum has been fixed.
\end{proof}

For the boundary variant of Section~\ref{sec:multiple-origin-intervals}, the
ordinary endpoint conditions are encoded in \(\mathcal V_X\); they are not
recoverable from \((\Gamma,[\mathbf a])\) alone.

A common positive rescaling of the weights is immaterial.

\begin{corollary}[Projective weight invariance]
\label{cor:weighted-common-rescaling}
Fix the metric graph and closed trace domain \(\mathcal V_X\). If
\[
  \widetilde a_e=c\,a_e
  \qquad
  (c>0)
\]
on every edge, then the two resulting scalar Hamiltonians are unitarily
equivalent. Hence, for fixed \(\mathcal V_X\), the scalar theory depends on
the projective weight class
\[
  [\mathbf a].
\]
\end{corollary}

\begin{proof}
Multiplication by \(\sqrt c\) gives a unitary map from
\(L^2(\Gamma,c\mathbf a\,dx)\) to
\(L^2(\Gamma,\mathbf a\,dx)\), and the form is carried to the original one.
\end{proof}

For branching, topology fixes the arms while the projective weights fix the bright-mode probabilities \(P_{j\leftarrow0}=a_j/a_0\).

\subsection{Scalar equivalence and information loss}
\label{subsec:scalar-equivalent-resolutions}
\label{subsec:what-survives-scalar}

\begin{corollary}[Scalar \(H^1\)-equivalence]
\label{cor:scalar-H1-equivalence}
Let
\[
  X_1\longrightarrow\Gamma_1,
  \qquad
  X_2\longrightarrow\Gamma_2
\]
satisfy the hypotheses of
Theorem~\ref{thm:weighted-Hausdorff-reduction}, with closed trace domains
\(\mathcal V_1\) and \(\mathcal V_2\). If a metric-graph isomorphism
\[
  F:\Gamma_1\longrightarrow\Gamma_2
\]
preserves the projective edge weights and carries \(\mathcal V_1\) onto
\(\mathcal V_2\), then the resulting scalar closed forms and Friedrichs
Hamiltonians are unitarily equivalent.
\end{corollary}

Information loss is explicit: \(\mathbb L_k\) forgets origin multiplicity, while an \(m\)-to-\(n\) junction reduces its incidence matrix \(C=(C_{ij})\) to row and column sums. Equal marginals therefore give the same scalar vertex scattering, although the reduced graph still retains valence, lengths, cycles, and projective weights.

Observable equivalence is coarser still. Root incidence on a resolving tree suppresses the internal hierarchy, whereas reverse incidence can access dark sectors; recombination restores sensitivity through \(J(p),K(p)\), and compactification transfers the same path information to spectral quantization. Even then, leading Weyl asymptotics retain only total metric length.

Recombination restores direct phase sensitivity to internal path data. For the generalized theta graph, the external scattering law
depends on
\(J(p)=\sum_j a_j\csc(p\ell_j), \qquad K(p)=\sum_j a_j\cot(p\ell_j)\).
Cycles can restore observable sensitivity to internal metric structure
that a forward experiment on an acyclic network loses.

\subsection{Limits of the scalar invariant}
\label{subsec:limits-propagation-graph}

The weighted graph alone is complete only for \(H^1_{\rm cont}(\Gamma)\). In general the reduced scalar datum is \((\Gamma,[\mathbf a],\mathcal V_X)\), with \(\mathcal V_X\) recording additional closed trace conditions; this determines the flat scalar form and Friedrichs operator but not bundle or higher-smooth data.

The same \((\Gamma,[\mathbf a])\) can therefore support different scalar or bundle trace domains, while the oriented first-order derivative retains further orientation and deficiency data.

\section{Bundle domains, holonomy, and smooth resolution}
\label{sec:bundle-smooth-resolution}
\label{sec:bundle-versus-Hamiltonian}
\label{sec:smooth-structure-limits}

The weighted scalar graph does not determine bundle domains or higher smooth compatibility: unitary transport changes value traces, whereas branch reparametrizations may enter only higher jets.

\subsection{Rank-one transport and nodal collapse}
\label{subsec:bundle-fixed-graph}
\label{subsec:bundle-local-equaliser}
\label{subsec:bundle-jets-versus-H1}

Let \(q:X\longrightarrow\Gamma\) be a finite graph resolution with projective edge weights
\([\mathbf a]\), and let \(E\longrightarrow X\) be a rank-one Hermitian bundle with locally constant unitary transition
functions. After choosing frames on the regular graph edges, sections are
represented by scalar edge functions. The bundle changes the admissible
endpoint relations, not the local differential expressions.

Consider a multiple-origin defect joining the same two regular edge germs,
with relative transports
\(z_1,\ldots,z_k\in U(1)\).
The endpoint equaliser is
\[
  \mathcal E_v^0
  =
  \left\{
    (\psi_-,\psi_+)\in\mathbb C^2:
    \psi_-=z_i\psi_+
    \text{ for every }i
  \right\}.
\]

\begin{proposition}[Rank-one endpoint dichotomy]
\label{prop:rank-one-endpoint-dichotomy}
If
\[
  z_1=\cdots=z_k=:z,
\]
then
\[
  \mathcal E_v^0
  =
  \{(z\xi,\xi):\xi\in\mathbb C\}
\]
is one-dimensional. If two transport factors differ, then
\[
  \mathcal E_v^0=\{0\}.
\]
\end{proposition}

\begin{proof}
The coherent case is immediate. If \(z_i\neq z_j\), then \(z_i\psi_+=\psi_-=z_j\psi_+\) forces \(\psi_+=0\), hence also \(\psi_-=0\).
\end{proof}

In the identity-glued models, the same fibre relation holds at every smooth derivative order. Incompatible transport therefore forces a flat smooth core but only a node after \(H^1\)-closure. This collapses the \(2^{k-1}\) labelled spin structures on \(\mathbb L_k\) to two \(H^1\)-domains, and with several split points only the active defect set survives.

\subsection{Cycle holonomy and analytic cutting}
\label{subsec:bundle-cycle-holonomy}
\label{subsec:bundle-active-erases-holonomy}

A cycle can retain coherent bundle information that no individual local
equaliser detects. If all local transports are coherent, local frame changes
remove them except for their product
\(Z=e^{i\theta}\in U(1)\).
On a circle of length \(L\), the resulting domain is
\(H^1_Z(0,L) = \{\psi\in H^1(0,L):\psi(L)=Z\psi(0)\}\),
and the Friedrichs operator satisfies the corresponding derivative relation
\(\psi'(L)=Z\psi'(0)\).
Its spectrum is
\(\lambda_n(Z) = \left( \frac{2\pi n+\theta}{L} \right)^2, \qquad n\in\mathbb Z\).

\begin{theorem}[Holonomy under nodal cutting]
\label{thm:holonomy-cut-principle}
For the rank-one circle resolutions of
Sections~\ref{sec:circle-k} and \ref{sec:multisplit-circle}, coherent
\(U(1)\) holonomy affects the resulting \(H^1\)-theory precisely while the
corresponding cycle remains uncut by an active nodal defect. Once such a
defect imposes zero trace, the global holonomy no longer affects the
Friedrichs Hamiltonian.
\end{theorem}

\begin{proof}
With no active defect, local frame changes reduce the coherent transports to
their product \(Z\), giving the quasiperiodic trace relation and the spectrum
displayed above. An active defect imposes zero trace and decomposes the form
domain into Dirichlet path components. Constant phase changes on those
components remove every remaining coherent local transport, so the closed
form and its Friedrichs operator are independent of \(Z\).
\end{proof}

\subsection{Bundle classes, Hamiltonians, and extension data}
\label{subsec:bundle-not-operator}
\label{subsec:bundle-first-order-different}
\label{subsec:extension-not-bundle}

\begin{proposition}[Hamiltonian equivalence criterion]
\label{prop:bundle-Hamiltonian-equivalence}
Let \(E_1,E_2\to X\) be rank-one bundle sectors over the same weighted
propagation graph. If their regular-locus Hilbert identifications agree and
\[
  \overline{\mathcal C_{E_1}}^{\,H^1}
  =
  \overline{\mathcal C_{E_2}}^{\,H^1},
\]
then their resulting kinetic forms and Friedrichs Hamiltonians coincide
under that identification.
\end{proposition}

\begin{proof}
The two forms have the same edge weights, differential expression, and
closed form domain. Their associated self-adjoint operators are therefore
identical.
\end{proof}

The same statement applies to the geometrically selected first-order
operator once an orientation is fixed. However, the first-order and
second-order theories ask different questions. The positive form has a
unique self-adjoint Friedrichs realisation; the edgewise derivative may be
self-adjoint, may admit a family of self-adjoint extensions, or may have
unequal deficiency indices and no self-adjoint extension.

\subsection{\texorpdfstring{Smooth resolution beyond \(H^1\)}{Smooth resolution beyond H1}}
\label{subsec:smooth-data-beyond-graph}
\label{subsec:smooth-transported-jets}
\label{subsec:smooth-H1-universality}
\label{subsec:why-second-order-H1}

The smooth resolution can contain information that is absent even from the
bundle-sensitive \(H^1\)-domain. Near a resolved vertex, two branch
coordinates may be related by a smooth germ
\(h(t) = \lambda t+c_2t^2+c_3t^3+\cdots, \qquad \lambda\neq0\).
Derivative transport depends on the jet of \(h\):
\((f\circ h)'(0)=h'(0)f'(0)\),
\((f\circ h)''(0) = h'(0)^2f''(0)+h''(0)f'(0)\).
Value matching does not determine higher smooth compatibility.

Let \(E_\pi^r\) denote the transported order-\(r\) jet equaliser at a resolved vertex.
The scalar closure theorem, in particular
Corollary~\ref{cor:universal-first-order-graph-closure}, gives under the
hypotheses used here,
\(\overline{\mathscr A_M}^{\,H^1_\oplus(\Gamma)} = H^1_{\mathrm{cont}}(\Gamma)\).

\begin{theorem}[Second-order smooth-resolution blindness]
\label{thm:second-order-smooth-blindness}
Let
\[
  \pi_i:M_i\longrightarrow\Gamma,
  \qquad
  i=1,2,
\]
be scalar graph resolutions satisfying first-order scalar universality and
inducing the same projective edge weights and edge metrics. Then their
resolution-induced scalar kinetic forms and weighted Kirchhoff Hamiltonians are
unitarily equivalent. Differences carried solely by higher transported
endpoint jets are invisible to this second-order Friedrichs theory.
\end{theorem}

\begin{proof}
First-order scalar universality identifies both closed form domains with
\(H^1_{\mathrm{cont}}(\Gamma)\). Equality of the projective weights gives,
after the common unitary normalization of
Corollary~\ref{cor:weighted-common-rescaling}, the same Hilbert norm and
kinetic form. The associated self-adjoint operators are therefore unitarily
equivalent. Higher endpoint-jet relations do not enter either closed form.
\end{proof}

Variationally, constraints lost in the closed \(H^1\)-form domain are not restored merely because the associated operator acts edgewise on \(H^2\).

\subsection{Higher Sobolev and formal visibility}
\label{subsec:sobolev-jet-visibility}
\label{subsec:tangent-identity-return}
\label{subsec:formal-prolongation}
\label{subsec:second-order-formal-blindness}

Higher Sobolev orders successively expose endpoint jets: in one dimension \(u^{(j)}(v)\) is continuous on \(W^{m,p}\) when \(j<m-1/p\). A tangent-to-identity return germ \(h(t)=t+a t^{q+1}+O(t^{q+2})\) acts trivially on values but can alter these higher traces, and finite compatibility may still fail formal prolongation through \(P_\pi^r=E_\pi^r/\operatorname{pr}_r(E_\pi^\infty)\).

Finite-order compatibility can itself fail to prolong to the full formal
system. If
\(I_\pi^r := \operatorname{pr}_r(E_\pi^\infty)\),
then \(P_\pi^r:=E_\pi^r/I_\pi^r\) measures the order-\(r\) compatibility directions that fail to prolong
formally.

Sections~\ref{sec:extension-failure}--\ref{sec:formal-holonomy} quantify these higher defects; detecting them requires probes beyond the ordinary first- and second-order models.

\section{Hierarchy of quantum detectability}
\label{sec:detectability-hierarchy}

The preceding results define a hierarchy of equivalence relations indexed by the chosen probe.

\subsection{Scalar and bundle reductions}
\label{subsec:detectability-data}
\label{subsec:detectability-scalar}
\label{subsec:detectability-measure}
\label{subsec:detectability-bundles}

At \(H^1\)-level the scalar Friedrichs theory factors through \((\Gamma,[\mathbf a],\mathcal V_X)\), reducing to \((\Gamma,[\mathbf a])\) for boundaryless finite smooth graph resolutions. Rank-one transport refines this only through the value equaliser (plus one coherent holonomy on an uncut cycle), while finite-rank gluing replaces it by \(V_{\mathbf W}=\bigcap_{a=2}^k\ker(W_a-I)\).

\begin{theorem}[Rank-one \(H^1\)-reduction]
\label{thm:rank-one-H1-bundle-reduction}
For the rank-one multiple-origin systems treated in
Part~\ref{part:quantum-dynamics}, the resulting \(H^1\)-domain is determined
by the active nodal defects. In the circle models of
Sections~\ref{sec:circle-k} and \ref{sec:multisplit-circle}, when no active
defect cuts the circle, one coherent rank-one holonomy remains as additional
trace data. Distinct local transport data that induce the same value
equalisers and, where applicable, the same surviving circle holonomy produce
the same resulting \(H^1\)-domain.
\end{theorem}

\begin{proof}
At each exceptional fibre,
Proposition~\ref{prop:rank-one-endpoint-dichotomy} leaves either the coherent
one-dimensional value equaliser or the zero trace. Applying the exact
first-order closure results for the line and circle models gives the active
nodal decomposition. On a circle without an active node, the local coherent
phases reduce to the product holonomy; with a node, that phase is removed by
Theorem~\ref{thm:holonomy-cut-principle}.
\end{proof}

\begin{theorem}[Finite-rank first-order bundle reduction]
\label{thm:finite-rank-H1-bundle-reduction}
For the finite-rank multiple-origin lines of
Theorem~\ref{thm:unitary-Lk-main}, the closed \(H^1\)-domain, flat
Friedrichs operator, and stationary scattering are determined by the common
fixed space \(V_{\mathbf W}\). In particular,
\[
  \mathcal H_{\mathbf W}^1
  =
  \{\psi\in H^1(\mathbb R;\mathbb C^r):\psi(0)\in V_{\mathbf W}\},
  \qquad
  T_{\mathbf W}=P_{V_{\mathbf W}}.
\]
The rank-one coherent/nodal dichotomy is the special case in which the
common fixed space has dimension \(1\) or \(0\).
\end{theorem}

\begin{proof}
This is the reduction established in
Theorem~\ref{thm:unitary-Lk-main}: the invariant component is free across the
defect, while the orthogonal complement has zero trace and splits into
Dirichlet half-line sectors. Both the form domain and the scattering operator
therefore depend on the tuple \(\mathbf W\) only through
\(V_{\mathbf W}\).
\end{proof}

\subsection{Operator-order dependence}
\label{subsec:detectability-operator-order}

The same \(H^1\)-geometry can support a unique Friedrichs Hamiltonian but different first-order extension behaviour. The explicit indices are summarised below; the compact and open theta families both have \((k-1,k-1)\) despite different cycle rank.

\subsection{Observable dependence}
\label{subsec:detectability-scattering}
\label{subsec:detectability-spectrum}

Even after an operator is fixed, a restricted measurement can discard
further information.

For root incidence on a resolving tree,
\(P_{\lambda\leftarrow\mathrm{root}} = \frac{c_\lambda}{\sum_\mu c_\mu}\),
so the internal branching hierarchy is absent from the forward
probabilities. If root-to-leaf lengths also agree, the corresponding
root-incidence amplitudes agree. Reverse or leaf-to-leaf experiments can
still distinguish the graphs by exciting dark sectors.

\subsection{Smooth information beyond conventional dynamics}
\label{subsec:detectability-smooth}

Ordinary first- and second-order theories use \(H^1\)-closed trace domains. Higher Sobolev closures and formal prolongation can retain derivative and return data erased at this level, so operator indistinguishability does not imply equivalence of the underlying smooth non-Hausdorff resolutions.

\subsection{Detectability synthesis}
\label{subsec:quantum-detectability-theorem}

\begin{principle}[Probe-dependent detectability]
\label{thm:quantum-detectability}
For the finite graph-resolved one-dimensional systems and finite-rank bundle
sectors analysed in this paper:

\begin{enumerate}[label=(\roman*),leftmargin=2.2em]

  \item the resolution-induced scalar \(H^1\)-Friedrichs theory factors
        through \((\Gamma,[\mathbf a],\mathcal V_X)\), where
        \(\mathcal V_X\) is the completed scalar trace domain; in the
        boundaryless first-order-universal case this reduces to
        \((\Gamma,[\mathbf a])\);

  \item exceptional-fibre multiplicity is scalar-visible only through
        changes in the reduced graph, metric, projective weights, or a
        nonuniversal closed trace domain;

  \item finite-rank bundle transport can remain visible through value-level
        fibre equalisers such as \(V_{\mathbf W}\); in rank one this becomes
        the coherent/nodal dichotomy, and coherent holonomy can additionally
        survive on uncut cycles;

  \item for fixed Hilbert measure and flat kinetic form, bundle sectors with
        the same closed \(H^1\)-domain have the same associated Friedrichs
        Hamiltonian;

  \item naturally oriented first-order operators retain self-adjointness
        and deficiency information that is not equivalent to the
        second-order scattering or spectral data;

  \item scattering and spectra retain only information present in the
        reduced operator, and restricted measurements can discard still
        more;

  \item higher transported-jet and formal data erased by \(H^1\)-completion
        require higher Sobolev or higher-order probes.

\end{enumerate}
\end{principle}

The principle is restricted to the explicit classes treated here and does not assert a classification of arbitrary non-Hausdorff quantum systems.

\subsection{Non-inference consequences}
\label{subsec:detectability-noninference}

Several converse inferences fail without additional hypotheses.

Equality of scalar Hamiltonians does not imply equality of
exceptional-fibre multiplicities or resolving incidence. Equality of
resolution-induced bundle Hamiltonians does not imply bundle isomorphism when
different bundle classes collapse to the same \(H^1\)-domain. Equality of
deficiency indices does not determine the closed first-order operator.
Agreement of a restricted scattering experiment does not imply operator
equivalence, as the rooted-tree examples show. Finally, equality of the
conventional first- and second-order reductions does not determine higher
smooth transported jets.

\section{Conclusion}
\label{sec:conclusion}

The line with two origins exhibits the central analytic distinction between local smooth geometry and the data retained by quantum completion. Scalars factor through the ordinary line, whereas relative-sign transport produces a genuine component-extension defect: finite regularity gives a jet kernel and smooth regularity the flat ideal. For general finite branch systems, transported jets give the formal equaliser and Whitney realisability the residual extension condition; exponentially tangent branches show that these are distinct. Conical tameness removes the residual obstruction, while Sobolev completion retains only traces visible in the range \(j<m-c/p\).

For finite smooth graph resolutions, the scalar theory becomes an endpoint calculus. The descended smooth core realises the full formal endpoint equaliser, integer Sobolev closures are its finite projections, and first-order scalar closure is continuity on the Hausdorff graph. Higher orders can distinguish smooth resolutions with the same first-order reduction; cyclic transport adds signed return holonomy and finite prolongation profiles. With a marked presentation measure, the flat positive derivative form gives the weighted Kirchhoff Hamiltonian, while the oriented derivative has its own deficiency theory. Multiple origins, branching trees, and recombination thereby separate exceptional multiplicity, propagation valence, cycle structure, first-order indices, and spectral or scattering complexity.

Finite-rank unitary gluing supplies an independent internal reduction. On \(\mathbb L_k\), the common fixed space \(V_{\mathbf W}=\bigcap_{a=2}^k\ker(W_a-I)\) is exactly the transmitting \(H^1\)-sector of the flat Friedrichs operator. For compact-group gluing this becomes the invariant space of the generated subgroup. The invariant-completion number is bounded both dimensionally and by topological generation; connected compact semisimple groups require at most two relative gluing matrices, and the low-degree \(SU(3)\) examples give explicit optimal ranks and a quartic rank-two projector.

The resulting hierarchy of equivalences depends on the chosen probe. Scalar reduction retains the weighted Hausdorff graph and closed scalar trace domain; bundle transport can reintroduce value-level constraints and cycle holonomy; higher Sobolev and formal probes can retain smooth transported jets invisible to ordinary \(H^1\)-based dynamics. Every resolution-induced identification is relative to the marked analytic resolution, presentation measure, bundle transport, and chosen flat form. Other measures, connections, potentials, or self-adjoint extensions need not preserve these equivalences.
\appendix
\section{Conical cutoff estimates}
\label{app:conical-cutoffs}

Let \(S\subseteq M_2\) be a clean multiple-incidence stratum. Fix a Riemannian metric and a
tubular map
\(\tau: \mathcal N_\varepsilon \subseteq NS \longrightarrow \mathcal U\).
In a local trivialisation of the normal bundle, write points as
\((p,v), \qquad p\in S, \quad v\in N_pS\).
For \(v\neq0\), set
\(r=|v|, \qquad \omega=\frac{v}{|v|}\).

All estimates are proved in adapted coordinate charts and local bundle
frames. On relatively compact coordinate neighbourhoods, changes of
coordinates and frames have bounded derivatives through every fixed
finite order, so the conclusions are independent of these choices.

\subsection{Angular cutoffs}
\label{appsubsec:angular-cutoffs}

Let
\(\Omega_a^-\Subset_S\Omega_a^+ \subseteq \mathbb S(NS), \qquad a\in A_S\),
be the angular regions from
Definition~\ref{def:conical-tameness}. \(\overline{\Omega_a^-} \subseteq \Omega_a^+\),
and
\[
  \overline{\Omega_a^+}
  \cap
  \overline{\Omega_b^+}
  =
  \varnothing
  \qquad
  (a\neq b).
\]

\begin{lemma}[Angular cutoff construction]
\label{lem:angular-cutoff-construction}
For every
\[
  a\in A_S,
\]
there exists
\[
  \vartheta_a
  \in
  C^\infty
  \left(
    \mathbb S(NS),
    [0,1]
  \right)
\]
such that
\[
  \vartheta_a=1
  \quad
  \text{on }\Omega_a^-,
\]
and
\[
  \operatorname{supp}\vartheta_a
  \subseteq
  \Omega_a^+.
\]
The functions may be chosen with pairwise disjoint supports.
\end{lemma}

\begin{proof}
The smooth Urysohn lemma on the unit normal sphere bundle gives a
function equal to one on the closed set
\(\overline{\Omega_a^-}\) and supported in \(\Omega_a^+\). Since the
sets \(\Omega_a^+\) have pairwise disjoint closures, the resulting
supports are pairwise disjoint.
\end{proof}

After decreasing the positive radius function \(\eta\), assume that \(\tau(p,v)\) is defined whenever
\(|v|<2\eta(p)\).
This shrinking preserves the conical-tameness inclusions.

Choose \(\beta\in C^\infty([0,\infty),[0,1])\) such that
\(\beta(t)=1 \quad (0\leq t\leq1)\),
and
\(\beta(t)=0 \quad (t\geq2)\).
For \(v\neq0\), define
\[
  \chi_a(p,v)
  :=
  \beta
  \left(
    \frac{|v|}{\eta(p)}
  \right)
  \vartheta_a
  \left(
    p,\frac{v}{|v|}
  \right).
\]
The function \(\chi_a\) need not extend continuously to the zero
section. It will only be multiplied by sections whose prescribed jets
vanish there.

Near \(S\), conical tameness gives
\(\chi_a=1 \quad \text{on }V_a\),
and
\(\chi_a=0 \quad \text{on }V_b, \qquad b\neq a\).

\begin{lemma}[Derivative bounds for angular cutoffs]
\label{lem:angular-cutoff-derivative-bounds}
Let
\[
  K\Subset S.
\]
For every pair of multi-indices \(\alpha,\beta\), there exists
\[
  C_{K,\alpha,\beta}>0
\]
such that
\[
  \boxed{
  \left|
    \partial_p^\alpha
    \partial_v^\beta
    \chi_a(p,v)
  \right|
  \leq
  C_{K,\alpha,\beta}
  |v|^{-|\beta|}
  }
\]
whenever
\[
  p\in K,
  \qquad
  0<|v|<2\eta(p).
\]
\end{lemma}

\begin{proof}
The map \(v\longmapsto\frac{v}{|v|}\) is homogeneous of degree zero. Its derivative of order \(\ell\) is
homogeneous of degree \(-\ell\), and hence is bounded by \(C_\ell|v|^{-\ell}\) on the punctured normal fibre. Repeated differentiation therefore
gives
\[
  \left|
    \partial_p^{\alpha_1}
    \partial_v^{\beta_1}
    \vartheta_a
    \left(
      p,\frac{v}{|v|}
    \right)
  \right|
  \leq
  C
  |v|^{-|\beta_1|}.
\]

Over the compact set \(K\), the function \(\eta\), its reciprocal, and
all derivatives needed at the fixed order are bounded. Derivatives of \(\beta \left( \frac{|v|}{\eta(p)} \right)\) are therefore bounded by
\(C|v|^{-|\beta_2|}\).
For terms involving derivatives of \(\beta\), one has \(|v|\asymp\eta(p)\) on their support, so powers of \(\eta(p)^{-1}\) may be replaced by
corresponding powers of \(|v|^{-1}\).

Leibniz' rule gives the stated estimate.
\end{proof}

\subsection{Taylor decay along the zero section}
\label{appsubsec:taylor-decay}

Let \(F\longrightarrow\mathcal U\) be a smooth finite-rank vector bundle.

\begin{lemma}[Normal Taylor decay]
\label{lem:normal-taylor-decay}
Let
\[
  r\in\Gamma^k(\mathcal U,F)
\]
and suppose that its full ambient \(k\)-jet vanishes along \(S\):
\[
  j_S^k r=0.
\]
Then, in every adapted coordinate chart and local bundle frame,
\[
  \boxed{
  \partial_p^\alpha
  \partial_v^\beta r(p,v)
  =
  o
  \left(
    |v|^{k-|\alpha|-|\beta|}
  \right)
  }
\]
locally uniformly in \(p\), whenever
\[
  |\alpha|+|\beta|\leq k.
\]
\end{lemma}

\begin{proof}
Fix \(\alpha\) and \(\beta\), and set
\(d:=|\alpha|+|\beta|\).
For fixed \(p\), apply Taylor's theorem in the normal variable to \(v\longmapsto \partial_p^\alpha \partial_v^\beta r(p,v)\) through order \(k-d\). Every Taylor coefficient is an ambient
derivative of \(r\) along the zero section of total order at most
\(k\), and hence vanishes. The Peano remainder gives
\(\partial_p^\alpha \partial_v^\beta r(p,v) = o \left( |v|^{k-d} \right)\).
Uniformity for \(p\) in compact subsets follows from uniform
continuity of the derivatives of order \(k\).
\end{proof}

If \(r\) is smooth and flat along \(S\), then for every \(N\),
\[
  \left|
    \partial_p^\alpha
    \partial_v^\beta r(p,v)
  \right|
  \leq
  C_{K,N,\alpha,\beta}
  |v|^N
\]
for \(p\) in a fixed compact subset of \(S\) and \(v\) sufficiently
small.

\subsection{Quantitative cutoff-flatness}
\label{appsubsec:detailed-cutoff-flatness}

Let
\(r\in\Gamma^k(\mathcal U,F), \qquad j_S^k r=0\).
On \(\mathcal U\setminus S\), set
\(u_a:=\chi_ar\).

\begin{proposition}[Quantitative cutoff-flatness]
\label{prop:quantitative-cutoff-flatness}
For every pair of multi-indices \(\alpha,\beta\) satisfying
\[
  |\alpha|+|\beta|\leq k,
\]
one has
\[
  \boxed{
  \partial_p^\alpha
  \partial_v^\beta u_a(p,v)
  =
  o
  \left(
    |v|^{k-|\alpha|-|\beta|}
  \right)
  }
\]
locally uniformly in \(p\).

The zero extension
\[
  [\chi_ar]_0(x)
  :=
  \begin{cases}
    \chi_a(x)r(x),&x\notin S,\\
    0,&x\in S
  \end{cases}
\]
belongs to
\[
  \Gamma^k(\mathcal U,F)
\]
and satisfies
\[
  j_S^k[\chi_ar]_0=0.
\]
\end{proposition}

\begin{proof}
Leibniz' rule gives
\[
  \partial_p^\alpha
  \partial_v^\beta
  (\chi_ar)
  =
  \sum_{
    \substack{
      \alpha_1+\alpha_2=\alpha\\
      \beta_1+\beta_2=\beta
    }
  }
  C_{\alpha_1,\alpha_2,\beta_1,\beta_2}
  \left(
    \partial_p^{\alpha_1}
    \partial_v^{\beta_1}\chi_a
  \right)
  \left(
    \partial_p^{\alpha_2}
    \partial_v^{\beta_2}r
  \right).
\]
By
Lemma~\ref{lem:angular-cutoff-derivative-bounds},
the first factor is
\(O \left( |v|^{-|\beta_1|} \right)\).
By
Lemma~\ref{lem:normal-taylor-decay},
the second factor is
\(o \left( |v|^{ k-|\alpha_2|-|\beta_2| } \right)\).
Their product is therefore
\[
  o
  \left(
    |v|^{
      k-|\alpha_2|-|\beta_1|-|\beta_2|
    }
  \right)
  =
  o
  \left(
    |v|^{
      k-|\alpha_2|-|\beta|
    }
  \right).
\]
Since
\(|\alpha_2|\leq|\alpha|\),
this implies
\(o \left( |v|^{k-|\alpha|-|\beta|} \right)\).
Summing the finitely many Leibniz terms proves the estimate.

In particular, every derivative through order \(k\) tends
continuously to zero along the zero section. Defining all those
derivatives to be zero on \(S\), and inducting on the derivative order
in adapted coordinates, shows that the zero extension is \(C^k\).
Every derivative through order \(k\) vanishes on \(S\), so its
\(k\)-jet there is zero.
\end{proof}

If \(r\) is smooth and flat along \(S\), the same argument at every
finite order gives \([\chi_ar]_0 \in \Gamma^\infty(\mathcal U,F)\) and
\(j_S^\infty[\chi_ar]_0=0\).

\section{Sobolev perturbation estimates}
\label{app:sobolev-perturbations}

\subsection{Insertion of a prescribed supercritical jet}
\label{appsubsec:supercritical-insertion}

Work first in a product neighbourhood
\(\mathbb R^{n-c}_y \times \mathbb R^c_z, \qquad Z=\{z=0\}\).
Let \(F\) be a finite-dimensional coefficient space, and let
\[
  a
  \in
  C_c^\infty
  \left(
    \mathbb R^{n-c};
    F\otimes
    \operatorname{Sym}^j(\mathbb R^c)^*
  \right).
\]
Define the corresponding homogeneous polynomial by
\(P_a(y,z) := \frac{1}{j!} a(y) \left[ z^{\otimes j} \right]\).
With the standard symmetric-derivative convention,
\(\partial_z^{(j)}P_a(y,0)=a(y)\).

Choose \(\psi\in C_c^\infty(\mathbb R^c)\) with \(\psi=1\) near the origin, and set
\(v_\varepsilon(y,z) := P_a(y,z) \psi \left( \frac{z}{\varepsilon} \right)\).
By homogeneity,
\[
  v_\varepsilon(y,z)
  =
  \varepsilon^j
  P_a
  \left(
    y,\frac{z}{\varepsilon}
  \right)
  \psi
  \left(
    \frac{z}{\varepsilon}
  \right).
\]

\begin{lemma}[Trace of the inserted jet]
\label{lem:trace-inserted-jet}
For every \(\varepsilon>0\),
\[
  \partial_z^{(\ell)}
  v_\varepsilon(y,0)
  =
  0
  \qquad
  (0\leq\ell<j),
\]
and
\[
  \partial_z^{(j)}
  v_\varepsilon(y,0)
  =
  a(y).
\]
\end{lemma}

\begin{proof}
Near \(z=0\), the cutoff \(\psi(z/\varepsilon)\) is identically one,
so
\(v_\varepsilon=P_a\).
The assertion follows from the homogeneity and normalization of
\(P_a\).
\end{proof}

\begin{lemma}[Jet-insertion scaling]
\label{lem:general-jet-insertion-scaling}
For
\[
  |\alpha|+|\beta|\leq m,
\]
one has
\[
  \left\|
    \partial_y^\alpha
    \partial_z^\beta
    v_\varepsilon
  \right\|_{L^p}
  \leq
  C_{\alpha,\beta}
  \varepsilon^{j-|\beta|+c/p}.
\]
For
\[
  0<\varepsilon\leq1,
\]
\[
  \boxed{
  \|v_\varepsilon\|_{W^{m,p}}
  \leq
  C
  \varepsilon^{j-m+c/p}.
  }
\]
\end{lemma}

\begin{proof}
Differentiation in \(y\) replaces \(a(y)\) by one of finitely many
compactly supported derivatives. After \(|\beta|\) differentiations in
\(z\), every resulting term has the form
\[
  \varepsilon^{j-|\beta|}
  b_{\alpha,\beta}
  \left(
    y,\frac{z}{\varepsilon}
  \right),
\]
where \(b_{\alpha,\beta}\) is smooth, compactly supported in \(y\), and
supported in a fixed bounded subset of the rescaled normal variable.

Changing variables \(z=\varepsilon w\) gives
\[
  \left\|
    \partial_y^\alpha
    \partial_z^\beta
    v_\varepsilon
  \right\|_{L^p}^p
  \leq
  C
  \varepsilon^{p(j-|\beta|)+c}.
\]
Taking \(p\)-th roots proves the first estimate. Among all derivatives
of total order at most \(m\), the smallest exponent is
\(j-m+\frac{c}{p}\),
which gives the second estimate.
\end{proof}

\begin{corollary}[Supercritical erasure]
\label{cor:supercritical-erasure}
If
\[
  j>m-\frac{c}{p},
\]
then
\[
  v_\varepsilon
  \longrightarrow0
  \qquad
  \text{in }W^{m,p},
\]
while its order-\(j\) normal trace remains equal to \(a\).

An independently variable pure order-\(j\) trace condition is not
closed in \(W^{m,p}\) in the supercritical range.
\end{corollary}

Adapted tubular charts, local frames, and a finite partition of unity give the same scaling exponent for compactly supported bundle data.

\subsection{One-sided flat cutoff convergence}
\label{appsubsec:flat-cutoff-convergence}

Let \(\varrho\in C^\infty(\mathbb R)\) be flat at \(0\) and satisfy
\(\varrho(t)=0 \qquad (t\leq0)\).
Choose \(\chi\in C^\infty(\mathbb R,[0,1])\) such that
\(\chi(t)=0 \quad (t\leq1)\),
and
\(\chi(t)=1 \quad (t\geq2)\).
Set
\(\chi_\varepsilon(t) := \chi(t/\varepsilon)\).

Let \(\eta\in C_c^\infty(\mathbb R^d)\) and define
\(v(x) := \eta(x)\varrho(x_1)\),
\(v_\varepsilon(x) := \eta(x) \chi_\varepsilon(x_1) \varrho(x_1)\).

\begin{lemma}[One-sided flat cutoff estimate]
\label{lem:one-sided-flat-cutoff-estimate}
For every
\[
  m\in\mathbb N_0,
  \qquad
  1\leq p<\infty,
\]
and every \(N\in\mathbb N\), there exists
\[
  C_{m,N}>0
\]
such that
\[
  \boxed{
  \|v-v_\varepsilon\|_{W^{m,p}}
  \leq
  C_{m,N}
  \varepsilon^{N-m+1/p}
  }
\]
for all sufficiently small \(\varepsilon>0\). In particular,
\[
  v_\varepsilon
  \longrightarrow
  v
  \qquad
  \text{in }W^{m,p}(\mathbb R^d).
\]
\end{lemma}

\begin{proof}
Set
\[
  r_\varepsilon
  :=
  v-v_\varepsilon
  =
  \eta(x)
  \left(
    1-\chi_\varepsilon(x_1)
  \right)
  \varrho(x_1).
\]
Because \(\varrho\) vanishes on the negative half-line,
\[
  \operatorname{supp}r_\varepsilon
  \subseteq
  \left\{
    0\leq x_1\leq2\varepsilon
  \right\}
  \cap
  \operatorname{supp}\eta.
\]

For
\(|\alpha|\leq m\),
Leibniz' rule expresses \(D^\alpha r_\varepsilon\) as a finite sum of
terms bounded by
\[
  C
  \varepsilon^{-\ell}
  \left|
    \varrho^{(q)}(x_1)
  \right|,
  \qquad
  0\leq\ell,q\leq m.
\]
Flatness gives, for every \(N\), \(\left| \varrho^{(q)}(t) \right| \leq C_{N,q}|t|^N\) near the origin. On the support of \(r_\varepsilon\),
\(|x_1|\leq2\varepsilon\),
and hence
\(|D^\alpha r_\varepsilon| \leq C_{\alpha,N} \varepsilon^{N-m}\).

The support has measure \(O(\varepsilon)\) in the shrinking coordinate
and remains in a fixed compact set in the remaining variables.
Therefore
\(\|D^\alpha r_\varepsilon\|_{L^p} \leq C_{\alpha,N} \varepsilon^{N-m+1/p}\).
Summing over \(|\alpha|\leq m\) proves the estimate. Choosing \(N>m-\frac1p\) gives convergence.
\end{proof}

A finite partition of unity gives the same estimate for compactly supported bundle sections.

\subsection{Application to the exponentially tangent system}
\label{appsubsec:exponential-application}

Retain the notation of
Section~\ref{subsec:exponentially-tangent-branches}.
Let \(s_1\) be the compactly supported localisation used in
Section~\ref{subsec:exponential-sobolev-erasure}, and define
\(s_{1,\varepsilon}\) by replacing \(\rho(x)\) on the second source
component with
\(\chi_\varepsilon(x)\rho(x)\).

\begin{proposition}[Extendability of the approximants]
\label{prop:extendability-exponential-approximants}
For every
\[
  \varepsilon>0,
\]
the section \(s_{1,\varepsilon}\) is smooth, compactly supported, and
belongs to
\[
  \mathfrak R_A^\infty.
\]
\end{proposition}

\begin{proof}
Let \(q_\varepsilon := \overline{\sigma}_1[s_{1,\varepsilon}]\) be the globally defined prolonged representative of the second branch
on \(M_2=\mathbb R^2\). It is smooth and compactly supported. Since
\(\chi_\varepsilon(x)=0 \qquad (x\leq\varepsilon)\),
the support of \(q_\varepsilon\) is bounded away from the
multiple-incidence point.

Set
\(K_\varepsilon := \operatorname{supp}q_\varepsilon \cap \overline{V_1}^{\,M_2}\).
This is compact. On the compact range of \(x\)-coordinates meeting
\(K_\varepsilon\), the two ribbon closures \(\overline{V_0}^{\,M_2} \quad\text{and}\quad \overline{V_1}^{\,M_2}\) are separated by a positive distance. Hence
\(K_\varepsilon \cap \overline{V_0}^{\,M_2} = \varnothing\).

Choose \(\theta_\varepsilon \in C^\infty(\mathbb R^2,[0,1])\) such that \(\theta_\varepsilon=1\) on a neighbourhood of \(K_\varepsilon\), and \(\theta_\varepsilon=0\) on a neighbourhood of
\(\overline{V_0}^{\,M_2}\). Define
\(s_{2,\varepsilon} := \theta_\varepsilon q_\varepsilon\).

On \(V_0\),
\(s_{2,\varepsilon}=0\),
which is the forced value from the first source branch. On \(V_1\), if
\(q_\varepsilon\neq0\), then the point belongs to
\(K_\varepsilon\) and \(\theta_\varepsilon=1\); if
\(q_\varepsilon=0\), the equality is automatic. Therefore
\(s_{2,\varepsilon}|_{U_2} = \sigma[s_{1,\varepsilon}]\).
The pair \((s_{1,\varepsilon},s_{2,\varepsilon})\) is a compatible smooth pair.
\end{proof}

\begin{proposition}[Convergence of the approximants]
\label{prop:convergence-exponential-approximants}
For every
\[
  m\in\mathbb N_0,
  \qquad
  1\leq p<\infty,
\]
\[
  \boxed{
  s_{1,\varepsilon}
  \longrightarrow
  s_1
  \qquad
  \text{in }W^{m,p}(M_1,E_1).
  }
\]
\end{proposition}

\begin{proof}
The difference vanishes on the first source component. In the fixed
coordinates on the second source component, it has the form
\(\eta(x,y) \left( 1-\chi_\varepsilon(x) \right) \rho(x)\),
where \(\eta\) is smooth and compactly supported. Since \(\rho\) is
flat at the origin,
Lemma~\ref{lem:one-sided-flat-cutoff-estimate}
applies.
\end{proof}

Together, the preceding propositions give \(s_1 \in \overline{ \mathfrak R_A^{\infty;m,p} }^{\,W^{m,p}}\) for every finite Sobolev order, even though
\(s_1\notin\mathfrak R_A^\infty\).

\subsection{Scope of the perturbation arguments}
\label{appsubsec:perturbation-scope}

\begin{enumerate}[label=\textup{(\roman*)}]
  \item The supercritical insertion changes a prescribed
        order-\(j\) trace while converging to zero in \(W^{m,p}\). It
        applies when \(j>m-\frac{c}{p}\)         and when that coefficient can be varied independently of all
        lower visible trace data.

  \item The flat-cutoff approximation removes a smooth cross-branch
        obstruction supported arbitrarily close to the
        multiple-incidence locus. Its convergence depends on
        infinite-order flatness, not on the supercritical trace
        inequality.
\end{enumerate}

Neither argument determines the critical case
\(j=m-\frac{c}{p}\),
and neither implies that every smooth Whitney obstruction disappears
under finite Sobolev completion.

\bibliographystyle{plainnat}
\bibliography{references}

\end{document}